\documentclass[a4paper, amsfonts, amssymb, amsmath, prx, showkeys, nofootinbib, longbibliography]{revtex4-2}
\usepackage[english]{babel}
\usepackage[utf8]{inputenc}

\usepackage{amsthm}
\usepackage{amsmath} 
\usepackage{mathtools}
\mathtoolsset{showonlyrefs}
\usepackage{physics}
\usepackage{xcolor}
\usepackage{graphicx}
\usepackage[left=18mm,right=18mm,top=35mm,columnsep=15pt]{geometry} 
\usepackage{adjustbox}
\usepackage{placeins}
\usepackage[T1]{fontenc}
\usepackage{lipsum}
\usepackage{csquotes}
\usepackage{tikz}
\usepackage{float}
\usepackage{mathdots}
\usepackage{youngtab}
\usepackage{ytableau} 
\usepackage{pgfplots} 
\usepackage{amsthm}
\usepackage{stmaryrd}
\usepackage{makecell}

\usepackage{tikz}
\usetikzlibrary{quantikz2}

\newcommand{\nw}{\wireoverride{n}}
\newcommand{\mctrl}[1]{
  \gate[style={inner sep=-1, minimum height=1.pt, minimum width=1pt}]{}\vqw{#1}
}

\newcommand{\select}{\mbox{S\scriptsize ELECT}}
\newcommand{\sel}{\mbox{S\scriptsize EL}}
\newcommand{\sw}{\mbox{S\scriptsize{WAP}}}
\newcommand{\up}{\mbox{U\scriptsize{P}}}
\newcommand{\swup}{\mbox{\sw \up}}
\newcommand{\selswap}{\mbox{\select \sw}}
\newcommand{\adder}{\mbox{A\scriptsize{DDER}}}

\newcommand{\data}{\mbox{D\scriptsize{ATA}}}

\definecolor{quantinuumteal}{cmyk}{0.44, 0.0000, 0.07, 0.37}
\usepackage[unicode=true,colorlinks=true, allcolors=quantinuumteal]{hyperref} 

\theoremstyle{definition}
\newtheorem{definition}{Definition}

\newtheorem{theorem}{Theorem}

\newtheorem{proposition}{Proposition}

\newtheorem{corollary}{Corollary}

\newtheorem{lemma}[theorem]{Lemma}

\newtheorem*{example}{Example}

\theoremstyle{remark}
\newtheorem*{remark}{Remark}

\begin{document}
\title{The Quantum Paldus Transform: Efficient Circuits with Applications}

\author{Jędrzej Burkat}
\email{jbb55@cam.ac.uk}
\affiliation{Quantinuum, Terrington House, 13--15 Hills Road, Cambridge CB2 1NL, UK \\ Cavendish Laboratory, Department of Physics, University of Cambridge, Cambridge CB3 0US, UK}

\author{Nathan Fitzpatrick}
\email{nathan.fitzpatrick@quantinuum.com}
\affiliation{Quantinuum, Terrington House, 13--15 Hills Road, Cambridge CB2 1NL, UK}


\begin{abstract}

We present the Quantum Paldus Transform: an efficient quantum algorithm for block-diagonalising fermionic, spin-free Hamiltonians in the second quantisation. Our algorithm implements an isometry between the occupation number basis of a fermionic Fock space of $2d$ modes, and the Gelfand-Tsetlin (GT) states spanning irreducible representations of the group $U(d) \times SU(2)$. The latter forms a basis indexed by well-defined values of total particle number $N$, global spin $S$, spin projection $M$, and $U(d)$ GT patterns. This realises the antisymmetric unitary-unitary duality discovered by Howe and developed into the Unitary Group Approach (UGA) for computational chemistry by Paldus and Shavitt in the 1970s. The Paldus transform lends tools from the UGA readily applicable to quantum computational chemistry, leading to maximally sparse representations of spin-free Hamiltonians, efficient preparation of Configuration State Functions, and a direct interpretation of quantum chemistry reduced density matrix elements in terms of $SU(2)$ angular momentum coupling. The transform also enables the encoding of quantum information into novel Decoherence-Free Subsystems for use in communication and error mitigation. Our work can be seen as a generalisation of the quantum Schur transform for the second quantisation, made tractable by the Pauli exclusion principle. Alongside self-contained derivations of the underlying dualities we provide fault-tolerant circuit compilation methods with full gate counts for the Paldus transform, resulting in $\mathcal{O}(d^3)$ Toffoli complexity, where a transform on $50$ spatial orbitals would require a modest $5500$ Toffoli gates. This paves the way for significant advancements in quantum simulation on quantum computers enabled by the UGA paradigm.

\end{abstract}

\maketitle
\onecolumngrid

\tableofcontents

\section{Introduction}

Group and representation theory is an indispensable tool in quantum mechanics. In fundamental physics, it is widely known for its success in the formation of the standard model, and has since become a framework for systematising the search for grand unified theories. In quantum chemistry, point groups are a standard tool for classifying molecules and predicting their properties. In the theory of quantum information, groups and representations have played an equally important role. Many problems where quantum computers are set to achieve the most significant speed-ups (such as database search and factoring) reduce to hidden subgroup problems (HSPs) \cite{mosca1999, Jozsa_2001, lomont2004}. The group Fourier transform forms a key subroutine in quantum algorithms for solving such problems. However, it is not limited to HSPs in its utility, also appearing in another protocol known as generalised phase estimation \cite{harrowphd, BCH_2005, childs2007}, with applications in estimating representation-theoretic quantities \cite{bravyi2024kronecker, larocca2025} or quantum property testing and tomography \cite{hu2024, montanaro2018}.

Another related class of group-theoretic quantum algorithms are symmetry-adapted basis transformations, which map the computational basis into disjoint subspaces of the Hilbert space. Such subspaces are closed under the action of a group, meaning that group action (represented through a unitary matrix) takes on a block-diagonal form in the new basis. In the language of representation theory, this constitutes a transform into invariant subspaces of the irreducible representations of the group: the new, block-diagonal form of the group action is exactly its decomposition into irreps. The most well-known of these is the \textit{quantum Schur transform} \cite{BCH_2005, BCH_2006, harrowphd, Krovi_2019, wills2024}, followed by the recently introduced mixed Schur transform \cite{nguyen2023, grinko2023, Grinko_2025}. The two protocols form a subroutine in a variety of quantum information tasks, from tomography and learning \cite{Keyl_2001, christandl2006spectra} to quantum communication \cite{Bartlett_2003}, error mitigation \cite{Kempe_2001, lidar_2014}, as well as Hamiltonian simulation \cite{Gu_2021} and equivariant quantum machine learning \cite{zheng2022, Zheng_2023, Nguyen_2024}. At the heart of these applications lies the ability to decompose the representations of $U(d) \times S_n$ on the computational basis of $n$ qudits, given by qudit permutations (representing the symmetric group $S_n$) and tensor product of single-qudit gates $U^{\otimes n}$ (representing the unitary group $U(d)$). This is achieved by transforming the computational basis into the \textit{Schur basis}, under which the group action takes on a block-diagonal form and realises the isomorphism:

\begin{equation}
(\mathbb{C}^d)^{\otimes n} \overset{S_n \times U(n)}{\cong} \ \bigoplus_{\lambda} W^{U(d)}_\lambda \otimes W^{S_n}_\lambda \label{eq:swduality}\\
\end{equation}

\noindent By Schur-Weyl duality, the multiplicities of $U(d)$ irreps are the dimension of $S_n$ irreps, and vice versa. Typical implementations of the Schur transform output the Schur basis in the tensor product form $\ket{\lambda} \otimes \ket{q} \otimes \ket{p}$, where $\lambda$ labels a joint irrep of the symmetric and unitary groups. Under $U(d)$ ($S_n$) group action only the $\ket{q}$ ($\ket{p}$) register is affected, with the other $\ket{p}$ ($\ket{q}$) register left untouched as an irrep multiplicity label. By virtue of the tensor product form, in the Schur basis information about the irrep index and multiplicity for either group can be directly manipulated, or retrieved by a measurement. 

Over the last half century, similar techniques have also been used to great effect in classical computational quantum chemistry, through a formalism known as the \textit{Unitary Group Approach} (UGA) \cite{Paldus2020, Paldus2020a, Paldus2020b}. The key observation of UGA is that the chemical spin-free Hamiltonians can effectively be treated as a representation of the $\mathfrak{u}(d) \times \mathfrak{u}(2)$ Lie algebra, and thus their action will block-diagonalise in an appropriate basis, known as the \textit{Gelfand-Tsetlin (GT) basis}. Here $\mathfrak{u}(d)$ acts on the $d$-dimensional spatial orbital Hilbert space, and $\mathfrak{u}(2)$ acts on the $2$-dimensional spin Hilbert space. By working in the GT basis, the UGA has enabled more efficient calculations of spin-free Hamiltonian elements owing to their increased sparsity, as well as simple rules for their evaluation~\cite{Robb1984}. Up until now, it was not known how to implement this formalism in a quantum computational setting. In this work, we bridge the gap between the classical unitary group approach and quantum computational chemistry. For a fermionic system of $d$ orbitals (i.e. $2d$ fermionic modes) we derive rules for decomposing the Fock space into symmetry sectors with common overall spin ($S$) and particle number ($N$), spanned by GT states indexed by their spin projection ($M$) and $U(d)$ Gelfand-Tsetlin patterns. The latter of these indices can be easily compressed into a $2d$-bit string $\mathbf{d}$, referred to as a \textit{step vector}. Through associating evolution by spin-free Hamiltonians on the fermionic Fock space with the representation of the group $U(d) \times SU(2)$ on the antisymmetric \textit{exterior product space}, the fermionic Fock space can be shown to decompose into irrep spaces as:

\begin{equation}
    \bigwedge( \mathbb{C}^d \otimes \mathbb{C}^2 ) \overset{U(d) \times SU(2)}{\cong} \ \bigoplus_{N=0}^{2d} \bigoplus_{S=0}^{N/2} W^{U(d)}_{N, S} \otimes W^{SU(2)}_{S} \label{eq:intro_paldus_duality}
\end{equation}

In the antisymmetric representation, the groups $U(d)$ and $SU(2)$ constitute a \textit{reductive dual pair}, leading to a specific instance of $GL(n) \times GL(m)$ \textit{Howe duality} \cite{Howe1989, Howe_1995}, which we refer to as \textit{Paldus duality}. Remarkably, Paldus duality and its corollaries are a physically motivated and direct consequence of the Pauli exclusion principle, which forbids any two electrons inhabiting one spatial orbital from carrying an equal spin projection, and mathematically indentifies the multi-electron Fock space with an exterior product of the one-electron Hilbert space. To the best of our knowledge, this work presents the first explicit realisation of this duality within a quantum computational framework. We introduce a quantum algorithm that implements the isomorphism in Equation \eqref{eq:intro_paldus_duality}, termed the \textit{Quantum Paldus Transform}. This transformation efficiently block-diagonalises spin-free quantum chemistry Hamiltonians, as well as other Hamiltonians of analogous structure such as the Fermi-Hubbard model. Formally, the quantum Paldus transform is an isometry between Gelfand-Tsetlin states expressed in the computational basis (as superpositions over bitstrings) and the \textit{UGA basis}, wherein each index of the GT state is encoded in a distinct quantum register:

\begin{equation}
     \ket{N, S, M; \mathbf{d}} \overset{U_{\text{Paldus}}}{\longrightarrow} \ \ket{N} \otimes \ket{S} \otimes \ket{M} \otimes \ket{\mathbf{d}} \label{eq:intro_uga_basis}
\end{equation}

Our algorithm for the quantum Paldus transform is efficient, and we provide an explicit circuit construction requiring $\mathcal{O}(d^3 \text{poly}(\log d))$ two-qubit gates and $\mathcal{O}(\log(d))$ ancillary qubits, together with a detailed $T$ gate and Toffoli resource estimate for a fault-tolerant implementation. However, it can also be performed in linear time complexity: we show that with the $SU(2)$ Clebsch-Gordan transform construction of \cite{BCH_2005} as a subroutine, the Paldus transform admits an implementation with $\mathcal{O}(d \log d)$ two-qubit gates. Following the algorithm, we discuss applications of the Paldus transform in, and beyond the realm of quantum computational chemistry: configuration state function (CSF) preparation, maximally-sparse Hamiltonian simulation, as well as encoding quantum information into novel decoherence-free subsystems (DFSs). 

\subsection{Summary of Results and Outline}

Our main contributions are as follows:
\begin{itemize}
    \item We derive an algorithm for the \textit{Quantum Paldus Transform}: an isometry between the Fock space basis and the Gelfand-Tsetlin basis which generalises the qubit Schur transform to the second quantisation formalism. We provide \textit{two} efficient implementations: an explicit algorithm readily implementable on near-term quantum hardware, as well as a theoretical algorithm with linear complexity making use of the Clebsch-Gordan transforms of \cite{BCH_2005, BCH_2006}. In addition, we give detailed resource estimates for the algorithm using modern circuit primitives in Appendix 
    \ref{sec:circuit_compilation}.
    
    \item We consider the \textit{ladder operators} $E_{ij} = \sum_\mu a_{i \mu}^\dagger a_{j \mu}$ and $\mathcal{E}_{\mu \nu} = \sum_{i} a_{i \mu}^\dagger a_{i \nu}$ constructed out of fermionic creation and annihilation operators, and show that they form a representation of the Lie algebra $\mathfrak{u}(d) \oplus \mathfrak{u}(2)$. This allows us to identify the full set of unitaries block-diagonised under the GT basis (and thus the quantum Paldus transform) to be exponents of Hamiltonians in the \textit{universal enveloping algebra} of $\mathfrak{u}(d) \oplus \mathfrak{u}(2)$. This is shown to include the Fermi-Hubbard Hamiltonian, as well as any second-quantised spin-free Hamiltonian of the form: 
    \begin{equation}
        H = \sum_{ij, \mu} h_{ij} a_{i \mu}^\dagger a_{j \mu} + \frac{1}{2} \sum_{ijkl, \mu \nu} v_{ij, kl} a_{i \mu}^\dagger a_{j \nu}^\dagger a_{l \nu} a_{k \mu}
    \end{equation}
    
    \item As an example, we explicitly work out the form of $U(d)$ and $U(2)$ representations generated by Hamiltonians belonging to $\mathfrak{u}(d)$ or $\mathfrak{u}(2)$ -- single-body Hamiltonians exclusively containing terms linear in $E_{ij}, \mathcal{E}_{\mu \nu}$ -- and show that such unitaries always form instances of \textit{matchgate circuits} composed of parity-preserving two-qubit gates. As arbitrary matchgate circuits do not typically block-diagonalise beyond subspaces of constant Hamming weight, this identifies a new class of matchgate circuits with a maximally-sparse form under the GT basis.
\end{itemize}

In addition to our algorithm, we offer a selection of applications for the quantum Paldus transform:
\begin{itemize}
    \item \textbf{Hamiltonian Simulation.} Our algorithm block-diagonalises a wide range of Hamiltonians as well as their induced unitary evolution. By virtue of the irrep decomposition, this is the \textit{most sparse representation} achievable for spin-free Hamiltonians -- meaning that the symmetry sectors which it acts on have the lowest possible dimensions. We anticipate this property to become extremely useful in the context of quantum simulation, where a dimensionality reduction often leads to a speed-up. One method for achieving this speed-up is the Hamiltonian fast-forwarding algorithm of \cite{Gu_2021}, where the structure of the Schur transform is used to implement each `block' of $e^{-iHt}$ as a unitary controlled on an irrep label. This strategy is especially promising when evolving wavefunctions constrained to a given symmetry sector: then, Hamiltonian simulation only needs to be carried out in the (polynomially smaller) block that the state is known to reside in.   
    \item \textbf{Preparation of Configuration State Functions (CSFs).} Our work allows for an efficient preparation of any CSF belonging to a given symmetry sector, i.e. a wavefunction with well-defined values of $N, S, M$. This is done by preparing the desired wavefunction in the basis to the right of Equation \eqref{eq:intro_uga_basis}, and applying the inverse Paldus transform to map it into the occupation (Fock state) basis. Similarly, one can project any input state into a given symmetry sector by applying the Paldus transform and measuring the $N, S, M$ registers of the output. To give an example, we provide an $\mathcal{O}(\sqrt{d})$ time protocol for preparing uniform superpositions of all CSFs, i.e. states of the form:
    \[
        \ket{\Psi} = \frac{1}{\sqrt{|\{ \mathbf{d} \}|}}  \sum_{\mathbf{w} \in \{ \mathbf{d} \} } \ket{N}_N \otimes \ket{S}_S \otimes \ket{S}_M \otimes \ket{\mathbf{w}}_\mathbf{d}
    \]
    where the $SU(2)$ irrep labels are set to their highest values $M = S$, and $\{ \mathbf{d} \}$ is the set of all valid $U(d)$ GT patterns. 
    \item \textbf{Encoding and Decoding into Decoherence-Free Subsystems (DFSs).} Taking advantage of the underlying $U(d) \times U(2)$ duality in the Paldus transform we are able to construct new types of DFSs made up of unitary evolution by Hamiltonians in the universal enveloping algebra $\mathfrak{U}(\mathfrak{u}(2))$. Considering a simple example, we show that this allows for protecting quantum information against decoherence by unitaries $U^{\otimes d}$, where $U$ is a member of the group $U(2)$ embedded in a two-qubit gate. Notably, this is a gate with three nontrivial parameters whose action the `standard' qubit Schur transform cannot protect against. 
\end{itemize}

Finally, we offer self-contained derivations of the mathematical techniques which enable our construction:
\begin{itemize}
     \item Using the theory of symmetric polynomials, we prove the \textit{antisymmetric unitary-unitary duality} between the groups $U(m)$ and $U(n)$ in the antisymmetric representation of the group $U(mn)$. As the two groups form a reductive dual pair, this allows the decomposition of an exterior product space $\bigwedge( \mathbb{C}^m \otimes \mathbb{C}^n )$ into tensor products of irrep spaces of $U(m)$ and $U(n)$. We pay particular attention to the case of $U(d) \times SU(2)$, where we interpret the exterior product space of the representation as the Fock space of $2d$ fermionic modes.
    
    \item Focusing specifically on the $U(d) \times SU(2)$ example, we derive the branching rules of this representation under group subduction and lifting. The branching is multiplicity-free, and allows for the construction of a Gelfand-Tsetlin basis for the antisymmetric representation, which we explicitly identify. We relate our derivations to the theory of \textit{Shavitt graphs}: structures that resemble Bratelli diagrams and form a useful graphical tool in the unitary group approach. 
\end{itemize}

\

Our paper is structured as follows: in Section \ref{sec:background} we begin by deriving the Paldus duality and Gelfand-Tsetlin basis, which form the mathematical preliminaries of the Paldus transform. Using basic properties of Schur polynomials, we derive the antisymmetric unitary-unitary duality, with the $U(d) \times SU(2)$ Paldus duality as a corollary. We explicitly construct the relevant Gelfand-Tsetlin states, clarifying their relation to Shavitt graphs and $SU(2)$ spin coupling in the process. In Section \ref{sec:paldus_transform} we define the quantum Paldus transform, and present a circuit for its efficient implementation. We take a pedagogical approach, giving detailed specifics of each step in the circuit, which we hope finds use for readers unfamiliar with the related Schur transform. Section \ref{sec:applications} is dedicated to applications of the Paldus transform: we begin by introducing the unitary group approach and representation-theoretic properties of spin-free Hamiltonians, with examples of the $U(d) \times U(2)$ representations they induce (and which block-diagonalise under the algorithm). Applications to Hamiltonian evolution, state preparation, and decoherence-free subsystems are then presented along with contrasts between our algorithm and the Schur transform. 

The rest of the paper contains additional derivations and proofs: Appendix \ref{app:uga_history} gives a brief history of the UGA. Appendix \ref{app:liegroups} states key results from the theory of Lie groups and representations used throughout the main text. Appendix \ref{app:symmetricpolynomials} reviews the relevant theory of symmetric polynomials, including proofs of the Cauchy identity and Pieri's rule referred to in the main text. Appendix \ref{app:dualities} contains derivations of dualities used in the main text, including Schur-Weyl duality and the double centralizer theorem. Appendix \ref{app:unitarygroup} derives branching rules for subductions of $U(n)$ representations from the perspective of symmetric polynomials. Appendix \ref{app:dimension_formula} contains a short calculation of the dimensionality of the $U(d) \times SU(2)$ Gelfand-Tsetlin basis. Finally, Appendix \ref{sec:circuit_compilation} gives details and thorough resource estimates of how to implement the quantum Paldus transform via leading circuit compilation techniques for data lookup, multiplexing and incrementations in the algorithm. 

Given the extensive page count, we recommend a selective reading: Section \ref{sec:background} may be initially skipped by those primarily interested in the algorithm or its applications, whereas individual appendices mostly form self-contained, additional content for the mathematically minded audience and practitioners looking to implement our algorithm.

\section{Background} \label{sec:background}

In this section we present mathematical results which enable the construction of the Paldus transform, and the associated Gelfand-Tsetlin basis. By studying \textit{antisymmetric representations} of the group $U(m) \times U(n)$ through the theory of symmetric polynomials we will arrive at the Paldus duality, which allows for a partition of the fermionic Fock space into symmetry sectors indexed by the quantum numbers $( N, S, M, \mathbf{d})$. Practically, this is achieved by mapping the computational basis into the Gelfand-Tsetlin basis, whose explicit form we derive. Such states -- along with the subgroup branching used to obtain them -- admit a simple graphical representation through a tool known as the Shavitt graph, which we elucidate towards the end. Supplementary to this section are Appendices \ref{app:symmetricpolynomials}, \ref{app:dualities} and \ref{app:unitarygroup}, which provide additional derivations, as well as a self-contained introduction into the theory of symmetric polynomials. Whilst the intention of this section is to build intuition and provide the minimal mathematical background required ahead of Section \ref{sec:paldus_transform}, it may be skipped on a first reading. 

\subsection{Second Quantisation and the Exterior Product Space}\label{sec:ON}

As the antisymmetric representation of groups relates very closely to the Fock space representation of quantum systems, we begin with a short review of second quantisation. Throughout the text, we will identify the Fock space of $2d$ fermionic mode with the exterior algebra $\bigwedge \mathbb{C}^{2d}$ of a single-particle Hilbert space. The equivalence between the two is well-known, and we provide a brief overview here for completeness. We will consider a chemical system of $d$ spatial orbitals occupied by electrons, each of which (by the Pauli exclusion principle) can only contain up to two electrons with opposite spin projections -- i.e, a system with $d$ spatial degrees of freedom and $2$ spin degrees of freedom. A single electron may be then represented as a vector in the $2d$-dimensional single-particle space $\mathbb{C}^{2d} \cong \mathbb{C}^d \otimes \mathbb{C}^2$, with the basis states $\ket{e_1}, \ket{e_2}, \ldots, \ket{e_{2d}}$ representing the spin-orbital it occupies. For two electrons, we can then represent the state of the system as a vector in the tensor product space $\mathbb{C}^{2d} \otimes \mathbb{C}^{2d}$, with the basis states $\ket{e_i} \otimes \ket{e_j}$ representing the occupation of spin-orbitals $i$ and $j$. However, for such a representation to be physical, we must also impose the indistinguishability of two particles and antisymmetry of their state vector under permutation. Thus, the two-particle states can be represented as vectors of the form:

\begin{equation}
\ket{i, j} = \frac{1}{\sqrt{2}}\left(\ket{e_i} \otimes \ket{e_j} - \ket{e_j} \otimes \ket{e_i}\right) \ \ i < j \in \llbracket 2d \rrbracket^2
\end{equation}

\noindent Such a state is antisymmetric under the exchange of particles represented by the action of a permutation $\hat{\sigma}_{ij}$:  $\hat{\sigma}_{ij} \ket{i, j} = - \ket{i, j}$. In particular, this ensures that $\ket{i, i} = 0$ so the antisymmetrised basis states obey the Pauli exclusion principle.  Moving an electron into an \textit{unoccupied} orbital does not necessarily incur a negative sign:

\begin{equation}
  \hat{\sigma}_{ik} \ket{i, j} = \frac{1}{\sqrt{2}}\left(\ket{e_k} \otimes \ket{e_j} - \ket{e_j} \otimes \ket{e_k}\right) 
  = \begin{cases} \ket{k, j} & \text{ for } k < j \\ -\ket{k, j} & \text{ for } k > j \end{cases}
\end{equation}

\noindent Keeping in mind that all valid states representing two electrons must obey this antisymmetry, we can denote their space as the antisymmetric tensor product space $\bigwedge^2 \mathbb{C}^{2d} $. Extending this idea, we can write any state $\ket{i_1, \ldots, i_k}$ of $k$ electrons inhabiting the spin-orbitals $\{ i_1, \ldots, i_k \}$ as an antisymmetrised tensor product:

\begin{align}
  \ket{i_1, i_2, \ldots, i_k}  
  & = \frac{1}{\sqrt{k!}} \sum_{\sigma \in S_k} \text{sgn}(\sigma) \ket{e_{i_{\sigma(1)}}} \otimes \ket{e_{i_{\sigma(2)}}} \otimes \ldots \otimes \ket{e_{i_{\sigma(k)}}}, \ {1 \leq i_1 < \ldots < i_k \leq 2d} \label{eq:wedgeproductstate} \\
  &\in \bigwedge^k \mathbb{C}^{2d}
\end{align}

\noindent where $\dim (  \bigwedge^k \mathbb{C}^{2d} ) = {2d \choose k}$. We can then take the direct sum of $k$-particle spaces to obtain the space of $\{0, 1, \hdots, m \}$-electron states:

\begin{equation}
  \bigwedge^{\bullet m} \mathbb{C}^{2d} = \mathbb{C} \oplus \mathbb{C}^{2d} \oplus \bigwedge^2 \mathbb{C}^{2d} \oplus \cdots \oplus \bigwedge^{m} \mathbb{C}^{2d} = \bigoplus_{k=0}^{m} \bigwedge^k \mathbb{C}^{2d}
\end{equation}

\noindent where each state in $\bigwedge^k \mathbb{C}^{2d} $ is of the form in Equation \eqref{eq:wedgeproductstate} and $ \mathbb{C}$ denotes the one-dimensional vacuum space. For $m > 2d$ the antisymmetric product space $\bigwedge^m \mathbb{C}^{2d} $ is dimensionless, so we also write the full electronic Fock space on $2d$ spin-orbitals as:

\begin{equation}
  \bigwedge \mathbb{C}^{2d} = \bigwedge^{\bullet 2d} \mathbb{C}^{2d}  = \bigoplus_{k=0}^{2d} \bigwedge^k \mathbb{C}^{2d}
\end{equation}

\noindent As we implicitly assume that the indices in Equation \eqref{eq:wedgeproductstate} are ordered, we can write each state as a using a shorthand $\ket{x_1 x_2 \ldots x_{2d}}$, where each $x_i \in \{ 0, 1 \}$ denotes whether the $i^\text{th}$ spin-orbital is occupied or not. For example, $\ket{0000} \cong \ket{\emptyset}$, or $\ket{0111} \cong \ket{2, 3, 4}$. This mapping allows us to represent the antisymmetrised multi-electron wavefunction on a space of $2d$ qubits. The states in $\bigwedge \mathbb{C}^{2d} $ obey the correct exchange symmetry under permutations, whereas for states in $(\mathbb{C}^2)^{\otimes 2d}$ the exchange symmetries are encoded in the action of \textit{creation and annihilation} operators $ a_i, a_j^\dagger$, which satisfy the following anticommutation relations:

\begin{equation}
  \{a_i, a_j\} = \{a_i^\dagger, a_j^\dagger\} = 0, \quad \{a_i, a_j^\dagger\} = \delta_{ij}
\end{equation}

\noindent The creation and annihilation operators map between different $k$-particle spaces as $a_i: \bigwedge^k \mathbb{C}^{2d} \rightarrow \bigwedge^{k-1} \mathbb{C}^{2d}$ and $a^\dagger_i : \bigwedge^k \mathbb{C}^{2d} \rightarrow \bigwedge^{k+1} \mathbb{C}^{2d}$. A qubit state $\ket{\mathbf{x}}$ representing a multi-electron state from $\bigwedge\mathbb{C}^{2d}$ may then be written as:

\begin{equation}
    \ket{\mathbf{x}} =  (a_1^\dagger)^{x_1}  (a_2^\dagger)^{x_2} \hdots  (a_{2d}^\dagger)^{x_{2d}} \ket{\mathbf{0}} \cong \ket{1^{x_1}, 2^{x_2}, \hdots, (2d)^{x_{2d}}} \label{eq:fockstate}
\end{equation}

\noindent A permutation of multi-electron states in $\bigwedge \mathbb{C}^{2d} $ is then represented on bitstrings in $(\mathbb{C}^2)^{\otimes 2d}$ by a permutation of the indices in Equation \eqref{eq:fockstate}:

\begin{equation}
    \hat{\sigma} \ket{\mathbf{x}} =  (a_{\sigma(1)}^\dagger)^{x_{\sigma(1)}}  (a_{\sigma(2)}^\dagger)^{x_{\sigma(2)}} \hdots  (a_{\sigma(2d)}^\dagger)^{x_{\sigma(2d)}} \ket{\mathbf{0}}
\end{equation}

\noindent Which after re-arranging into the original ordering of Equation \eqref{eq:fockstate} (using the anticommutation relations) reproduces the correct $\pm 1$ sign of fermion exchange. 

Clearly, the antisymmetrised tensor products on $\bigwedge \mathbb{C}^{2d}$ are in one-to-one correspondence with the occupation number states on $(\mathbb{C}^2)^{\otimes 2d}$. Consequently, the two spaces are isomorphic as representations of the fully antisymmetric irrep of the symmetric group -- in the former, group action consists of permutations of the tensor products in Equation \eqref{eq:wedgeproductstate}, and in the latter it is given by permutations on the ordering of creation operators in Equation \eqref{eq:fockstate} (assuming an appropriate choice of creation and annihilation operators). We are free to take advantage of this isomorphism; we will work in $\bigwedge \mathbb{C}^{2d}$ for the derivation of Paldus duality, and apply the result to $(\mathbb{C}^2)^{\otimes 2d}$ when explicitly constructing the Gelfand-Tsetlin states and the quantum Paldus transform.

\subsection{Schur Polynomials and Semi-Standard Young Tableaux}\label{sec:sym}

Schur polynomials $s_\lambda(x_1,\ldots,x_d)$ are a family of symmetric polynomials with the defining property of remaining invariant under any permutation of the variables. There is a one-to-one correspondence between the partitions of $d$ into $\lambda = (\lambda_1,\ldots,\lambda_d)$ and the Schur polynomials. A detailed introduction, along with many proofs of the relevant statements, is presented in Appendix~\ref{app:symmetricpolynomials}. In particular, many results concerning Schur polynomials rely on the famous bialternant formula discovered by Cauchy~\cite{Cauchy_2009} as stated in Theorem~\ref{thm:bialternant}. For the purposes of this work it will be sufficient to utilise the Semi-Standard Young Tableaux (SSYT) form of Schur polynomials, which is isomorphic to the conventional Schur polynomial form (as we show in Theorem~\ref{thm:un_branching}).

\begin{definition}[SSYT Schur Polynomial]
  A Semi-Standard Young Tableau (SSYT) is a filling of a Young diagram with positive integers such that the entries are weakly increasing along rows and strictly increasing down columns. The Schur polynomial $s_\lambda(x_1,\ldots,x_d)$ is the sum of monomials over all SSYT $T$ of shape $\lambda$ with entries from $\{1,\ldots,d\}$, where the monomial $x^T$ is defined as the product of the variables $x_i$ corresponding to the entries of $T$.
  \label{def:schur_polynomial}
\end{definition}

\begin{example}
The Schur polynomial \( s_{(2,1)}(x_1, x_2, x_3) \) corresponding to the partition \( \lambda = (2,1)\) and variables \( x_1, x_2, x_3 \) is obtained from the sum over all semi-standard Young tableaux \( T \) of shape \( (2,1) \) with entries from \( \{1, 2, 3\} \). Then, \( x^{T} = x_1^{m_1} x_2^{m_2} x_3^{m_3} \):

\begin{equation}
\begin{ytableau}
1 & 1 \\
2
\end{ytableau}
\quad
\begin{ytableau}
1 & 1 \\
3
\end{ytableau}
\quad
\begin{ytableau}
1 & 2 \\
2
\end{ytableau}
\quad
\begin{ytableau}
1 & 2 \\
3
\end{ytableau}
\quad
\begin{ytableau}
1 & 3 \\
2
\end{ytableau}
\quad
\begin{ytableau}
1 & 3 \\
3
\end{ytableau}
\quad
\begin{ytableau}
2 & 2 \\
3
\end{ytableau}
\quad
\begin{ytableau}
2 & 3 \\
3
\end{ytableau}
\end{equation}

\noindent Adding all the monomials, we obtain the Schur polynomial:

\begin{equation}
s_{(2,1)}(x_1, x_2, x_3) = x_1^2 x_2 + x_1^2 x_3 + x_1 x_2^2 + 2x_1 x_2 x_3 + x_1 x_3^2 + x_2^2 x_3 + x_2 x_3^2
\end{equation}
\end{example}

\noindent Remarkably, the Schur polynomials provide the characters of the irreducible representations of the unitary group, which can be derived from the celebrated Weyl character formula~\cite{Weyl1925}, presented in Definition~\ref{def:weyl_character_formula} (a detailed exposition on the unitary group is provided in Appendix~\ref{app:unitarygroup}). The variables in the Schur polynomials then correspond to the complex eigenvalues in the standard representation of \( U(d) \).

\begin{definition}[Standard Irrep of $U(d)$]
  The standard irrep of the unitary group $U(d)$ is its $d \times d$ matrix representation. Every $U(d)$ matrix can be brought to diagonal form by a unitary transformation, which does not affect its trace. Therefore, every $U(d)$ matrix is conjugate to a diagonal matrix:
  
  \begin{equation}
      U = W \begin{pmatrix}
      z_1 & 0 & \cdots & 0 \\
      0 & z_2 & \cdots & 0 \\
      \vdots & \vdots & \ddots & \vdots \\
      0 & 0 & \cdots & z_d 
      \end{pmatrix} W^{-1}
  \end{equation}
  
  \noindent where the unique diagonal matrix elements $z_i = e^{i\theta_i}$ define a conjugacy class of $U(d)$, parameterised by the set $\{z_1,z_2,..,z_d\}$. The trace of $U$ gives the character of its conjugacy class in the standard representation of $U(d)$, represented by a single box Young diagram $\lambda  = (1)$:
  
  \begin{equation}
    \chi^{U(d)}_{(1)}(z_1,z_2, \hdots ,z_d) = \sum_{i=1}^{d} z_i
  \end{equation}

  \label{def:unitary_standard_irrep}
  
  \end{definition}
  
  \noindent There are clearly infinitely many conjugacy classes, parameterised by the set of angles $\{\theta_i\}$. Its Gelfand-Tsetlin basis is known as the standard basis $\{\psi_{1},\hdots, \psi_{d}\}$ and is the same for all conjugacy classes, satisfying $U \psi_{i} = e^{i\theta_i} \psi_i$.
  Larger, polynomial irreducible representations of $U(d)$ can be constructed by taking the $k$-th order tensor product of the standard representation with itself. By diagonalising the basis with respect to the Cartan subalgebra of the Lie algebra $\mathfrak{u}(d)$ (see Definition \ref{def:cartan} and Appendix \ref{app:unitarygroup}), one obtains a set of eigenvalues that form a partition $\lambda = (\lambda_1,\dots,\lambda_n)$. The characters of these representations are given by the Schur polynomials $s_\lambda(x_1,\dots,x_n)$, which can be derived from the semi-standard Young tableaux associated with $\lambda$:

  \begin{equation}
    \chi^{U(d)}_\lambda (z_1, \hdots ,z_d) = s_\lambda(z_1, \hdots ,z_d)
  \end{equation}

  \noindent This connection between the characters of $U(d)$ irreps and Schur polynomials allows the decomposition of $U(d)$ irreps to be expressed in terms of symmetric polynomial identities. We will use it to derive relevant branching rules, as well as the decomposition of the antisymmetric representation of $U(d) \times U(2)$ into its irreducible representations.

\subsection{Antisymmetric Unitary-Unitary Duality} \label{sec:uuduality}

In many fermionic systems, it is possible to factorise the wavefunction. For example, in electronic systems it can be factorised into dual spatial and spin components. This section presents the underlying theory of this factorisation, with more details and extended proofs presented in the Appendix~\ref{app:shavitt_branching}. In our discussion we will work with the antisymmetric representations, which can be related to the fermionic Fock space reviewed in Section \ref{sec:ON}. The $k$-fold antisymmetric tensor product spaces $\bigwedge^k \mathbb{C}^{n}$ introduced there are isomorphic to the antisymmetric representation of the permutation group $S_k$, with the projection given by Equation \eqref{eq:wedgeproductstate}. Following the notation of Hamermesh \cite{hamermesh1962group}, the irreducible representations of $S_k$ are labeled by $(\lambda)$ and those of $U(n)$ by $\{\lambda\}$, where $\lambda \vdash k$ is an ordered partition of $k$. This antisymmetric representation is then denoted by $(1^k)$, i.e. a partition of $k$ integers into $k$ parts, corresponding to a Young diagram consisting of a single column of $k$ boxes. By Schur–Weyl duality (explained in Appendix~\ref{app:dualities}), this representation carries an irreducible representation of $U(n)$ denoted by $\{1^k\}$. Using Weyl's character formula (shown in Definition~\ref{def:weyl_character_formula}), it is well known that the character of the $U(n)$ representation on $\bigwedge^k \mathbb{C}^n$ is given by:

\begin{equation}
  \begin{split}
  \chi_{U(n)}\Bigl( \bigwedge^k \mathbb{C}^n \Bigr) &= \sum_{1 \le i_1 < i_2 < \cdots < i_k \le n} x_{i_1} x_{i_2} \hdots x_{i_k} \\
  &= e_k(x_1,x_2,\dots,x_n),
  \end{split}
\end{equation}

\noindent where $e_k$ is the $k^\text{th}$ elementary symmetric polynomial (presented in Definition~\ref{def:elementary_symmetric_polynomial}) and the $x_i$ are the eigenvalues of a unitary operator in the standard $n\times n$ representation of $U(n)$. Thus, the character of the full fermionic Fock space is:

\begin{equation}
  \begin{split}
  \chi_{U(n)}\Bigl( \bigwedge \mathbb{C}^n \Bigr) &= \chi_{U(n)}\Bigl(\bigoplus_{k=0}^n \bigwedge^k \mathbb{C}^n\Bigr) \\
  &= \sum_{k=0}^n e_k(x_1,x_2,\dots,x_n) \\
  &= \prod_{i=1}^n (1 + x_i).
  \end{split}
\end{equation}

\noindent Remarkably, this character is closely related to the generating function of the elementary symmetric polynomials,
$E_n(t) = \prod_{i=1}^n (1 + x_i t)$, which underpins the connection between symmetric polynomials and the representation theory of the unitary group.

When fermions exist in a direct product space $\mathbb{C}^m \otimes \mathbb{C}^m$, such as the spatial orbital and spin spaces of an electron in a molecule, a dual pair is formed. In representation theory, a dual pair (or reductive dual pair) refers to a pair of subalgebras or groups that act on a vector space in a mutually commuting manner, with each serving as the centralizer of the other within a larger symmetry group or algebra. This leads to the following decomposition:

\begin{theorem}[Antisymmetric Unitary-Unitary Duality] 
  The irreducible decomposition of the exterior power of the tensor product space $\mathbb{C}^m \otimes \mathbb{C}^n$ on which the elements of $U(n) \otimes U(m)$ act is given by:

  \begin{align}
    \bigwedge(\mathbb{C}^m \otimes \mathbb{C}^n) &: U(mn) \downarrow U(m) \times U(n) \cong \bigoplus_{k=0}^{m \times n} \bigoplus_{\substack{\lambda_k \vdash k \\ \ell(\tilde{\lambda}_k) \le m \\ \ell(\lambda_k) \le n }} W^{U(m)}_{\tilde{\lambda}_k} \otimes W^{U(n)}_{\lambda_k}
  \end{align}

  \noindent where $U(mn) \downarrow U(m) \times U(n)$ refers to the \textit{restriction} of $U(mn)$ to its subgroup $U(m) \times U(n)$, and $W_{\lambda_k}$ and $W_{\tilde{\lambda}_k}$ are the irreducible representations of $U(m)$ and $U(n)$ indexed by mutually conjugate partitions $\tilde{\lambda}_k$ and $\lambda_k$. The sum runs over all partitions $\lambda_k \vdash k$ with at most $n$ rows and $m$ columns. The conjugate partition $\tilde{\lambda}_k$ is obtained by transposing the Young diagram of $\lambda_k$.
  \label{thm:ffspace_umn_decomp}
  \end{theorem}
  
  \begin{proof} 
  The famous dual Cauchy identity (derived in Theorem~\ref{thm:dualcauchy}) can be used to obtain the decomposition into irreducible representations via the characters. The reducible character of the exterior power of the tensor product space $\mathbb{C}^m \otimes \mathbb{C}^n$ is given by:
  
  \begin{equation}
    \chi_{U(m) \times U(n)}\Bigl( \bigwedge(\mathbb{C}^m \otimes \mathbb{C}^n) \Bigr) = \prod^m_i \prod^n_j (1 + x_i y_j),
  \end{equation}

  \noindent where $x_i$ and $y_j$ are the eigenvalues of the unitary operators in the standard $m \times m$ and $n \times n$ representations of $U(m)$ and $U(n)$ respectively. For two sets of variables $x = (x_1, x_2, \ldots, x_n)$ and $y = (y_1, y_2, \ldots, y_m)$, we have:
  
  \begin{equation}
      \prod_{i=1}^{n} \prod_{j=1}^{m} (1 + x_i y_j) = \sum_{k=1}^{m \times n} \sum_{\substack{\lambda_k \vdash k \\ \ell(\tilde{\lambda}_k) \le m \\ \ell(\lambda_k) \le n }} s_{\tilde{\lambda}_k}(x) s_{\lambda_k}(y). \label{eq:background_cauchy}
  \end{equation}
  
  \noindent The Schur polynomials $s_{\tilde{\lambda}_k}(x)$ and $s_{\lambda_k}(y)$ are the characters of the irreducible representations of the $k^\text{th}$ tensor products of $U(m)$ and $U(n)$ respectively, thus Equation \eqref{eq:background_cauchy} gives the stated decomposition of the antisymmetric representation into irreps.
  \end{proof}
  \begin{remark}
    The elements of the subgroup \( U(m) \times U(n) \) embedded in \( \mathbb{C}^{mn} \) are given by \( \{ u \otimes v \mid u \in U(m), v \in U(n) \} \). Although both \( U(m) \) and \( U(n) \) are infinite-dimensional groups, the tensor product action of \( U(m) \times U(n) \) on \( \mathbb{C}^{mn} \) constitutes a subset of the unitary group \( U(mn) \). This is because the dimension of \( U(m) \times U(n) \) is \( m^2 + n^2 \), corresponding to the unique matrix elements, whereas the dimension of \( U(mn) \) is \( (mn)^2 \).
  \end{remark}

\subsection{Paldus Duality} \label{sec:paldus_duality}

For molecular systems, the Pauli exclusion principle dictates that electrons are paired into $d$ spatial orbitals and transform according to global $SU(2)$ symmetry in their spin degrees of freedom. Therefore, we study the decomposition of the antisymmetric representation $\{1^{2d}\}$ under the following group restriction chain:

\begin{equation}
  \bigwedge (\mathbb{C}^d \otimes \mathbb{C}^2): U(2d) \downarrow  U(d) \times U(2) \downarrow   U(d) \times SU(2). \label{eq:paldus_duality}
\end{equation}

\noindent The fully antisymmetric Fock space can be decomposed under $U(n) \times SU(2)$, which is referred to here as \textit{Paldus Duality}. In 1979, Paldus employed a basis adadpted to this decomposition to calculate the matrix elements of spin-adapted quantum chemistry Hamiltonians, thereby determining the eigenvalues and eigenspectrum of the quantum chemistry wavefunction~\cite{Paldus1979,Paldus1980,Paldus2020,Paldus2020a,Paldus2020b}. This constitutes a reduction from the more general $GL(m) \times GL(n)$ Howe duality~\cite{MathSoc1989} discovered later in 1988.

\begin{theorem}[Paldus Duality]
For an antisymmetric representation of $U(2d)$, the decomposition into irreps under restriction to the group $U(d) \times SU(2)$ is given by:

\begin{align}
    \bigwedge(\mathbb{C}^d \otimes \mathbb{C}^2) &: U(2d) \downarrow U(d) \times SU(2) \cong \bigoplus_{\substack{N=0}}^{2d} \bigoplus_{S=0}^{N / 2} W^{U(d)}_{(N,S)} \otimes W^{SU(2)}_S
\end{align}
where $W^{U(d)}_{(N,S)}$ and $W^{SU(2)}_S$ are the irreducible representations of $U(d)$ and $SU(2)$ respectively, with $N$ being the total number particles in occupied spin-orbitals and $S$ the global spin.
\label{thm:paldus_duality}
\end{theorem}

\begin{proof}
Proceeding from Theorem~\ref{thm:ffspace_umn_decomp} for $m=d$ and $n=2$, we have
\[
\begin{split}
\bigwedge(\mathbb{C}^d \otimes \mathbb{C}^2) &: U(2d) \downarrow U(d) \times U(2) \cong \bigoplus_{k=0}^{2d} \bigoplus_{\substack{\lambda_k \vdash k \\ \ell(\tilde{\lambda}_k) \le d \\ \ell(\lambda_k) \le 2 }} W^{U(d)}_{\tilde{\lambda}_k} \otimes W^{U(2)}_{\lambda_k}\,.
\end{split}
\]
Here, $\lambda_k = (\lambda_1, \lambda_2)$ denotes a partition with at most two parts (rows) and $d$ columns, and its conjugate $\tilde{\lambda}_k$ corresponds to a partition with at most $d$ rows and two columns. By applying the branching rule from $U(2)$ to $SU(2)$ (see Appendix~\ref{app:u2su2}), the irreducible representations of $SU(2)$ are obtained via the total spin $S$, where the irreducible representation of $U(2)$ indexed by $\lambda = (\lambda_1, \lambda_2)$, corresponds to $S = \frac{1}{2}(\lambda_1 -  \lambda_2)$. Moreover, the total number of particles $N$ is given by the number of boxes in the partition $\tilde{\lambda}_k$.
\end{proof}

\begin{figure}[!htbp]
\centering
\includegraphics[width=7cm]{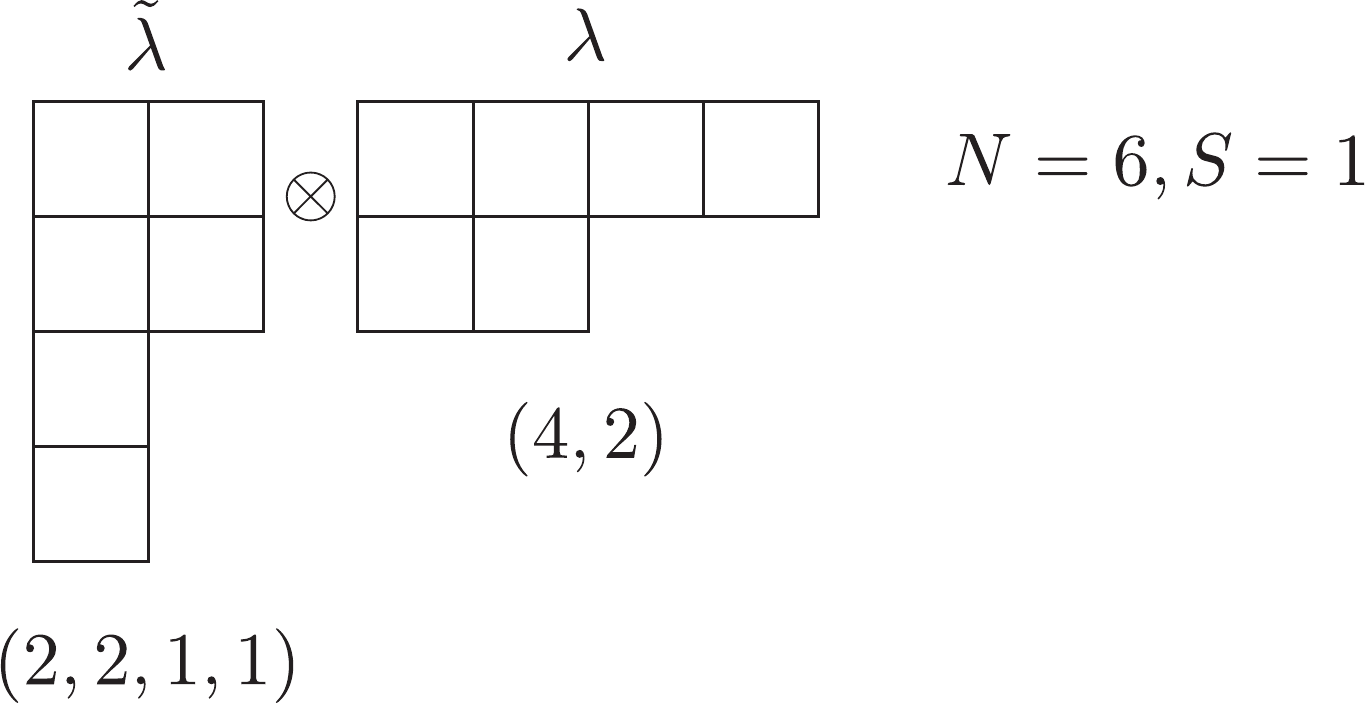}
\caption{$U(d) \times SU(2)$ irrep $(N=6,S=1)$, which is expressed by the partition $\lambda = (4,2)$ and its conjugate $\tilde{\lambda} = (2,2,1,1)$. $N$ is calculated from the number of boxes in both partitions and the $S$ is calculated from the difference between the number of boxes in the first and second row of $\lambda$, given by $S = \frac{1}{2}(\lambda_1 - \lambda_2) = \frac{1}{2}(4-2) = 1$.}
\label{fgr:paldus_dual_example}
\end{figure}

The decomposition of the fermionic Fock space provided by Theorem \ref{thm:paldus_duality} will re-appear in Section \ref{sec:paldus_transform}, where we will present an algorithm to explicitly instantiate the isomorphism. Notably, the Theorem implies the existence of a basis for $\bigwedge(\mathbb{C}^d \otimes \mathbb{C}^2)$ of the form $\ket{N, S, M; \mathbf{d}}$, in which states are indexed by the $(N,S)$ irrep space they belong to, their $SU(2)$ index $M$ and their $U(d)$ index $\mathbf{d}$. The action of $SU(2)$ then leaves all but the $M$ label unaffected, and the action of $U(d)$ only affects the $\mathbf{d}$ label. We refer to states of this form as \textit{Gelfand-Tsetlin (GT) states}. As a sanity check, we can verify that the number of GT states agrees with the dimension of $\bigwedge(\mathbb{C}^d \otimes \mathbb{C}^2)$. Each subspace $W^{SU(2)}_S$ has dimension $(2S + 1)$, whereas the dimensions of spaces $W^{U(d)}_{(N,S)}$ can be written as $T^d_{S, N}$ and evaluated with the hook-content formula:

\begin{equation}
  T^d_{S, N} = \frac{2S+1}{N+1} {d + 1 \choose \frac{N}{2} - S } {d + 1 \choose \frac{N}{2} + S + 1}. \label{eq:tds_formula}
\end{equation}

\noindent Carrying out sum of $\dim(W^{U(d)}_{(N,S)}) \times \dim (W^{SU(2)}_S)$ over all values of $N$ and $S$, the total dimension of the right-hand side of Equation \eqref{eq:paldus_duality} can be verified to equal $2^{2d}$, and a short calculation to do this is presented in Appendix \ref{app:dimension_formula}.

\begin{remark}
  Hamiltonians constructed from \( U(d) \times U(2) \) representations can be decomposed under \( U(d) \times SU(2) \), as the \( U(1) \) component of \( U(2) \cong SU(2) \times U(1) \) contributes a particle number-dependent, but often negligible overall phase. In this symmetry-adapted basis, the Hamiltonian block-diagonalises into subspaces corresponding to the $(N, S)$ irreducible representations of \( U(d) \times SU(2) \). Hence symmetric polynomial identities for the characters of \( U(d) \times U(2) \) are used to decompose the representations into irreps, and \( U(d) \times SU(2) \) Clebsch-Gordan coefficients construct the basis states spanning these irreps.
\end{remark}

\subsection{Subgroup Branching and Shavitt Graphs} \label{sec:shavitt}

In this section we define the basis upon which the Paldus transform is constructed. As observed previously, each Gelfand-Tsetlin state can be indexed by the quantum numbers $(N, S, M)$ and a $U(d)$ index $\mathbf{d}$, which has so far been left unidentified. To make this complete, we will make use of the \text{Shavitt graph} \cite{Shavitt1978,Shepard2006} -- a diagrammatic tool in the UGA for the representation of GT states as walks on a graph, which in turn allow a compact description by a $2d$-bit string, $\mathbf{d}$.

As is standard in representation theory literature, we may take advantage of the multiplicity-free branching property of the antisymmetric $U(d) \times U(2)$ representation and identify GT basis elements $\ket{N, S, M, \mathbf{d}}$ according to the irrep subspaces of the subgroups $U(d) \times U(2) \supset U(d-1) \times U(2) \supset \cdots \supset U(1) \times U(2)$ they belong to under a recursive restriction of $U(i)\times U(2)$ group action to $U(i-1)\times U(2)$. In the antisymmetric representation, multiplicity-free branching is guaranteed by Howe duality, and also follows from the double centralizer theorem (see Appendix \ref{app:dualities}). This implies that at each such restriction (i.e, $U(i) \times U(2) \downarrow U(i-1) \times U(2)$), each irreducible representation space of $U(i) \times U(2)$ decomposes into a direct sum of $U(i-1) \times U(2)$ irrep subspaces, with multiplicity of at most one. As a consequence, each GT basis element $\ket{N, S, M, \mathbf{d}}$ of $U(d) \times U(2)$ belongs to a unique sequence of $U(d-1) \times U(2), \hdots, U(1)\times U(2)$ irrep spaces under the restriction chain. The specification of this sequence is sufficient to unambiguously identify each GT basis state.

By a theorem known as Frobenius reciprocity there is an equivalent, second method of identifying GT basis states via branching which we will follow. Rather than recursively restricting $U(d) \times U(2)$ down its subgroup tower and listing the irrep spaces each state occupies, one can instead start with the group $U(1) \times U(2)$ and recursively \textit{lift} the antisymmetric representation via $(U(i-1) \times U(1)) \times U(2)$ to representations of $U(i) \times U(2)$. In our case, each $U(i) \times U(2)$ irrep at the $i^\text{th}$ step can be denoted by a Young diagram $\lambda_i$. Therefore, we can list track of all possible irrep sequences by beginning with the irrep indexed by the (null) Young diagram $\lambda_0 = \emptyset$, and repeatedly listing all irreps $\lambda_i$ which contain $\lambda_{i-1}$ after the $i^\text{th}$ lifting step. Each such sequence $\{ \lambda_d, \lambda_{d-1}, \hdots, \lambda_1, \emptyset \}$ is known as a \textit{Gelfand Tableau}. Formally, this is a sequence of \textit{interleaving} partitions $\{\lambda_{d}, \hdots, \lambda_1 \}$, where $\lambda_{i}$ interleaves $\lambda_{i-1}$ if the latter is obtained by removing a subset of boxes from $\lambda_i$. The exact rules for interleaving are derived, along with more details on Gelfand tableaux in Theorem~\ref{thm:gt_branching_betweenness}. Explicitly writing each element of each $\lambda_i$ in the sequence then forms a \textit{Gelfand-Tsetlin pattern} $\mathbf{m}$, which is the unique $U(d)$ index of a Gelfand-Tsetlin state:
\begin{equation}
    \mathbf{m} = 
  \begin{bmatrix}
  \lambda_d\\
  \lambda_{d-1}\\\
  \vdots\\
  \lambda_2\\
  \lambda_1
  \end{bmatrix} 
  = 
  \begin{bmatrix}
  \lambda_{1,d} &  & \lambda_{2,d} &  & \cdots & & \lambda_{d-1,d} &  & \lambda_{d,d}  \\
   & \lambda_{1,d-1} & &  \lambda_{2,d-1} & & \cdots &  &  \lambda_{d-1,d-1} &  \\
   &  & \ddots &  & \cdots &  & \iddots &  &  \\
   &  &  & \lambda_{1,2} &  & \lambda_{2,2} &  &  &  \\
   &  &  &  & \lambda_{1,1} &  &  &  &   
  \end{bmatrix} 
  \end{equation}

\noindent Finding all possible sequences $\mathbf{m}$ amounts to identifying the full Gelfand-Tsetlin basis for the $U(d)$ component of the antisymmetric $U(d) \times U(2)$ representation. In our particular case, the conjugacy between partitions $\lambda$ which index $U(d)$, and $U(2)$ irreps in Theorem \ref{thm:paldus_duality} greatly simplifies keeping track of the branching structure, forcing any irrep in the $U(d)$ chain to only have at most two columns (as shown in Figure \ref{fgr:paldus_dual_example}). Practically, this means that $\lambda$ can be represented by a simple tuple $(a, b, c)$, denoting the number of two-row, one-row and zero-row parts respectively. This is illustrated in Figure~\ref{fgr:WEYL}. As the three elements of the $(a, b, c)$ tuple sum to the number of spatial orbitals as $a + b + c = d$, if $d$ is known then $c$ can be additionally omitted. 

The $(a, b, c)$ representation gives rise to a compact way of representing the GT patterns $\mathbf{m}$. Because the Young diagram $\lambda_i$ in each row is constrained to have values $\lambda_{j, i} \in \{0,1,2\}$, the form of $\mathbf{m}$ simplifies to a sequence $\{ (a_d, b_d, c_d), \hdots, (a_1, b_1, c_1)\}$. In Appendix \ref{app:shavitt_branching}, we also prove Theorem \ref{thm:shavitt_branching}, which severely restricts the branching rules for lifting antisymmetric $U(i-1) \times U(2)$ representations. The main takeaway of our result is that under each lifting step, an irrep $\lambda_{i-1}$ can only branch into \textit{at most four} irreps $\lambda_i$, obtained by `adding' up to two boxes to $\lambda_{i-1}$ in four different ways. This in turn allows one to compress the form of $\textbf{m}$ even further. Instead of listing the $(a, b, c)$ tuples at each step, we can recover their sequence through a series of \textit{steps} $\mathbf{d}_i$, which modify the values of $(a, b, c)$ as follows:
\begin{equation}
(a_{i-1}, b_{i-1}, c_{i-1}) \overset{\mathbf{d}_i}{\longrightarrow} (a_i, b_i, c_i) = \begin{cases} (a_{i-1}, b_{i-1}, c_{i-1} + 1) & \text{if } \mathbf{d}_i = 0, \\ (a_{i-1}, b_{i-1} + 1, c_{i-1}) & \text{if } \mathbf{d}_i = 1, \\ (a_{i-1} + 1, b_{i-1} - 1, c_{i-1} + 1) & \text{if } \mathbf{d}_i = 2, \\ (a_{i-1} + 1, b_{i-1}, c_{i-1}) & \text{if } \mathbf{d}_i = 3. \end{cases}
\end{equation}

\noindent Concatenating all the steps $\mathbf{d}_i$ then forms a \textit{step vector}, $\mathbf{d}$. With the above rules and starting at $(a_0, b_0, c_0) = (0,0,0)$, applying each step from $\mathbf{d}$ can be used to find the tuple $(a_i, b_i, c_i)$ in each row of $\mathbf{m}$, and thus the overall GT pattern. The effect of different step vectors on values of $(a, b, c)$ is graphically shown in Figure~\ref{fgr:CSFALL}. 

As a sequence $\{ \lambda_d, \hdots, \lambda_1 \}$ in a Gelfand tableau can be represented by a semi-standard Young tableau (SSYT) on its final partition $\lambda_d$, the step vector $\mathbf{d}$ forms a set of `instructions' on how to build a SSYT, shown in Figure \ref{fig:branching_111}. We therefore have a one-to-one mapping between the Gelfand-Tsetlin patterns, SSYTs and step vectors. Out of the three, the step vector $\mathbf{d}$ is the most compact and thus suitable for use in a computational setting -- because each $\mathbf{d}_i \in \{ 0, 1, 2, 3 \}$, we can represent it \textit{in binary} using a $2d$-bit string. We will make use of this in constructing the quantum Paldus transform; our chosen encoding is presented in Table~\ref{fig:step_vectors_binary}.

\begin{table}[!htbp]
  \centering
  \begin{ruledtabular}
  \begin{tabular}{c c c c c c c}
   $\mathbf{d}_i$  &  $ \Delta a_i$& $ \Delta b_i = 2 \Delta S_i $ & $ \Delta c_i$ & $\Delta N_i$ & $\Delta \bar{c_i}$  & $\Delta \bar{c} a_i$\\
  \colrule
  0 & 0 & 0  & 1& 0 & 0 & \texttt{00} \\
  1  & 0 &$ + 1 $  & 0 & 1& 1 & \texttt{10} \\
  2  & 1 &$ -1 $ & 1 & 1& 0 & \texttt{01}\\
  3 & 1 &0  & 0& 2 & 1 & \texttt{11}\\
  \end{tabular}
  \end{ruledtabular}
      \caption{Effect of lifting via steps $\mathbf{d}_i \in \{0, 1, 2, 3\}$ on the values of tuples $(a_i, b_i, c_i)$. The tuples represents the number of \{0, 1, 2\}-parts in the partition $\lambda_i$ of a Gelfand-Tsetlin pattern. Their differences between each step can be related to changes in the quantum numbers $N$ and $S$ through Equation \eqref{eqn:GUGARULES}. Each digit $\mathbf{d}_i$ of the step vector $\mathbf{d}$ can be represented using two bits in binary, by concatenating the negation of $\Delta c_i$ with $\Delta a_i$. The resulting $2d$-bit string $\mathbf{d}$ completely describes the $U(d)$ GT pattern. Together with a value of the spin projection $M$, this uniquely defines the $U(d) \times U(2)$ GT basis states $\ket{N, S, M; \mathbf{d}}$.}
  \label{fig:step_vectors_binary}
\end{table}

\begin{figure}[htbp!]
  \centering

  \centering
  \hspace*{+0cm}
  \begin{minipage}{.5\textwidth}
    \includegraphics[width=4.5cm]{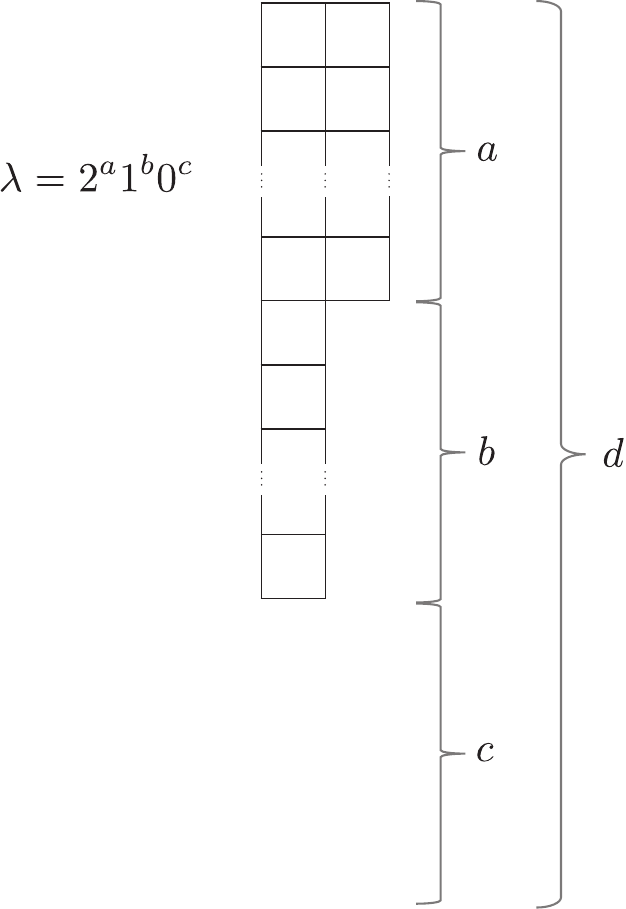}
  \end{minipage}%
  \begin{minipage}{.5\textwidth}
    \includegraphics[width=8cm]{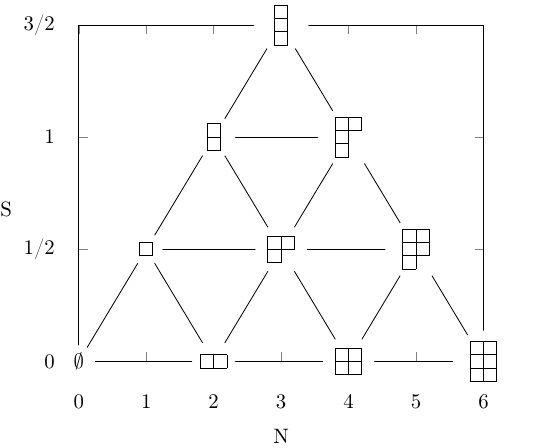}
  \end{minipage}
    \caption{\textit{Left:} Young Diagram representing the $d$-part partition $\lambda =  (2, \hdots ,2,1,\hdots,1,0,\hdots,0)$, which is abbreviated to $(2^a, 1^b, 0^c)$. The restriction to two columns imposed by Theorem \ref{thm:paldus_duality} allows for a compact representation of $\lambda$ through the tuple $(a, b, c)$. \textit{Right:} Irreducible representations of in the antisymmetric $U(3) \times U(2)$ representation. Each Young diagram denotes an $(N, S)$ irrep of $U(3)$. $U(2)$ irrep labels are obtained by conjugating the Young diagrams.}
    \label{fgr:WEYL}
  \end{figure}

The step vector also admits a further, graphical interpretation. Similarly to Bratelli diagrams in representation theory, $\mathbf{d}$ can be drawn as a walk on a \textit{Shavitt Graph} \cite{Shavitt1978,Shepard2006}, composed of nodes at levels $i = 0, 1, \hdots, d$. Each node at the $i^\text{th}$ level denotes an irrep appearing in the decomposition of the antisymmetric representation of $U(i) \times U(2)$, and vertices denote possible branchings under group subduction and lifting. Step vector $\mathbf{d}$ are displayed by upwards walks from the $i=0$ (tail) node to a $d=3$ (head) node. This elegantly captures the sequence of irrep spaces each GT state resides in under recursive subgroup branching of $U(d) \times U(2)$, and illustrates the equivalence between the step vector $\mathbf{d}$ and GT pattern $\mathbf{m}$. In a Shavitt graph, only the subgroup chain $U(d) \supset \cdots \supset U(1)$ is displayed, and the $U(2)$ irrep at each step is implicitly given through conjugation of the label $\lambda$. A full Shavitt graph for $d=3$ is shown in Figure \ref{fgr:CSFALL}, whereas a graph for the $\lambda=(1,1,1)$ irrep of $U(3) \times U(2)$ is shown in Figure \ref{fig:branching_111}.

\begin{figure}[!htbp]
  \centering
    \includegraphics[width=12cm]{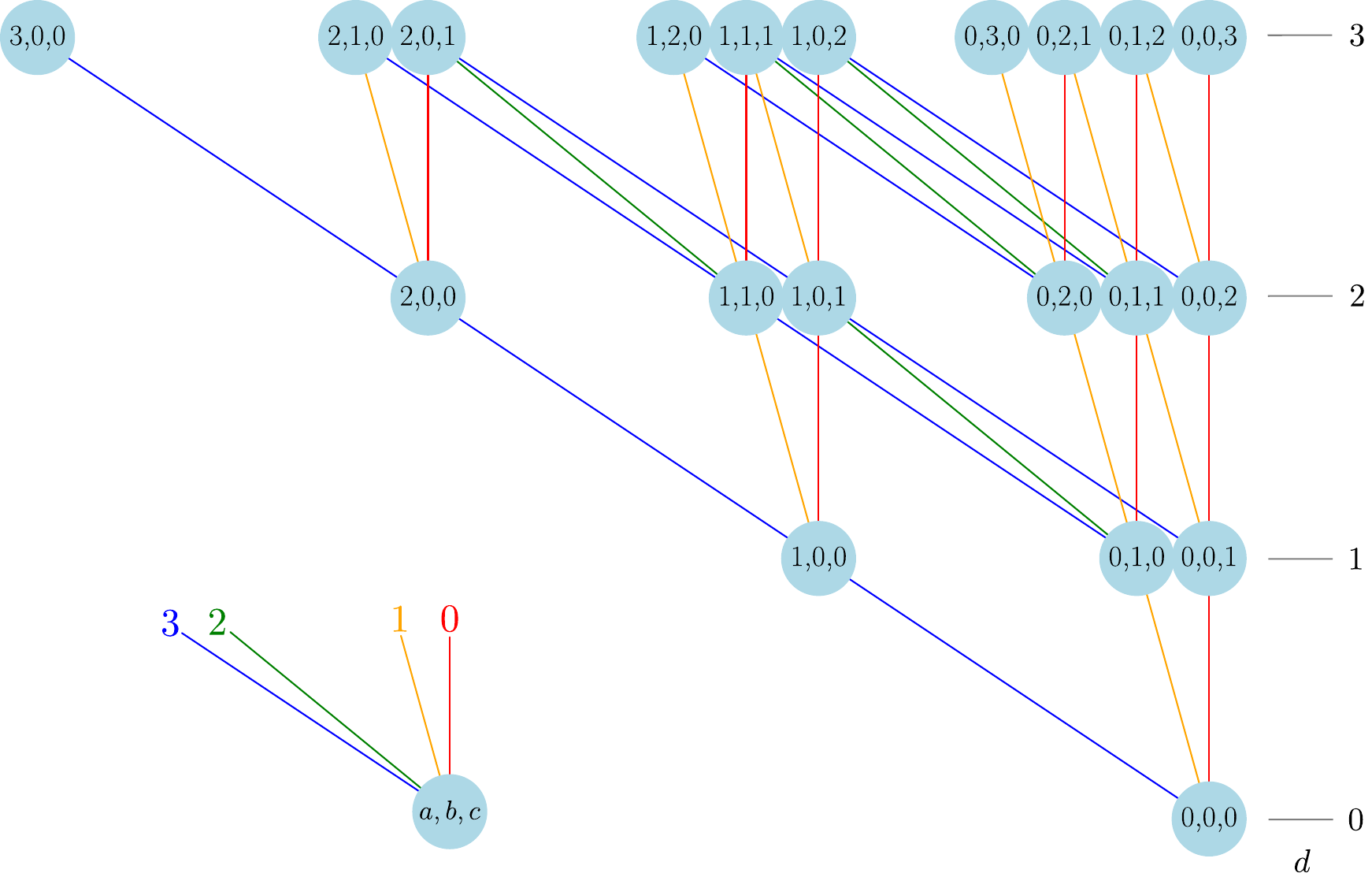}
    \caption{A superposition of all Shavitt graphs for $d=3$. The nodes denote $(a_i,b_i,c_i)$ triples defining the irreps at each link of the subgroup chain. The steps $\mathbf{d}_i$ are shown by the four coloured arcs, with a full step vector $\mathbf{d}$ consisting of a walk from a $d=0$ (tail) to the $d=3$ (head) nodes. The Gelfand-Tsetlin basis states spanning the irrep $\lambda $ are defined by unique paths from head to tail of the Shavitt graph, which is equivalent to the unique subduction chain $U(d) \times U(2) \supset \cdots \supset U(1) \times U(2)$.}
    \label{fgr:CSFALL}
  \end{figure}

\begin{figure}[!htbp]
  \centering
  \includegraphics[width=0.95\textwidth]{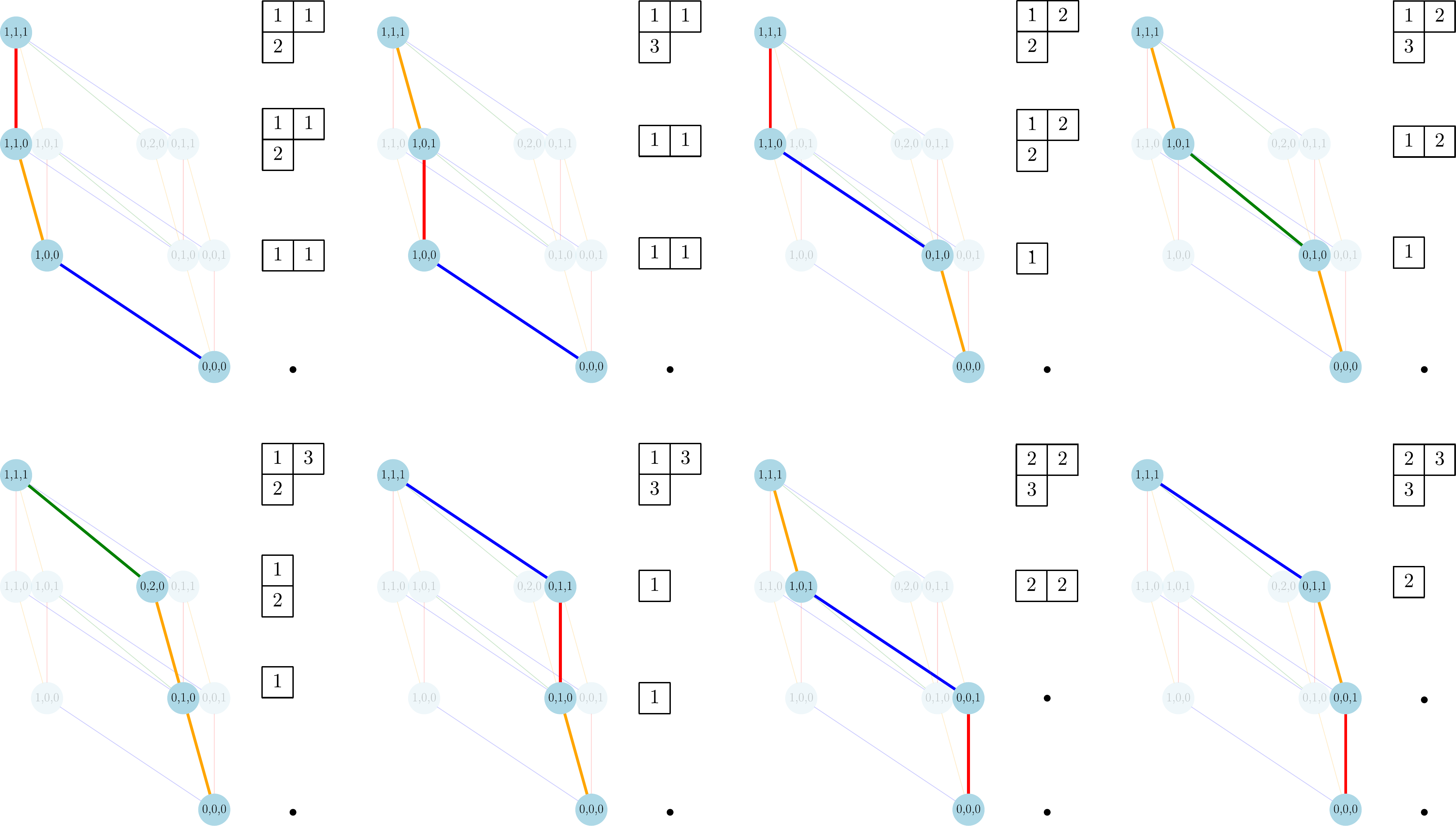}
   \caption{Shavitt graph for the $U(3) \times U(2)$ irrep $\lambda = (1,1,1)$ via the subgroup chain $U(3)\times U(2) \supset U(2) \times U(2) \supset U(1) \times U(2)$. Each walk denotes a Gelfand-Tsetlin state, with the associated semi-standard Young tableau also shown. The collection of all walks leads to a Gelfand-Tsetlin basis of the $U(3)$ irrep.}
  \label{fig:branching_111}
\end{figure}

\newpage
    
\subsection{\texorpdfstring{$SU(2)$}{SU(2)} Spin Coupling and the Gelfand-Tsetlin Basis} \label{sec:SU2}

So far the discussion has centered around the recursive coupling of $U(d) \times U(2)$. Bringing our discussion to the $U(d) \times SU(2)$ case is straightforward, and uses the $U(2) \downarrow SU(2)$ branching rule which we derive via the characters in Theorem~\ref{thm:u2su2} of Appendix \ref{app:unitarygroup}. This then implies that $U(2)$ conjugate tableaux can genuinely express the $SU(2)$ spin irrep $S$ from Theorem~\ref{thm:paldus_duality}, and that from $U(d) \times U(2)$ branching we can recover the well-known $SU(2)$ Clebsch-Gordan coupling rules \cite{Brink:1975}. In other words, each $U(d)$ Gelfand-Tsetlin state carries an $SU(2)$ spin eigenfunction of coupled spin-$\frac{1}{2}$ primitives. The $SU(2)$ irrep $S$ is given by the difference of the number of boxes in the two rows of the $U(2)$ tableaux via $S = \frac{1}{2}(\lambda_1 - \lambda_2)$. Because of this, we can interpret each step $\mathbf{d}_i$ as a simultaneous coupling of spatial orbitals and spin. In Section \ref{sec:shavitt} we introduced the $(a, b, c)$ tuples as a shorthand for partitions $\lambda$ in a lifting procedure: they can be related to the quantum numbers $N$ and $S$ through the following rules:

\begin{equation}
  \begin{split}
      &N = 2a + b \\
      &S = b/2 \\
      &d = a + b + c \\ 
  \end{split}
  \label{eqn:GUGARULES}
\end{equation}

From Equation \eqref{eqn:GUGARULES} it can be seen, for example, that only a change in $b$ values (steps $\mathbf{d}_i = 1, 2$) can affect the total spin $S$ of the irrep, and the number of boxes $N$ in a Young diagram corresponds to the number of particles in the system. The effects of the four couplings on $(N, S)$ are summarised in Table~\ref{fig:step_vectors_binary} in the previous section. In practice it is often more convenient to use the $(N, S)$ quantum numbers (with orbital number $d$ fixed) to express the GT states, both from a quantum chemistry standpoint and for the implementation of the Paldus transform on a quantum computer. Just like with the $(a, b, c)$ tuples, we can also draw Shavitt graphs in terms of $(N,S)$ values. The Shavitt graph may then be interpreted as an extension of \textit{geneological spin coupling} \cite{Pauncz1979}, with the main difference over the standard case being the additional ability to couple unoccupied, and doubly-occupied spatial orbitals. This link to geneological coupling is made clear in Figure~\ref{fig:quga_block_diagonals}, where the $(a,b,c)$ tuples are mapped to $(N,S,n_M)$ values, where $n_M = 2S + 1$ are the spin multiplicities. Just like with the graphs in Figures \ref{fgr:CSFALL} and \ref{fig:branching_111}, each basis state in an irrep space corresponds to a walk from the $(0, 0, 1)$ node at the tail of the graph to one of the head nodes.

\subsubsection{Explicit Construction of Gelfand-Tsetlin States} \label{sec:gt_states_construction}

In the preceding sections we have discussed the decomposition of $\bigwedge(\mathbb{C}^d \otimes \mathbb{C}^2)$ into irrep spaces of $U(d) \times SU(2)$, and in the process defined their spanning Gelfand-Tsetlin basis states $\ket{N, S, M; \mathbf{d}}$. These states can be interpreted as paths on Shavitt graphs, which themselves can be seen as a variant of a Bratelli diagram or a geneological coupling procedure. However, from this discussion alone it is not clear how to \textit{construct} these states explicitly in the occupation number (computational) basis, i.e. how identify the coefficients $c_{x_{1 \uparrow} x_{1 \downarrow} \cdots x_{d \uparrow} x_{d \downarrow}}$ in the expansion (with $\uparrow = +\frac{1}{2}, \downarrow = -\frac{1}{2}$ and $x_{i, \sigma}$ an occupation number):

\begin{equation}
\ket{N, S, M; \mathbf{d}} = \sum_{x_{1 \uparrow} x_{1 \downarrow} \cdots x_{d \uparrow} x_{d \downarrow}} c_{x_{1 \uparrow} x_{1 \downarrow} \cdots x_{d \uparrow} x_{d \downarrow}} \ket{x_{1 \uparrow} x_{1 \downarrow} \cdots x_{d \uparrow} x_{d \downarrow}}. \label{eq:gt_state_expansion}
\end{equation}

\noindent In the present case of $U(d) \times SU(2)$, the above states turn out to be remarkably simple to construct. Previous work by Paldus, Moshinsky and others \cite{Moshinsky1971, Paldus1980} has shown that states $\ket{N, S, M; \mathbf{d}}$ are equivalent (up to an overall phase) to the the \textit{Yamanouchi-Kotani} states \cite{Pauncz1979}, which arise from the sequential Clebsch-Gordan coupling of spin-$\frac{1}{2}$, or spin-$0$ orbitals. They are defined as:

\begin{equation}
  \ket{N, S, M; \mathbf{d}}  = \left( \sum_{m_1, m_2} \sum_{m_{[2]}, m_3} \cdots \sum_{m_{[n-1]}, m_n}  \left[ \prod_{i=1}^{d} C^{s_{[i]} m_{[i]}}_{s_{[i-1]} m_{[i-1]} ; s_i m_i} \mathcal{A}_i^\dagger(m_i, \mathbf{d}_i) \right] \right) | 0 \rangle^{\otimes 2d} \label{eq:yk_states_formula}
\end{equation}

\noindent Where the square bracket contains a product of $d$ many CG coefficients, each arising from the coupling of a system with spin (projection) $(s_{[i-1]}, m_{[i-1]})$ and an orbital with $(s_i, m_i$) into a larger system with $(s_{[i]}, m_{[i]})$. Clearly, for unoccupied or doubly-occupied orbitals $s_i = m_i = 0$, and for singly-occupied orbitals $s_i = \frac{1}{2}$ and $m_i = \pm \frac{1}{2}$. All intermediate spin projection values $m_{[i]}$ are summed over, barring the final $M$ which forms the $U(2)$ GT pattern. We take the `starting' spin to be $s_{[0]} = m_{[0]} = 0$, and $s_{[i]} = s_{[i-1]} + \Delta S_i$ to be the total spin at each step of the sequential coupling. To raise or lower the overall spin in the $i^\text{th}$ step we must couple a spin-$\frac{1}{2}$ orbital (steps $\mathbf{d}_i = 1, 2$), meaning that $\Delta S_i = \pm \frac{1}{2}$. To leave the overall spin unchanged at the $i^\text{th}$ step, we couple a spin-$0$ orbital (steps $\mathbf{d}_i = 0, 3$), in which case $\Delta S_i = 0$ and $s_{[i]} = s_{[i-1]}$. Every term in the sum contains a product of operators $\mathcal{A}_i^\dagger(m_i, \mathbf{d}_i)$, each of which is made out of fermionic creation operators as follows:

\begin{equation}
  \mathcal{A}_i^\dagger(m_i, \mathbf{d}_i) = \begin{cases}
    \mathbf{1} & \text{if } \Delta N_i = 0 \ (\mathbf{d}_i = 0) \\
    a^\dagger_{i, m_i} & \text{if } \Delta N_i = 1 \ (\mathbf{d}_i = 1, 2) \\
    a^\dagger_{i, +\frac{1}{2}} a^\dagger_{i, -\frac{1}{2}} & \text{if } \Delta N_i = 2 \ (\mathbf{d}_i = 3)
  \end{cases}
\end{equation}

\noindent Each $\mathcal{A}_i^\dagger(m_i, \mathbf{d}_i)$ is essentially a creation operator on the $i^\text{th}$ orbital, acting in accordance with how that orbital has been coupled (i.e, one of the four possible $\mathbf{d}_i$'s as shown in Table \ref{fig:step_vectors_binary}). For example, the fourth case ($\mathbf{d}_i = 3$) corresponds to the coupling of a spin-$0$ orbital containing two spin-$\frac{1}{2}$ electrons pre-coupled into a singlet:

\begin{align}
\mathcal{A}^\dagger_i(0, 3) &= \frac{1}{\sqrt{2}} \left( C^{0, 0}_{\frac{1}{2}, +\frac{1}{2}; \frac{1}{2}, -\frac{1}{2} } a^\dagger_{i, +\frac{1}{2}} a^\dagger_{i, -\frac{1}{2}} + C^{0, 0}_{\frac{1}{2}, -\frac{1}{2} ; \frac{1}{2}, +\frac{1}{2} } a^\dagger_{i, -\frac{1}{2}} a^\dagger_{i, +\frac{1}{2}} \right) \\
& = \frac{1}{2} a^\dagger_{i, +\frac{1}{2}} a^\dagger_{i, -\frac{1}{2}} - \frac{1}{2} a^\dagger_{i, -\frac{1}{2}} a^\dagger_{i, +\frac{1}{2}} \\
&= a^\dagger_{i, +\frac{1}{2}} a^\dagger_{i, -\frac{1}{2}}
\end{align}

The product in Equation \eqref{eq:yk_states_formula} is defined to multiply the operators $\mathcal{A}_i^\dagger(m_i, \mathbf{d}_i)$ `to the right', i.e. every term appearing in the sum is of the form:

\begin{equation}
C_1 C_2 \cdots C_d \times \mathcal{A}_1^\dagger(m_1, \mathbf{d}_1) \mathcal{A}_2^\dagger(m_2, \mathbf{d}_2) \cdots \mathcal{A}_d^\dagger(m_d, \mathbf{d}_d) \ket{0}^{\otimes 2d},
\end{equation}

\noindent where $C_i$ are Clebsch-Gordan coefficients. With the Jordan-Wigner fermion-to-qubit mapping (defined in Equation \eqref{eq:jw_creation_annihilation}), we can thus write each term as:

\begin{equation}
C_1 C_2 \cdots C_d \times (\frac{1}{2}(X -iY))^{a_1} Z^{b_1} \ket{0} \otimes (\frac{1}{2}(X -iY))^{a_2} Z^{b_2}  \ket{0} \otimes \cdots \otimes (\frac{1}{2}(X -iY))^{a_d} Z^{b_d}  \ket{0}, \ \ \ a_i, b_i \in \{0, 1\}
\end{equation}

\noindent Thus, for every term in the sum and every qubit in the $2d$-fold tensor product, any $Z$ Pauli matrix present is applied to a $|0\rangle$ state first, meaning that there are no intermediate minus signs to keep track of other than those introduced by the Clebsch-Gordan coefficients (i.e. the $Z$ operators can be dropped). As an example of the above, GT states constructed via Equation \eqref{eq:yk_states_formula} are given for $d=1$ and $d=2$ in Figures \ref{fig:d1_gt_states} and \ref{fig:d2_gt_states}. 

As a consequence of the absence of intermediate minus signs in Equation \eqref{eq:yk_states_formula}, the GT states of interest \textit{structurally resemble} another family of states, referred to as the \textit{Schur basis} in quantum information \cite{harrowphd} or \textit{spin eigenfunctions} in quantum chemistry \cite{Pauncz1979}. They are states of the form:

\begin{equation}
|S, M; \mathbf{s} \rangle = \sum_{m_1, m_2} \sum_{m_{[2]}, m_3} \cdot \cdot \sum_{m_{[n-1]}, m_n} C^{s_{[2]}, m_{[2]}}_{\frac{1}{2} m_1 ;\frac{1}{2}  m_2} C^{s_{[3]}, m_{[3]}}_{s_{[2]}, m_{[2]} ;\frac{1}{2}  m_3} \cdot \cdot \ C^{s_{[n-1]}, m_{[n-1]}}_{s_{[n-2]}, m_{[n-2]} ;\frac{1}{2}  m_{n-2}} C^{S M}_{s_{[n-1]}, m_{[n-1]} ;\frac{1}{2}  m_n} |m_1 \hdots m_n \rangle \label{eq:schur_states_formula}
\end{equation}

\noindent Where $\mathbf{s} = \{ s_{[1]}, s_{[2]}, ..., s_{[n-1]}, S\}$ is a sequence of spins obtained by a sequential coupling of $n$ spin-$\frac{1}{2}$ particles, $s_{[1]} = \frac{1}{2}$, and $s_{[i+1]} = s_{[i]} \pm \frac{1}{2}$. This definition coincides with Equation \eqref{eq:yk_states_formula}, with the only difference being that in addition to coupling spin-$\frac{1}{2}$ orbitals, in the $U(d) \times SU(2)$ GT basis one can also couple spin-$0$ orbitals in two different ways: one when both of its spin-orbitals are both empty, and another when they are both occupied and pre-coupled into a singlet.  When one does this, the intermediate spin values read $s_{[i]} = s_{[i-1]}, m_{[i]}=m_{[i-1]}, s_i = 0, m_i = 0$. The $i^\text{th}$ Clebsch-Gordan coefficient then reads $C^{s_{[i]}, m_{[i]}}_{s_{[i-1]}, m_{[i-1]}; 0, 0} = 1$, and the action of the operator $\mathcal{A}_i^\dagger(0, \mathbf{d}_i)$ `tensors in' a state $\ket{00}$ or $\ket{11}$ into the $(2i)^\text{th}$ position for all bitstrings in the sum. 

What the above tells us is that the coefficients $c_{x_{1 \uparrow} x_{1 \downarrow} \cdots x_{d \uparrow} x_{d \downarrow}}$ of $U(d) \times SU(2)$ GT states in Equation \eqref{eq:yk_states_formula} will \textit{always} correspond to the coefficients of some Schur state. We can show this by considering two cases: the first is when $\mathbf{d}$ contains no instances of $\mathbf{d}_i = 0, 3$ (i.e. none of the steps come from the first or last row of Table \ref{fig:step_vectors_binary}). Then, the GT state (in the occupation basis) is given by taking a Schur state on $d$ qubits, and re-writing each $\ket{m_i}$ on the right-hand side of Equation \eqref{eq:schur_states_formula} by $\beta = \ket{\uparrow} = \ket{10}$ and $\alpha = \ket{\downarrow} = \ket{01}$. On the other hand, if the step vector $\mathbf{d}$ contains $k$ instances of $d_i = 0, 3$, then the GT state will be obtained by taking a Schur state on $d - k$ qubits, re-writing each $\ket{m_i}$ with $\{ \ket{10}, \ket{01} \}$ as before, and finally tensoring in each $\ket{00}$ or $\ket{11}$ at the appropriate positions. This works because the Clebsch-Gordan coefficient associated with coupling a spin-$0$ system is always equal to  $1$. It is well-known that in the first quantisation spin eigenfunctions play a very useful role by virtue of their well-defined values of global spin. Their close relation to the $U(d) \times SU(2)$ GT states tells us that methods for preparing Schur states will prove useful in our case, and that our GT states will \textit{reduce} to spin eigenfunctions for systems where each orbital is singly-occupied. In essence, the GT basis of interest is a very straightforward generalisation of the Schur states to a second quantised setting.

\subsubsection{Spin Coupling as Rotations} \label{sec:spin_coupling_rotations}

With the formula for constructing GT states by sequential coupling in Equation \eqref{eq:yk_states_formula}, we are ready to provide some intuition behind the circuit for the quantum Paldus transform. In Section \ref{sec:paldus_transform} we present the transform as an \textit{isometry} between the states $\ket{N, S, M; \mathbf{d}}$ and a \textit{UGA basis}, for which each quantum number is stored locally in a separate subsystem as $\ket{N} \otimes \ket{S} \otimes \ket{M} \otimes \ket{\mathbf{d}}$. As a precursor, we end this section by showing that the construction of GT states can be seen as sequence of controlled rotations. This is because the Clebsch-Gordan coefficients which constitute the amplitudes in Equation \eqref{eq:gt_state_expansion} form the components of a rotation, given by:

\begin{align}
  C^{S - \frac{1}{2} , M }_{S, M+ \frac{1}{2}; \frac{1}{2}, -\frac{1}{2} } &= \sqrt{\frac{1}{2}} \sqrt{1 + \frac{2M}{2S + 1}} = \cos(\theta_{S, M}), & C^{S + \frac{1}{2} , M}_{S, M + \frac{1}{2} ; \frac{1}{2}, -\frac{1}{2} } &= \sqrt{\frac{1}{2}} \sqrt{1 - \frac{2M}{2S + 1}} = \sin(\theta_{S, M}), \label{eq:CGMatrix} \\
  C^{S - \frac{1}{2} , M }_{S, M- \frac{1}{2}; \frac{1}{2}, +\frac{1}{2} } &= - \sqrt{\frac{1}{2}} \sqrt{1 - \frac{2M}{2S + 1}} = -\sin(\theta_{S, M}), & C^{S + \frac{1}{2} , M }_{S, M - \frac{1}{2}; \frac{1}{2}, +\frac{1}{2} }  &= \sqrt{\frac{1}{2}} \sqrt{1 + \frac{2M}{2S + 1}} = \cos(\theta_{S, M}) \nonumber
\end{align}

\noindent Therefore, a unitary change of basis between the left and right-hand sides of Equation \eqref{eq:gt_state_expansion} can be obtained by repeatedly rotating states of the form $\ket{N, S, M'; \mathbf{d}} \otimes \ket{x_{i \uparrow} x_{i \downarrow}} $ into new states $\ket{N', S', M; \mathbf{d}'}$, where we can write the former as:

\begin{equation}
  \ket{N, S, M'; \mathbf{d}} \otimes \ket{x_{i \uparrow} x_{i \downarrow}} = \begin{cases} \ket{N, S, M; \mathbf{d}} \otimes \ket{n_i = 0, s_i=0, m_i=0} & \text{ if } x_{i \uparrow} x_{i \downarrow} = 00, \\
    \ket{N, S, M - 1/2; \mathbf{d}} \otimes \ket{n_i = 1, s_i=1/2, m_i=1/2} & \text{ if } x_{i \uparrow} x_{i \downarrow} = 10, \\
    \ket{N, S, M + 1/2; \mathbf{d}} \otimes \ket{n_i = 1, s_i=1/2, m_i=-1/2} & \text{ if } x_{i \uparrow} x_{i \downarrow} = 01, \\
    \ket{N, S, M; \mathbf{d}} \otimes \ket{n_i = 2, s_i=0, m_i=0} & \text{ if } x_{i \uparrow} x_{i \downarrow} = 11 \end{cases}
\end{equation} 

Because of the mentioned absence of minus signs in Equation \eqref{eq:yk_states_formula}, this is completely analogous to how the Schur states and spin eigenfunctions are constructed, with the main difference being that some of the rotations elements (those corresponding to coupling spin-$0$ orbitals) are trivial with the CG coefficients $C^{S, M}_{S, M ; 0, 0} = 1$. We demonstrate this in Figures \ref{fgr:mz_matrix_arcs} and \ref{fgr:mz_gt_coupling}. In our diagrams, a yellow arc signifies coupling to raise $(S, N)$ by $(+\frac{1}{2}, 1)$, a green arc signifies coupling to change $(S, N)$ by $(-\frac{1}{2}, 1)$, a red arc couples an empty orbital (leaving $S$ and $N$ unchanged), and a blue arc couples a doubly-occupied orbital (also leaving $S$ unchanged, but raising $N$ by $2$). Crucially, the corresponding Clebsch-Gordan coefficients for $S+\frac{1}{2}$ and $S-\frac{1}{2}$ coupling depend on the spin projection $M$, reflected in the expressions for $\theta_{S, M}$. Together with the trivial CG coefficients for when $S$ is unchanged, the full $R(\theta_{S, M})$ is a $4\times4$ Givens rotation matrix.

\begin{figure}[h!]
  \centering
    \includegraphics[width=\textwidth]{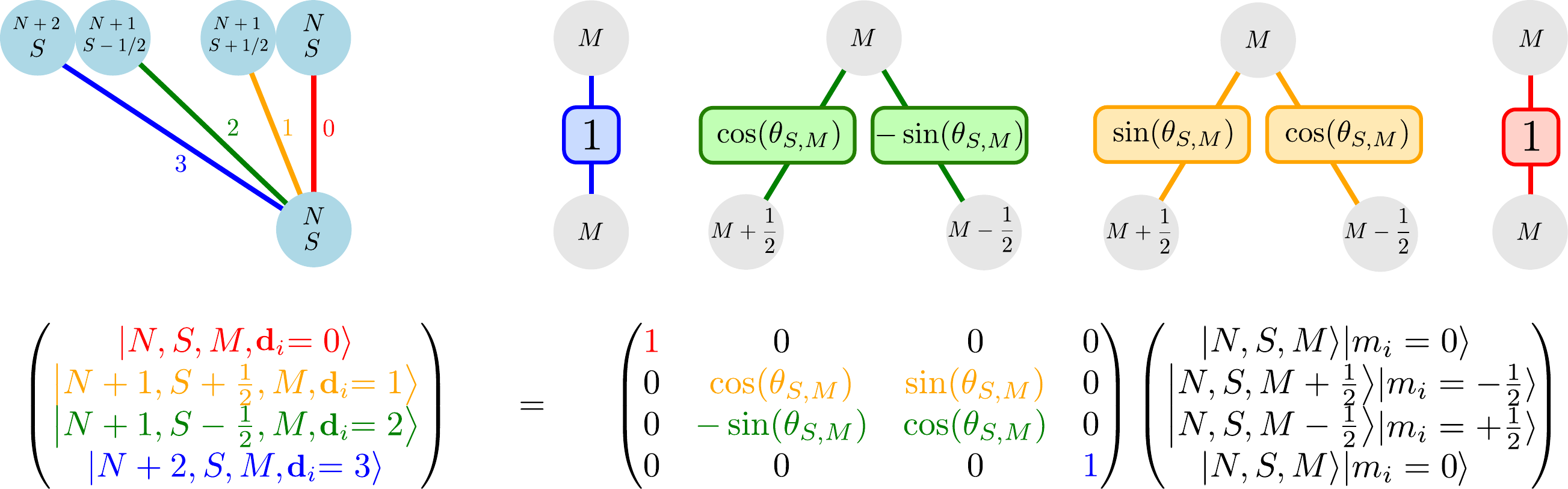}
    \caption{\textit{Top left:} Coupling by steps $\mathbf{d}_i \in \{0,1,2,3\}$ out of a node $N, S$ of a Shavitt graph. The $(\Delta S, \Delta N)$ for each $\mathbf{d}_i$ are represented inside the nodes.   \textit{Top right:} The corresponding coupling of $M$ values which appear in the spin coupling graphs of Figure \ref{fgr:mz_gt_coupling}. \textit{Bottom:} The overall effect of the coupling can be seen as a rotation from tensor product states $\ket{N, S, M'; \mathbf{d}} \otimes \ket{m_i}$ (bottom nodes) into Gelfand-Tsetlin states $\ket{N', S', M; \mathbf{d}'}$ (top nodes).}
    \label{fgr:mz_matrix_arcs}
\end{figure}

\newpage

In Figure \ref{fgr:mz_gt_coupling}, the correspondence between Shavitt graph walks (given by $\mathbf{d}$) and $SU(2)$ spin coupling is shown. Each Shavitt graph walk (left) gives rise to a spin coupling graph (right), with the top node indicating a final choice of $M$ which is undeteremined by $\mathbf{d}$. The resulting $U(d) \times SU(2)$ GT states $\ket{N, S, M; \mathbf{d}}$ are then constructed by summing over all possible upwards walks from the bottom node to an uppermost $M$ node. Each step in such a walk corresponds to `tensoring' of a state $\ket{x_{i \uparrow} x_{i \downarrow}}$ multiplied by the weighing (CG coefficient) of the edge. Steps to the left ($\nwarrow$) correspond to $\alpha = \ket{01}$ indicating $\Delta M = -\frac{1}{2}$, steps to the right ($\nearrow$) correspond to $\beta = \ket{10}$ indicating $\Delta M = +\frac{1}{2}$, and steps upwards ($\uparrow$) correspond to $\ket{00}$ (for red edges) or $\ket{11}$ (for blue edges) both indicating $\Delta M = 0$. This method performs exactly the same calculation as Equation \eqref{eq:yk_states_formula}.

To make Figure \ref{fgr:mz_gt_coupling} more intuitive, give two examples. The state $\ket{N=3, S=\frac{3}{2}, M=\frac{1}{2}; \mathbf{d = 10, 10, 10}}$ (with $\mathbf{d}$ in binary via Table \ref{fig:step_vectors_binary}) in the top right of Figure \ref{fgr:mz_gt_coupling} is calculated as follows:

\begin{align}
\ket{3, \frac{3}{2}, \frac{1}{2}; \mathbf{10, 10, 10}} &= \left(  \ket{\nwarrow}  \otimes  \frac{1}{\sqrt{2}} \ket{\nearrow} \otimes  \sqrt{\frac{2}{3}} \ket{\nearrow} \right) +  \left(  \ket{\nearrow}  \otimes \frac{1}{\sqrt{2}} \ket{\nwarrow} \otimes  \sqrt{\frac{2}{3}} \ket{\nearrow} \right) +  \left( \ket{\nearrow}  \otimes \ket{\nearrow} \otimes  \frac{1}{\sqrt{3}}  \ket{\nwarrow} \right) \\
&= \frac{1}{\sqrt{3}} \ket{011010} + \frac{1}{\sqrt{3}} \ket{100110} + \frac{1}{\sqrt{3}} \ket{101001},
\end{align}

\noindent whereas the GT basis state $\ket{2, \frac{1}{2}, 0; \mathbf{10, 00, 10}}$ is given by:

\begin{equation}
  \ket{2, \frac{1}{2}, 0; \mathbf{10, 00, 10}} = \left(  \ket{\nwarrow}  \otimes { \ket{\color{purple} \uparrow} } \otimes  \frac{1}{\sqrt{2}} \ket{\nearrow} \right) +  \left(  \ket{\nearrow}  \otimes { \ket{\color{purple} \uparrow} } \otimes  \frac{1}{\sqrt{2}} \ket{\nwarrow} \right)  = \frac{1}{\sqrt{2}} \ket{010010} +\frac{1}{\sqrt{2}} \ket{100001}.
\end{equation}

\noindent The form of the above states should reaffirm the relation of the GT basis states to Schur states, as remarked on in Section \ref{sec:gt_states_construction}. As we will see in the following sections, working in the GT basis offers numerous advantages, one of them being the block-diagonalisation of $U(d) \times U(2)$ antisymmetric representations, which in turn correspond to evolution by spin-free Hamiltonians. This is shown in Figure \ref{fig:quga_block_diagonals}. With the GT basis defined this way, we recover the full dimensionality of the $2^{2d}$ occupation number states -- a proof of this is given in Appendix \ref{app:dimension_formula}.

\begin{figure}[]
  \centering
    \includegraphics[width=\textwidth]{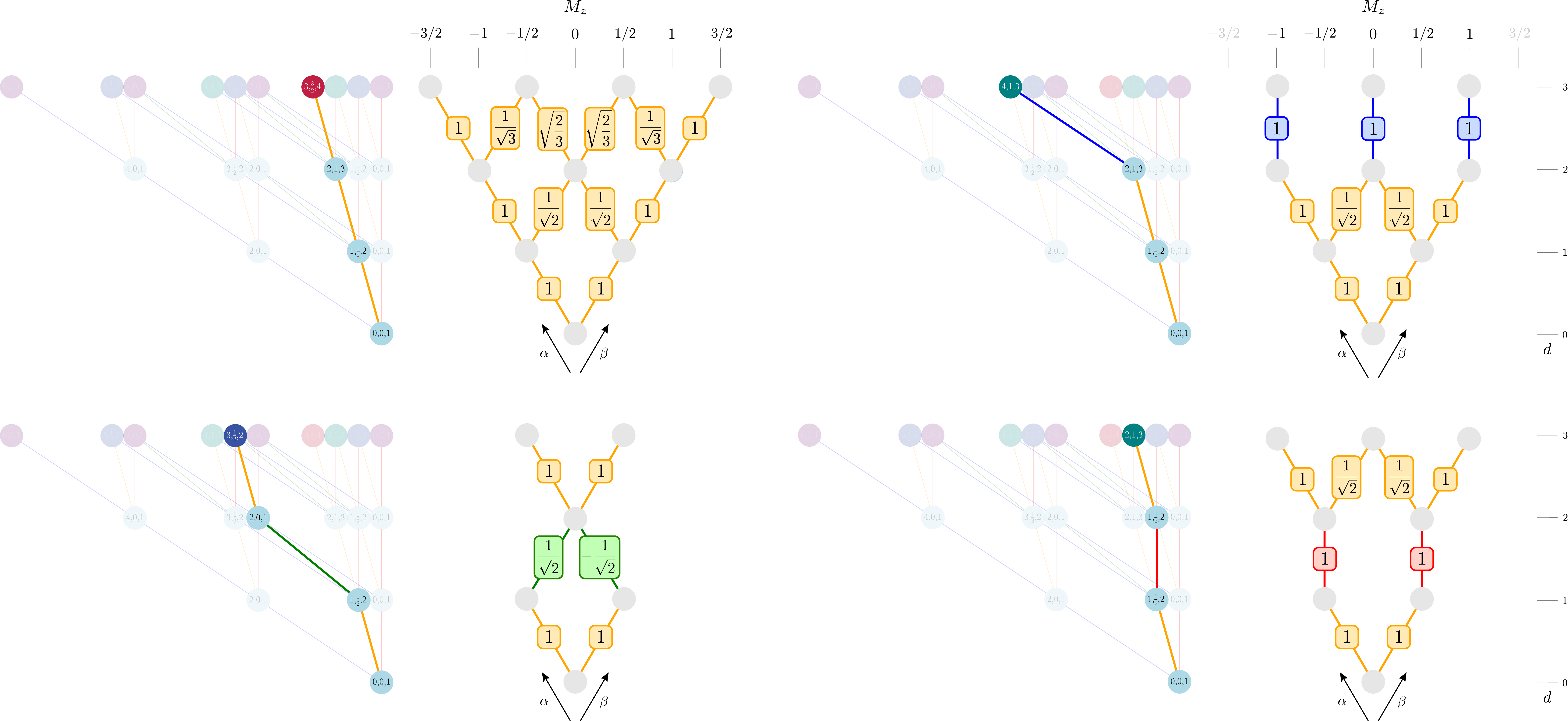}
    \caption{Four sub-plots illustrating examples of genealogical coupling of spin primitives $\ket{x_{i \uparrow} x_{i \downarrow}}$. Each spin coupling graph is obtained out of a walk on a Shavitt graph. The Clebsch-Gordan coupling coefficients are annotated on the arcs, with colour-coding indicating the type of step $\mathbf{d}_i \in \{0,1,2,3\}$ in the corresponding Shavitt walk. The spin coupling graphs can be used to explicitly write down GT basis states in terms of the occupation number basis: see the main text for details.}
    \label{fgr:mz_gt_coupling}
\end{figure}

\newpage

\begin{figure}[htb!]
  \includegraphics[width=17.3cm]{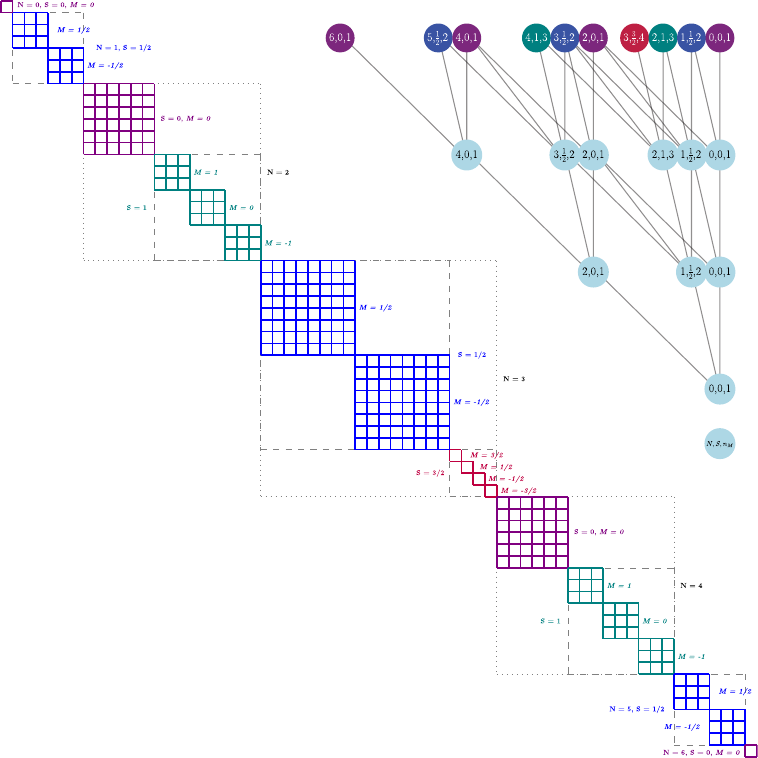}
\caption{\textit{Left: Block-diagonalisation of Antisymmetric $U(3) \times U(2)$ Representations in the Gelfand-Tsetlin Basis.} The diagram represents matrix elements of $U(d)$ group action for $d=3$ under the $U(3) \times U(2)$ GT basis: examples include evolution by spin-free Hamiltonians \eqref{eq:hud_hamiltonian} and the Hubbard model \eqref{eq:fermihubbardladder}. Solid boxes represent irreducible representations $(N, S)$, spanned by GT basis states $\ket{N, S, M, \mathbf{d}}$. In the block-diagonalisation with respect to $U(d)$, the spin projection $M$ gives rise to $2S+1$ multiplicities of each irrep. Dotted lines group blocks of common particle number $(N)$. Dashed lines group simultaneous irreps of $U(3) \times U(2)$. Colours represent the value of $S$ (Purple $= 0$, Blue $= 1/2$, Green $= 1$, Red $= 3/2$). Under the action of $U(2)$ one obtains a similar block-diagonalisation, with the multiplicities and dimensions of each block swapped. \textit{Right: Full Shavitt Graph for $U(3) \times U(2)$}. Each node represents an $(N, S)$ irrep of $U(3) \times U(2)$ and is labeled $(N, S, n_M)$, with the spin multiplicity written as $n_M = 2S + 1$. Each walk from the $(0, 0, 1)$ node to an uppermost node represents a unique $U(3)$ GT basis state, with the nodes traversed indicating the irreps under group subduction which the state transforms according to. The nodes are coloured to match up with those of the left-hand side blocks.}
\label{fig:quga_block_diagonals}
\end{figure}

\newpage

\

\newpage

\section{The Quantum Paldus Transform} \label{sec:paldus_transform}

Over the 50 years since its conception, the unitary group approach has enjoyed great success in CSF calculations, establishing itself as a mature and well-understood tool in computational chemistry. Although past research has utilised concepts from UGA \cite{gandon2024, Dobrautz2019}, up until now it was not known how to fully implement the formalism in a quantum computational setting. Continuing on from our derivations of the mathematical structure of $U(d) \times SU(2)$ branching and Gelfand-Tsetlin states, we now derive an efficient algorithm for the quantum Paldus transform -- a change of basis for re-expressing fermionic occupation number states explicitly in terms of their GT patterns. Our algorithm has wide applicability, from the asymptotically efficient simulation of quantum chemistry Hamiltonians to novel quantum information protocols for noise mitigation. 

The circuit for the quantum Paldus transform is given in Figure \ref{fig:paldus_transform}. It proceeds analogously to the Schur transform, with a cascade of Clebsch-Gordan transforms -- however, our CG transforms differ from those of \cite{BCH_2006} by their ability to couple spin-$0$ subsystems (both empty and paired spatial orbitals), and the resulting need to also keep track of the particle number $N$ of the GT states. Each such CG transform is composed of multi-qubit incrementers and a multiplexed sequence of controlled Givens rotations with Clebsch-Gordan coefficients as the rotation matrix elements. In effect, consecutive CG transforms simultaneously perform Clebsch-Gordan couplings of the input states. We give the algorithm as an abstract set of steps, which overall perform the desired change of basis. Then, we discuss the finer, technical details and give an explicit construction of the circuit on qubits. Finally, we discuss matters of complexity, giving the asymptotic scaling of elementary gates and ancilla qubits required as well as how to implement the controlled Givens rotations present in our algorithm in a fault-tolerant manner. More detailed resource estimates of the fault-tolerant implementation are also provided in Appendix \ref{sec:circuit_compilation}.

\

\begin{figure}[h!]
  \includegraphics[width=17.5cm]{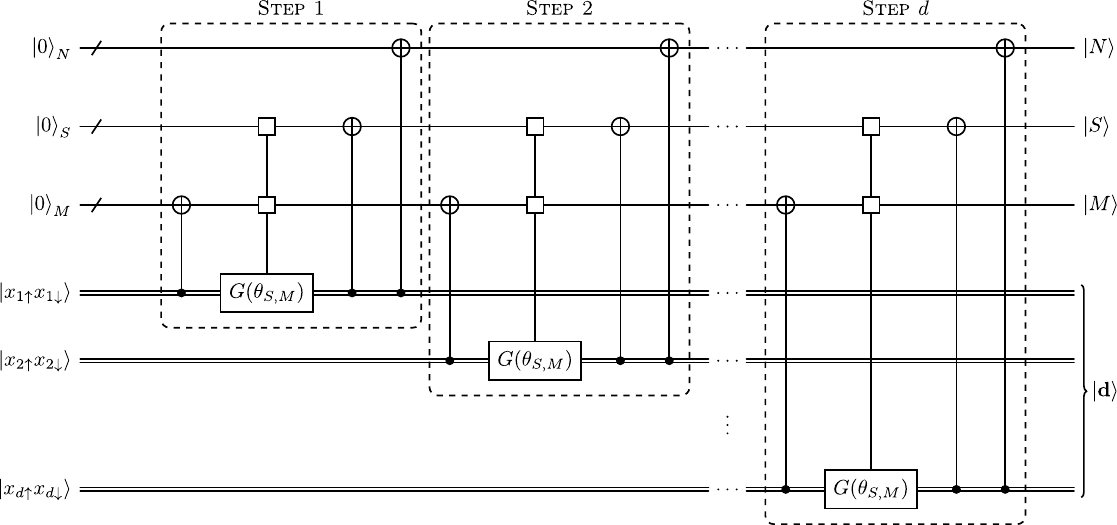}
  \caption{\textit{The Quantum Paldus Transform Circuit.} The Paldus transform is a quantum circuit that maps the occupation number basis to the UGA basis. Each step is a Clebsch-Gordan transform (Figure \ref{fig:cg_transform}) composed of $S, M, N$ incrementers (Figures \ref{fig:incrementer_s},  \ref{fig:incrementer_m}, \ref{fig:incrementer_n}) and a sequence of controlled Givens rotations (Figure \ref{fig:controlled_givens_rotations}).}
  \label{fig:paldus_transform}
\end{figure}

\

In Section \ref{sec:paldus_duality} we have derived and gained understanding of how the fermionic Fock space decomposes under the antisymmetric representation of the group $U(d) \times SU(2)$, giving rise to the isomorphism:

\begin{equation}
  (\mathbb{C}^2)^{\otimes 2d} \cong \bigwedge(\mathbb{C}^d \otimes \mathbb{C}^2) \cong \bigoplus_{\substack{N=0}}^{2d} \bigoplus_{S=0}^{N / 2} W^{U(d)}_{(N,S)} \otimes W^{SU(2)}_S
\end{equation}

\noindent Then, in Section \ref{sec:SU2} we have shown how to construct a basis \textit{adapted} to the irrep decomposition, formed of the Gelfand-Tsetlin states $\ket{N, S, M; \mathbf{d}}$. The GT states have the advantage of holding well-defined values of quantum numbers relevant in the study of molecular systems. They are always expressible as linear combinations of Fock states, with Clebsch-Gordan coefficients as the probability amplitudes. However, even though the quantum numbers of each GT state are well-defined, they are not \textit{locally accessible} through a simple measurement; rather, they are encoded \textit{globally} in the coefficients of the Fock states. From a practical point of view, this is not helpful -- unless we can represent each quantum number $(N, S, M, \mathbf{d})$ on a separate register, having a basis adapted to the GT states is of limited utility. To address this, we define the \textit{UGA basis}, composed of states of the form:

\begin{equation}
  \ket{N}_N \otimes \ket{S}_S \otimes \ket{M}_M \otimes \ket{\mathbf{d}}_{\mathbf{d}}
\end{equation} 

\noindent Where the values of $(N, S, M, \mathbf{d})$ have the same physical meaning as before, but are now stored in separate subsystems, with the $N, S, M$ values represented in binary using $n_N, n_S, n_M$ qubits. The step vectors $\mathbf{d}$ are also represented as $2d$-bit strings, via rules given in Table \ref{fig:step_vectors_binary}. For this and subsequent sections, we will denote $\mathbf{d}$ using bitstrings, with $\mathbf{d}_i$ representing their $i^\text{th}$ bit. The branching rules of Section \ref{sec:shavitt} imply that the relevant quantum numbers satisfy the constraints:

\begin{align}
    &\sum_{i=1}^{2d} \mathbf{d}_i = N \leq 2d, & 
    &\sum_{i=1}^{k } (\mathbf{d}_{2i-1} - \mathbf{d}_{2i}) \geq 0, \ \ k \in \llbracket d-1 \rrbracket, \\
    &\sum_{i=1}^d (\mathbf{d}_{2i-1} - \mathbf{d}_{2i})  = 2S \leq d, &
    &M \in \{-S, -S + 1, \hdots, S -1, S\} \label{eq:ugabasisconstraints}.
\end{align} 

The UGA basis forms the quantum computational analogue of the \textit{electronic Gelfand basis} used in the original unitary group approach. Being defined with respect to the same labels as the Gelfand-Tsetlin states, its basis vectors are a representation of the latter. However, we refrain from using the same name for two reasons: the first is that UGA basis states contain information about the spin projection $M$, whereas in the original unitary group approach $M$ is implicitly taken to be its highest value, i.e. $M=S$. For many applications involving spin-free Hamiltonians this is unproblematic, as such Hamiltonians do not affect the value of $M$ and so it can be set to any arbitrary value. However, when constructing a mapping from the Fock states to the UGA basis having freedom in $M$ ensures that the latter spans a $2^{2d}$-dimensional space. The second is that Gelfand-Tsetlin states are still expressible using the $2d$-qubit computational basis via Equation \eqref{eq:yk_states_formula}, so the difference between the two is how they are \textit{encoded}, rather than their properties. For the case of $d=1$, the Gelfand-Tsetlin states in the occupation number basis and UGA basis are shown in Figure \ref{fig:d1_gt_states}. For $d=2$, they are shown in Figure \ref{fig:d2_gt_states}.

In order to write the UGA basis states, typically additional qubits are needed: $c = n_N + n_S + n_M$ qubits for the $N, S, M$ registers, and $2d$ qubits for the $\mathbf{d}$ register. The upside of this is the UGA states exhibit the properties of the GT states in a far more natural and practically useful manner, owing to the tensor product structure. For example, in the GT basis the antisymmetric representation of group elements $Q(g)\in U(d)$, $R(h) \in U(2)$ reduces to separate action on the $U(d)$ and $U(2)$ irrep labels:

\begin{equation}
  Q(g) \otimes R(h) \left( \ket{N} \otimes \ket{S} \otimes \ket{M} \otimes \ket{\mathbf{d}} \right) = \ket{N} \otimes \ket{S} \otimes \left(  \sum_{M'=-S}^S [\tilde{R}_{N, S}(h)]_{M', M} \ket{M'} \right)\otimes \left(\sum_{\mathbf{d}'} [\tilde{Q}_{N, S}(g)]_{\mathbf{d', d}} \ket{\mathbf{d}'} \right)
\end{equation}

\noindent where $[\tilde{Q}_{N, S}(g)]$ is a $T^d_{N, S} \times T^d_{N, S}$ irreducible representation matrix which only acts on $\ket{\mathbf{d}}$ states belonging to a $(N, S)$ subspace, and $[\tilde{R}_{N, S}(h)]$ is a $(2S+1) \times (2S+1)$ irrep matrix which only acts on $\ket{M}$ states. The group action of $U(d) \times U(2)$ is thus manifestly separable, acting on each corresponding subsystem independently. 

\begin{definition}[Quantum Paldus Transform]
The quantum Paldus transform is an \textit{isometry} mapping the Gelfand-Tsetlin states of the antisymmetric representation of $U(d) \times U(2)$ from the occupation basis to the UGA basis. Written as an unitary $U_P$ in the enlarged Hilbert space of $n_N + n_S + n_M + 2d$ qubits, it acts as:

\begin{equation}
  U_P \left(  \ket{0}^{\otimes (n_N + n_S + n_M)} \otimes \ket{N, S, M; \mathbf{d}} \right) = |N \rangle_N \otimes |S \rangle_S \otimes |M \rangle_M \otimes |\mathbf{d} \rangle_\mathbf{d}.
\end{equation}

\noindent The output UGA basis states are mutually orthogonal, and span a $2^{2d}$-dimensional subspace of the $(c + 2d)$ qubits which $U_P$ acts on. 
\end{definition}

One can verify (Appendix \ref{app:dimension_formula}) that the space of \textit{valid} GT states is $2^{2d}$-dimensional, so this isometry is indeed possible. The algorithm for $U_P$ transfers information about $N, S, M$ onto the ancillary registers, while the original $2d$ qubits of the input state become a carrier for the step vectors $\mathbf{d}$. Thus, any quantum computational task on the original Fock space of $2d$ qubits can be accomplished in the UGA basis through our algorithm.

The circuit for $U_P$ is shown in Figure \ref{fig:paldus_transform}. It consists of a cascade of Clebsch-Gordan transforms, each of which takes in two qubits from the $\mathbf{d}$ register. The Clebsch-Gordan transforms are controlled by the $N, S, M$ registers, which are updated accordingly to simulated the sequential branching illustated in Figure~\ref{fgr:mz_matrix_arcs}. The overall circuit is then constructed with a complexity of $\mathcal{O}(d^3)$ in terms of Toffoli gates.

\begin{figure}[htb!]
  \begin{align}
    \ket{0}_N \otimes \ket{0}_S \otimes \ket{0}_M \otimes \ket{\mathbf{00}}_\mathbf{d} &\leftrightarrow \ket{00} &
    \ket{1}_N \otimes \ket{1/2}_S \otimes \ket{1/2}_M \otimes \ket{\mathbf{10}}_\mathbf{d} &\leftrightarrow \ket{10} \\
     \ket{2}_N \otimes \ket{0}_S \otimes \ket{0}_M \otimes \ket{\mathbf{11}}_\mathbf{d} &\leftrightarrow \ket{11}  &
     \ket{1}_N \otimes \ket{1/2}_S \otimes \ket{-1/2}_M \otimes \ket{\mathbf{10}}_\mathbf{d} &\leftrightarrow \ket{01} 
  \end{align}
  \caption{\textit{GT states for $d=1$}. The same GT states are written out in the UGA basis (left) and occupation basis (right).}
  \label{fig:d1_gt_states}
\end{figure}

\begin{figure}[h!]
  \begin{align*}
      \ket{0}_N \otimes \ket{0}_S \otimes \ket{0}_M \otimes \ket{\mathbf{00, 00}}_\mathbf{d} &\leftrightarrow \ket{0000} & \ket{4}_N \otimes \ket{0}_S \otimes \ket{0}_M \otimes \ket{\mathbf{11, 11}}_\mathbf{d} &\leftrightarrow \ket{1111}  \\
      & & & \\
      \ket{1}_N \otimes \ket{1/2}_S \otimes \ket{1/2}_M \otimes \ket{\mathbf{00, 10}}_\mathbf{d} &\leftrightarrow \ket{0010} & \ket{3}_N \otimes \ket{1/2}_S \otimes \ket{-1/2}_M \otimes \ket{\mathbf{11, 10}}_\mathbf{d} &\leftrightarrow \ket{1101} \\
      \ket{1}_N \otimes \ket{1/2}_S \otimes \ket{1/2}_M \otimes \ket{\mathbf{10, 00}}_\mathbf{d} &\leftrightarrow \ket{1000} & \ket{3}_N \otimes \ket{1/2}_S \otimes \ket{-1/2}_M \otimes \ket{\mathbf{10, 11}}_\mathbf{d} &\leftrightarrow \ket{0111} \\
      & & & \\
      \ket{1}_N \otimes \ket{1/2}_S \otimes \ket{-1/2}_M \otimes \ket{\mathbf{00, 10}}_\mathbf{d} &\leftrightarrow \ket{0001} & \ket{3}_N \otimes \ket{1/2}_S \otimes \ket{1/2}_M \otimes \ket{\mathbf{11, 10}}_\mathbf{d} &\leftrightarrow \ket{1110} \\
      \ket{1}_N \otimes \ket{1/2}_S \otimes \ket{-1/2}_M \otimes \ket{\mathbf{10, 00}}_\mathbf{d} &\leftrightarrow \ket{0100} & \ket{3}_N \otimes \ket{1/2}_S \otimes \ket{1/2}_M \otimes \ket{\mathbf{10, 11}}_\mathbf{d} &\leftrightarrow \ket{1011} \\
      & & & \\
      & & \ket{2}_N \otimes \ket{1}_S \otimes \ket{1}_M \otimes \ket{\mathbf{10, 10}}_\mathbf{d} &\leftrightarrow \ket{1010} \\
      \ket{2}_N \otimes \ket{0}_S \otimes \ket{0}_M \otimes \ket{\mathbf{00, 11}}_\mathbf{d} &\leftrightarrow \ket{0011}  & &  \\
       \ket{2}_N \otimes \ket{0}_S \otimes \ket{0}_M \otimes \ket{\mathbf{10, 01}}_\mathbf{d} &\leftrightarrow \frac{1}{\sqrt{2}}(\ket{1001} - \ket{0110}) & \ket{2}_N \otimes \ket{1}_S \otimes \ket{0}_M \otimes \ket{\mathbf{10, 10}}_\mathbf{d} &\leftrightarrow \frac{1}{\sqrt{2}}(\ket{1001} + \ket{0110}) \\
       \ket{2}_N \otimes \ket{0}_S \otimes \ket{0}_M \otimes \ket{\mathbf{11, 00}}_\mathbf{d} &\leftrightarrow \ket{1100} & & \\
      & & \ket{2}_N \otimes \ket{1}_S \otimes \ket{-1}_M \otimes \ket{\mathbf{10, 10}}_\mathbf{d} &\leftrightarrow \ket{0101}
    \end{align*}
  \caption{\textit{GT states for $d=2$}. GT states are written out in the UGA basis (left) and occupation basis (right). They have been grouped into irreps of $U(d)$, spanned by distinct step vectors $\mathbf{d}$ with common values of $S, N$. Each appears with a multiplicity of $2S + 1$.}
  \label{fig:d2_gt_states}
  \end{figure}

\subsection{Clebsch-Gordan Transforms for the Paldus Transform} \label{sec:cgtransform}

Our algorithm for the Paldus transform algorithm is expressible as a sequence of \textit{Clebsch-Gordan transforms}. This approach mirrors that of the Schur transform by Bacon, Chuang, and Harrow \cite{BCH_2005, BCH_2006}, and works due to the structural similarities between the Schur states (Equation \eqref{eq:schur_states_formula}) and the GT basis in the Jordan-Wigner representation, as observed in Section \ref{sec:gt_states_construction}. We can take advantage of these similarities and accordingly modify the Schur transform to obtain the Paldus transform. As noted in Section \ref{sec:spin_coupling_rotations}, the GT states contain products of CG coefficients as their probability amplitudes, which come about from a sequence of controlled rotations, each of which corresponds to a CG coupling (Figure \ref{fgr:mz_matrix_arcs}). We now show how to implement such couplings algorithmically. 

The idea behind the Clebsch-Gordan transform is to take in two states -- belonging to separate subsystems with quantum numbers $(N, S, M')$ and $(n_i, s_i, m_i)$ -- and decompose them into a superposition of states inhabiting a larger system with quantum numbers $(N', S', M)$, completely akin to how one would approach the task of coupling angular momenta in textbook quantum mechanics. If $S$ and $s_i$ are arbitrary, this is straightforward to write out on paper, but complicated to implement algorithmically on a quantum computer. However, if we couple orbitals one at a time so that $s_i \in \{0, \frac{1}{2}\}$, then the corresponding decomposition is much simpler. This mirrors our interpretation of the \textit{step vectors} (shown in Figure~\ref{fgr:mz_matrix_arcs}), which as we have seen correspond to a sequential coupling procedure. Thus, we seek a quantum algorithm which takes in states of the form $\ket{N} \ket{S} \ket{M'} \ket{\mathbf{d}} \otimes \ket{x_{i \uparrow} x_{i \downarrow}}$ and outputs a superposition of states $\ket{N'} \ket{S'} \ket{M} \ket{\mathbf{d}'}$, where $\mathbf{d}'$ is the original step vector appended with two bits from $\{\mathbf{00, 10, 01, 11}\}$ at the end, to indicate one of the four types of coupling which has taken place. We therefore define the CG transform to be a unitary acting on $n_N + n_S + n_M + 2$ qubits which performs the following mapping:

\begin{align}
  \big( \ket{N}_N \otimes \ket{S}_S \otimes \ket{M}_{M'} \otimes \ket{\mathbf{d}}_\mathbf{d} \big) \otimes \ket{00} &\longrightarrow \ket{N}_{N'} \otimes \ket{S}_{S'} \otimes \ket{M}_M \otimes \ket{\mathbf{d, 00}}_{\mathbf{d}'} \label{eq:cg_transform_rules}  \\ 
  \big( \ket{N}_N \otimes \ket{S}_S \otimes \ket{M + 1/2}_{M'} \otimes \ket{\mathbf{d}}_\mathbf{d} \big) \otimes \ket{01} &\longrightarrow  C^{S - \frac{1}{2}, M }_{S, M + \frac{1}{2}; \frac{1}{2}, -\frac{1}{2} }\ket{N+1}_{N'} \otimes \ket{S - 1/2 }_{S'} \otimes \ket{M }_M \otimes \ket{\mathbf{d, 01}}_{\mathbf{d}'}  \\ 
  &+  C^{S + \frac{1}{2} , M  }_{S,  M + \frac{1}{2}; \frac{1}{2}, -\frac{1}{2} }\ket{N+1}_{N'} \otimes \ket{S + 1/2 }_{S'} \otimes \ket{M }_M \otimes \ket{\mathbf{d, 10}}_{\mathbf{d}'} \\
  \big( \ket{N}_N \otimes \ket{S}_S \otimes \ket{M - 1/2}_{M'} \otimes \ket{\mathbf{d}}_\mathbf{d} \big) \otimes \ket{10} &\longrightarrow   C^{S - \frac{1}{2} , M }_{S, M -\frac{1}{2}; \frac{1}{2}, +\frac{1}{2} }\ket{N+1}_{N'} \otimes \ket{S - 1/2 }_{S'} \otimes \ket{M }_M \otimes \ket{\mathbf{d, 01}}_{\mathbf{d}'}\\ 
  &+  C^{S + \frac{1}{2} , M }_{S, M - \frac{1}{2}; \frac{1}{2}, +\frac{1}{2} }\ket{N+1}_{N'} \otimes \ket{S + 1/2 }_{S'} \otimes \ket{M  }_M \otimes \ket{\mathbf{d, 10}}_{\mathbf{d}'} \\
  \big( \ket{N}_N \otimes \ket{S}_S \otimes \ket{M}_{M'} \otimes \ket{\mathbf{d}}_\mathbf{d} \big) \otimes \ket{11} &\longrightarrow  \ket{N+2}_{N'} \otimes \ket{S}_{S'} \otimes \ket{M}_M \otimes \ket{\mathbf{d, 11}}_{\mathbf{d}'}
\end{align}

\noindent To correctly apply the Clebsch-Gordan transform, we require a quantum circuit which implements Equation \eqref{eq:cg_transform_rules} for all possible input values of $(N, S, M)$ and $(x_{i \uparrow}, x_{i \downarrow})$ on the left-hand side. To see how this can be broken down into simple steps, we begin with a few observations. The first is that the step vectors $\mathbf{d}$ of the larger input states bear no effect on the mapping, and aside from concatenation of bits at their end remain unchanged. This means that the circuit for the CG transform does not act on the $\mathbf{d}$ register. The second is that we have chosen to index the couplings by the \textit{incoming} values of $S$ (giving outgoing spin values $S' \in \{S, S \pm \frac{1}{2} \}$) and \textit{outgoing} values of $M$ (resulting from incoming spin projections $M' \in \{M, M \pm \frac{1}{2} \}$). With this choice, the Clebsch-Gordan coefficients appearing in the expressions form the elements of $2 \times 2$ rotation matrices $R(\theta_{S, M})$, given by Equation \eqref{eq:CGMatrix} of Section \ref{sec:spin_coupling_rotations}. 

From these observations, we can see that the Clebsch-Gordan transform consists of updates on the $(N, S, M)$ registers and rotations of the $\ket{x_{i \uparrow} x_{i \downarrow}}$ state. After updating $M'$ to its outgoing value $M$, there is a sequence of controlled rotations, followed by an update of the $N, S$ registers with an addition and subtraction controlled on the final two digits of $\mathbf{d}$. The circuit for the CG transform, implementing the mapping defined in Equation \eqref{eq:cg_transform_rules} is illustrated in Figure \ref{fig:cg_transform}. This performs the spin coupling shown in Figure \ref{fgr:mz_matrix_arcs}, and is composed of four primary steps:

\begin{enumerate}
  \item \textbf{Spin Projection ($M$) Incrementer}: The $M$ register of the incoming state is incremented by $\pm \frac{1}{2}$ or $0$, depending on the spin projection of the coupled $\ket{x_{i \uparrow} x_{i \downarrow}}$ state. Physically, this corresponds to conservation of angular momentum projections in Clebsch-Gordan coupling.
  \item \textbf{Controlled Givens Rotations}: A sequence of multiplexed Givens rotations $G(\theta_{S, M})$, controlled on the $S$ and $M$ labels, is applied to the $\ket{x_{i \uparrow} x_{i \downarrow}}$ state. This corresponds to the multiplication by CG coefficients in Clebsch-Gordan coupling. The output forms the $2$-qubit step vector element $\ket{\mathbf{d}_{2i-1} \mathbf{d}_{2i}}$ of the resulting GT state. 
  \item \textbf{Spin ($S$) Incrementer}: The spin register $S$ of the GT state is incremented by $\pm \frac{1}{2}$ or $0$, depending on the coupling type which took place in Step $3$. This `updates' the global spin quantum number of the GT state to its correct value. 
  \item \textbf{Particle Number ($N$) Incrementer}: The particle number register $N$ of the GT state is incremented by $\{0, 1, 2\} $, depending on the coupling type which took place in Step $3$. This ensures that the $N$ label of the output GT state reflects the correct number of particles in the system. 
\end{enumerate}

\noindent Each of these steps will be explained in detail in the following subsections.

\begin{figure*}[htbp!]
  \includegraphics[width=9cm]{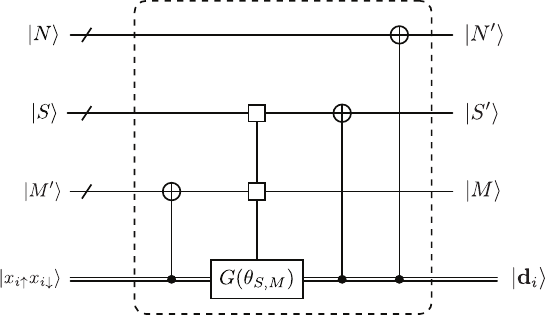}
  \caption{\textit{The Clebsch-Gordan Transform Circuit.} The circuit is composed of an $M$ incrementers (Figure \ref{fig:incrementer_m}), a sequence of controlled Givens rotations (Figure \ref{fig:controlled_givens_rotations}), and $N, S$ incrementers (Figures \ref{fig:incrementer_s}, \ref{fig:incrementer_n}).}
  \label{fig:cg_transform}
\end{figure*}

\subsubsection{Incrementing the Spin Projection Register (\texorpdfstring{$M$}{M})}

The rules for updating the $M$ register are simple: it is raised by $+\frac{1}{2}$ if $x_{i \uparrow} x_{i \downarrow} = 10$, lowered by $-\frac{1}{2}$ if $x_{i \uparrow} x_{i \downarrow} = 01$, and unchanged otherwise. In this step, the incoming values $M'$ of the spin projection register are updated to their outgoing values $M$ :

\begin{align}
  \left( \ket{N}_N \otimes \ket{S}_S \otimes \ket{M}_{M'} \otimes \ket{\mathbf{d}}_\mathbf{d} \right) \otimes \ket{00} &\longrightarrow \ket{N}_N \otimes \ket{S}_S \otimes \ket{M}_M \otimes \ket{\mathbf{d}}_\mathbf{d} \otimes \ket{00} \\
  \left( \ket{N}_N \otimes \ket{S}_S \otimes \ket{M +1/2}_{M'} \otimes \ket{\mathbf{d}}_\mathbf{d} \right) \otimes \ket{01} &\longrightarrow \ket{N}_N \otimes \ket{S}_S \otimes \ket{M}_M \otimes \ket{\mathbf{d}}_\mathbf{d} \otimes \ket{01} \\
  \left( \ket{N}_N \otimes \ket{S}_S \otimes \ket{M-1/2}_{M'} \otimes \ket{\mathbf{d}}_\mathbf{d} \right) \otimes \ket{10} &\longrightarrow \ket{N}_N \otimes \ket{S}_S \otimes \ket{M}_M \otimes \ket{\mathbf{d}}_\mathbf{d} \otimes \ket{10} \\
  \left( \ket{N}_N \otimes \ket{S}_S \otimes \ket{M}_{M'} \otimes \ket{\mathbf{d}}_\mathbf{d} \right) \otimes \ket{11} &\longrightarrow \ket{N}_N \otimes \ket{S}_S \otimes \ket{M}_M \otimes \ket{\mathbf{d}}_\mathbf{d} \otimes \ket{11}
\end{align}

The $M$ register stores the $SU(2)$ GT pattern (i.e. spin projection), or equivalently indexes the multiplicity of each $U(d)$ irrep. The value of $2M$ is a signed integer ranging from $-2S \leq 2M \leq 2S$. As $2S$ can take all possible integer values between $0$ and $d$, there are $(2d+1)$ possible values of $2M$ and so this register consists of $n_M = \lceil \log_2(2d + 1) \rceil$ qubits. As the values of $2M$ can be negative, we use a two's complement scheme to store $2M$ in binary (i.e, the leftmost bit controls the sign of the integer, and the remaining bits its magnitude). In step $1$ of each Clebsch-Gordan transform, $2M$ can increase or decrease by $\pm 1$ or remain unchanged. To achieve this, we use controlled increment and decrement gates, where a decrement gate is the inverse of an increment gate. This is shown in Figure \ref{fig:incrementer_m}.

\begin{figure*}[htbp!]
  \begin{equation}
  INC^{M}_i (\ket{2M}_{M'} \otimes \ket{x_{i, \uparrow} x_{i, \downarrow}}) = \ket{2M + x_{i, \uparrow} - x_{i, \downarrow}}_M \otimes \ket{x_{i, \uparrow} x_{i, \downarrow}}, \ \ \ 1 \leq i \leq d, \ \ \ x_{i, \mu} \in \{0, 1\}
\end{equation}
  \includegraphics[width=16cm]{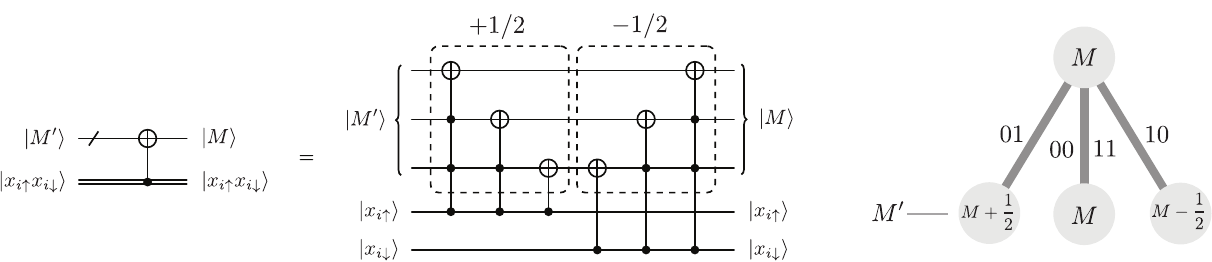}
  \caption{\textit{Left:} Incrementer circuit for updating the $M$ register in the Paldus transform. $M$ is stored as a signed integer $2M$ through the two's complement scheme. \textit{Right:} Incrementation represented as branching on the $SU(2)$ spin projection graph, where the arcs are the cases of $|x_{i\uparrow} x_{i\downarrow} \rangle$} 
  \label{fig:incrementer_m}
\end{figure*}

\subsubsection{Controlled Givens Rotations}

The second step of the Clebsch-Gordan transform is a sequence of rotations in the $\{ \ket{10}, \ket{01} \}$ subspace of the last two qubits, with the rotation angles $\theta_{S, M}$ only dependent on the incoming values $S$ and outgoing values $M$. To impose this, we can control each rotation on the registers $\ket{S}_{S} \ket{M}_M$, with the $\ket{x_{i \uparrow} x_{i \downarrow}}$ register as the target of $G(\theta_{S, M})$. This is written as:

\begin{align}
  \ket{N}_N \otimes \ket{S}_S \otimes \ket{M}_M \otimes \ket{\mathbf{d}}_\mathbf{d}\otimes \ket{00} &\longrightarrow \ket{N}_N \otimes \ket{S}_{S} \otimes \ket{M}_M \otimes \ket{\mathbf{d, 00}}_{\mathbf{d}'} \label{eq:cg_step_2} \\
  \ket{N}_N \otimes \ket{S}_S \otimes \ket{M}_M \otimes \ket{\mathbf{d}}_\mathbf{d} \otimes \ket{01} &\longrightarrow \cos(\theta_{S, M})  \ket{N}_N \otimes \ket{S}_S \otimes \ket{M}_M \otimes \ket{\mathbf{d, 01}}_{\mathbf{d}'}   \\
  &+ \sin(\theta_{S, M}) \ket{N}_N \otimes \ket{S}_S \otimes \ket{M }_M \otimes \ket{\mathbf{d, 10}}_{\mathbf{d}'} \\
  \ket{N}_N \otimes \ket{S}_S \otimes \ket{M}_M \otimes \ket{\mathbf{d}}_\mathbf{d} \otimes \ket{10} &\longrightarrow  -\sin(\theta_{S, M}) \ket{N}_N \otimes \ket{S}_S \otimes \ket{M}_M \otimes \ket{\mathbf{d, 01}}_{\mathbf{d}'} \\
  & + \cos(\theta_{S, M}) \ket{N}_N \otimes \ket{S}_S \otimes \ket{M}_M \otimes \ket{\mathbf{d, 10}}_{\mathbf{d}'} \\
  \ket{N}_N \otimes \ket{S}_S \otimes \ket{M}_M \otimes \ket{\mathbf{d}}_\mathbf{d} \otimes \ket{11} &\longrightarrow \ket{N}_N \otimes \ket{S}_S \otimes \ket{M}_M \otimes \ket{\mathbf{d, 11}}_{\mathbf{d}'}
\end{align}

\noindent This can be achieved by performing a sequence of \textit{controlled Givens rotations}, illustated in Figure \ref{fig:controlled_givens_rotations}. Each such rotation corresponds to up to two couplings (rows of Equation \eqref{eq:cg_transform_rules}), and is controlled by the values of $S$ and $M$. That is, the controlled rotations of step $2$ take the form:

\begin{equation}
  G(\theta_{S, M})  = I_N \otimes \ket{S} \bra{S} \otimes \ket{M} \bra{M} \otimes I_\mathbf{d} \otimes \begin{pmatrix} 1 & 0 & 0 & 0 \\ 0 & \sqrt{\frac{S+M + 1/2}{2S + 1}} & \sqrt{\frac{S - M + 1/2}{2S + 1}} & 0 \\ 0 & -\sqrt{\frac{S - M + 1/2}{2S + 1}} & \sqrt{\frac{S+M + 1/2}{2S + 1}} & 0 \\ 0 & 0 & 0 & 1 \end{pmatrix} 
  \end{equation}

\noindent The corresponding sequence of rotations (each controlled by a possible combination of $S, M$ up to in the $i^\text{th}$ Clebsch-Gordan transform) multiplies the input states by the correct CG coefficients present in Equation \eqref{eq:cg_transform_rules}, and outputs the \textit{step vector elements} $\mathbf{d}_{2i-1} \mathbf{d}_{2i}$ in place of the input occupation number register. These step vector elements are appended at the end of the incoming step vectors $\mathbf{d}$, and denote the four different types of coupling which have taken place. They can then be used to update the $S$ and $N$ registers to their appropriate values (c.f. Table \ref{fig:step_vectors_binary}) in steps $3$ and $4$. The gate $G(\theta_{S, M})$ is illustrated in Figure \ref{fig:controlled_givens_rotations}.

  \begin{figure*}[htb!]
  \includegraphics[width=12cm]{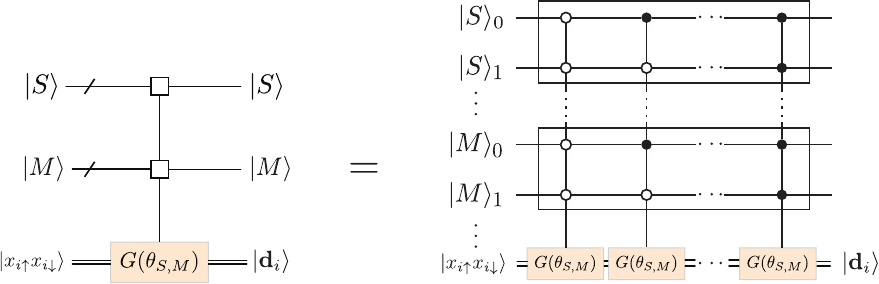}
  \caption{\textit{Controlled Givens Rotations Circuit.} The controlled Givens rotation $G(\theta_{S, M})$ is controlled on the $S, M$ registers and acts on the $\ket{x_{i \uparrow} x_{i \downarrow}}$ state. The rotation angle $\theta_{S, M}$ is determined by the incoming value of $S$ and outgoing value of $M$. The square box represents a control across all binary indices on the control register.}
  \label{fig:controlled_givens_rotations}
\end{figure*}

If the values of $S, M$ are represented in binary, each $2 \times 2$ rotation is controlled on a computational basis state on at most $2 + n_S + n_M = \mathcal{O}(\log (d))$ qubits. A simple argument (given in \cite{Nielsen2010}, page 189) tells us that each individual rotation in the sequence is efficiently implementable and can be decomposed into $\mathcal{O}(\text{poly}(\log(d)))$ $2 \times 2$ unitary gates. 

To give an exact number of controlled Givens rotations needed for the full Paldus transform on $d$ orbitals, we number each Clebsch-Gordan transform as $i = 1, 2, ..., d$, and count the number of controlled rotations required at the $i^\text{th}$ step. Fixing an incoming value of $S$, the possible incoming values of $M'$ are $M' \in \{ S, S - 1, \hdots, -S \}$. After the initial incrementation step, the possible outgoing values of $M$ are $M \in \{S + \frac{1}{2}, S - \frac{1}{2}, S - \frac{3}{2}, \hdots, - S - \frac{1}{2} \}$. This means there are $2S + 2$ distinct values of $(S, M)$ to control on. However, whenever $M = S + \frac{1}{2}$ the rotation $G(\theta_{S, S + \frac{1}{2}}) = I$ is trivial and does not contribute to runtime. Therefore, each possible incoming value of $S$ requires $2S+1$ controlled rotations. At the $i^\text{th}$ step, the possible incoming values of $S$ are $\{ 0, \frac{1}{2}, 1, ..., (i-1)/2\}$, so the number of controlled Givens rotations one needs to implement is given by $\sum_{j=1}^i (2 \times (j - 1)/2 + 1) = i(i+1)/2$. To count their overall number for $U_P$, we sum over these for $1 \leq i \leq d$, giving $d(d+1)(d+2)/6 = \mathcal{O}(d^3)$ controlled Givens rotations. Thus we see that their count scales as $\mathcal{O}(d^3)$, which is polynomial in the number of qubits. This contributes the most to the runtime, as the incrementers used in the other steps of the Clebsch-Gordan transforms always act on $\mathcal{O}(\log d)$ qubits, and can be decomposed into $\mathcal{O}(\log d)$ elementary gates (See Figure \ref{fig:incrementer_circuit}). Thus, our implementation of the Paldus transform is efficient, requiring $\mathcal{O}(d^3 \cdot \text{poly}(\log(d)))$ two-qubit gates.

Finally, we note that although the implementation presented in this section is $\mathcal{O}(d^3 \cdot \text{poly}(\log(d)))$, the Paldus transform also admits an \textit{approximate} implementation up to error $\varepsilon$ in $\mathcal{O}(d \cdot \text{poly}(\log(d), \log(1/\varepsilon)) )$ elementary gates, which is asymptotically more efficient. In fact, this implementation can be achieved with the exact same Clebsch-Gordan circuit $U_{CG}$ used in the Schur transform, so long as $U_{CG}$ can be implemented as a controlled unitary. We discuss this in Section \ref{sec:paldus_cg}.

\subsubsection{Incrementing the Spin Register (\texorpdfstring{$S$}{S})}

Following the controlled rotations of step $2$, the next step is to update the $S$ register to its appropriate outgoing value. This means incrementing $S$ by $\pm \frac{1}{2}$ whenever the final two bits of $\mathbf{d}$ read $\mathbf{10}$ or $\mathbf{01}$. The $S$ register stores twice the value of the total spin of an irrep in binary, i.e. the positive integer $2S$ which changes by $\pm 1$ over the course of each Clebsch-Gordan transform. Because $2S$ is an integer ranging from $0 \leq 2S \leq d$, we choose to store it in binary, using $n_S = \lceil \log_2(d + 1) \rceil$ qubits. For example $|0\cdots00\rangle_S = |0\rangle_S$, $|0\cdots01\rangle_S = |\frac{1}{2}\rangle_S$, and $|0\cdots10\rangle_S = |1\rangle_S$.

To correctly update $2S$ according to the final two bits of $\mathbf{d}$, we use controlled increment and decrement gates acting on the penultimate and final qubit in each CG transform, shown in Figure \ref{fig:incrementer_s}. The $S$ incrementer acts as:

\begin{equation}
  INC^{S}_i (\ket{2S}_S \otimes \ket{\mathbf{d}_{2i - 1} \mathbf{d}_{2i}}) = \ket{2S + \mathbf{d}_{2i-1} - \mathbf{d}_{2i}}_{S'} \otimes \ket{\mathbf{d}_{2i - 1} \mathbf{d}_{2i}}, \ \ \ 1 \leq i \leq d, \ \ \  \mathbf{d}_{2i - 1}, \mathbf{d}_{2i} \in \{\mathbf{0, 1}\}
\end{equation}

\begin{figure*}[htb!]
  \includegraphics[width=16cm]{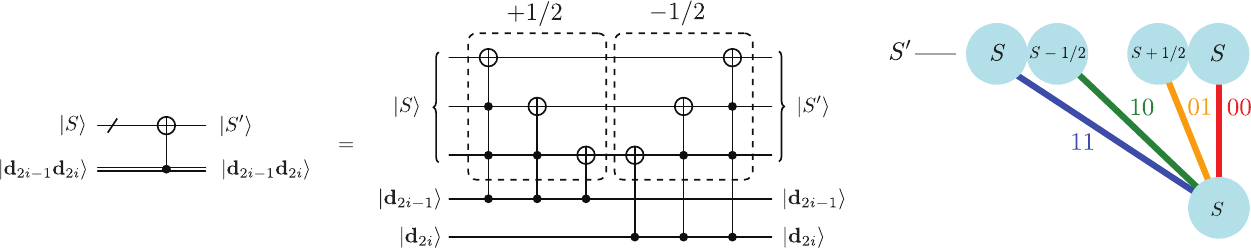}
  \caption{\textit{Left:} Incrementer circuit for updating the $S$ register in the Paldus transform. $S$ is stored as a binary integer $2S$. \textit{Right:} The incrementation represented as branching on the Shavitt graph changing the $S$ values, where the arcs are the cases of $\mathbf{d}_{2i - 1} \mathbf{d}_{2i}$.}
  \label{fig:incrementer_s}
\end{figure*}

If the final two qubits (after the controlled rotation of step $2$) read $\ket{ \mathbf{10}}$, for example, an increment gate is triggered by the first of the two and takes $2S \rightarrow 2S + 1$. On the other hand, $\ket{ \mathbf{11}}$ will activate both the incrementer and decrementer, performing the desired action of leaving $2S$ unchanged. The controlled increment step is necessary for our implementation of the Paldus transform and cannot be delegated elsewhere in the circuit, however the $S$ register can optionally be `decoupled' at the end of the algorithm, by performing every $INC^{S}_i$ gate in reverse. By the end of the decoupling procedure, $\ket{S}_S \rightarrow \ket{0}_S$. This is possible because information about $S, N$ is implicit in $\mathbf{d}$, and allows for an alternative `compressed' representation of the UGA basis with fewer qubits (as states of the form $\ket{M} \otimes \ket{\mathbf{d}}$), at the cost of losing the ability to measure or control on $\ket{S}$.

\subsubsection{Incrementing the Particle Number Register (\texorpdfstring{$N$}{N})} \label{sec:encoding_n}

The final, fourth step of the Clebsch-Gordan transform consists of updating the particle number register ($N$) to its correct value, where the number of particles `added on' to the system is implicit in the final two digits of the step vector $\mathbf{d}$. For the GT states on $2d$ spin-orbitals, the particle number can range as $0 \leq N \leq 2d$, and in the occupation basis any GT basis vector $\ket{N, S, M; \mathbf{d}} = \sum_\mathbf{x} c_\mathbf{x} \ket{x_{1 \uparrow} x_{1 \downarrow} \cdots  x_{d \uparrow} x_{d \downarrow} }$ satisfies $N = |\mathbf{x}|$, where $|\mathbf{x}|$ is the Hamming weight of the bitstring $\mathbf{x}$. Similarly, in the UGA basis the corresponding vector $\ket{N} \ket{S} \ket{M} \ket{\mathbf{d}}$ satisfies $N = |\mathbf{d}|$. In fact, the action of updating the $N$ register commutes with every step of the Paldus transform, since none of the steps explicitly depend on the particle number. In some applications, having access to the $N$ register is also not required -- therefore, one can treat the $N$ update as an optional step in the algorithm.

In order to store $N$, we again choose binary representation on $n_N = \lceil \log_2(2d + 1) \rceil$ qubits (with the least significant bit on the right, i.e. $\ket{0}_N = \ket{00...0}$, $\ket{1}_N = \ket{00...1}$, and so on). It is updated by a simple sequence of $2d$ controlled increment gates, which are a cascade of multi-controlled $CNOT$ gates shown in Figure \ref{fig:incrementer_n}. Each controlled incrementer adds the $i^\text{th}$ bit of $\mathbf{d}$ to the binary digit in $N$.

\begin{equation}
INC^{N}_i (|N\rangle_N \otimes |\mathbf{d}_{2i-1} \mathbf{d}_{2i} \rangle) = |N + \mathbf{d}_{2i-1} +  \mathbf{d}_{2i} \rangle_{N'} \otimes |\mathbf{d}_{2i-1} \mathbf{d}_{2i} \rangle, \ \ \ 1 \leq i \leq d, \ \ \  \mathbf{d}_{2i - 1}, \mathbf{d}_{2i} \in \{\mathbf{0, 1}\}
\end{equation}

\begin{figure*}[htb!]
  \includegraphics[width=16cm]{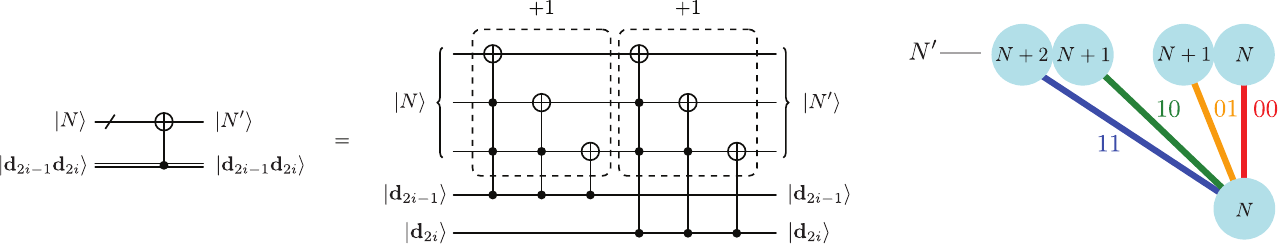}
  \caption{\textit{Left:} Incrementers for updating the $N$ register in the Paldus transform. \textit{Right:} The incrementation represented as branching on the Shavitt graph changing the $N$ values, where the arcs are the cases of $\mathbf{d}_{2i - 1} \mathbf{d}_{2i}$.}
  \label{fig:incrementer_n}
\end{figure*}

\subsection{Fault-Tolerant Controlled Givens Rotations}

In a fault-tolerant setting the implementation of small continuous-angle rotations directly is costly, as their discretisation incurs significant overheads in terms of T gates due to the need for magic state distillation~\cite{Litinski2019magicstate} or repeat-until-success circuits~\cite{Kliuchnikov2023shorterquantum}. Instead, the Givens rotations (shown in Figure~\ref{fig:givens_rotation_circuit}) are implemented using fault-tolerant circuit primitives, with various choices of data lookup oracles (introduced in Appendix~\ref{sec:data_lookup}), adders (Appendix~\ref{sec:adder}), and phase gradients (as shown in Figure~\ref{fig:phase_gradient}) being leveraged. An example of a simpler controlled $R_y$ rotation using these primitives is also presented in Appendix~\ref{sec:rotations}.

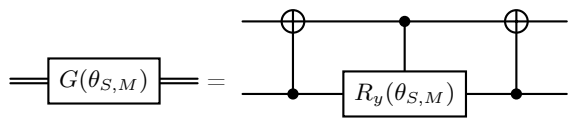
\begin{figure}
  \centering
  \begin{quantikz}[wire types={c}]
     & \gate{G(\theta_{S,M})} & \\
  \end{quantikz}
  =
  \begin{quantikz}
    & \targ{}   & \ctrl{1}                 & \targ{} &  \\
    & \ctrl{-1} & \gate{R_y(\theta_{S,M})} & \ctrl{-1} & \\
  \end{quantikz}
  \caption{\textit{Givens Rotation Circuit.} The Givens rotation $G(\theta_{S,M})$ is implemented using a controlled $R_y(\theta_{S,M})$ rotation and two $CX$ gates.}
    \label{fig:givens_rotation_circuit}
\end{figure}

In this approach, data lookup is used to implement increments of the controlled $G(\theta_{S,M})$ rotation in parallel, where the increments are stored as fixed-point binary fractions in the lookup data indexed on the $S$ and $M$ registers. The controlled $R_y$ rotation in the Givens rotation then activates depending on whether the data increment bit is $0$ or $1$ via a controlled operation. The error in the rotation discretisation decreases with the number of data qubits $q$ by $q = \mathcal{O}(\log(1/\epsilon))$. The data lookup oracle approach is shown in Figure~\ref{fig:data_lookup_rotations}. Importantly, this method allows for the implementation of the controlled Givens rotations in parallel for each $S,M$ index, which is a significant advantage over doing so sequentially.

\begin{figure*}[htbp!]
  \includegraphics[width=\textwidth]{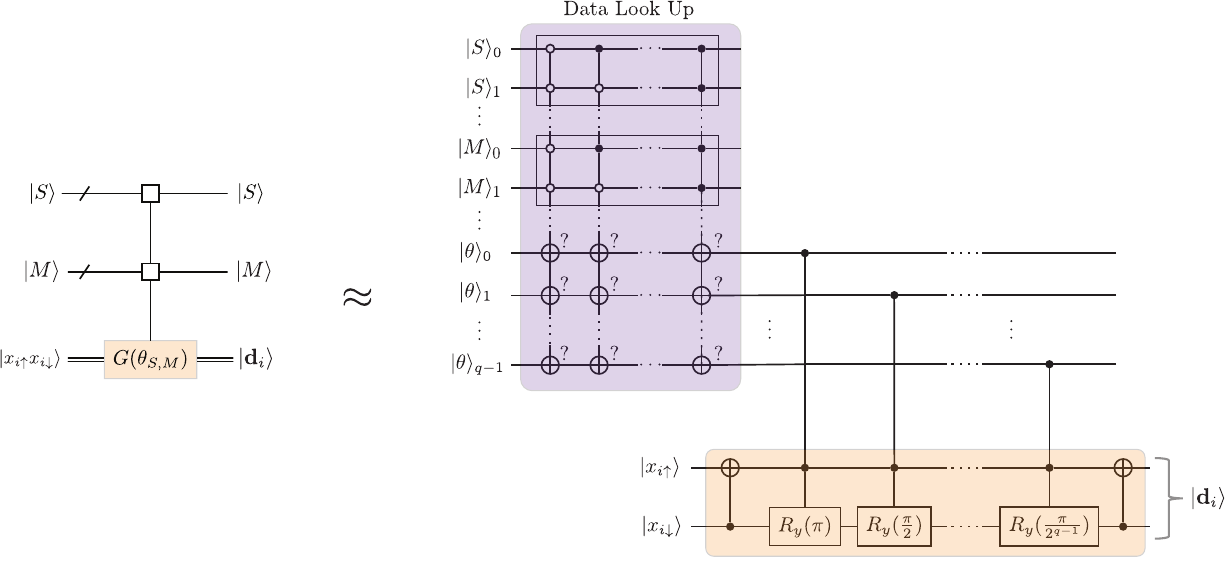}
  \caption{\textit{Data Lookup Givens Rotations.} The rotations are indexed on the control registers $\ket{S}_S$ and $\ket{M}_M$, and a data lookup oracle is used to implement the controlled Givens rotations by storing the fixed-point binary fraction in a lookup table. Each bit of the binary fraction corresponds to an increment of the overall rotation $\theta_{S,M}$ by $\pi/2^n$ for $n$ bits. This rotation increment is controlled on the value of the bit in the lookup table. The controlled $R_y$ gate is then used to perform the Givens rotation.}
  \label{fig:data_lookup_rotations}
\end{figure*}

In practice to save having to implement costly $R_y$ rotations repeatedly for every step of the Paldus transform, we can use a phase gradients (shown in Figure~\ref{fig:phase_gradient}) at a one-off cost as a resource to implement the rotations for each iteration of the Paldus transform. Following the construction introduced in Appendix~\ref{sec:rotations} the rotations then can be implemented using a phase gradient and a doubly controlled adder (see Figure~\ref{fig:gidney_adder_2controlled}). The phase gradient $\mathcal{F}$ initialisation is used once, and requires $\mathcal{O}(q\log(1/\epsilon))$ $T$ gates~\cite{Gidney2018halvingcostof}, where $q$ is the number of phase gradient qubits which is equal to the number of data qubits. The doubly-controlled adder scales with a Toffoli complexity as $3q$ and $T$ complexity as $14q$~\cite{Gidney2018halvingcostof}. This is illustrated in Figure~\ref{fig:Givens_rotation_phase gradient}. 

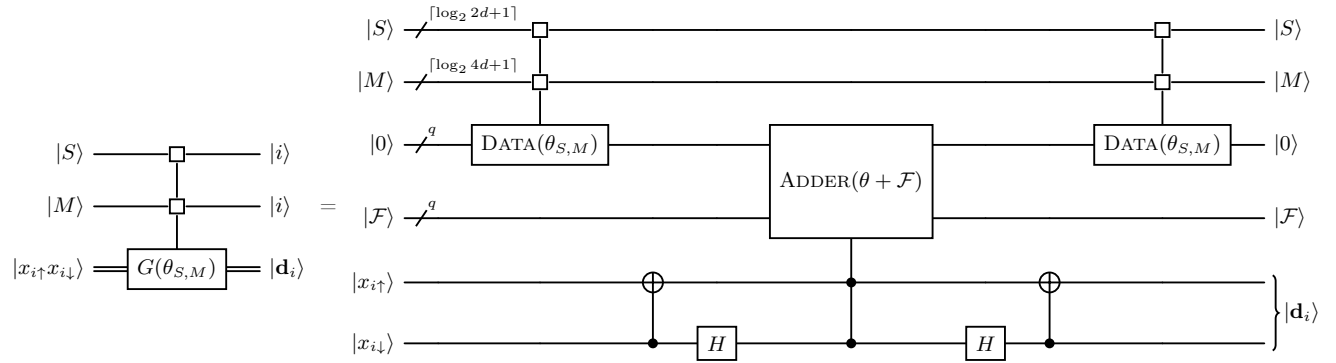
\begin{figure}[!htbp]
  \centering
\begin{adjustbox}{width=\textwidth}
\begin{quantikz}[wire types={q,q,c}]
  \lstick{$|S\rangle$} & \mctrl{1} & \rstick{$|i\rangle$} \\
  \lstick{$|M\rangle$} & \mctrl{1} & \rstick{$|i\rangle$} \\
  \lstick{$|x_{i\uparrow}x_{i\downarrow}\rangle$} & \gate{G(\theta_{S,M})} & \rstick{$|\mathbf{d}_i\rangle$} \\
\end{quantikz}
=
\begin{quantikz}
  \lstick{$|S\rangle$}                &\qwbundle{\lceil \log_2 2d + 1 \rceil}   & \mctrl{1}                   &            &           &                                       &          &           & \mctrl{1}                   & \rstick{$|S\rangle$}   \\
  \lstick{$|M\rangle$}                &\qwbundle{\lceil \log_2 4d + 1 \rceil}   & \mctrl{1}                   &            &           &                                       &          &           & \mctrl{1}                   & \rstick{$|M\rangle$}   \\
  \lstick{$|0\rangle$}                &\qwbundle{q}                             & \gate{\data(\theta_{S,M})}  &            &           &\gate[2]{\adder(\theta + \mathcal{F})} &          &           & \gate{\data(\theta_{S,M})}  & \rstick{$|0\rangle$}   \\
  \lstick{$|\mathcal{F}\rangle$}      &\qwbundle{q}                             &                             &            &           &                                       &          &           &                             & \rstick{$|\mathcal{F}\rangle$}  \\
  \lstick{$|x_{i\uparrow}\rangle$}    &                                         &                             &\targ{}     &           &\ctrl{-1}                              &          &\targ{}    &                             & \rstick[2]{$|\mathbf{d}_i\rangle$}\\
  \lstick{$|x_{i\downarrow} \rangle$} &                                         &                             &\ctrl{-1}   & \gate{H}  &\ctrl{-1}                              &\gate{H}  &\ctrl{-1}  &                             &                \\
\end{quantikz}
\end{adjustbox}
  \caption{\textit{Phase Gradient Parallel Givens Rotation.} The phase gradient $\mathcal{F}$ is initialised with a one off cost of $\mathcal{O}(q\log(1/\epsilon))$ $T$ gates. The controlled Givens rotation is then implemented using the phase gradient and a doubly-controlled controlled adder. The rotation of the Given rotation $\theta_{S,M}$ is implemented using the data lookup oracle $\data(\theta)$ which stored the fixed point binary fraction of the $R_y$ rotation angle. The data lookup is then uncomputed. }
  \label{fig:Givens_rotation_phase gradient}
\end{figure}

The number of possible $S$ values is $2d + 1$ (ranging from $0$ in half-integer increments), and the number of possible $M$ values is $4d + 1$. This results in a total of $(S, M)$ possible elements equal to $I = (2d+1)(4d+1) = 8d^2 + 6d + 1$. However, this is inefficient because $-S \leq M \leq S$, with $M$ increasing in integer steps. Therefore, for all $S$ values, of which there are $d + 1$, the total number of allowed $(S, M)$ elements at each step of the Paldus transform is $L = (d+1)(d+2)/2 = (d^2 + 3d + 2)/2$. The unary iteration and $\selswap$ methods which we will leverage must iterate over all $8d^2 + 6d + 1$ elements. In contrast, the multi-index data lookup method (see Appendix~\ref{sec:multi_index_data_select}) only needs to iterate over $(d^2 + 3d + 2)/2$ elements by flattening the indices and avoiding undefined entries. This represents a reduction by a factor of approximately $16$ compared to the other methods. The resource estimates are presented in Table~\ref{tab:resouces_cost_gives}.

\begin{table}[!htbp]
  \centering
  \begin{ruledtabular}
  \begin{tabular}{l r r r}
  Compilation & Toffoli Count & Clean Qubits & Dirty Qubits \\
  \colrule
      Unary Iteration~\cite{Babbush2018LinearT}                                   & $2I + 3q$                                       & $2\lceil \log_2(I) \rceil + 2q$           & 0 \\
      Clean $\selswap$~\cite{Low2024tradingtgatesdirty,Berry2019qubitizationof}   & $2\lceil I/k\rceil + q(k - 1)+ k + 3q$          & $\lceil \log_2(\lceil I\rceil ) \rceil  + \lceil \log_2(\lceil I/k\rceil ) \rceil + k (q+ 2) + 1 $   & 0\\        
      Dirty $\selswap$~\cite{Berry2019qubitizationof}                             & $2\lceil I/k\rceil+ 4q(k - 1) + 4k + 3q$        & $\lceil \log_2(\lceil I\rceil ) \rceil  + \lceil \log_2(\lceil I/k\rceil ) + 3q +1$          & $(k -1)q$\\
      $^{**}$Multi Index $\data$ & \makecell{$2(2\lceil \log_2 L \rceil + 2\lceil L/k \rceil$ \\ $+ 4q(k-1) + 4(d+2))+ 3q$}         & $2\log_2(2d+1) + 3\log_2(L) + 3q + 1$ & $(k -1)q$\\
  \end{tabular}
  \end{ruledtabular}
       \caption{\textit{Costs for the Fault-Tolerant Givens Rotation.} $I = 8d^2 + 6d + 1$ is looped over all $(S,M)$ combinations, whereas $L = (d^2 + 3d + 2)/2$ are the valid $(S, M)$ combinations. $^{**}$ Multi Index $\data$ lookup is used from Table~\ref{tab:double_index_cost} which is derived in Appendix~\ref{sec:multi_index_data_select}, where the smaller index runs over the $S$ elements which scale as $2d +1$ and has been substituted into the $I$ of Table~\ref{tab:double_index_cost}. $T$ counts can be approximated as $4\times$ Toffoli gates. These resource estimates include measurement-based uncomputation, removing the cost of the uncomputed Toffoli. $q$ is the number of data qubits required to implement the rotation angle at precision $\log_2(1/\epsilon)$ and $k$ is the number of additional $q$ qubit registers needed (either dirty or clean). Where each $q$ bit adder has a $3q$ qubit cost due to measurement-based uncomputation and the factor of $\lceil \log_2(\lceil I/k\rceil ) \rceil - 1$ comes from the smaller unary iteration in the $\selswap$ method.}
  \label{tab:resouces_cost_gives}
  \end{table}

Because of the cascading nature of the algorithm, it is expected that there will always be a large number of dirty qubits available at runtime for each controlled Givens rotation, leading us to expect the present implementation to be the practical method of choice for implementing the controlled rotation.

\subsection{Toffoli Gate Counts}

An end-to-end Toffoli gate cost for a Quantum Paldus Transform over $d$ spatial orbitals can be achieved by combining the $N,S,M$ incrementers with the chosen FT Givens rotation implementation (presented in Table~\ref{tab:resouces_cost_gives} and then cumulatively summing over all $d'\le d$. A detailed breakdown of this approach is outlined in Appendix~\ref{sec:total_gate_complexity}, with the results presented in Table~\ref{tab:full_gate_counts} and shown in Figure~\ref{fig:full_gate_counts}.

\begin{figure}[!htbp]
  \centering
  \includegraphics[width=\textwidth]{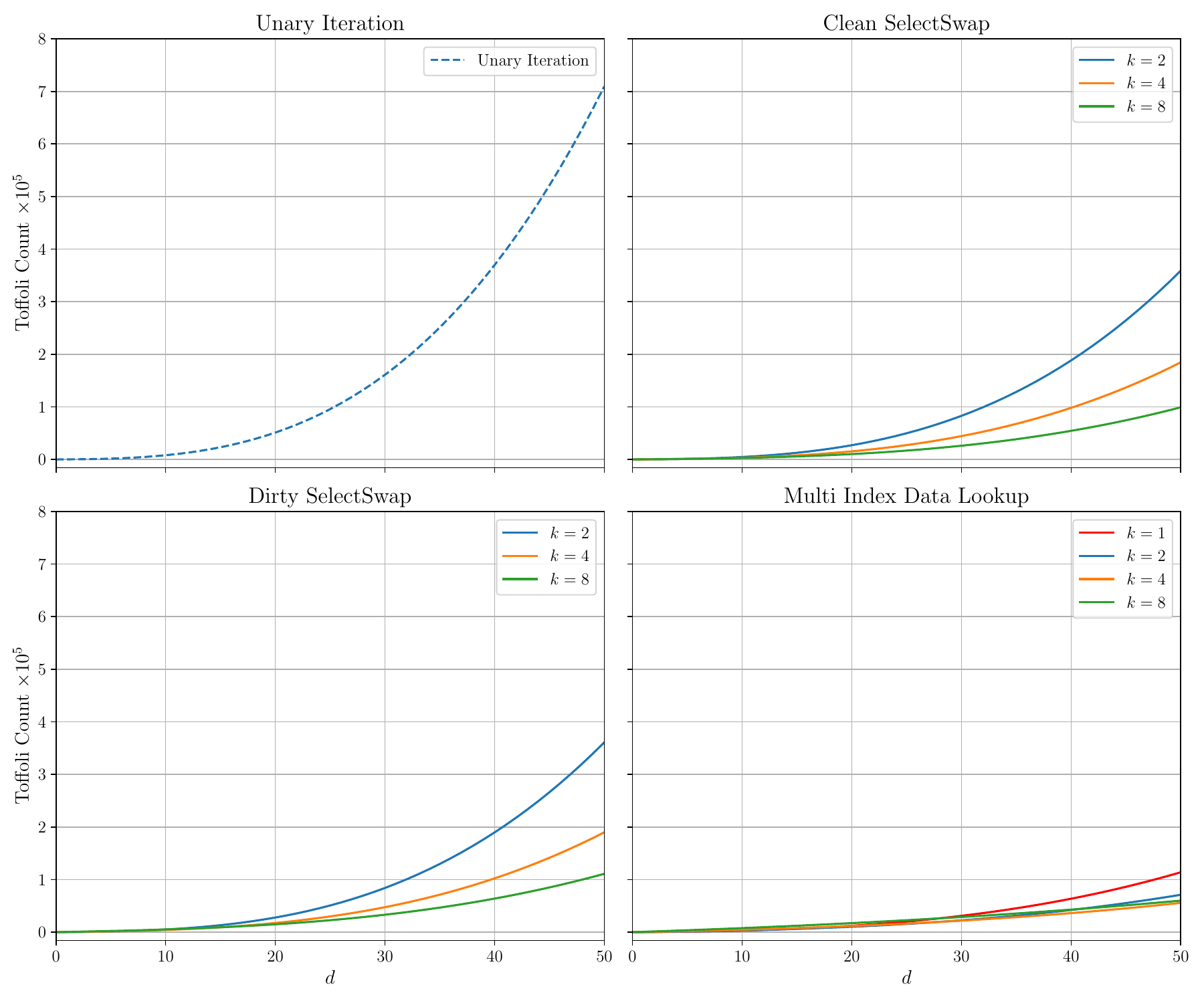}
  \caption{\textit{Toffoli cost of the Paldus transform as a function of the number of spatial orbitals $d$.} The cost is computed using the Givens rotation variants from Table~\ref{tab:resouces_cost_gives} combined with the $N,S,M$ incrementers, with a cumulative sum over all $d'\le d$ according to the formulas in Table~\ref{tab:full_gate_counts}. Where there are $10$ data qubits encoding the rotation precision leading to an error $\epsilon = \frac{2\pi}{2^{10}}$ and $k$ is the number of duplicated data registers.}
  \label{fig:full_gate_counts}
\end{figure}

Rather strikingly, although from Table~\ref{tab:full_gate_counts}, the asymptotic Toffoli gate complexity of the algorithm is $\mathcal{O}(d^3)$. It can be seen that with the Multi-Index strategy for the data lookup step of the FT Givens rotation, due to the sparsity of the allowed $S,M$ pairs, for up to $50$ spatial orbitals, the algorithm scales close to linear. For $50$ spatial orbitals with $k = 4$ duplicated rotation data registers, only $~5500$ Toffoli gates are required, which would be in the regime of quantum advantage for system size for chemistry applications.

A promising direction for future improvement is the observation that many of the controlled $R_y(\theta_{S,M})$ rotations are repeated across multiple iteration steps. Instead of redundantly storing these rotations in a lookup table at every iteration, one could pre-store all the binary increments of the rotations at the outset and then query them as needed throughout the algorithm. This approach would markedly reduce the Toffoli gate cost, further enhancing the efficiency of the algorithm. We leave the investigation of this optimisation for future work.

\newpage

\section{Applications of the Quantum Paldus Transform} \label{sec:applications}

One of the biggest advantages of the UGA basis is its resolution of quantum numbers $(N, S, M)$ into separate registers. Accordingly, the most obvious use of the Paldus transform is the measurement of values $(N, S, M)$ on a generic state $\ket{\psi}$ made possible by its transformation into the UGA basis. With the straightforward representation of their respective values in binary, in the UGA basis a simple measurement of the bitstrings $\mathbf{N}, \mathbf{S}, \mathbf{M}$ will yield the quantum numbers of the state.  More precisely, it reveals the eigenvalues of the operators $\hat{N}, \hat{S}^2, \hat{M}$ (see Section \ref{sec:uga}). However, for measurement of $N$ and $M$ itself, the transform is not necessary. In fact, this information is implicit in the occupation basis already, since $N = |\mathbf{x}|$, and $M = \frac{1}{2} \sum_{i=1}^d (x_{i \uparrow} - x_{i \downarrow})$, so a simple computational basis measurement will yield $N, M$. On the other hand, projective measurements onto $S$ are typically \textit{nontrivial} in the occupation basis: even if information about global spin is well-defined, it is stored globally in the amplitudes of the state. Applying the Paldus transform to such states makes recovering this information straightforward, allowing one to \textit{project} and \textit{post-select} wavefunctions onto a subspace with a particular total spin. The post-measurement wavefunction can then be brought back into the occupation basis by applying the inverse Paldus transform. Overall, our algorithm enables an efficient construction of a \textit{global spin measurement POVM}, given by:

\begin{equation}
U_P^\dagger \left( I_N \otimes |S \rangle \langle S|_S \otimes I_M \otimes I_\mathbf{d} \right) U_P = I_{n_N + n_S + n_M} \otimes \Pi_S, \ \ \ S\in \{ 0, \frac{1}{2}, 1, ..., \frac{d}{2} \}
\end{equation}

\noindent Where $\Pi_S$ is a projection operator onto a subspace of the Fock space, composed of wavefunctions with a \textit{well-defined} value of $S$. Prior to implementing $U_P$ we append ancillas, but after implementing $U_P^\dagger$ they all read $\ket{0}$ (which is reflected in the tensor product with $I_{n_N + n_S + n_M}$ on the right-hand side). We obtain similar POVMs for the $N, M$ and $\mathbf{d}$ measurements, which can be applied by measurements of the spin projection and particle number register. We anticipate this to find use in state preparation, however also expect that similarly to the case of \textit{Weak} and \textit{Strong Schur Sampling} \cite{harrowphd,childs2007, cervero2023, cervero2024} in the first quantised picture, the measurement of $N, S, M$ or $\mathbf{d}$ will turn out to also be possible with a circuit of lower complexity than $U_P$ through a procedure known as \textit{Generalised Phase Estimation}~\cite{Fitzpatrick_2025}. We therefore turn our focus to applications which go beyond projections, and explicitly make use of the GT basis and its symmetries. 

To demonstrate the broad utility of the quantum Paldus transform, we now present a selection of its applications.  We will begin by studying the form of antisymmetric $U(d) \times U(2)$ representations from an operator point of view, showing that the GT basis is invariant under the action of Hamiltonians built out of $\mathfrak{u}(d) \oplus \mathfrak{u}(2)$ \textit{ladder operators}, which form Lie algebra representations constructed from the creation and annihilation operators on the underlying Fock space. This will enable us to isolate classes of Hamiltonian and circuits which block-diagonalise under the Paldus transform in Sections \ref{sec:uga} and \ref{sec:matchgate_circuits}. The ability to block-diagonalise unlocks speed-ups for Hamiltonian simulation with the Paldus transform as a subroutine, which are described in Sections \ref{sec:fastforwarding} and \ref{sec:chemistry}. We then turn to applications for state preparation in Section \ref{sec:CSFprep}. Taking advantage of the relationship between the Paldus and Schur transforms presented in Section \ref{sec:paldusschur}, we end with an application of the former in the context of noise mitigation -- encoding and decoding into \textit{Decoherence-Free Subsystems} -- detailed in Section \ref{sec:dfs}.

\subsection{The Unitary Group Approach} \label{sec:uga}

Having worked out the decomposition of the fermionic Fock space and resulting Gelfand-Tsetlin basis in Section \ref{sec:background}, we now turn to identifying the representations of $U(d) \times U(2)$ explicitly. Our work concerns the \textit{electronic structure problem} of quantum chemistry, in which one wishes to calculate the properties of a system with electrons confided to $d$ many spatial orbitals (which could correspond to an atom, or sites on a lattice in certain condensed matter models). As previously discussed, this corresponds to a vector space of $2d$ fermionic modes. In a quantum computational setting, such modes are typically represented with qubits, where each qubit state corresponds to a Fock state. The exact correspondence is then determined by a \textit{fermion-to-qubit mapping}. For example, in our circuit for the Paldus transform we have implicitly used the \textit{Jordan-Wigner mapping}, in which the $n=2d$ qubit computational basis states are indexed by $(i, \mu)$, where $i\in \{1, ..., d\}$ labels the orbital and $\mu = \pm \frac{1}{2}$ the spin projection (we also use the shorthand $\mu  \in \{ \uparrow, \downarrow \}$). The occupation number of the $(i, \mu)^\text{th}$ spin-orbital is then given by the state of the $(2i - (\mu + \frac{1}{2}))^\text{th}$ qubit. The creation and annihilation operators are given by:
  
\begin{align}
    a^\dagger_{j \mu} &= \underbrace{Z \otimes \cdots \otimes Z}_{2j - (\mu + 3/2)} \otimes \frac{1}{2} \left(X - i Y \right) \otimes \underbrace{I \otimes \cdots \otimes I}_{2d - 2j + (\mu+1/2)}, &
    a_{j \mu} &= \underbrace{Z \otimes \cdots \otimes Z}_{2j - (\mu + 3/2)} \otimes \frac{1}{2} \left(X + i Y \right) \otimes \underbrace{I \otimes \cdots \otimes I}_{2d - 2j + (\mu+1/2)} \label{eq:jw_creation_annihilation}
\end{align}
  
\noindent For example, with $2d=4$ spin-orbitals we have:
  
  \begin{align*}
    a^\dagger_{1 \uparrow} &= \frac{1}{2} \left(X - i Y \right) \otimes I \otimes I \otimes I, &
    a^\dagger_{1 \downarrow} &= Z \otimes \frac{1}{2} \left(X - i Y \right) \otimes I \otimes I, \\
     a^\dagger_{2 \uparrow} &= Z \otimes Z \otimes \frac{1}{2} \left(X - i Y \right) \otimes I,  &
    a^\dagger_{2 \downarrow} &= Z \otimes Z \otimes Z \otimes \frac{1}{2} \left(X - i Y \right) \\
  \end{align*} 

\subsubsection{\texorpdfstring{$\mathfrak{u}(d) \oplus \mathfrak{u}(2)$}{u(d) + u(2)} Ladder Operators} \label{sec:ladder_operators}

The \textit{Unitary Group Approach} is a formalism which exploits the Lie group-theoretic symmetries present in second-quantised \textit{spin-free} Hamiltonians. Such Hamiltonians appear in a wide variety of quantum chemistry problems, and take the form \cite{Babbush_2018}:

\begin{align}
    H &= T + V  \nonumber \\
    &= \sum_{ij, \mu} h_{ij} a_{i \mu}^\dagger a_{j \mu} + \frac{1}{2} \sum_{ijkl, \mu \nu} v_{ij, kl} a_{i \mu}^\dagger a_{j \nu}^\dagger a_{l \nu} a_{k \mu}, \label{eq:spinfreehamiltonian}
\end{align}

Where $T$ contains the kinetic and single-body potential terms, $V$ contains two-body potential terms, and $h_{ij}$, $v_{ij, kl}$ are the one- and two-electron integrals. The creation and annihilation operators act on the $(i, \mu)^\text{th}$ fermionic modes. This Hamiltonian is said to be \textit{spin-free}, in the sense that there is no dependence on the $\mu$ in its free parameters. As a consequence, one can show that $H$ commutes with $\hat{S}^2$, and $\hat{M}$ (total spin, and spin projection) operators, acting independently of the global spin degrees of freedom. The power of UGA comes from recognising that such Hamiltonians are made up of elements of a representation of the $\mathfrak{u}(d)$ \textit{Lie algebra}. To make this clear, we may define a set of \textit{ladder operators}, which are obtained from sums of products of creation and annihilation operators:

\begin{align}
    E_{ij} &= \sum_\mu a^\dagger_{i\mu} a_{j\mu}, & \mathcal{E}_{\mu \nu} &= \sum_i a^\dagger_{i \mu} a_{i \nu} \label{eq:ladderdefinition}
\end{align}

\noindent Where $i \in \{1, \hdots, d\}$ and $\mu =  \pm\frac{1}{2}$. Clearly, $E_{ij}^\dagger = E_{ji}$ and $\mathcal{E}_{\mu \nu}^\dagger = \mathcal{E}_{\nu \mu}$. Using the anticommutation relations for fermionic operators, one can readily verify that the ladder operators commute as follows:

\begin{align}
    [E_{ij}, E_{kl}] & = \delta_{jk} E_{il} - \delta_{il} E_{kj}, & 
    [\mathcal{E}_{\mu \nu}, \mathcal{E}_{\sigma \tau}] & = \delta_{\nu \sigma} \mathcal{E}_{\mu \tau} - \delta_{\mu \tau} \mathcal{E}_{\sigma \nu}, & 
    [E_{ij}, \mathcal{E}_{\mu \nu}] & = 0 \label{eq:genscommutators}
\end{align}

\noindent The commutation relations of Equation \eqref{eq:genscommutators} replicate those of $d \times d$ and $2 \times 2$ \textit{matrix units}, which are matrices of a single unit entry in the $(i, j)^\text{th}$ position, i.e. $[e_{ij}]_{kl} = \delta_{ik} \delta_{jl}$. It is well-known that the matrix units $e_{ij}$ ($i, j \in \{ 1, \hdots, n \}$) form a basis for the $n \times n$ complex matrices, and along with the commutator they form standard representation of the Lie algebra $\mathfrak{gl}(n, \mathbb{C})$. By virtue of replicating the commutation relations, the ladder operators $\{E_{ij}\}$ and $\{ \mathcal{E}_{\mu \nu}\}$ then form a \textit{faithful} representation of the Lie algebras $\mathfrak{gl}(d, \mathbb{C})$ and $\mathfrak{gl}(2, \mathbb{C})$ respectively\footnote{We note in passing that although $E_{ij}$'s commute like $e_{ij}$'s, they do not multiply the same way -- in general, $E_{ij} E_{kl} \neq \delta_{jk} E_{il}$. This does not prohibit them from forming a representation.}. Because $E_{ij}$ and $\mathcal{E}_{\mu \nu}$ commute, we may thus treat the space of their complex linear combinations as a faithful representation of $\mathfrak{gl}(d, \mathbb{C}) \oplus \mathfrak{gl}(2, \mathbb{C})$ on the space of $2^{2d} \times 2^{2d}$ matrices. By the Lie group--Lie algebra correspondence (Theorem \ref{thm:homomorphisms} of Appendix \ref{app:liegroups}), there exists an associated (projective) representation of the Lie group $GL(d, \mathbb{C}) \times GL(2, \mathbb{C})$, with representation matrices of the form $M = e^{\sum_{ij} \alpha_{ij} E_{ij} + \sum_{\mu \nu} \beta_{\mu \nu} \mathcal{E}_{\mu \nu} }$. Under restriction to antihermitian exponents, this reduces to \textit{precisely} the $U(d) \times U(2)$ group action discussed in Theorem \ref{thm:paldus_duality}. In Appendix $\ref{app:liegroups}$ we provide a detailed review of the Lie-theoretic elements of Paldus duality, including relevant theorems and a definition of projective representations.

With the ladder operators defined in Equation \eqref{eq:ladderdefinition}, one can re-express the Hamiltonian of Equation \eqref{eq:spinfreehamiltonian} in terms of the $E_{ij}$'s, which will take the form:

\begin{equation}
    H = \sum_{ij} h_{ij} E_{ij} + \frac{1}{2} \sum_{ijkl} v_{ij, kl} ( E_{ik} E_{jl} - \delta_{ik} E_{il} ). \label{eq:ugaspinfreehamiltonian}
\end{equation}

\noindent We first consider the single-body term $T$, which is linear in the $E_{ij}$'s. Requiring that $h_{ij} = h_{ji}^*$ forces $T$ to be Hermitian, and thus belong to a representation of the Lie algebra $\mathfrak{u}(d)$. Any unitary evolution by $T$ of the form $U = e^{iT}$ will then form a representation of the Lie group $U(d)$ on the Fock space\footnote{Our discussion will focus on the representations of $U(d) \times U(2)$. However, many of the results also hold for representations of $GL(d, \mathbb{C}) \times GL(2, \mathbb{C})$.}. On the other hand, the two-body term $V$ is quadratic in the ladder operators, which in this representation are generally \textit{not} expressible through linear combinations of the $E_{ij}$'s. This puts the overall Hamiltonian $H$ in the \textit{universal enveloping algebra} $\mathfrak{U}(\mathfrak{u}(d))$ -- whilst this fact implies that $e^{iH}$ is not generally a representation of $U(d)$, it will not pose a problem for the development of the formalism, as we discuss in Appendix \ref{app:universalenveloping}.

Historically, UGA approaches for computational quantum chemistry have exploited the form of Equation \eqref{eq:ugaspinfreehamiltonian} by constructing a Gelfand-Tsetlin (GT) basis for the group $U(d) \times U(2)$, in which $H$ takes on a block-diagonal form. Calculations are then performed in this basis, and efficient algorithms have been developed for calculating matrix elements of $E_{ij}$ (in terms of the Shavitt graphs of Section \ref{sec:shavitt}), which were then used to calculate properties of the Hamiltonian. The vast majority of this work has focused on the $U(d)$ part of the group. In our development of UGA for quantum computation we expand on these ideas, providing an explicit quantum circuit to block-diagonalise any unitary $U = e^{iH}$ with $H \in \mathfrak{U}(\mathfrak{u}(d) \oplus \mathfrak{u}(2))$ (i.e. evolution by Hamiltonians polynomial in $\{ E_{ij}, \mathcal{E}_{\mu \nu} \}$). In fact, working with the universal enveloping algebra is necessary for UGA to be useful: in Section \ref{sec:matchgate_circuits} we show that unitary evolution by Hamiltonians belonging to either $\mathfrak{u}(d)$ or $\mathfrak{u}(2)$ forms matchgate circuits, which are known to be classically efficiently simulable.

In Section \ref{sec:gt_states_construction}, we have explicitly constructed the GT basis for $U(d) \times U(2)$ consisting of states of the form $\ket{N, S, M; \mathbf{d}}$. We can now identify the quantum numbers which index each basis state as the eigenvalues of the operators $\hat{N}, \hat{M}, \hat{S}^2$ corresponding to the physical quantities encoded by the state. The first two of these operators are given by:

\begin{align}
\hat{N} &= \sum_{i=1}^d E_{ii} = \sum_{\mu = \pm \frac{1}{2}} \mathcal{E}_{\mu \mu} = \sum_\mu \sum_i a_{i \mu}^\dagger a_{i \mu}, &
\hat{M} &= \frac{1}{2}(\mathcal{E}_{\uparrow \uparrow} - \mathcal{E}_{\downarrow \downarrow}) = \frac{1}{2} \sum_i (a_{i \uparrow}^\dagger a_{i \uparrow} - a_{i \downarrow}^\dagger a_{i \downarrow}) 
\end{align}

\noindent Whereas to express the total spin operator $\hat{S}^2$ we first write down the spin projection operators $\hat{S}_x^{(i)}, \hat{S}_y^{(i)}, \hat{S}_z^{(i)}$ for each orbital:

\begin{align}
\hat{S}_x^{(i)} &= \frac{1}{2} (a_{i \uparrow}^\dagger a_{i \downarrow} + a_{i \downarrow}^\dagger a_{i \uparrow}), &
  \hat{S}_z^{(i)} &= \frac{1}{2} (a_{i \uparrow}^\dagger a_{i \uparrow} - a_{i \downarrow}^\dagger a_{i \downarrow}), &
  \hat{S}_y^{(i)} &= \frac{i}{2} (a_{i \downarrow}^\dagger a_{i \uparrow} - a_{i \uparrow}^\dagger a_{i \downarrow}) \label{eq:spinoperators}
\end{align}

\noindent Then, the total spin projection operators are given by:

\begin{align}
  \hat{S}_x &= \sum_i \hat{S}_x^{(i)}, &
   \hat{S}_y &= \sum_i \hat{S}_y^{(i)}, &
  \hat{S}_z &= \sum_i \hat{S}_z^{(i)} = \hat{M} \label{eq:totalspinoperators}
\end{align}

\noindent The total spin operator is given by $\hat{S}^2 = \hat{S}_x^2 + \hat{S}_y^2 + \hat{S}_z^2$, which one may also write as $\hat{S}^2 = \hat{M}(\hat{M} + I) + \mathcal{E}_{\downarrow \uparrow} \mathcal{E}_{\uparrow \downarrow}$. 

For the original, in-depth discussion of the ladder operators we recommend the 1976 review by Paldus \cite{PALDUS1976131}, as well as the original work of Paldus or Matsen \cite{Paldus1974, Matsen}. 

\subsubsection{The Fermi-Hubbard Hamiltonian}

As an example of a computationally interesting spin-free Hamiltonian belonging to the universal enveloping algebra $\mathfrak{U}(\mathfrak{u}(d))$, we can consider the Fermi-Hubbard Hamiltonian. Its simulation is a nontrivial task, yet it is block-diagonalisable under the GT basis (and consequently, can be efficiently block diagonalised by the Paldus transform). The Fermi-Hubbard Hamiltonian is defined as:

\begin{equation}
H = -t \sum_{i , \sigma} ( a_{i, \sigma}^\dagger a_{i+1, \sigma} + a_{i+1, \sigma}^\dagger a_{i, \sigma}) + U \sum_i n_{i \uparrow} n_{i \downarrow}, \quad n_{i \sigma} = a_{i \sigma}^\dagger a_{i \sigma} \label{eq:fermihubbard}
\end{equation}

\noindent where $t$ is a hopping parameter, and $U$ is an on-site interaction strength describing the effect of electron repulsion. The terms $n_{i \sigma}$ are the number operators for the $(i, \sigma)^\text{th}$ spin-orbital. We can re-write this Hamiltonian in terms of the ladder operators $E_{ij}$. The hopping term is already in the form of Equation \eqref{eq:ladderdefinition}. On the other hand, the summands of the interaction term can be written as $n_{i \uparrow} n_{i \downarrow} = \frac{1}{2} E_{ii}^2 - \frac{1}{2} E_{ii}$. To show this, we square the $E_{ii}$ operator and use the anticommutation relations of the creation and annihilation operators:

\begin{align}
E_{ii}^2 &= (a_{i \uparrow}^\dagger a_{i \uparrow} + a_{i \downarrow}^\dagger a_{i \downarrow})^2 \\
&= a_{i \uparrow}^\dagger a_{i \uparrow} a_{i \uparrow}^\dagger a_{i \uparrow} + a_{i \uparrow}^\dagger a_{i \uparrow} a_{i \downarrow}^\dagger a_{i \downarrow} + a_{i \downarrow}^\dagger a_{i \downarrow} a_{i \uparrow}^\dagger a_{i \uparrow} + a_{i \downarrow}^\dagger a_{i \downarrow} a_{i \downarrow}^\dagger a_{i \downarrow} \\
&= a_{i \uparrow}^\dagger (I -  a_{i \uparrow}^\dagger a_{i \uparrow} ) a_{i \uparrow} + 2 a_{i \uparrow}^\dagger a_{i \uparrow} a_{i \downarrow}^\dagger a_{i \downarrow} + a_{i \downarrow}^\dagger (I -  a_{i \downarrow}^\dagger a_{i \downarrow} ) a_{i \downarrow} \\
&= 2n_{i \uparrow} n_{i \downarrow} + E_{ii}
\end{align}

\noindent Substituting this into Equation \eqref{eq:fermihubbard}, we get:

\begin{equation}
H = -t \sum_i (E_{i, i+1} + E_{i+1, i}) + \frac{U}{2} \sum_i ( E_{ii}^2 - E_{ii} ) \label{eq:fermihubbardladder}
\end{equation}

\noindent Clearly, $H$ lies outside the Lie algebra representation $\mathfrak{u}(d)$ by the presence of the $E_{ii}^2$ term, but only contains products of $E_{ij}$ operators, putting it in the universal enveloping algebra representation $\mathfrak{U}(\mathfrak{u}(d))$. As discussed in Appendix \ref{app:universalenveloping} the invariant subspaces of the two algebras coincide, hence evolution by the Fermi-Hubbard Hamiltonian can be block-diagonalised by the Paldus transform.

By including additional terms $E_{ij}$, we can also describe the model in two dimensions. In our view, the utility of UGA extends beyond just quantum chemistry and has the potential to aid in the simulation of a wide range of condensed matter systems. Further, explicit examples of Hamiltonians which are block-diagonalised by the Paldus transform (as well as other applications) are provided next.

\subsection{\texorpdfstring{$U(d) \times U(2)$}{U(d) x U(2)} Group Action and Matchgate Circuits} \label{sec:matchgate_circuits}

Here we consider Hamiltonians which belong to the Lie algebra representations $\mathfrak{u}(d)$ or $\mathfrak{u}(2)$, and explicitly identify the form of quantum circuits which describes their unitary evolution. This in turn gives an explicit form of the antisymmetric $U(d) \times U(2)$ representations introduced in Section \ref{sec:background}, instances of circuits block-diagonalised by the Paldus transform, as well as a worked example of ideas presented in Section \ref{sec:uga}. 

We seek to write down unitaries generated by $\{E_{ij}, \mathcal{E}_{\mu \nu}\}$-linear Hamiltonians which induce projective representations of $U(d) \times U(2)$, and will find that in the Jordan-Wigner mapping their unitary evolution is expressible in terms of \textit{matchgate circuits} \cite{Jozsa_2008, valiant_2002, knill2001, burkat2024}. The connection between free fermionic systems and matchgate circuits is well-known \cite{Terhal_2002, majsak2024, Wan_2023}, however examining them in the context of UGA can provide us with additional insights. For example, as we show in Section \ref{sec:dfs} our discussion leads to some interesting consequences for quantum communication, namely the ability to construct decoherence-free subsystems with respect to certain types of two-qubit noise characterisable by the action of a matchgate. To make the connection between Hamiltonians of our interest and matchgate circuits, we will make use of Majorana operators $c_i$. These are a set of $4d$ mutually anticommuting operators satisfying $\{ c_i, c_j \} = \delta_{ij}$. In the Jordan-Wigner mapping, they are Pauli strings given by:
  
  \begin{align}
    c_1 &= X \otimes I \otimes \hdots \otimes I, & c_2 &= Y \otimes I  \otimes \hdots \otimes I,  & c_3 &= Z \otimes X \otimes \hdots \otimes I, &  c_4 &= Z \otimes Y  \otimes \hdots \otimes I
   \end{align}
  
  \noindent and so on for $4 < i \leq 4d$.  We can write the creation and annihilation operators in terms of Majorana operators as:

  \begin{align}
     a_{1 \uparrow} &= \frac{1}{2}(c_1 + i c_2 ), &  a_{1 \downarrow} &= \frac{1}{2}(c_3 + i c_4 ), & a_{2 \uparrow} &= \frac{1}{2}(c_5 + i c_6 ), & a_{2 \downarrow} &= \frac{1}{2}(c_7 + i c_8 ) \label{eq:majorana_creation_annihilation}
   \end{align}

\noindent and so on for $2 < i \leq d$. 

Previous work \cite{Jozsa_2008} has shown that evolution $e^{iH}$ by Hamiltonians \textit{quadratic} in Majorana operators of the form $H = \sum_{i \neq j} h_{ij} c_i c_j$ can always be expressed as a \textit{matchgate circuit}, i.e. a circuit composed of polynomially many two-qubit gates $G(A, B)$ of the form:

\begin{equation}
  G(A, B) = \begin{pmatrix} a_{11} & 0 & 0 & a_{12} \\ 0 & b_{11} & b_{12} & 0 \\ 0 & b_{21} & b_{22} & 0 \\ a_{21} & 0 & 0 & a_{22} \end{pmatrix}, \ \text{ where } \ \det(A) = \det(B) \label{eq:matchgatedef}
\end{equation}

\noindent for which each $G(A, B)$ gate acts on \textit{neighbouring} qubit pairs. Two examples of matchgate circuits are shown in Figure \ref{fig:fswap_sequence_matchgate_circuit}. In turn, circuits of this form are known to be classically efficiently simulable, so long as every $G(A, B)$ satisfies Equation \eqref{eq:matchgatedef} and acts on neighbouring qubits. In the rest of this section, we will use the relation between quadratic Hamiltonians and representations of $\mathfrak{u}(d)$ and $\mathfrak{u}(2)$ in the Jordan-Wigner mapping to express the latter in terms of matchgate circuits. This will give us a concrete form of the induced $U(d) \times U(2)$ representations.

  \subsubsection{Single-body \texorpdfstring{$\mathfrak{u}(2)$}{u(2)} Hamiltonians} \label{sec:u2_hamiltonians}
  
Here, we consider the case of evolution by Hamiltonians belonging to the Lie algebra $\mathfrak{u}(2)$ and characterise the action of unitaries $e^{iH_{\mathfrak{u}(2)}}$, where $H_{\mathfrak{u}(2)}$ contains purely terms linear in $\mathcal{E}_{\mu \nu}$. Physically, a Hamiltonian of this form can be used describe $d$ identical, noninteracting orbitals in a magnetic field, each with kinetic energy $\alpha_N(a_{i \uparrow}^\dagger a_{i \uparrow} + a_{i \downarrow}^\dagger a_{i \downarrow})$ and magnetic field interaction $\alpha_x \hat{S}_x^{(i)} + \alpha_y \hat{S}_y^{(i)} + \alpha_z \hat{S}_z^{(i)}$ (for $\alpha_i \in \mathbb{R}$). The overall Hamiltonian $H_{\mathfrak{u}(2)}$ is then a sum over all orbitals (c.f. Equations \eqref{eq:spinoperators} and \eqref{eq:totalspinoperators}), given by:
  
  \begin{align*}
  H_{\mathfrak{u}(2)} &= \alpha_{x} \hat{S}_x + \alpha_{y} \hat{S}_y + \alpha_{z} \hat{S}_z + \alpha_{N} \hat{N}  \\
  &= \frac{\alpha_{x}}{2} ( \mathcal{E}_{\uparrow \downarrow} + \mathcal{E}_{\downarrow \uparrow}) + \frac{i \alpha_{y}}{2}  ( \mathcal{E}_{\downarrow \uparrow} - \mathcal{E}_{\uparrow \downarrow}) +  \frac{\alpha_{z}}{2} ( \mathcal{E}_{\uparrow \uparrow} -  \mathcal{E}_{\downarrow \downarrow}) + \alpha_{N}( \mathcal{E}_{\uparrow \uparrow} +  \mathcal{E}_{\downarrow \downarrow})
  \end{align*}
  
\noindent We will now show in the Jordan-Wigner representation, evolution by $H_{\mathfrak{u}(2)}$ is always expressible as a $d$-fold tensor product of two-qubit gates $(WV)^{\otimes d}$, where $V$ is an $SU(2)$ operation embedded in the $\{ \ket{01}, \ket{10}\}$ subspace and $W$ is an element of $U(1)$ (in the form of a controlled phase). As discussed previously, $e^{iH_{\mathfrak{u}(2)}}$ generates a projective representation of $U(2)$, and in this case we will have a clear decomposition of a $U(2)$ representation into $SU(2)$ and $U(1)$ components, reflecting the well-known double covering of the group $U(2)$ by $SU(2) \times U(1)$. 

We begin by re-writing each ladder operator $\mathcal{E}_{\mu \nu}$ explicitly in terms of the creation and annihilation operators via Equation \eqref{eq:ladderdefinition}. After substituting in the expressions for $a_{i \mu}$ and $a_{i \mu}^\dagger$ in terms of Majorana operators $c_i$ (given in Equation \eqref{eq:majorana_creation_annihilation}) and using the anticommutation relation $\{ c_i, c_j \} = \delta_{ij}$, we can write the $\mathcal{E}_{\mu \nu}$'s in terms of $c_i$'s as follows:

  \begin{align}
  \mathcal{E}_{\uparrow \uparrow} &= \frac{1}{2}(I  + i c_1 c_2) + \frac{1}{2}(I  + i c_5 c_6) + \hdots + \frac{1}{2}(I  + i c_{4d-3} c_{4d-2}) \\
  \mathcal{E}_{\downarrow \downarrow} &= \frac{1}{2}(I + i c_3 c_4) + \frac{1}{2}(I + i c_7 c_8) + \hdots + \frac{1}{2}(I + i c_{4d-1} c_{4d}) \\
  \mathcal{E}_{\uparrow \downarrow} &= \frac{1}{4}( c_1 c_3 + c_2 c_4 +i c_1 c_4 - ic_2 c_3) + \frac{1}{4}( c_5 c_7 + c_6 c_8 +i c_5 c_8 - ic_6 c_7) + \hdots \\
  \mathcal{E}_{\downarrow \uparrow} &= \frac{1}{4}( - c_1 c_3 - c_2 c_4 +i c_1 c_4 - ic_2 c_3) + \frac{1}{4}( - c_5 c_7 - c_6 c_8 +i c_5 c_8 - ic_6 c_7) + \hdots
  \end{align}
  
  \noindent where $\mathcal{E}_{\mu \nu}$ has been written out as a sum of $d$ bracketed terms, each acting nontrivially on up to two neighbouring qubits. Now, we collect all terms from $H_{\mathfrak{u}(2)}$ acting on the $1^\text{st}$ and $2^\text{nd}$ qubit lines (i.e. bracketed terms including $c_1, c_2, c_3, c_4$). After substituting in the expressions for each $c_i$ in terms of Pauli operators, we can write out these terms explicitly as:
  
  \begin{align}
  \alpha_x \hat{S}_x^{(1)} + \alpha_y \hat{S}_y^{(1)} + \alpha_z \hat{S}_z^{(1)} + \alpha_N \hat{N}^{(1)} &= \frac{\alpha_x}{4} (ic_1 c_4 - ic_2 c_3) + \frac{\alpha_y}{4} (ic_1 c_3 + ic_2 c_4) + \frac{\alpha_z}{4} ( ic_1 c_2 -  ic_3 c_4)  \\
  & \ + \frac{\alpha_N}{2}(2 \times I + i c_1 c_2 + i c_3 c_4) \\
  &= \frac{1}{2} \left( \begin{bmatrix} 0 & 0 & 0 & 0 \\ 0 & - \alpha_z & \alpha_x - i \alpha_y & 0 \\ 0 & \alpha_x + i \alpha_y & \alpha_z &0  \\ 0 &0  & 0 & 0 \end{bmatrix}  + \begin{bmatrix} 0 & 0 & 0 & 0 \\ 0 & \alpha_N & 0 & 0 \\ 0 & 0 & \alpha_N & 0 \\ 0 & 0 & 0 & 2\alpha_N \end{bmatrix} \right) \otimes I \otimes \hdots \otimes I
  \end{align}
  
  \noindent Similarly, the terms acting on the $3^\text{rd}$ and $4^\text{th}$ qubit lines (which include $c_5, c_6, c_7, c_8$) follow an identical expression, but are preceded by a tensor product with $I \otimes I$. There are $d$ such (mutually commuting) terms in the Hamiltonian. Furthermore, the first and second square matrix within each term also commute -- the first is clearly an element of $\mathfrak{su}(2)$, and the second is a diagonal matrix which exponentiates to a controlled phase gate. Therefore, we can write the evolution by $H_{\mathfrak{u}(2)}$ as $(WV)^{\otimes d}$:

  \begin{equation}
  e^{iH_{\mathfrak{u}(2)}} = \bigotimes_{i=1}^d e^{i(\alpha_x \hat{S}_x^{(1)} + \alpha_y \hat{S}_y^{(1)} + \alpha_z \hat{S}_z^{(1)})} e^{i \alpha_N \hat{N}^{(1)}} = (WV) \otimes \hdots \otimes (WV) = (WV)^{\otimes d},
  \label{eq:hu2_hamiltonian}
  \end{equation}
  
  \noindent where the exact expression for $WV$ is given by:
  
  \begin{equation}
  WV = \begin{pmatrix}1 & 0 & 0 & 0 \\ 0 & u_{11} & u_{12} & 0 \\ 0 & u_{21} & u_{22} & 0 \\ 0 & 0 & 0 & 1 \end{pmatrix} \begin{pmatrix} 1 & 0 & 0 & 0 \\ 0 & e^{\frac{i \alpha_N}{2}} & 0 & 0 \\ 0 & 0 & e^{\frac{i \alpha_N}{2}} & 0 \\ 0 & 0 & 0 & e^{i \alpha_N} \end{pmatrix}, \ \text{where } \ U \in SU(2) \label{eq:wv_circuit}
  \end{equation}
  
  \noindent Physically, the matrix $W$ can be seen to describe the precession present due to interaction with a magnetic field, whereas the matrix $V$ describes the kinetic energy of the overall system, dependent on the total number of particles. The presence of $\hat{N}$ in $H_{\mathfrak{u}(2)}$ also determines whether its action on the Gelfand-Tsetlin basis $\ket{N, S, M; \mathbf{d}}$ is dependent on the particle number $N$. In its absence, the action on subspaces $V_{N, S}$ and $V_{2d -N, S}$ is identical. On the other hand, if $\alpha_N \neq 0$ then unitary evolution by $H_{\mathfrak{u}(2)}$ will introduce an $N$-dependent phase $e^{i N \alpha_N / 2}$ on each $V_{N, S}$, which corresponds to the total, particle number-proportional kinetic energy of the system.
  
  Mathematically, exponentiating $H_{\mathfrak{u}(2)}$ gives rise to a polynomial representation of $U(2)$ of the form $R_{N, S}(U) \cong \det U^N R_{S}(U)$. Because $H_{\mathfrak{u}(2)}$ is in the Lie algebra representation of $\mathfrak{u}(2)$, the induced unitary will only affect the $M$-component of the GT basis, which we may summarise as:
  
  \begin{equation}
    e^{iH_{\mathfrak{u}(2)}} \ket{N, S, M; \mathbf{d}} = \sum_{S=0}^{d} \sum_{M'=-S}^S e^{i N \phi} [R_S(U)]_{M', M} \ket{N, S, M'; \mathbf{d}} \label{eq:u2_action}
  \end{equation}

  \noindent Where $\phi = \alpha_N / 2$, and $[R_S(U)]_{M', M}$ form the components of $(2S+1) \times (2S+1)$ irrep matrices. Finally, we note that the unitary $(WV)^{\otimes d}$ satisfies Equation \eqref{eq:matchgatedef}, and is therefore a matchgate.
  
  Whilst $H_{\mathfrak{u}(2)}$ is generally uninteresting from a quantum simulation perspective, we will find that the symmetry of the GT basis under its action gives rise to an interesting consequence: an error-free protocol for encoding quantum information into a \textit{decoherence-free subsystem} \cite{Kempe_2001}, achieved by encoding quantum information into the $\mathbf{d}$ register (which for an appropriate choice of $(N, S)$ is exponentially large) to protect against noise of the form $e^{iH_{\mathfrak{u}(2)}}$ This encoding is made possible by the quantum Paldus transform, and we return to it in Section \ref{sec:dfs}.

   \begin{figure*}[htb!]
    \includegraphics[width=8cm]{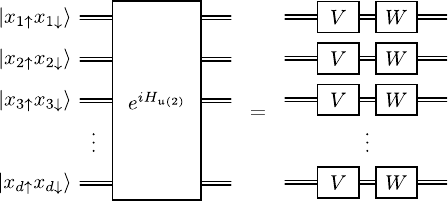}
    \caption{\textit{Evolution by Single-body $\mathfrak{u}(2)$ Hamiltonians.} The gate $e^{iH_{\mathfrak{u}(2)}}$ is the unitary evolution by a Hamiltonian of the form in Equation \eqref{eq:hu2_hamiltonian}. Each $WV$ is a gate described by Equation \eqref{eq:wv_circuit} and Figure \ref{fig:wv_gate}.} \label{fig:u2_circuit}
  \end{figure*}
  
  \begin{figure*}[htb!]
    \includegraphics[width=14cm]{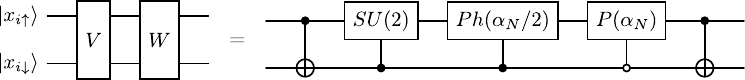}
    \caption{\textit{Decomposition of WV Gates (Equation (\ref{eq:wv_circuit})) Into a Quantum Circuit}. Gates controlled on the $\ket{1}$ ($\ket{0}$) state correspond to action in the $\{ \ket{01}, \ket{10} \}$ ($\{ \ket{00}, \ket{11} \}$) subspace of $\ket{x_{i \uparrow} x_{i \downarrow}}$. $Ph(\alpha_N / 2) = \text{diag}(e^{i \alpha_N/2}, e^{i \alpha_N/2}) \otimes |1 \rangle \langle 1 |$ is a controlled global phase gate. $P(\alpha_N) = \text{diag}(1, e^{i \alpha_N}) \otimes |0 \rangle \langle 0 |$ is the standard controlled phase gate. This forms an example of a 2-qubit matchgate.}
    \label{fig:wv_gate}
  \end{figure*}

  \subsubsection{Single-body \texorpdfstring{$\mathfrak{u}(d)$}{u(d)} Hamiltonians}

  We now proceed in analogy to the previous section, looking closely at the case of evolution by spin-free one-body Hamiltonians, i.e. members of the Lie algebra representation $\mathfrak{u}(d)$, composed of terms linear in $E_{ij}$. Hamiltonians in this class can be written as:
  
  \begin{equation}
  H_{\mathfrak{u}(d)} = \sum_{i < j}( \beta_{ij} F_{ij} + \beta^{ij} F^{ij}) + \sum_{i=1}^d \beta_{ii} F_{ii}, \quad \beta_{ij}, \beta^{ij} \in \mathbb{R}, \label{eq:hud_hamiltonian}
  \end{equation}
  
  \noindent where we have rewritten the basis of generators $E_{ij}$ in terms of their Hermitian linear combinations:

  \begin{align}
    F_{ii} &= E_{ii}, &
    F_{ij} &= \frac{1}{2} (E_{ij} + E_{ji}),  &
    F^{ij} &= \frac{i}{2} (E_{ij} - E_{ji})  \label{eq:f_generators}
  \end{align}
  
  \noindent We will find that evolution by $\mathfrak{u}(d)$ Hamiltonians is again expressible as a matchgate circuit. In order to make this clearer, we first \textit{re-order} the Jordan-Wigner operators, i.e. apply a circuit $\sigma$ which permutes the occupation basis:
  
  \begin{equation}
  \sigma \ket{x_{1 \uparrow} x_{1 \downarrow} x_{2 \uparrow} x_{2 \downarrow} \cdots x_{d \uparrow} x_{d \downarrow}} = \ket{x_{1 \uparrow} x_{2 \uparrow} x_{3 \uparrow} \hdots x_{d \uparrow} x_{1 \downarrow} x_{2 \downarrow} x_{3 \downarrow} \cdots x_{d \downarrow} }
  \end{equation}
  
  \noindent This can be achieved by a series of \textit{fermionic} $fSWAP$ gates, which re-order the the fermionic modes whilst preserving the form of the creation and annihilation operators \cite{Kivlichan2018QuantumForm}. A full re-ordering requires $d(d-1)/2$ $fSWAP$ gates, as shown in Figure \ref{fig:fswap_sequence_matchgate_circuit}. After re-ordering, the creation and annihilation operators can be written as:
  
  \begin{align*}
    \tilde{a}^\dagger_{j \mu} &= \underbrace{Z \otimes \cdots \otimes Z}_{j + d(1/2 - \mu) - 1} \otimes \frac{1}{2} \left(X - i Y \right) \otimes \underbrace{I \otimes \cdots \otimes I}_{2d - j - d(1/2 - \mu)} = \frac{1}{2}(\tilde{c}_{2j + d(1 - 2\mu) - 1} - i \tilde{c}_{2j + d(1 - 2\mu)} ), \\
    \tilde{a}_{j \mu} &= \underbrace{Z \otimes \cdots \otimes Z}_{j + d(1/2 - \mu) - 1} \otimes \frac{1}{2} \left(X + i Y \right) \otimes \underbrace{I \otimes \cdots \otimes I}_{2d - j - d(1/2 - \mu)} = \frac{1}{2}(\tilde{c}_{2j + d(1 - 2\mu) - 1} + i \tilde{c}_{2j + d(1 - 2\mu)} )\nonumber
  \end{align*}
  
  \noindent Again, we have written the ladder operators in terms of Majoranas $\tilde{c}_i$, with tildes to denote the re-ordering. Following the definition of $E_{ij} = \sum_\mu a_{i \mu}^\dagger a_{j \mu}$ and $F_{ij}, F^{ij}$ from Equation \eqref{eq:f_generators}, we again re-write the re-ordered Hermitian $\mathfrak{u}(d)$ generators in terms of Majorana operators as:
  
  \begin{align*}
    \tilde{F}_{ii} &= \frac{1}{2} (I + i \tilde{c}_{2i - 1} \tilde{c}_{2i}) + \frac{1}{2} (I + i \tilde{c}_{2d + 2i - 1 } \tilde{c}_{2d + 2i}) = \tilde{F}^{\uparrow}_{ii} + \tilde{F}^{\downarrow}_{ii}  \\
    \tilde{F}_{ij} &= \frac{1}{2}(i\tilde{c}_{2i - 1} \tilde{c}_{2j} - i \tilde{c}_{2i} \tilde{c}_{2j - 1}) + \frac{1}{2} (i \tilde{c}_{2d +2i - 1} \tilde{c}_{2d +2j} -i \tilde{c}_{2d +2i} \tilde{c}_{2d +2j - 1}) = \tilde{F}^{\uparrow}_{ij} + \tilde{F}^{\downarrow}_{ij}\\
    \tilde{F}^{ij} &= \frac{1}{2}(i\tilde{c}_{2i - 1} \tilde{c}_{2j-1} + i \tilde{c}_{2i} \tilde{c}_{2j} ) + \frac{1}{2} (i \tilde{c}_{2d +2i - 1} \tilde{c}_{2d +2j - 1} +i \tilde{c}_{2d +2i} \tilde{c}_{2d +2j}) = \tilde{F}^{\uparrow, ij} + \tilde{F}^{\downarrow, ij}
  \end{align*}
  
  \noindent Where we have split the first two generators into a sum of `partial' generators $\tilde{F}^{\uparrow}_{ij} + \tilde{F}^{\downarrow}_{ij}$, and the third into a sum $\tilde{F}^{\uparrow, ij} + \tilde{F}^{\downarrow, ij}$. Importantly, all partial generators indexed by $\uparrow$ commute with those indexed by $\downarrow$. Similarly to the $\mathfrak{u}(2)$ case, this lets us group up the $\uparrow$ and $\downarrow$ terms into two Hamiltonians, which act identically on the first and last $d$ qubits. We then write the evolution by $H_{\mathfrak{u}(d)}$ (in the re-ordered basis) as:
  
  \begin{align}
    \sigma e^{iH_{\mathfrak{u}(d)}} \sigma^\dagger &= e^{i \sum_{ij} (\beta_{ij} \tilde{F}^\uparrow_{ij} + \beta^{ij} \tilde{F}^{\uparrow ij} + \beta_{ii} \tilde{F}^{\uparrow}_{ii}  ) + i \sum_{ij} (\beta_{ij} \tilde{F}^\downarrow_{ij} + \beta^{ij} \tilde{F}^{\downarrow ij} + \beta_{ii} \tilde{F}^{\downarrow}_{ii}  ) } \\
    &= e^{i \sum_{ij} (\beta_{ij} \tilde{F}^\uparrow_{ij} + \beta^{ij} \tilde{F}^{\uparrow ij} + \beta_{ii} \tilde{F}^{\uparrow}_{ii}  ) } e^{i \sum_{ij} (\beta_{ij} \tilde{F}^\downarrow_{ij} + \beta^{ij} \tilde{F}^{\downarrow ij} + \beta_{ii} \tilde{F}^{\downarrow}_{ii}  ) } \\
  &= e^{i \sum_{ij} (\beta_{ij} \tilde{F}^\uparrow_{ij} + \beta^{ij} \tilde{F}^{\uparrow ij} + \beta_{ii} \tilde{F}^{\uparrow}_{ii}  ) } \otimes e^{i \sum_{ij} (\beta_{ij} \tilde{F}^\uparrow_{ij} + \beta^{ij} \tilde{F}^{\uparrow ij} + \beta_{ii} \tilde{F}^{\uparrow}_{ii}  )} \\ 
  &=  G \otimes G \label{eq:ud_g_circuit}
  \end{align}
  
  \noindent In the first line, we have split the re-ordered Hamiltonian $\tilde{H}_{\mathfrak{u}(d)} = \tilde{H}^\uparrow + \tilde{H}^\downarrow$ into two terms indexed by $\uparrow$ and $\downarrow$, which act identically on the first $d$ and last $d$ qubit lines. Writing these out in full, every non-identity term in the exponent is a product of two Majorana operators $\tilde{c}_i \tilde{c}_j$ (where $i, j \leq 2d$). For example, terms in the $\uparrow$ part are always a Pauli string on the first $1 \leq i \leq d$ qubits and an identity on the last $d \leq i \leq 2d$ qubits, and vice versa for the $\downarrow$ part. Because of this, $\tilde{H}^\uparrow$ and $\tilde{H}^\downarrow$ commute, so we can write the first line as a product of two matrix exponentials. Finally, because $\tilde{H}^\uparrow$ and $\tilde{H}^\downarrow$ act identically on their respective qubit lines, they can be re-expressed as a tensor product of the form $G \otimes G$. 
  
  Because each term in the exponent of $G$ only contains terms with $I$ and $\tilde{c}_i \tilde{c}_j$, it satisfies the definition of quadratic Hamiltonians. Therefore, each $G$ can be re-expressed as a product of $\mathcal{O}(\text{poly}(d))$ nearest-neighbour matchgates, forming a matchgate circuit. Each $fSWAP$ gate is a also matchgate, satisfying Equation \eqref{eq:matchgatedef}, so the circuit $\sigma$ which re-orders the fermionic modes is another matchgate circuit. 
  
  We have thus shown that unitary evolution by $H_{\mathfrak{u}(d)}$ is expressible as $\sigma^\dagger ( G \otimes G ) \sigma$, a \textit{classically simulable matchgate circuit}. This is unsurprising, given that single-body Hamiltonians are also known to be efficiently simulable. However, this decomposition gives a more concrete form of $U(d)$ representations obtained via $H_{\mathfrak{u}(d)}$, and sheds light on the importance of the universal enveloping algebra $\mathfrak{U}(\mathfrak{u}(d))$ in Hamiltonian simulation. In general, when $H$ belongs to this class, the situation is not so simple and the Hamiltonian evolution does not admit a simple decomposition into elementary gates. Indeed, most spin-free Hamiltonians of interest are not expressible as $H_{\mathfrak{u}(d)}$, but rather $H_{\mathfrak{U}(\mathfrak{u}(d))}$, where $\mathfrak{U}(\mathfrak{u}(d))$ is the universal enveloping of $\mathfrak{u}(d)$. This is the set of all possible Hermitian operators expressible as sums of $E_{ij}$'s and their products, as discussed at the beginning of Section \ref{sec:uga} and Appendix \ref{app:universalenveloping}.

\begin{figure*}[htb!]
    \includegraphics[width=9cm]{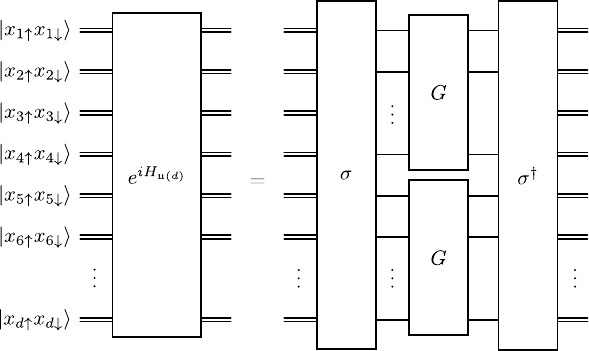}
    \caption{\textit{Evolution by Single-body $\mathfrak{u}(d)$ Hamiltonians.} The gate $e^{iH_{\mathfrak{u}(d)}}$ is the unitary evolution by a Hamiltonian in the form of Equation \eqref{eq:hud_hamiltonian}. With a re-ordering of qubit lines through a sequence of $fSWAP$ gates $\sigma$, this can be described as a tensor product of matchgate circuits (Figure \ref{fig:fswap_sequence_matchgate_circuit}). } \label{fig:ud_circuit}
  \end{figure*}
  
  \begin{figure*}[htb!]
    \centering
    \hspace*{+0cm}
    \begin{minipage}{.5\textwidth}
        \includegraphics[width=7.5cm]{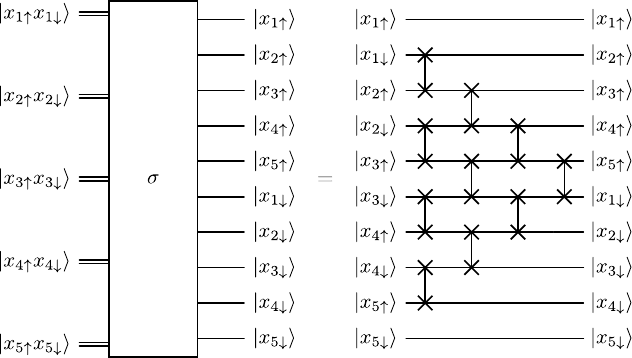}
    \end{minipage}%
    \begin{minipage}{.5\textwidth}
        \includegraphics[width=7.5cm]{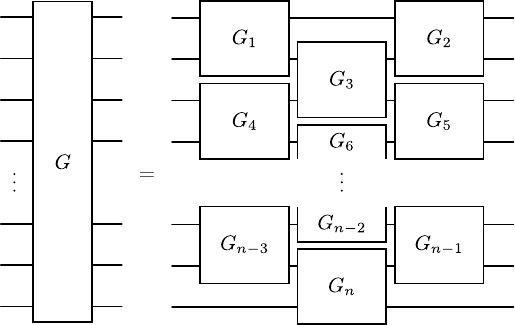}
    \end{minipage}
    \caption{\textit{Left:} The circuit $\sigma$ used for re-ordering occupation number labels. Each element of this circuit is an $fSWAP$ gate. \textit{Right:} A matchgate circuit. Each $G_i$ gate satisfies Equation \eqref{eq:matchgatedef}. Figure \ref{fig:ud_circuit} describes the unitary evolution by spin-free Hamiltonians as gate compositions of the form $e^{i H_{\mathfrak{u}(d)}} = \sigma^\dagger ( G \otimes G ) \sigma$.}
    \label{fig:fswap_sequence_matchgate_circuit}
  \end{figure*}

\subsection{Fast-forwarding Quantum Evolution in the Second Quantisation} \label{sec:fastforwarding}

The quantum Paldus transform is readily applicable to the \textit{Hamiltonian fast-forwarding} algorithm proposed by \cite{Gu_2021}. In their work, the authors use properties of the Schur basis $\ket{S} \otimes \ket{M} \otimes \ket{Y}$ to simulate the action of permutation-invariant $SU(2)$ Hamiltonians with an exponential speedup relative to working in the computational basis. The idea behind this approach is that because each $SU(2)$ irrep space is of only a polynomially large dimension ($2S + 1$), after transforming some initial state $\ket{\psi}$ into the Schur basis one can implement unitaries controlled on the $\ket{S}$ register in $\mathcal{O}(\text{poly}(n))$ time to quickly simulate its evolution. Although each $SU(2)$ irrep has exponentially large multiplicity (i.e. spin degeneracy, given by $d_S$ in Equation \eqref{eq:spin_degeneracy}), in the Schur basis each multiplicity space is acted on simulatenously by the controlled operation, leading to the speedup. This is then applied to the Lipkin-Meshkov-Glick model of nuclear physics \cite{meshkov1965validity}. Our basis satisfies the block-diagonalisation requirements of the fast-forwarding theorem in \cite{Gu_2021}, and in particular this enables fast-forwarded simulation of $\mathfrak{U}(\mathfrak{u}(2))$ Hamiltonians polynomial in the $\{ \mathcal{E}_{\mu \nu}\}$ ladder operators. In particular, whilst using the Schur transform restricts us to first-quantised models, the Paldus transform extends this to scenarios where a second quantisation is more appropriate.

\subsection{Quantum Chemistry in the UGA basis} \label{sec:chemistry}

\begin{figure}[htbp!]
  \includegraphics[width=12cm]{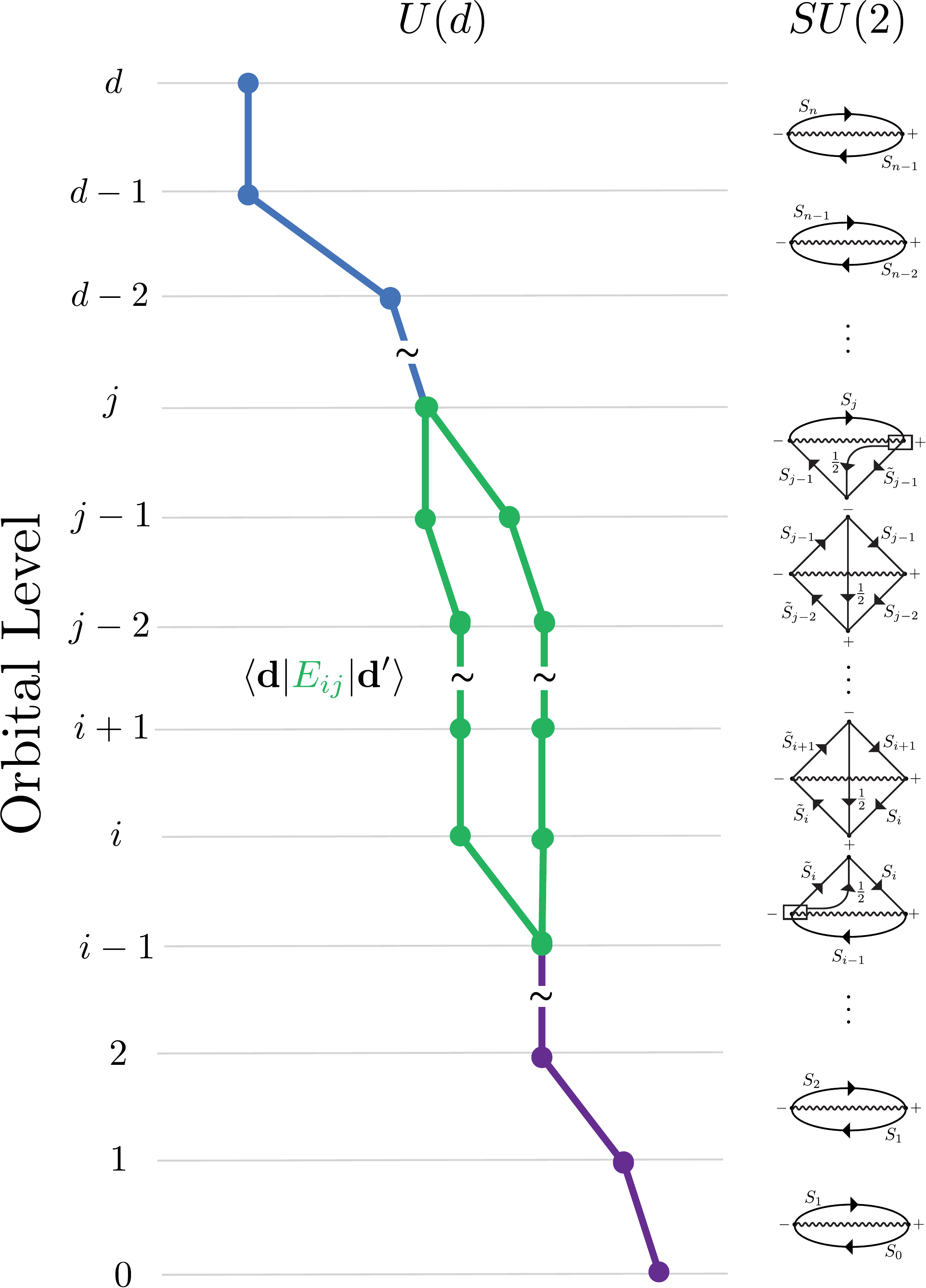}
  \caption{\textit{Shavitt Loops.} The Shavitt graph is a representation of the UGA basis states as a graph, where each edge is a step vector $\mathbf{d}$. The Shavitt loop is a closed path in the graph which represents the overlap of the $U(d)$ raising operator $\langle \mathbf{d} | E_{ij} | \mathbf{d}' \rangle$ (shown in green). Each segment of the loop carries an $SU(2)$ angular momentum coupling tensor which is shown diagrammatically using Jucys' segmented network construction \cite{Brink:1975}, calculated from the $S$ values of the pair of top and bottom nodes of the loop segment. The product of these coefficients gives the matrix element, as was shown by Paldus and Boyle \cite{Paldus1980}.}
  \label{fig:shavitt_loops}
\end{figure}

At the beginning of Section \ref{sec:applications}, we have mentioned applications of the Paldus transform for state preparation in the occupation basis via projections. Conversely, one can use the efficient Paldus transform to carry out quantum chemistry calculations in the UGA basis \textit{itself}, making use of the greater sparsity obtained in the block-diagonal form of $U(d) \times U(2)$ representations. Here we briefly summarise some of these potential techniques. 

In general, the dimension $T^d_{S, N}$ (Equation \eqref{eq:weyl_dimension}) of any $U(d)$ irrep in the GT basis is at least \textit{polynomially} smaller than the full dimension of the Fock space, and the $U(2)$ irreps of dimension $(2S+1)$ are exponentially smaller. Carrying out simulation techniques such as trotterization, or block encoding thus has the potential to be more efficient in the UGA basis. To aid in this, there exist plenty of computational techniques for calculating the matrix elements of ladder operators $\bra{\mathbf{d}} E_{i, i+1} \ket{\mathbf{d}'}$~\cite{Robb1984}, and together with the recurrence relations $[E_{i,j}, E_{j, j+1}] = E_{i, j+1}$ they can be used to calculate matrix elements for any ladder operator~\cite{Hegarty01121979}. These could be used to implement the ladder operators $E_{ij}$ in the UGA basis directly on a quantum computer. If we order the step vectors in a particular way \cite{PALDUS1976131}, one can show that similarly to the Jordan-Wigner case the ladder operators $E_{ij}$ are $|i-j|$-local, that is, $\bra{\mathbf{d}} E_{ij} \ket{\mathbf{d}'}$ is zero if step vectors $\mathbf{d}, \mathbf{d}'$ differ outside the range $[i, j]$. 

Furthermore, each $E_{i, i+1}$ can contain at most two nonzero elements. The exact values of the matrix elements can be calculated from the \textit{overlaps} of two Shavitt graph paths induced by the step vectors $\mathbf{d}, \mathbf{d}'$, in a technique known as the \textit{Graphical Unitary Group Approach} (GUGA). An example of a GUGA calculation of a matrix element $\bra{\mathbf{d}} E_{ij} \ket{\mathbf{d}'}$ via Shavitt path overlaps is shown in Figure \ref{fig:shavitt_loops}, Where the $SU(2)$ angular momentum coupling \cite{Brink:1975,Paldus1980}, carried by each step vector $\bra{\mathbf{d}}$ leads directly to the matrix element value. Calculation of density matrix elements in the UGA basis using these methods on a quantum computer will be the topic of future work.

\subsection{Preparing Uniform Superpositions of CSFs} \label{sec:CSFprep}

The ability to project and post-select via the Paldus transform lets us prepare \textit{Configuration State Functions} (CSFs) with well-defined quantum numbers for use as initial states for quantum algorithms (for example, approximating ground state energies of a chemical system). This is useful, for example, in the context of variational algorithms where one starts out by preparing an \textit{ansatz} \cite{Puig_2025, Truger_2024}, or in finding ground state energies via phase estimation with the ground state known to lie in some $(N, S)$ symmetry sector. One way to prepare an ansatz with well-defined quantum numbers is to prepare a superposition of all bitstrings through a Hadamard transform, apply the Paldus transform and post-select on the desired quantum numbers. Then, one can apply $U_P^\dagger$ to obtain the desired CSF mixture in the occupation basis. 

Here we consider a slightly different approach, which is especially useful in state preparation for algorithms involving spin-free Hamiltonians $H_{\mathfrak{U}(\mathfrak{u}(d))}$. As noted previously, such Hamiltonians \textit{do not affect} the spin projection $M$, so the ground state (and its energy) will be independent of this value. Similarly to the classical unitary group approach, in this scenario the $M$ label of each GT state can be set arbitrarily -- we choose this to be the highest value, i.e. $M = S$. With this in mind we seek to prepare a uniform superposition of \textit{all step vectors}, given by a state in the UGA basis of the form:

\begin{equation}
  \ket{\Psi} = \frac{1}{\sqrt{|\{ \mathbf{d} \}|}}  \sum_{\mathbf{w} \in \{ \mathbf{d} \} } \ket{N}_N \otimes \ket{S}_S \otimes \ket{S}_M \otimes \ket{\mathbf{w}}_\mathbf{d} \label{eq:all_csf_superposition}
\end{equation}

\noindent Where the sum is evaluated over all valid step vectors $\{ \mathbf{d} \}$, the quantum numbers $N, S$ are those corresponding to each $\mathbf{w}$ (c.f. Table \ref{fig:step_vectors_binary}), and for each constituent UGA basis state the spin projection is set to $M=S$. This is essentially a linear combination of all GT states of $U(d)$ ignoring the multiplicities $(2S+1)$ of their subspaces.

The total number of valid step vectors is smaller than the number of $2d$-bit strings, as not every string forms a step vector (for example, any step vector starting with $\mathbf{01}$ is not valid). A valid step vector for a subspace $(N, S)$ satisfies the aforementioned constraints in Equation \eqref{eq:ugabasisconstraints}. The total number of $\mathbf{d}$'s which satisfy this condition is given by:

\begin{equation}
  T^d = | \{ \mathbf{d} \} | = {2d + 1 \choose d }
\end{equation}

\noindent which we derive in Appendix \ref{app:dimension_formula}. Whilst it is possible to prepare $\ket{\Psi}$ in a variety of ways (for example, using amplitude amplification) we present a more practical method based on post-selection. The idea is to prepare a superposition of all bitstrings and eliminate those which form invalid step-vectors. To do this, we exploit the two's complement scheme used to store $M$ values in the UGA basis. Recall that in the two's complement scheme, the first bit is used to store the sign of the number, being $0$ for a nonnegative integer and $1$ for a negative integer. This means that measuring the \textit{first bit} of the $M$ register will collapse it onto one of two subspaces: $\{ \ket{M < 0} \}$ or $\{\ket{M \geq 0} \}$. For example, if we prepare a uniform superposition of all $2$-bit strings, append an ancilla register $\ket{0}_M$, and apply the incrementer circuit of Figure \ref{fig:incrementer_m} we will get the state $\frac{1}{4}( \ket{0}_M \ket{00} + \ket{1/2}_M \ket{10} + \ket{-1/2}_M \ket{01} + \ket{0}_M \ket{11})$. Measuring $0$ on the first qubit of $M$ will collapse this state to $\frac{1}{\sqrt{3}}(\ket{0}_M \ket{\mathbf{00}}_\mathbf{d} + \ket{1/2}_M \ket{\mathbf{10}}_\mathbf{d} + \ket{0}_M \ket{\mathbf{11}}_\mathbf{d} )$, which is a uniform superposition of all valid step vectors (with $S = M$) for $d=1$. Repeating this process by appending a superposition of all bitstrings onto the end, applying the incrementer, and post-selecting on $M \geq 0$ we obtain a uniform superposition of all states $\frac{1}{\sqrt{|\{\mathbf{d}\}|}} \sum_{\mathbf{d}} \ket{S}_M \ket{\mathbf{d}}_\mathbf{d}$. Appending ancilla registers $\ket{0}_N \ket{0}_S$ and setting them to their respective values with another series of incrementers then returns $\ket{\Psi}$. In summary, the algorithm to prepare the uniform CSF superposition is as follows:

\begin{enumerate}
  \item Prepare the blank state $\ket{0}_M \otimes \ket{0}_\mathbf{d}$. \\  \\
  \textit{For $i = 1, ..., d$:}
  \item Apply a Hadamard gate to qubits $2i-1$ and $2i$ of the $\mathbf{d}$ register.
  \item Apply a controlled incrementer gate (of the form in Figure \ref{fig:incrementer_m}), with qubits $2i-1$ and $2i$ as the control, and $M$ as the target.
  \item Measure the \textit{first} qubit of $M$. If the measurement reads $1$, discard the state and start again. If it reads $0$, repeat steps $2$ -- $4$ for $i \rightarrow i + 1$.  
  \item Append blank ancilla registers for $S, N$ and apply controlled incrementers (Figures \ref{fig:incrementer_s}, \ref{fig:incrementer_n}) to obtain $\ket{\Psi}$.

\end{enumerate}

\noindent To show that the above requires a polynomial number of repetitions for a constant probability of success, we note that at each step the probability of success (i.e, measuring $0$ as the first qubit of $M$) is given by:

\begin{equation}
  \Pr(S_i \geq 0 | S_{i-1}, \hdots, S_1 \geq 0) = \frac{|\{ \mathbf{d}\}|_i}{4 \times |\{ \mathbf{d}\}|_{i-1}} = \frac{{2i + 1 \choose i}}{4 \times {2i - 1 \choose i}} 
\end{equation}

\noindent Hence, the probability of success after $d$ steps is given by:

\begin{equation}
\Pr(S_i \geq 0 \ \forall i = 1, ..., d) = \prod_{i=1}^d \Pr(S_i \geq 0 | S_{i-1}, \hdots, S_1 \geq 0) = \frac{1}{4^d} {2d + 1 \choose d}
\end{equation}

\noindent which is just the ratio between $2d$-bit strings which form valid step vectors and all $2d$-bit strings. Asymptotically, the combinatorial factor behaves like the central binomial coefficient, with a scaling of $\mathcal{O}(4^d / \sqrt{\pi d})$, so the success probability scales as $\mathcal{O}(1 / \sqrt{d})$. Hence to prepare $\ket{\Psi}$ with a success probability of $1 - \delta$, we will need to repeat the above process $\mathcal{O}(\sqrt{d} \log(1 / \delta))$ times.

Once the state $\ket{\Psi}$ is prepared, one can perform further projections to form equal superpositions of $U(d)$ GT states within an $(N, S)$ irrep, with a success probability given by $T^d_{N, S} / T^d$ (where $T^d_{N, S}$ is the dimension of the irrep given in Appendix \ref{app:dimension_formula}). At this point, amplitude amplification can also be used to boost the relative amplitude of states within the desired projection sector, increasing the probability of success. Values of $M$ for each term making up $\ket{\Psi}$ in Equation \eqref{eq:all_csf_superposition} can also be adjusted with incrementers controlled on $\ket{S}_S$ with $\ket{S}_M$ as the target. Similarly, applying a sequence of controlled unitaries acting as $\ket{S}_S \otimes \ket{S}_M \rightarrow \frac{1}{\sqrt{2S+1}}\ket{S}_S \otimes \sum_{M=-S}^S \ket{M}$ onto the $S, M$ registers can also be used to prepare a superposition of \textit{all} UGA basis states, i.e. a state of the form:

\begin{equation}
\frac{1}{2^d} \sum_{S=0}^{d/2} \sum_{N=2S}^{2d - 2S} \sum_{M=-S}^S \sum_{\mathbf{d}} \ket{N}_N \otimes \ket{S}_S \otimes \ket{M}_M \otimes \ket{\mathbf{d}}_\mathbf{d}
\end{equation}

\noindent After the desired superposition is obtained, the inverse Paldus transform can be applied to represent the state in the occupation basis for use in a quantum chemistry simulation.

\begin{figure*}[htb!]
  \includegraphics[width=17cm]{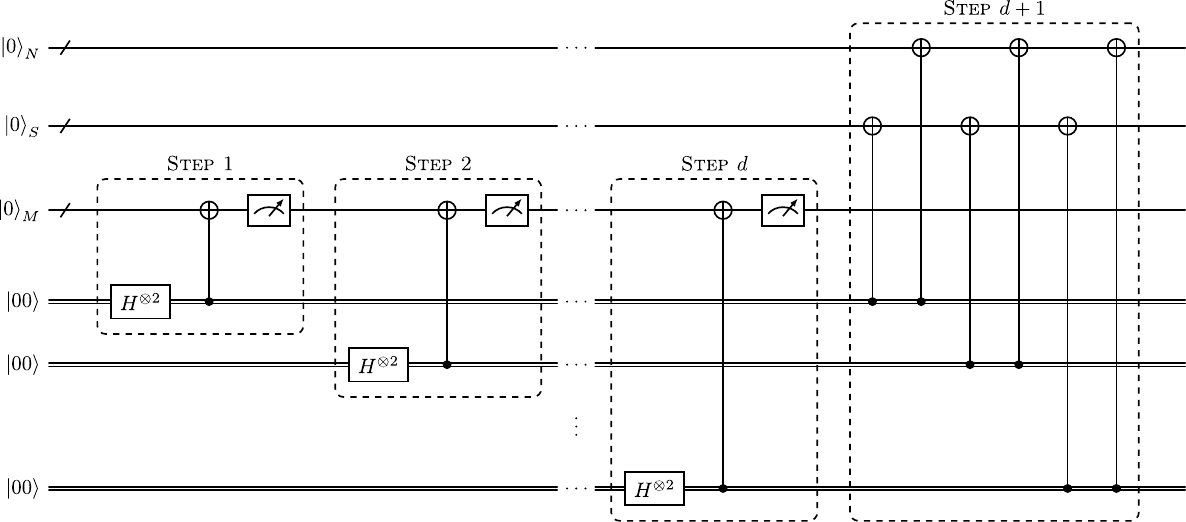}
  \caption{\textit{Preparing a Uniform Superposition of CSFs.} Each step is an application of $H \otimes H$ to a qubit pair, followed by a measurement of the \textit{first} qubit of the $M$ register. The final state is the uniform superposition of valid step vectors $\ket{\Psi}$ given in Equation \eqref{eq:all_csf_superposition}. The circuit succeeds with a probability of $\mathcal{O}(1 / \sqrt{d})$, where $d$ is the number of orbitals. The resulting state may be brought into the occupation number basis by applying the inverse Paldus transform (Figure \ref{fig:paldus_transform}).}
  \label{fig:all_csf_preparation}
\end{figure*}

\subsection{Paldus Transform and the Schur Transform} \label{sec:paldusschur}

The quantum Schur and Paldus transforms are closely related. Here we discuss their similarities, and identify applications of the Schur transform which carry over to the Paldus transform. With an appropriate mapping of $d$ qubit states into $2d$-qubit states, we will show that any application of the former can be emulated by the latter. We have already used some of these similarities to derive a circuit for the Paldus transform through a cascade of Clebsch-Gordan transforms, and to make this complete we also provide alternative $\mathcal{O}(d \cdot \text{poly}(\log(d), \log(1/\varepsilon)) )$ circuit for the Paldus transform, which uses the same building blocks as the $SU(2)$ Schur transform but with the Clebsch-Gordan transforms as controlled unitaries.

\subsubsection{Emulating the Qubit Schur Transform with the Paldus Transform} \label{sec:schur_emulation}

As observed in Section \ref{sec:gt_states_construction}, each GT state in the $2d + 2k$ qubit occupation basis may be obtained from a (computational basis) Schur state on $d$ qubits, with the mapping $\ket{\uparrow} \rightarrow \ket{10}$ and $\ket{\downarrow} \rightarrow \ket{01}$ followed by a tensor product of $k$ many $ \{ \ket{00}, \ket{11} \}$ states. For example, the singlet state $\frac{1}{\sqrt{2}}(\ket{\uparrow \downarrow} - \ket{\downarrow \uparrow})$ is mappable to $\frac{1}{\sqrt{2}}(\ket{1001} - \ket{0110})$, or $\frac{1}{\sqrt{2}}(\ket{100001} - \ket{010010})$, or $\frac{1}{\sqrt{2}}(\ket{00101101} - \ket{00011110})$, and so on. In the UGA setting, all of these states are $S=0$ singlets with differing particle numbers $N$. Despite the presence of spin-$0$ orbitals, their relative amplitudes are the same, as the Clebsch-Gordan coefficients corresponding to spin-$0$ coupling are always unity.

With this in mind, let us consider an occupation basis state acted on by a Paldus transform which consists exclusively of half-occupied orbitals (i.e. $\ket{x_{i \uparrow} x_{i \downarrow}} \in \{ \ket{10}, \ket{01} \}$ for all $i$). All UGA basis states in the resulting superposition will contain $N = d$ particles. Following the definition of the Clebsch-Gordan transform in Section \ref{sec:cgtransform}, we also see that each term in the output step vectors will be an odd-parity bitstring (i.e. $\ket{\mathbf{d}}_i \in \{ \ket{\mathbf{10}}, \ket{\mathbf{01}} \}$), as only the second and third row of the four possible inputs of Equation \eqref{eq:cg_transform_rules} are realised on each step. Such step vectors are in a one-to-one correspondence with $d$-bit \textit{Yamanouchi symbols}, which index the $S_d$ irrep labels of $SU(2)$ Schur states. The mapping is given by $\mathbf{10} \leftrightarrow 0$ and $\mathbf{01} \leftrightarrow 1$. Therefore, on inputs with half-occupied orbitals only, the action of the Paldus transform is identical to the Schur transform, with the output states forming a superposition of states $\ket{d}_N \otimes \ket{S}_S \otimes \ket{M}_M \otimes \ket{\mathbf{Y}}_\mathbf{d}$, where the step vectors $\mathbf{Y}$ correspond to Yamanouchi symbols encoded in $2d$ qubits. Consequently, the $SU(2)$ Schur transform on $d$ qubits can be emulated by the Paldus transform with $d$ ancillary qubits. This is achieved by `encoding' each qubit of the input state into the $\{ \ket{10}, \ket{01} \}$ subspace of two qubits using the circuit in Figure \ref{fig:paldus_schur_gadget}, applying the Paldus transform, and then `decoding' each qubit pair of the resulting step vector by applying the encoder in reverse and discarding the ancillae. Putting aside differences in notation, the resulting states will be identical to the superposition of the $\ket{S} \otimes \ket{m} \otimes \ket{p}$ states output by the algorithm of \cite{BCH_2005}. Therefore, any application of the $SU(2)$ Schur transform can be emulated by the Paldus transform with a polynomial overhead in qubits. The converse is not true, which we demonstrate in Section \ref{sec:dfs} with a protocol reliant on the Paldus transform that cannot be performed with the qubit Schur transform.

\begin{figure*}[htb!]
  \includegraphics[width=8cm]{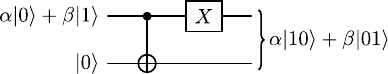}
  \caption{\textit{Encoding Qubit States into Singly-occupied Fock Basis States.} Following the discussion above, appending each register of an $n$-qubit state with such circuit allows for the emulation of the Schur transform with the Paldus transform.} \label{fig:paldus_schur_gadget}
\end{figure*}

\subsubsection{Linear-depth Circuits for the Paldus Transform} \label{sec:paldus_cg}

In the preceding discussion, we have found similarities between Schur basis states and the $U(d) \times U(2)$ GT states, and shown that the Paldus transform can emulate the applications of the Schur transform. We now relate the \textit{structure} of the two algorithms more closely. We have previously described the Paldus transform as a generalisation of the Schur transform for the second quantisation -- here we offer an alternative view of the Paldus transform as a sequence of \textit{controlled} Clebsch-Gordan transforms. This will allow us to construct a linear-depth circuit for the Paldus transform, using the same building blocks as the $SU(2)$ Schur transform.

Our observation relies on the close analogy between GT, and Schur states discussed in-depth in Sections \ref{sec:gt_states_construction} as well as Section \ref{sec:schur_emulation}. When deriving the Clebsch-Gordan transform for the Paldus transform, we have noted that rotations in Clebsch-Gordan coupling only occur on inputs $\ket{x_{i \uparrow} x_{i \downarrow} } \in \{ \ket{10}, \ket{01 }\}$ (c.f. Equation \eqref{eq:cg_transform_rules}). On such inputs, the coupling is identical to that of the qubit Clebsch-Gordan transform in \cite{BCH_2006} (which we will refer to as $U_{CG}$), but succeeded by the $N$ register update (Section \ref{sec:encoding_n}). This raises the question: \textit{given a black box for the $SU(2)$ Clebsch-Gordan transform, $U_{CG}$, can we use it to construct the Paldus transform?} The answer is yes, as long as $U_{CG}$ can be implemented as a controlled unitary. To show this, we consider our input states $\ket{x_{i \uparrow} x_{i \downarrow} }$ to which we apply a simple $CNOT$ gate (controlled on the first bit), mapping them as:

\begin{align}
  \ket{00} & \rightarrow \ket{00}, &
  \ket{01} & \rightarrow \ket{01}, &
  \ket{10} & \rightarrow \ket{11}, &
   \ket{11} & \rightarrow \ket{10} \label{eq:cnot}
\end{align}

\noindent After this mapping, we obtain two-qubit states $\ket{\tilde{x}_{i,1} \tilde{x}_{i,2}}$ with the \textit{second} qubit denoting whether the orbital is singly-occupied ($\tilde{x}_{i,2} = 1$), or doubly-occupied ($\tilde{x}_{i,2} = 0$). Then, for singly-occupied orbitals the first qubit $\tilde{x}_{i,1}$ is set to $1$ if the orbital is spin-up, and $0$ if it is spin-down. We only need to update the $(S, M)$ registers and apply controlled rotations if this is the case, which can be achieved by the $U_{CG}$ gate \textit{controlled} on the second qubit. When $\tilde{x}_{i,2} = 0$ no rotations are implemented, in agreement with the first and last line of Equation \eqref{eq:cg_transform_rules}. Following the action of $U_{CG}$, the first qubit $\tilde{x}_{i,1}$ turns into a Yamanouchi symbol $y_i$, indicating whether $S$ was raised ($y_i = 1$) or lowered ($y_i=0$). However, the second qubit $\tilde{x}_{i,2} = 1$ is unchanged, so we can apply another $CNOT$ to turn the resulting two-qubit state into a step vector element $\ket{\mathbf{d}_{2i-1} \mathbf{d}_{2i}}$. Finally, we can update the $N$ register just as we do in the circuit of Section \ref{sec:paldus_transform}, which is accomplished by an incrementer. This is summarised in Figure \ref{fig:paldus_transform_cg}.

This alternative implementation of the Paldus transform works because the controlled $U_{CG}$ + $N$ incrementer steps implements the same mapping as the Clebsch-Gordan transform described in Section \ref{sec:cgtransform}. However, there are subtle differences we wish to highlight. The first is that $U_{CG}$ and `our' CG transform perform different types of decompositions: the former according to Equation \eqref{eq:swduality}, whilst the latter according to Paldus duality in Section \ref{sec:paldus_duality}. In the former, there is no $N$ register to keep track of and the step vector elements $\ket{\mathbf{d}}$ are in general not mappable to Yamanouchi symbols. It is only when $U_{CG}$ is controlled and succeeded by the $N$ incrementer that the two transforms coincide. The second difference is that the $U_{CG}$ circuit in \cite{BCH_2005, BCH_2006} does not necessarily store the $(S, M)$ values in the same scheme as the circuit of Section \ref{sec:paldus_transform} -- for example, they could be stored in unary encoding, as done in \cite{Gu_2021}. This does not significantly affect the applications of the algorithm, but may lead to a different ancilla qubit count. In fact, the circuit in the implementation of \cite{BCH_2005} does not store the $S$ value at all, leaving it implicit in the outputs $\ket{\lambda} \ket{M}$ (where $\lambda$ is an $S_n$ irrep label). In the qubit case $S$ can be reconstructed from $\lambda$, so this again is not a problem. The only constraint we impose on $U_{CG}$ is that it stores the $S_n$ irrep labels in the \textit{Yamanouchi encoding}, so that its output $\ket{y_i}$ is a Yamanouchi symbol mappable to a step vector element via Equation \eqref{eq:cnot}. Importantly, sub-linear complexity Clebsch-Gordan transforms for the Yamanouchi encoding exist, with a thorough implementation given in \cite{Grinko_2025}.

In summary, with this construction we are free to use the more efficient (but \textit{approximate}) Clebsch-Gordan transform of \cite{BCH_2005, Grinko_2025} which is correct up to precision $\varepsilon$ but only runs in $\mathcal{O}(\text{poly}(\log(d), \log(1/\varepsilon)) )$ time, with an extra control on $\ket{\tilde{x}_{i,2}}$ giving a negligible overhead. This transform is then called $d$ times, giving a linear-depth, $\mathcal{O}(d \cdot \text{poly}(\log(d), \log(1/\varepsilon)) )$ algorithm for the Paldus transform. As this method relies on quantum arithmetic circuits, it is not always the preferrable choice. Nonetheless, its existence brings down the theoretical complexity of many of our proposed applications.

\begin{figure*}[htb!]
  \includegraphics[width=15.25cm]{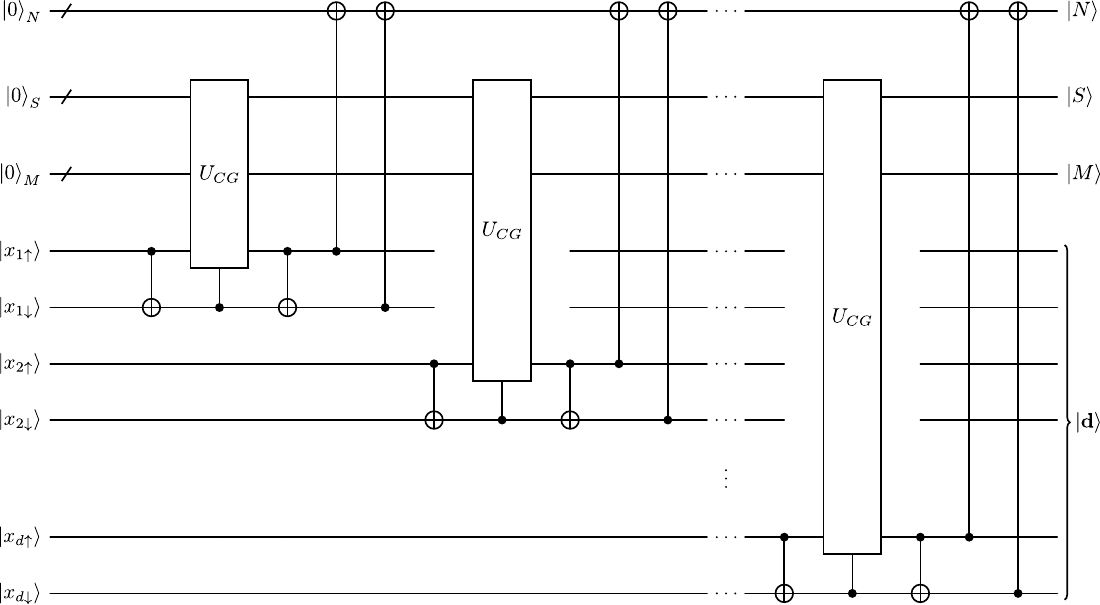}
  \caption{\textit{Alternative Circuit for the Paldus Transform.} The Paldus transform can be implemented using $d$ many controlled Clebsch-Gordan transforms $U_{CG}$, which form the building blocks of the $SU(2)$ Schur transform. Using the algorithm of \cite{BCH_2005} for $U_{CG}$, this gives an $\mathcal{O}(d)$ circuit for the Paldus transform.}
  \label{fig:paldus_transform_cg}
\end{figure*}

\subsection{Communication Through \texorpdfstring{$\mathfrak{u}(2)$}{u(2)} Decoherence-Free Subsystems} \label{sec:dfs}

As an application unique to the Paldus transform, we consider the problem of \textit{communication through $\mathfrak{u}(2)$ Decoherence-Free Subsystems (DFSs)}, an adaptation of the DFS protocol of \cite{Kempe_2001} and one of the early applications of the quantum Schur transform. The Paldus transform turns out to enable error-free, asymptotically optimal communication of a quantum state through an error channel described by a unitary of the form $e^{iH_{\mathfrak{U}(\mathfrak{u}(2))}}$ discussed in Section \ref{sec:ladder_operators}. We first review the problem for the $SU(2)$ case. Then, we present our version of protocol adapted to $H_{\mathfrak{u}(2)}$ noise, which easily generalises to $H_{\mathfrak{U}(\mathfrak{u}(2))}$. Finally, we argue that in this adapted setting, the qubit Schur transform is unable to protect information against the noise, giving an explicit example of an application where the Paldus transform exceeds the capability of the Schur transform. 

To review the DFS protocol \cite{Kempe_2001, Lidar_2003, lidar_2014}, consider the following problem: Alice wishes to communicate with Bob, by sending him an $n$-qubit state $\ket{\psi}$ one qubit at a time. Eve, an adversary, attempts to scramble their communication by intercepting each qubit and applying a gate $U$ (known only to her) before sending it over to Bob. Bob receives a state $U^{\otimes n} \ket{\psi}$, and must recover the original state $\ket{\psi}$ from the scrambled state. This is an example of collective decoherence, where each subsystem couples identically to a bath (in fact, $U$ does not need to be a unitary, and can be any operation from $GL(2, \mathbb{C})$). In order to protect her state, Alice \textit{encodes} it in the Schur basis, preparing a state of the form 

\begin{equation}
  |\tilde{\psi} \rangle = \ket{S} \otimes \ket{M} \otimes \ket{\psi} = \ket{S} \otimes \ket{M} \otimes \sum_Y c_Y \ket{Y}
\end{equation}

\noindent where $S$ is the angular momentum of the $SU(2) \times S_n$ irreducible representation with the largest dimension $d_S$:

\begin{equation}
  d_S = { n \choose n/2 - S} - { n \choose n/2 - S - 1} \label{eq:spin_degeneracy}
\end{equation}

\noindent $M$ are the spin projection values ($SU(2)$ irrep labels), and $Y$ are the Yamanouchi symbols ($S_n$ irrep labels). As $S, M$ are held fixed, the $n$-qubit register $\ket{Y}$ serves the purpose of Alice's `logical basis' carrying information about the state she wishes to send over. She then applies the inverse Schur transform, discards the ancillary qubits and sends the state to Bob. Upon receiving the state, Bob appends ancillae, and applies the Schur transform to recover the state sent over by Alice encoded in the $\ket{Y}$ register. This succeeds because in the GT basis, the operation $U^{\otimes n}$ will only act on the $\ket{M}$ register, leaving $\ket{Y}$ unaffected (recall that $Y$ indexes the \textit{multiplicities} of $SU(2)$ irreps, which by definition are unchanged under group action):

\begin{align}
  \left( U_{Sch} U^{\otimes n} U_{Sch}^\dagger \right) |\tilde{\psi} \rangle &= \left(  \bigoplus_{S} I_S \otimes W_S \otimes I_{Y} \right) \ket{S} \otimes \ket{M} \otimes \ket{\psi} \\
  &= \ket{S} \otimes W_S \ket{M} \otimes \sum_Y c_Y \ket{Y} \\
  & = \ket{S} \otimes \left( \sum_{M' = -S}^S [W_S]_{M', M} \ket{M'} \right) \otimes \ket{\psi}
\end{align}

\noindent The $W_S$ are $(2S+1)$-dimensional matrices acting uniformly on the $SU(2)$ irrep spaces, leaving $\ket{Y}$ unchanged. In a more general setting with non-unitary decoherence, the resulting state will be:

\begin{equation}
\ket{S} \bra{S} \otimes \rho_M \otimes \ket{\psi} \bra{\psi},
\end{equation}

\noindent where the decoherence only affects the $SU(2)$ irrep register, again leaving the desired state unaffected by noise. Assuming no errors in the Schur transform, this protocol is asymptotically efficient, as for increasing $n$ the number of logical qubits Alice is able to send over, $k = \log_2 (\max d_S)$, satisfies $\lim_{n \rightarrow \infty} \frac{k}{n} = 1 + \mathcal{O}(\frac{\log_2 n}{n})$.

\

Let us now turn to our protocol, where the statement of the problem differs slightly: Alice wishes to communicate a $2d$ qubit state $\ket{\psi}$ to Bob. Eve is again tasked with scrambling their communication, but at her disposal is any operation $U = e^{iK}$, where $K$ is a member of the universal enveloping algebra of $\mathfrak{U}(\mathfrak{u}(2))$, i.e. a linear combination of products of operators $\{ \hat{S}_x, \hat{S}_y, \hat{S}_z, \hat{N} \}$ from Section \ref{sec:ladder_operators}. She may act with this gate on all the qubits at once, though if $K = H_{\mathfrak{u}(2)} = \alpha_{x} \hat{S}_x + \alpha_{y} \hat{S}_y + \alpha_{z} \hat{S}_z + \alpha_{N} \hat{N}$ is a single-body $\mathfrak{u}(2)$ Hamiltonian then she applies a \textit{matchgate} $WV$ (Equation \eqref{eq:wv_circuit}) to each pair of qubits which pass by her. This scenario resembles collective decoherence by $XY$, $ZZ$ and controlled phase gates on a quantum computer. Similarly to before, Alice can encode her state in the UGA basis, confined to the largest $(N, S)$ irrep of dimension $T^d_{S, N}$ (Equation \eqref{eq:weyl_dimension}). The state she prepares is:

\begin{equation}
| \tilde{\psi} \rangle = \ket{N} \otimes \ket{S} \otimes \ket{M} \otimes \ket{\psi} = \ket{N} \otimes \ket{S} \otimes \ket{M} \otimes \sum_{\mathbf{d}} c_{\mathbf{d}} \ket{\mathbf{d}}.
\end{equation}
In fact, preparing the $\ket{N}$ register is not necessary, although the $\ket{\mathbf{d}}$ states in the superposition must be valid step vectors for the chosen subspace (i.e. $| \mathbf{d} | = N$). Following Alice's application of $U_P^\dagger$, Eve's interference with the state and Bob's application of $U_P$, the resulting state will be given by: 

\begin{align}
  \left( U_P (WV)^{\otimes n} U_P^\dagger \right) |\tilde{\psi} \rangle &= \left(  \bigoplus_{N, S} I_N \otimes I_S \otimes W_{N, S} \otimes I_{\mathbf{d}} \right) \ket{N} \otimes \ket{S} \otimes \ket{M} \otimes \ket{\psi} \nonumber \\
  & = \ket{N} \otimes \ket{S} \otimes \ket{M'} \otimes \ket{\psi} \label{eq:dfs_action}
\end{align}

Where $W_{N,S}$ are again $(2S + 1) \times (2S+1)$ matrices acting on the $\ket{M}$ states of $(N, S)$ irreps. In the presence of particle number operators $\hat{N}$ in the exponent $K$, the action of $U$ in the UGA basis is dependent on the particle number $N$ of the $(N, S)$ subspace. We have `absorbed' this dependence into the operators $W_{N, S}$, however in general, irreps with distinct $(N, S)$ are acted on \textit{differently}. This is clearly visible our simple example of $U = (WV)^{\otimes n}$, where the phase component of Equation \eqref{eq:wv_circuit} multiplies each state by a particle number-dependent phase, and so $W_{N, S} = e^{i N \alpha_N / 2} W_{S}$. If $\alpha_N = 0$ then Eve's operation becomes a $SU(2)$ action which does not distinguish between irreps of equal $S$ but different $N$. In general, the largest subspace $(N, S)$ available to Alice will vary with $d$, but the choice of $(N=d, S=0)$ for $d$ even and $(N= d \pm 1, S=0)$ for $d$ odd is already optimal, satisfying $\lim_{n \rightarrow \infty} \frac{\log_2(T_{d, S, N})}{2d} = 1 + \mathcal{O}(\frac{\log_2 d}{d})$. This gives us the desired decoherence-free subsystem for $\mathfrak{u}(2)$ noise. 

In addition to proving that states $\ket{S} \otimes \ket{M} \otimes \ket{Y}$ form a decoherence-free subsystem for collective decoherence, the authors of \cite{Kempe_2001} go further, developing a stabilizer formalism for the DFS and proving that by acting on the logical qubits with appropriate exchange interaction Hamiltonians (evolution by $e^{i \mathbb{C}(S_n)}$, where $\mathbb{C}(S_n)$ is the $S_n$ group algebra) one can perform universal quantum computation within each DFS. Given the similarities between $SU(2)$ Schur states and the UGA basis states, we expect their proof will carry over to the DFS's formed by UGA states where universal gate sets spanning $SU(T^d_{S, N})$ will consist of interactions by $\mathfrak{u}(d)$ Hamiltonians. We leave this to future work, but note that this may be of interest to researchers wishing to perform DFS-based quantum error mitigation against $XY$ crosstalk on quantum computers without constructing a full error-correcting code.

\

We end this section with a simple argument to highlight that whilst our DFS protocol is achievable with the Paldus transform, noise mitigation against gates of the form $(WV)^{\otimes n}$ is \textit{not possible through the $SU(2)$ Schur transform}. This follows from the observation that $WV$ is an element of $U(2) \cong (SU(2) \times U(1)) / \mathbb{Z}_2$, with \textit{three} free parameters (two for the $SU(2)$ component, and $\alpha_N$ for $U(1)$), \textit{embedded} in a 2-qubit gate. In particular, the $U(1)$ component $\alpha_N$ is embedded in this gate in a way such that it is not `flushed out' by the global phase invariance of the quantum state. Because of the global phase invariance, in the original DFS protocol each single-qubit gate $U$ in Alice's operation $U^{\otimes n}$ only has \textit{two} nontrivial free parameters, so it will not be possible to factorise each two-qubit gate $WV$ into a product of the form $U \otimes U$, even if Alice and Bob can supplement $U_{Sch}$ with an additional circuit. Of course, if Alice and Bob instead perform the $SU(4)$ Schur transform (treating each pair of their qubits as a quatrit) they can secure their state against Eve, but given its higher resource overhead our protocol is more efficient and practical. Aside from its applications to quantum information, our described protocol \textit{elevates} the Paldus transform to an algorithm strictly more powerful than the standard qubit Schur transform -- any application of the latter may be emulated (with a qubit overhead) by the former, but the converse is not true.

\begin{figure*}[htb!]
  \includegraphics[width=14cm]{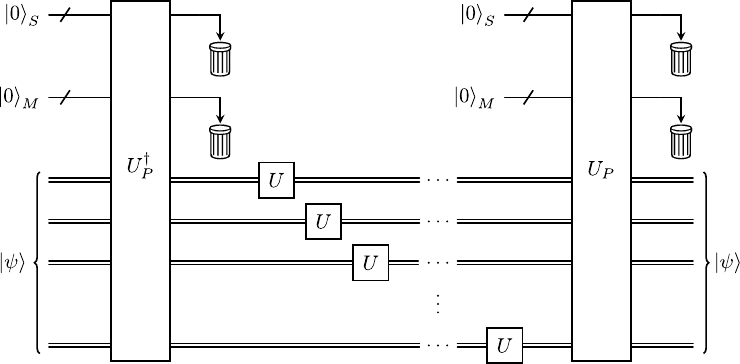}
  \caption{\textit{Encoding into a Decoherence-Free Subsystem.} The inverse Paldus transform enables Alice (left) to protect her quantum state against unknown noise of the form $U^{\otimes d}$, where each $U = WV$ is a gate from Figures \ref{fig:u2_circuit}, \ref{fig:wv_gate}. The noise affects the $M$ component of the GT basis only (Equation \eqref{eq:dfs_action}). Bob (right) recovers the state by applying the Paldus transform. The protocol gives a number of logical qubits which is asymptotically optimal.}
  \label{fig:dfs_circuit}
\end{figure*}

\section*{Acknowledgements} \label{sec:acknowledgements}
Throughout this project we have benefited from many helpful discussions, each of which helped improve the work and shape it into its final form. We would like to express our gratitude to Isaac Chuang, Werner Dobrautz, Richard East, Gabriel Greene-Diniz, Dmitry Grinko, Yuan Liu, Michael Robb, Frederic Sauvage and Sergii Strelchuk for their valuable inputs. We also thank Cristina Cirstoiu and Julien Sorci for reviewing the paper.

\appendix
\onecolumngrid

\section{Brief History of the Unitary Group Approach} \label{app:uga_history}
The discovery of electron spin by Stern and Gerlach in 1922 revolutionised our understanding of atomic structure \cite{sterng1922}. The revelation that electrons are antisymmetric under exchange was perhaps the most significant implication of their experiment. This observation fundamentally altered our understanding of atomic structure and quantum mechanics. This was first applied to a two-electron Helium wavefunction in 1926 by Heisenberg \cite{Heisenberg1926_0,Heisenberg1926_1}. One of the first applications of symmetry in quantum mechanics was Heisenberg's theory of ferromagnetism, which used the characters of representations of the symmetric group to ensure the resulting wavefunction was antisymmetric \cite{Heisenberg1928}. Following the incorporation of electron antisymmetry, the first ground state quantum chemical calculations were performed by Heitler and London in 1927 \cite{Heitler1927}, Pauling \cite{Pauling1931}, and separately Slater \cite{Slater1931a, Slater1931}, with valence bond theory. These methods were then extended to excited states using the Rayleigh-Ritz Variational method \cite{PhysRev.43.830} by MacDonald.

Group theory quickly became a useful tool in quantum mechanics, with applications in some of the earliest papers by Born, Heisenberg, and Jordan \cite{Heisenberg1926} on the presence of definite spin multiplicities due to global $SU(2)$ symmetry arising from spin angular momentum. It was further developed by Weyl and condensed into his famous \textit{Gruppentheorie und Quantenmechanik} \cite{weyl}. Then, Yamanouchi and Kotani related the properties of $SU(2)$ to the symmetric group with the so-called Yamanouchi-Kotani spin eigenfunctions \cite{kotanispin}. Despite Schrödinger's dislike of the new `gruppenpest', by 1967 the use of group theory in fields such as nuclear physics and quantum chemistry was widespread, and is summarised in a review by Löwdin \cite{RevModPhys.39.259}.

Much work was done by Coleman (a PhD student of Alfred Young) \cite{Coleman1968} towards the symmetric group, and many new methods exploiting this symmetry were pioneered by Löwdin \cite{RevModPhys.39.259, Klein1970}. Group theoretical methods were applied to second quantisation with nonorthogonal orbitals by Moshinsky \cite{Moshinsky1971}. Harter and Patterson applied group theory to a Slater determinant basis, simplifying calculations with their tableau formula \cite{Harter1973, Patterson1977}, later re-derived using angular momentum coupling rules by Drake and Schlesinger \cite{PhysRevA.15.1990}. Alongside these developments, Matsen presented the spin-free form of quantum chemistry in 1969 \cite{Matsen}, allowing explicit spin operators to be dropped from the Hamiltonian and incorporated into the basis functions. Salmon built on these ideas to derive the symmetric group formalism of quantum chemistry \cite{Salmon1974}, which was then developed further by Sarma \cite{Sarma1977} and finally formalised as the symmetric group approach by Duch and Karwowski \cite{Duch1982,Duch1985a,Duch1985,Karwowski1998}.

The unitary group approach traces back to Paldus, who built upon the ideas of Moshinsky from nuclear physics \cite{osti_4877459}, combined them with the Gelfand-Tsetlin basis \cite{Gelfand:1950ihs}, and applied them to the spin-free many-electron quantum chemistry wave function \cite{Paldus1974}. He then unified the UGA with that of Harter and Patterson's, employing Weyl tableaux in his method \cite{Paldus1976}. Matsen then applied the UGA to many-body theory and Green's functions \cite{MATSEN1978223}. Shavitt \cite{ISI:A1977DX38800015,Shavitt1978} then derived the graphical unitary group approach (GUGA), which expresses the basis functions as walks on graphs, with matrix elements calculated from overlapping walks. These methods were originally developed heuristically, and then derived algebraically using the theory of angular momentum coupling by Paldus and Boyle \cite{Paldus1980, Wormer2006}. The algebraic relationship between permutations and generators of the unitary group in the Configuration Interaction (CI) method was then established by Paldus and Wormer through the use of line-up permutations \cite{Paldus1979}. A vectorised computational implementation of the UGA was then given in \cite{Zarrabian1989}. The matrix elements were finally derived through UGA tensor operator algebras by Li and Paldus in 1990 \cite{Li1990, Li1993}, which give the true formulation of matrix elements in $U(n)$, rather than the previous $SU(2)$ angular momentum coupling work of Paldus and Boyle. We direct interested readers to the excellent reviews in \cite{PALDUS1976131, Roche1981, Robb1984}.

More recently, Ron Shepard has applied the GUGA approach of Shavitt towards efficient wavefunction truncations with his graphically-contracted function method, deriving an efficient nonlinear parameterisation of the UGA CI wavefunction \cite{Shepard2006, Shepard2006a, Shepard2007, Shepard2007,Shepard2008, Shepard2010, Shepard2014, Shepard2014, Shepard2015, Shepard2019}. Computationally, GUGA is also used extensively for matrix element calculations in the COLUMBUS program of Shavitt, Shepard, and Lischka \cite{columbus}. In the age of quantum computation, UGA has recently been applied to FCI quantum Monte Carlo \cite{Dobrautz2019}, and the symmetric group approach played a role in quantum algorithm design \cite{gandon2024}. Just before he passed away in 2023, Paldus collated much of his work on the UGA into a modern review: we instruct the interested reader to go there for more information \cite{Paldus2020, Paldus2020a, Paldus2020b}.

\section{Lie Theoretic Results}\label{app:liegroups}

In Section \ref{sec:uga}, we have made use of several results from the representation theory of Lie groups and Lie algebras. Here, we summarise some of the key concepts referred to in the main text. We assume basic knowledge of group theory and only state what is necessary. For an in-depth introduction to group and representation theory, we recommend the textbook by Fulton and Harris \cite{fultonandharris}. For an accessible introduction to matrix Lie groups and Lie algebras, we recommend the textbook by Hall \cite{hall2003lie}. For those looking for a very concise yet thorough set of lecture notes on both, we also recommend the course notes by Andre Lukas \cite{lukas2020}. As all groups we consider are matrix Lie groups, we will avoid mentions of differential geometry and closely follow the definitions in \cite{hall2003lie}.

\subsection{Definitions of Lie Groups, Algebras and Representations}

\begin{definition}[The General Linear Group] The general linear group $GL(n, \mathbb{C})$ ($GL(n, \mathbb{R})$) is the group of all invertible $n \times n$ complex (real) matrices with matrix multiplication.
\end{definition}

\begin{definition}[Matrix Lie Groups \cite{hall2003lie,lukas2020}] A matrix Lie group $G$ is a subgroup of $GL(n, \mathbb{C})$ which is also a \textit{closed} subset of $GL(n, \mathbb{C})$. Alternatively, a matrix group $G$ is a Lie group if $G$ is a differentiable manifold, with the group multiplication and inversion forming differentiable maps.
\end{definition}

\begin{definition}[Matrix Lie Algebras \cite{hall2003lie}] Let $G$ be a matrix Lie group. The Lie algebra of G, denoted $\mathfrak{g}$, is the set of all matrices $X$ such that $e^{itX} \in G$ for all $t \in \mathbb{R}$. $\mathfrak{g}$ is closed under addition, scalar multiplication from $\mathbb{R}$, and the Lie bracket $-i[X,Y] = -i(XY - YX)$.
\end{definition}

We follow the `physics' convention in which the exponential map $\mathfrak{g} \rightarrow G$ is $e^{itX}$, thus the factor of $-i$ in the Lie bracket definition is required. For example, the Pauli matrices $\{\sigma_x, \sigma_y, \sigma_z\}$ generate $SU(2)$ under this exponential map, but satisfy $[\sigma_i, \sigma_j] = i\epsilon_{ijk} \sigma_j$ and so including the factor of $-i$ in the Lie bracket ensures closure. Next we define the concept of a representation:

\begin{definition}[Lie Group and Lie Algebra Representations]
Suppose $G$ is a Lie group with Lie algebra $\mathfrak{g}$. A representation of $G$ is a group homomorphism $\Pi: G \rightarrow GL(V, \mathbb{F})$, where $V$ is a vector space over the field $\mathbb{F}$ and $\Pi$ satisfies $\Pi(g_1) \Pi(g_2) = \Pi(g_1 \circ g_2) \ \forall g_1, g_2 \in G$. A representation of $\mathfrak{g}$ is a Lie algebra homomorphism $\pi: \mathfrak{g} \rightarrow \mathfrak{gl}(V, \mathbb{F})$, where $\pi$ satisfies $\pi([X,Y]) = [\pi(X), \pi(Y)] \ \forall X, Y \in \mathfrak{g}$. If the homomorphism is one-to-one, then the representation is called \textit{faithful}.
\end{definition}

Often a Lie algebra is described in terms of a linearly independent basis $\{X_i\} \in \mathfrak{g}$ and \text{structure} constants $[X_i, X_j] = \sum_kf_{ij}^k X_k$ which describe how the basis elements commute. Showing that $\pi$ is a representation then amounts to showing that the structure constants are the same, i.e. the basis $\{ \pi(X_i) \}$ obeys the same commutation relations. This is clearly the case for the representation $\pi(e_{ij} \in \mathfrak{gl}(d, \mathbb{C})) = E_{ij}$ and $\pi(\tilde{e}_{\mu \nu} \in \mathfrak{gl}(2, \mathbb{C})) = \mathcal{E}_{\mu \nu}$ used in the main text. There is a one to one mapping between the $d^2$ (and $2^2$) matrix units, and the $d^2$ (and $2^2$) ladder operators, so the representation is faithful. As we make reference to the Lie algebras $\mathfrak{gl}(n, \mathbb{C})$, $\mathfrak{u}(n)$ and $\mathfrak{su}(n)$, we will define them next:

\begin{definition}[The General Linear Group Lie Algebra] The Lie algebra of $GL(n, \mathbb{C})$ ($GL(n, \mathbb{R})$), written as $\mathfrak{gl}(n, \mathbb{C})$ ($\mathfrak{gl}(n, \mathbb{R})$) is the vector space of all $n \times n$ complex (real) matrices. It is spanned by linear combinations over $\mathbb{C}$ ($\mathbb{R}$) of $d^2$ the matrix units $e_{ij}$ with matrix elements $[e_{ij}]_{kl} = \delta_{ik}\delta_{jl}$.
\end{definition}

\begin{definition}[The Unitary Group Lie Algebra] The Lie algebra of $U(n)$, written as $\mathfrak{u}(n)$ is the vector space of all complex, hermitian $n \times n$ matrices. It is spanned by linear combinations over $\mathbb{R}$ of the $d^2$ matrices $f_{ij} = e_{ij} - e_{ji}$, $f^{ij} = i(e_{ij} - e_{ji})$ (where $i < j$) and $f_{ii} = e_{ii}$.
  \label{def:unitaryliealgebra}
\end{definition}
We note that $\mathfrak{u}(n)$ is a \textit{real} Lie algebra, even if some of its basis elements are complex. If the field is $\mathbb{C}$, one obtains the \textit{complexified} Lie algebra $\mathfrak{u}(n)_\mathbb{C}$ which is isomorphic to $\mathfrak{gl}(n, \mathbb{C})$. The single-body Hamiltonians $H_{\mathfrak{u}(d)}$ and $H_{\mathfrak{u}(2)}$ belong to the former. 

\begin{definition}[The Special Unitary Group Lie Algebra] The Lie algebra of $SU(n)$, written as $\mathfrak{su}(n)$ is the vector space of all \textit{traceless} complex, hermitian $n \times n$ matrices. It is spanned by linear combinations over $\mathbb{R}$ of the $d^2 - 1$ matrices $f_{ij} = e_{ij} - e_{ji}$, $f^{ij} = i(e_{ij} - e_{ji})$ (where $i < j$), and $f^{ii} = e_{ii} - e_{i+1, i+1}$ (where $i < n$).
\end{definition}
In reference to our described representation, any Hamiltonian linear in $\mathcal{E}_{\mu \nu}$ which does not contain the term $\hat{N} = \mathcal{E}_{\uparrow \uparrow} + \mathcal{E}_{\downarrow \downarrow}$ belongs to a representation of $\mathfrak{su}(2)$. In this particular case, we can say for certain that the unitary evolution by such Hamiltonians forms a representation of $SU(2)$. This is because $SU(n)$ is \textit{simply-connected}, so by the Lie group-Lie algebra correspondence the Lie algebra representation $\pi$ carries a Lie group representation $\Pi$. The same is true for Hamiltonians linear in $E_{ij}$ for which the sum of coefficients of all $E_{ii}$ terms is zero. 

\subsection{Projective Representations}

With key definitions given, we now examine the exponential map from our Lie algebra representations $E_{ij}, \mathcal{E}_{\mu \nu}$, obtained by the unitary evolution $U = e^{iHt}$. First we recall a useful result relating Lie algebra representations to Lie group representations:

\begin{definition}[Universal Covering Groups \cite{hall2003lie}] Let $G$ be a connected Lie group. Then, a universal covering group of G is a simply-connected Lie group $\tilde{G}$ together with a Lie group homomorphism $\Phi : \tilde{G} \rightarrow G$ such that the associated Lie algebra homomorphism $\phi : \tilde{\mathfrak{g}} \rightarrow \mathfrak{g}$ is a Lie algebra isomorphism.
\end{definition}

For any connected Lie group, a universal covering group exists. In particular this means $GL(n, \mathbb{C})$ and $U(n)$ have universal covering groups. The group $SU(n)$ on the other hand is already simply connected, so it is its own universal cover. The property of (simple) connectedness is important when it comes to obtaining Lie group representations out of Lie algebra representations, as stated in the following theorem:

\begin{theorem}[The Homomorphisms Theorem \cite{hall2003lie}]
Let $G$ and $H$ be matrix Lie groups with Lie algebras $\mathfrak{g}$ and $\mathfrak{h}$. Let $\phi : \mathfrak{g} \rightarrow \mathfrak{h}$ be a Lie algebra homomorphism. If $G$ is simply connected, then there exists a unique Lie group homomorphism $\Phi: G \rightarrow H$ such that $\Phi(e^{iX}) = e^{i\phi(X)}$ for all $X \in \mathfrak{g}$. \label{thm:homomorphisms}
\end{theorem}
The above theorem is also referred to as the \textit{Lie Group--Lie Algebra Correspondence}, as it links representations of Lie algebras to representations of Lie groups. In our case, it allows us to state some concrete facts about the group representations obtained by exponentiating single-body Hamiltonians. For example, if $H$ belongs to a $\mathfrak{su}(n)$ representation we outright get a representation of $SU(n)$ from $e^{iH}$. In the other cases, Theorem \ref{thm:homomorphisms} implies that we get a representation of the universal covering groups of $GL(n)$ or $U(n)$. However, a theorem by Bargmann tells us that there representations are \textit{projective representations}, which preserve the group structure up to a constant:

\begin{definition}[Projective Representations] A projective representation of a group $G$ on a vector space $V$ over a field $\mathbb{F}$ is a group homomorphism $\Pi: G \rightarrow GL(V, \mathbb{F}) / \mathbb{F}^*$, where $GL(V, \mathbb{F}) / \mathbb{F}^*$ is the projective linear group. It satisfies the group structure up to a constant: 
\[
\Pi(g_1) \Pi(g_2) = c(g_1, g_2) \Pi(g_1 \circ g_2), \ \Pi(g) \in GL(V, \mathbb{F}), \ c(g_1, g_2) \in \mathbb{F}.
\]
\end{definition}

\begin{theorem}[Bargmann's Theorem]
Let $G$ be a connected Lie group with Lie algebra $\mathfrak{g}$. If $\pi: \mathfrak{g} \rightarrow \mathfrak{gl}(V, \mathbb{C})$ is an irreducible representation of $\mathfrak{g}$, then the map $\Phi : G \rightarrow GL(V, \mathbb{C}) / \mathbb{C}^*$ defined by $\Phi(e^{iX}) = e^{i\phi(X)}$ forms a projective representation of $G$.
\end{theorem}

A projective representation is in essence a representation where the group action on a state $\ket{\psi}$ is correct only `up to' some global rescaling of the state. However, in our case this does not matter, as we are working in a projective space, where states are normalised and equivalent up to a phase factor. Thus, throughout the main text we omit this technicality and refer to the evolution by $E_{ij}, \mathcal{E}_{\mu \nu}$-linear Hamiltonians as representations of the general linear and unitary groups. For a proof of Bargmann's theorem, see \cite{Bargmann1954}. For a more practical discussion, we recommend Chapter 5.4 of \cite{bouchard2020} which includes an explicit example of projective and non-projective representations of $SO(3)$.

\subsection{Hamiltonians in the Universal Enveloping Algebra} \label{app:universalenveloping}

We have so far established that exponentiation of $E_{ij}, \mathcal{E}_{\mu \nu}$-linear Hamiltonians gives rise to representations of certain matrix groups (depending on the choice of coefficients). In turn, such representations are reducible, breaking down into direct sums of irreducible representations (irreps) labelled by quantum numbers $(S, N)$. In addition, for $U(d)$ representations, there are labels $M$ which index the $2S + 1$ multiplicity of each irrep, and for $U(2)$ the multiplicities are labelled by step vectors $\mathbf{d}$. Each irreducible representation acts on an \textit{invariant subspace}, which is preserved under the group action. As we have seen in Sections \ref{sec:shavitt} and \ref{sec:SU2}, one can construct a Gelfand-Tsetlin basis for each invariant subspace, which is formed by vectors of the form $\ket{S, N, M, \mathbf{d}}$. Then, $U(d)$ acts on the $\mathbf{d}$ labels, and $U(2)$ acts on the $M$ labels.

The situation becomes less elegant when we consider exponentiation of Hamiltonians formed of \textit{products} of ladder operators. In this case, rather than belonging to the Lie algebra representation, the Hamiltonian belongs to a representation of the \textit{universal enveloping algebra}:

\begin{definition}[Universal Enveloping Algebra] Let $\mathfrak{g}$ be a Lie algebra with a basis $\{X_i\}$ and $\pi: \mathfrak{g} \rightarrow \mathfrak{gl}(V, \mathbb{C})$ a representation. The universal enveloping algebra of the representation $\pi(\mathfrak{g})$, denoted as $\mathfrak{U}(\mathfrak{g})$, is the algebra generated by the basis elements and their products:
\[
  \mathfrak{U}(\mathfrak{g}) = \text{span}_\mathbb{C} \{ \mathbf{1}, \pi(X_1), \ldots, \pi(X_n), \pi(X_1) \pi(X_2), \ldots, \pi(X_n)^2, \ldots, \pi(X_1) \pi(X_2) \pi(X_3), \ldots \},
\]
with the usual matrix multiplication and commutation relations of $\mathfrak{g}$: $\pi(X_i) \pi(X_j) - \pi(X_j) \pi(X_i) = i \sum_k f_{ij}^k \pi(X_k)$.
\end{definition}

This definition is much less rigorous than what appears in the mathematical literature, but it will suffice for our purposes. The main point we wish to make is that in general, exponentiating an element of $\mathfrak{U}(\mathfrak{g})$ \textit{does not} induce a representation of a group. As the broader class of Hamiltonians applicable to our work belongs to $\mathfrak{U}(\mathfrak{gl}(n, \mathbb{C}))$, in general we cannot say that the corresponding evolution forms a representation of any of the groups we have mentioned. However, from the definition of the matrix exponential, we can write any $e^{iX}$ as follows:

\begin{equation}
  e^{iX} = \mathbf{1} + iX  - \frac{X^2}{2!} - i \frac{X^3}{3!} + \ldots
\end{equation}

from which it is clear that if $W$ is an invariant subspace of $\mathfrak{g}$, then the same invariant subspaces are generated by $e^{i\mathfrak{g}}$ and $e^{i\mathfrak{U}(\mathfrak{g})}$. In other words, working in the GT basis we get the exact same block-diagonalisation of evolution by $H \in \mathfrak{g}$ and $H \in \mathfrak{U}(\mathfrak{g})$. Even though the latter cannot be expressed as a group action, we are still able to split it up into distinct invariant subspaces via the Paldus transform, and thus apply many of our techniques in the same way for applications such as Hamiltonian simulation or encoding information into decoherence-free subsystems.

\section{Symmetric Polynomials}
\label{app:symmetricpolynomials}

Let $R[x_1, x_2, \ldots, x_n]$ be a commutative ring of polynomials generated by the set $\{x_1, x_2, \ldots, x_n\}$, where the ring structure allows for the addition and multiplication of polynomial elements. The ring that defines the coefficients is typically $\mathbb{Z}$, $\mathbb{Q}$, $\mathbb{R}$, or $\mathbb{C}$. The symmetric polynomial ring has multiple generating sets: monomial, elementary, complete, and Schur polynomials, as well as power sums. The relationships between these generating sets lead to much of the structure that is exploited in the Paldus transform. The goal of this section is to provide a self-contained derivation of the Cauchy identities as they will be essential to the decomposition of the fermionic Fock space upon which the Paldus transform is based. We then discuss Schur-Weyl duality, fermionic Fock spaces and $S_N$ irreducible representations within the context of symmetric polynomials. The notation used is based on the conventions of the canonical reference by Macdonald~\cite{MacdonaldSymPoly}.

\subsection{Background}

\begin{definition}[Symmetric Polynomial] A polynomial $f(x_1, x_2, \ldots, x_n)$ in $n$ variables is symmetric if it remains unchanged under any permutation of the variables. Formally, for any permutation $\sigma \in S_n$ of the variables $x_1, x_2, \ldots, x_n$, the polynomial satisfies:

  \begin{equation}
      f(x_{\sigma(1)}, x_{\sigma(2)}, \ldots, x_{\sigma(n)}) = f(x_1, x_2, \ldots, x_n)
  \end{equation}
  \end{definition}
  
  \begin{remark}
      It is typical to use an infinite number of variables when working with symmetric polynomials, such as in the canonical reference by Macdonald~\cite{MacdonaldSymPoly}. However, in this work, a finite number of variables will be used throughout.
  \end{remark}

\begin{definition} [Symmetric Polynomial Ring] The ring of symmetric polynomials $\Lambda$ is the subring of $R[x_1, x_2, \ldots, x_n]$ consisting of all symmetric polynomials:

\[
\Lambda_n = R[x_1, x_2, \ldots, x_n]^{S_n} 
\]

\noindent Where it is obvious to see that the ring is graded by the homogeneous symmetric polynomials of degree $d$ by $\Lambda_n = \bigoplus_{d \geq 0} \Lambda_n^d$. 
\end{definition}

Symmetric polynomials will be essential for the derivation of the branching rules used to derive the quantum Paldus transform. In particular, they appear in the characters of the symmetric group $S_N$ and unitary group $U(n)$.

\begin{definition}[Partition]
A partition $\lambda$ of a positive integer $n$ is a nonincreasing sequence of positive integers $\lambda_1 \geq \lambda_2 \geq \ldots \geq \lambda_k > 0$ such that $\sum_{i=1}^k \lambda_i = n$. The length of the partition is the number of nonzero parts, denoted by $l(\lambda) = k$. The Young diagram of a partition is a left-justified array of boxes with $\lambda_i$ boxes in the $i$-th row. The conjugate partition $\tilde{\lambda}$ is obtained by reflecting the Young diagram in the main diagonal.
\end{definition}

\begin{example}
The partition $\lambda = (3,2,1)$ corresponds to the Young diagram:
\begin{equation}
  \begin{ytableau}
  \, & \, & \, \\
  \, & \, \\
  \,
  \end{ytableau}
\end{equation}
\end{example}

A useful combinatorial tool in the study of symmetric polynomials is the semi-standard Young tableau (SSYT).

\begin{definition}[Semi-standard Young Tableaux (SSYT)] A SSYT is a filling $T$ of the Young diagram of $\lambda$ with $n$ positive integers that satisfies the following conditions:
  \label{def:ssyt}
      \begin{enumerate}
          \item Rows are weakly increasing: The entries in each row are nondecreasing from left to right. Formally, for each row $i$ and columns $j$ and $j+1$,
        \begin{equation}
          T(i, j) \leq T(i, j+1).
        \end{equation}
          \item Columns are strictly increasing: The entries in each column are strictly increasing from top to bottom. Formally, for each column $j$ and rows $i$ and $i+1$,
        \begin{equation}
          T(i, j) < T(i+1, j).
        \end{equation}
      \end{enumerate}
\end{definition}
      
\begin{example}
      Consider the partition $\lambda = (3,2)$, which corresponds to the Young diagram:
    \begin{equation}
      \begin{ytableau}
      \, & \, & \, \\
      \, & \,
      \end{ytableau}
    \end{equation}
      
      \noindent An example of a semi-standard Young tableau of this shape with entries from $\{1, 2, 3\}$ is:
    \begin{equation}
      \begin{ytableau}
      1 & 1 & 2 \\
      2 & 3
      \end{ytableau}
    \end{equation}
\end{example}

The first of the symmetric polynomial bases we will consider are the monomial symmetric polynomials which have an intuitive interpretation in terms of partitions of the integer $n$ and serve as the entry point into the topic.

\begin{definition}[Monomial Symmetric Polynomial]
Let $\lambda = (\lambda_1, \lambda_2, \ldots, \lambda_k)$ be a partition of a nonnegative integer $n$, where $\lambda_1 \geq \lambda_2 \geq \cdots \geq \lambda_k > 0$. The monomial symmetric polynomial $m_\lambda(x_1, x_2, \ldots, x_n)$ is defined as:

\[
m_\lambda(x_1, x_2, \ldots, x_n) = \sum_{\alpha} x_1^{\alpha_1} x_2^{\alpha_2} \cdots x_n^{\alpha_n}
\]

\noindent where the sum is taken over the distinct permutations $\alpha = (\alpha_1, \alpha_2, \ldots, \alpha_n)$ of the multiset $(\lambda_1, \lambda_2, \ldots, \lambda_k, 0, \ldots, 0)$, with $n - k$ zeros appended to match the number of variables. 
\end{definition}

\begin{example}
$m_{(2,1)}(x_1, x_2, x_3)= x_1^2x_2 + x_1^2x_3 + x_2^2x_1 + x_2^2x_3 + x_3^2x_1 + x_3^2x_2 $
\end{example}

Another basis of the symmetric polynomials and perhaps the most fundamental to our derivation of the Paldus transform is the elementary symmetric polynomial.

\begin{definition}[Elementary Symmetric Polynomial] The elementary symmetric polynomial $e_k(x_1, x_2, \ldots, x_n)$ in $n$ variables is defined as the sum of all products of $k$ distinct variables. Formally, it is given by:

  \begin{equation}
      e_k(x_1, x_2, \ldots, x_n) = \sum_{1 \leq i_1 < i_2 < \ldots < i_k \leq n} x_{i_1} x_{i_2} \hdots x_{i_k}.
  \end{equation}
  \label{def:elementary_symmetric_polynomial}
\end{definition}

\noindent The elementary symmetric polynomials are represented by single-column SSYTs generated from $\lambda = (1^k)$. As mentioned in Section \ref{sec:ON}, they are equivalent to the $\chi^{U(d)}_{(1^k)}$ character of the $U(n)$ antisymmetric representation $\lambda = \{ 1^k \}$.
  
  \begin{example}
      The elementary symmetric polynomial $e_3(x_1, x_2, x_3, x_4)$ is given by:
  
      \begin{equation}
        e_3(x_1, x_2, x_3, x_4) = x_1x_2x_3 + x_1x_2x_4 + x_1x_3x_4 + x_2x_3x_4 = 
        \begin{ytableau}
          1 \\
          2 \\
          3 \\
      \end{ytableau}
      +
      \begin{ytableau}
          1 \\
          2 \\
          4 \\
      \end{ytableau}
      +
      \begin{ytableau}
          1 \\
          3 \\
          4 \\
      \end{ytableau}
      +
      \begin{ytableau}
          2 \\
          3 \\
          4 \\
      \end{ytableau}
      =
      \begin{ytableau}
        \, \\
        \, \\
        \, \\
      \end{ytableau}
    \end{equation}
  \end{example}

The generating function of the elementary symmetric polynomials is given by the product of the linear factors of the variables, which can be factorised as follows:

\begin{equation}
  E_n(t) = 1 + e_1t + e_2t^2 + \ldots + e_nt^n = \sum_{k=0}^n e_kt^k = \prod_{i=1}^n (1 + x_it)
\end{equation}

\begin{remark}
This the character of the totally antisymmetric full fermionic Fock space of $n$ particles
\end{remark}

The dual of the elementary symmetric polynomial is the complete symmetric polynomial, which is the sum of all products of $k$ variables, allowing for repetition.

\begin{definition}[Complete Symmetric Polynomial] The complete symmetric polynomial $h_k(x_1, x_2, \ldots, x_n)$ in $n$ variables is defined as the sum of all products of $k$ variables, allowing for repetition. Formally, it is given by:

\begin{equation}
    h_k(x_1, x_2, \ldots, x_n) = \sum_{1 \leq i_1 \leq i_2 \leq \ldots \leq i_k \leq n} x_{i_1} x_{i_2} \cdots x_{i_k}
\end{equation}
\end{definition}

\noindent The complete symmetric polynomials can be represented by Young diagrams and their SSYTs of a single row of length $k$.

\begin{example}
    The complete symmetric polynomial $h_3(x_1, x_2, x_3)$ is given by:

    \begin{equation}
      \begin{split}
      h_3(x_1, x_2, x_3) &= x_1^3 + x_1^2x_2 + x_1^2x_3 + x_1x_2^2 + x_1x_2x_3 + x_1x_3^2 + x_2^3 + x_2^2x_3 + x_2x_3^2 + x_3^3 \\ &=
      \begin{ytableau}
        1 & 1 & 1 \\
    \end{ytableau}+
    \begin{ytableau}
        1 & 1 & 2 \\
    \end{ytableau}+
    \begin{ytableau}
        1 & 1 & 3 \\
    \end{ytableau}+
    \begin{ytableau}
        1 & 2 & 2 \\
    \end{ytableau} +
    \begin{ytableau}
        1 & 2 & 3 \\
    \end{ytableau} \\ &+
    \begin{ytableau}
        1 & 3 & 3 \\
    \end{ytableau}+
    \begin{ytableau}
        2 & 2 & 2 \\
    \end{ytableau}+
    \begin{ytableau}
        2 & 2 & 3 \\
    \end{ytableau}+
    \begin{ytableau}
        2 & 3 & 3 \\
    \end{ytableau}+
    \begin{ytableau}
        3 & 3 & 3 \\
    \end{ytableau} \\
    & = \begin{ytableau}
      \, & \, & \, \\
    \end{ytableau}
  \end{split}
  \end{equation}
\end{example}

The generating function of the complete symmetric polynomials is given by the product of the linear factors of the reciprocals of the variables:

\begin{equation}
\begin{split}
  H(t) = \sum_{k=0}^\infty h_kt^k &= \prod^n_{i=1} ( 1 + x_i t + x_i^2 t^2 + \ldots) \\ &= \prod^n_{i=1} \frac{1}{(1 - x_it)}
\end{split}
\end{equation}

\begin{theorem}[Fundamental Theorem of Symmetric Polynomials] 
The elementary symmetric polynomials $e_k$ and the complete symmetric polynomials $h_k$ are algebraically independent and form a generating set for the ring of symmetric polynomials $\Lambda$.

\begin{equation}
  \Lambda_n = \mathbb{Z}[e_1, e_2, \ldots, e_n] = \mathbb{Z}[h_1, h_2, \ldots, h_n]
\end{equation}
\end{theorem}

\noindent The proof can be found in Macdonald (\cite{MacdonaldSymPoly}, page 20) and proceeds by showing that there exists a mapping between the products of elementary symmetric polynomials $e_\lambda = e_{\lambda_1}e_{\lambda_2} \cdots e_{\lambda_k}$, where $\lambda_i$ is the number of boxes in the $i^\text{th}$ row of $k$-partitions $\lambda$), and the monomial symmetric polynomials. This mapping is invertible and therefore the elementary symmetric polynomials form a basis for the ring of symmetric polynomials. By the properties of the $\omega$-involution introduced in Proposition \ref{prop:omega_involution} there is a one-to-one correspondence between the elementary symmetric polynomials and the complete symmetric polynomials. This means that the complete symmetric polynomials also form a basis for the ring of symmetric polynomials.

\begin{example} $e_2(x_1,x_2,x_3) \, e_1 (x_1,x_2,x_3)$ can be expressed in the monomial symmetric basis as:
  \[
\begin{aligned}
e_2(x_1,x_2,x_3) \, e_1 (x_1,x_2,x_3)
\;=\;&
\bigl(x_1x_2 + x_1x_3 + x_2x_3\bigr)\,\bigl(x_1 + x_2 + x_3\bigr)\\[6pt]
=\;&
x_1^2x_2 
\;+\;
x_1x_2^2 
\;+\;
x_1^2x_3 
\;+\;
x_1x_3^2
\;+\;
x_2^2x_3
\;+\;
x_2x_3^2
\;+\;
3\,x_1x_2x_3. \\
&= m_{(2,1)}(x_1,x_2,x_3)  + 3m_{(1,1,1)}(x_1,x_2,x_3) 
\end{aligned}
\]
\end{example}

\begin{proposition}[$\omega$-involution]
There exists an $\omega$-involution, which is a map that sends each elementary symmetric polynomial $e_k$ to the corresponding complete symmetric polynomial $h_k$ and vice versa. Concretely, we have:
\[
\omega(e_k) = h_k
\quad\text{and}\quad
\omega(h_k) = e_k.
\]
Graphically, this corresponds to transposing the Young diagrams associated with $e_k$ and $h_k$.
\label{prop:omega_involution}
\end{proposition}

\begin{proof}
Applying $\omega$ twice returns the original polynomial, establishing that it is indeed an involution. To see this, consider the product of the generating functions for $e_k$ and $h_k$:
\[
E(t) \;=\;\sum_{i \ge 0} e_i\, t^i,
\quad
H(-t) \;=\;\sum_{j \ge 0} h_j\, (-t)^j.
\]
Then observe that
\[
E(t)\,H(-t) 
\;=\;
\left(\sum_{i \ge 0} e_i\, t^i\right)
\left(\sum_{j \ge 0} h_j\, (-t)^j\right)
\;=\; \frac{\prod_{i=1} (1 + x_it)}{\prod_{i=1} (1 + x_it)} \;=\;
1.
\]
Equating coefficients shows that
\[
\sum_{i+j=n} (-1)^j\, e_i\, h_j
\;=\;
0
\quad \text{for all } n \ge 1,
\]
and the constant term ($n=0$) is $e_0 h_0 = 1$. Applying $\omega$ to both sides leaves the identity invariant, so $\omega^2$ recovers the original polynomials, and thus $\omega$ is an involution, thus revealing the dual relationship between the elementary and complete symmetric polynomials.
\end{proof} 

\begin{example} Using the $\omega$-involution, we can see that
  $\omega(h_{\lambda_1} h_{\lambda_2}\cdots h_{\lambda_k}) = \omega(h_{\lambda_1})\omega(h_{\lambda_2})\cdots \omega(h_{\lambda_k}) = e_{\lambda_1} e_{\lambda_2}\cdots e_{\lambda_k}$, where the monomials in $h_{\lambda_i}$ are weakly increasing rows in a Young diagram and the monomials in $e_{\lambda_i}$ are strictly increasing columns. For $\omega(h_3 h_2)$:

  \begin{equation}
    \omega(h_3 h_2) = \omega\Bigg(\begin{ytableau}
      \, & \, & \, \\
      \, & \, \\
    \end{ytableau}\Bigg) = \begin{ytableau}
      \, & \, \\
      \, & \, \\
      \, \\
      \end{ytableau} = e_3 e_2
  \end{equation}
\end{example}

Another basis of symmetric polynomials of interest is known as the Schur polynomials, which we have introduced in Definition \ref{def:schur_polynomial} in the main text. These are perhaps the most important because they are used to calculate the characters of the unitary group. They were first discovered by Cauchy in the ratio of the alternant form~\cite{Cauchy_2009}.

\begin{definition}[Alternant polynomial]
The alternant polynomial $a_{\lambda+\delta}(x_1, x_2, \ldots, x_n)$ is a polynomial in $n$ variables associated with the partition $\lambda + \delta$. It is defined as the determinant of the matrix whose $(i,j)^\text{th}$ entry is $x_i^{\lambda_j + n - j}$. This polynomial is antisymmetric under any permutation of the variables and is given by:

\[
a_{\lambda+\delta} = \det\bigl(x_i^{\lambda_j + n - j}\bigr)_{1 \leq i,j \leq n}
= 
\det\begin{pmatrix}
x_1^{\lambda_1+n-1} & x_2^{\lambda_1+n-1} & \cdots & x_n^{\lambda_1+n-1} \\
x_1^{\lambda_2+n-2} & x_2^{\lambda_2+n-2} & \cdots & x_n^{\lambda_2+n-2} \\
\vdots & \vdots & \ddots & \vdots \\
x_1^{\lambda_n} & x_2^{\lambda_n} & \cdots & x_n^{\lambda_n}
\end{pmatrix}.
\]

\noindent Here, $\delta = (n-1, n-2, \ldots, 1, 0)$ is the staircase partition. Thus, $\lambda + \delta = (\lambda_1 + n - 1, \lambda_2 + n - 2, \ldots, \lambda_n)$ has strictly decreasing parts, ensuring that the determinant is nonzero. The determinant may be equivalently expressed using the Leibniz formula:
\[
a_{\lambda+\delta} 
= \sum_{\sigma \in S_n}
\text{sgn}(\sigma)
\prod_{i=1}^n
x_i^{\lambda_i + n - \sigma(i)},
\]
where $S_n$ is the symmetric group on $n$ elements and $\mathrm{sgn}(\sigma)$ is the sign of the permutation $\sigma$.
\end{definition}
  
  \begin{remark}
    One of the key properties of determinants is that if two rows or two columns are the same, the determinant is zero. Hence, for a strictly decreasing sequence 
    \(\alpha = (\alpha_1, \alpha_2, \ldots, \alpha_n)\) (where \(\alpha_i > \alpha_{i+1}\)), the matrix $x_i^{\alpha_j}$ in
    \[
    \det\bigl(x_i^{\alpha_j}\bigr)_{1 \leq i,j \leq n} 
    = 
    \det\begin{pmatrix}
    x_1^{\alpha_1} & x_2^{\alpha_1} & \cdots & x_n^{\alpha_1} \\
    x_1^{\alpha_2} & x_2^{\alpha_2} & \cdots & x_n^{\alpha_2} \\
    \vdots & \vdots & \ddots & \vdots \\
    x_1^{\alpha_n} & x_2^{\alpha_n} & \cdots & x_n^{\alpha_n}
    \end{pmatrix}
    =
    \sum_{\sigma \in S_n} \mathrm{sgn}(\sigma)\,\prod_{i=1}^n x_{i}^{\alpha_{\sigma(i)}}
    \]
    cannot have any identical rows. Therefore, the determinant is nonzero precisely when
    \(\alpha\) has strictly decreasing parts. This justifies defining the alternant polynomial
    in terms of \(\lambda + \delta\), where \(\delta = (n-1, n-2, \ldots, 1, 0)\) ensures that
    the exponents are strictly decreasing, preventing the determinant from vanishing.
  \end{remark}
  \begin{definition}[Vandermonde Determinant]
  The Vandermonde determinant is a special case of the alternant polynomial for the partition \(\delta = (n-1, n-2, \ldots, 1, 0)\) (i.e., \(\lambda = (0,\ldots,0)\)). It is defined as 
  \[
  \Delta(x_1, x_2, \ldots, x_n) \;=\; a_\delta 
  \;=\; \det\bigl(x_i^{\,n-j}\bigr)_{1 \leq i,j \leq n}
  \;=\; \sum_{\sigma \in S_n} \mathrm{sgn}(\sigma)\,\prod_{i=1}^n x_i^{\,n - \sigma(i)}
  \;=\; \prod_{1 \leq i < j \leq n} (x_j - x_i).
  \]
  \label{def:vandermonde}
  \end{definition}

  \noindent The Vandermonde determinant is antisymmetric under any transposition of the variables:
  \[
  \Delta\bigl(x_{\sigma(1)}, \ldots, x_{\sigma(n)}\bigr)
  \;=\; \mathrm{sgn}(\sigma)\,\Delta(x_1, x_2, \ldots, x_n).
  \]
  It is a fundamental object in algebraic geometry and has broad applications in many areas of mathematics. In conjunction with the alternant polynomial, it can be used to define Schur polynomial via the determinantal identity:

  \begin{theorem}[Schur Polynomial -- Bialternant Identity]
  Let $\lambda$ be a partition with at most $n$ parts. The Schur polynomial $s_\lambda(x_1,\dots,x_n)$ can be written as:
  \[
  s_\lambda(x) \;=\; \frac{a_{\lambda+\delta}(x_1,\dots,x_n)}{a_{\delta}(x_1,\dots,x_n)}
  \;=\;
  \frac{\det\bigl(x_i^{\lambda_j + n - j}\bigr)_{1 \leq i,j \leq n}}{\Delta(x_1, x_2, \ldots, x_n)},
  \]
  where $\delta = (n-1, n-2, \ldots, 1, 0)$ and $\Delta(x_1, x_2, \ldots, x_n)$ is the Vandermonde determinant.
  \label{thm:bialternant}
  \end{theorem}
  
  \begin{proof}
  To see that $s_\lambda(x)$ is symmetric in $x_1, \ldots, x_n$, note that interchanging two variables in $\det\bigl(x_i^{\lambda_j + n - j}\bigr)$ contributes a factor of $-1$, while the same interchange in the denominator $\Delta(x_1, x_2, \ldots, x_n)$ also contributes $-1$. These signs cancel in the ratio, ensuring symmetry:
  \[
  s_\lambda(\ldots, x_i, x_j, \ldots)
  =
  s_\lambda(\ldots, x_j, x_i, \ldots).
  \]
  The mapping $(\lambda+\delta)_i = \lambda_i + n - i$ establishes a bijection for the exponents to partitions, completing the argument.
  \end{proof}

  An equivalent definition of the Schur polynomials uses the elementary and complete symmetric polynomial bases, known as the Jacobi-Trudy Identity. It is crucial for deriving the branching rules in the quantum Paldus transform.

  \begin{theorem}[Jacobi-Trudy Identity] The Schur polynomial $s_\lambda(x_1, x_2, \ldots, x_n)$ can be expressed as the determinant of a matrix with elements $(i,j)$ of the complete symmetric functions:

    \[
    s_\lambda(x) = \text{det}(h_{\lambda_i - i + j})_{1 \leq i, j \leq n} = \text{det}(e_{\tilde{\lambda}_i -i + j})_{1 \leq i, j \leq n}
    \]
    
    \noindent where the matrix is of size $n \times n$.
    \end{theorem}
    \begin{proof}
      Working with $n$ variables $x_1, x_2, \ldots, x_n$, for $1 \leq k \leq n$ let $e^{(k)}_r$ be the elementary symmetric polynomial of degree $r$ in all variables except $x_k$. Define
      \[
        E^{(k)}(t) = \sum_{q=0}^{n-1} e^{(k)}_q\,t^q \;=\; \prod_{i \neq k} (1 - x_i\,t).
      \]
      Reversing the indexing in the Jacobi-Trudy identity with $i = (n-1) - q$ and then canceling common factors in numerator and denominator gives
      \[
        H(t)\,E^{(k)}(-t) \;=\; (1 - x_k\,t)^{-1}.
      \]
      Expanding both sides:
      \[
        \sum_{p=0}^\infty h_p\,t^p \;\sum_{q=0}^{n-1} e^{(k)}_{q}\,(-t)^q 
        \;=\;
        \sum_{r=0}^\infty x_k^r\,t^r.
      \]
      Substitute $p + q = r$ and sum over $r$ on the left:
      \[
        \sum_{r=0}^\infty \Bigl(\,\sum_{q=0}^r h_{r-q}\,e^{(k)}_{q}\,(-1)^q\Bigr)\;t^r
        \;=\;
        \sum_{r=0}^\infty x_k^r\,t^r.
      \]
      Grouping the terms in $t^r$ and setting $r = \alpha_i$ yields 
      \[
        \sum_{q=0}^r h_{r-q}\,e^{(k)}_{q}\,(-1)^q \;=\; x_k^{\alpha_i}.
      \]
      Next, re-index with $q = (n-1) - i$, so $i = (n-1) - q$ runs from $0$ to $n-1$. Let $j = i + 1 = n - q$, which runs from $1$ to $n$. Thus, 
      \[
        \sum_{j=1}^n h_{r-n+j}\,e^{(k)}_{n-j}\,(-1)^{\,n-j} \;=\; x_k^{\alpha_i}.
      \]
      This sum can be viewed as a matrix product $\mathbf{H}_r\,\mathbf{M} = \mathbf{A}_r$, where $H_{ij} = h_{\alpha_i - n + j}$, $M_{jk} = e^{(k)}_{n-j}\,(-1)^{\,n-j}$, and $A_{ik} = x_k^{\alpha_i}$.  Taking determinants and using alternants gives 
      \[
        a_\alpha \;=\; \det(\mathbf{A}_\alpha) \;=\; \det(\mathbf{H}_\alpha)\,\det(\mathbf{M}).
      \]
       Then, $a_\delta = \det(\mathbf{H}_\delta)\,\det(\mathbf{M})$ with $\det(\mathbf{H}_\delta) = 1$ as it is lower triangular, (because $h_{i-j} = 0 \ \forall i < j$ and $h_0 = 1$). Therefore, $a_\alpha = \det(\mathbf{M})$. Therefore,
      \[
        a_\alpha \;=\; a_\delta\;\det(\mathbf{H}_\alpha).
      \]
      Expressing $a_\alpha$ as an alternant polynomial:
      \begin{equation}
        a_\alpha \;=\; a_\delta 
        \sum_{\sigma \in S_n} 
          \text{sgn}(\sigma)\,\prod_{i=1}^n h_{\alpha_i - i + \sigma(i)}.
      \label{eq:jt_derivation_identity}
      \end{equation}
      Rearranging terms yields the Jacobi-Trudy identity:
      \[
        s_\lambda 
        \;=\; 
        \frac{a_{\lambda + \delta}}{a_\delta}
        \;=\; 
        \sum_{\sigma \in S_n} \text{sgn}(\sigma)\,\prod_{i=1}^n h_{\lambda_i - i + \sigma(i)} = \det(h_{\lambda_i - i + j})_{1 \leq i, j \leq n}. 
      \]

      \noindent It is straightforward to show that the dual Jacobi-Trudy identity holds as well:

      \[
        s_\lambda = \omega(s_{\tilde{\lambda}}) = \det(\omega(h_{\tilde{\lambda}_i - i + j}))_{1 \leq i, j \leq n} = \det(e_{\lambda_i - i + j})_{1 \leq i, j \leq n}.
      \]

    \end{proof}

    \begin{example} For the partition \(\lambda = (2, 1)\) and its (self) conjugate partition \(\tilde{\lambda} = (2, 1)\), we will calculate the Schur polynomial \(s_{(2,1)}(x_1, x_2)\) using the Jacobi-Trudi identity and its dual. We can take determinants of the matrices and using the theory of alternants:
    
    \[
    a_\alpha = \det(\mathbf{A}_\alpha) = \det(\mathbf{H}_\alpha) \det(\mathbf{M})
    \]
    \[
    s_{(2,1)} = \det\begin{pmatrix}
    h_{\lambda_1 - 1 + 1} & h_{\lambda_1 - 1 + 2} \\
    h_{\lambda_2 - 2 + 1} & h_{\lambda_2 - 2 + 2}
    \end{pmatrix} = \det\begin{pmatrix}
    h_2 & h_3 \\
    h_0 & h_1
    \end{pmatrix}
    \]
    
    \noindent Expanding the determinant:
    
    \begin{equation}
      \begin{split}
        s_{(2,1)} &= h_2 h_1 - h_3 h_0 = h_2 h_1 - h_3 \\
        &= (x_1^2 + x_1 x_2 + x_2^2)(x_1 + x_2) - (x_1^3 + x_1^2 x_2 + x_1 x_2^2 + x_2^3)\cdot 1 \\
        &= x_1^3 + 2 x_1^2 x_2 + 2 x_1 x_2^2 + x_2^3 - x_1^3 - x_1^2 x_2 - x_1 x_2^2 - x_2^3 \\
        &= x_1^2 x_2 + x_1 x_2^2 
      \end{split}
    \end{equation}

    \noindent For the dual Jacobi-Trudi identity:
    
    \[
    s_{(2,1)} = \det\begin{pmatrix}
    e_{\lambda'_1 - 1 + 1} & e_{\lambda'_1 - 1 + 2} \\
    e_{\lambda'_2 - 2 + 1} & e_{\lambda'_2 - 2 + 2}
    \end{pmatrix} = \det\begin{pmatrix}
    e_2 & e_3 \\
    e_0 & e_1
    \end{pmatrix}
    \]
    
    \noindent Expanding the determinant where \(e_3 = 0\) (as there are only two variables):
    
    \[
    s_{(2,1)} = e_2 e_1 - e_3 e_0 = (x_1 x_2)(x_1 + x_2) - 0\cdot1 = x_1^2 x_2 + x_1 x_2^2
    \]
    
    \noindent Both identities yield the same Schur polynomial, confirming their equivalence.
    \end{example}
    
    \begin{corollary}
      The action of the $\omega-$involution on $s_\lambda(x)$ is therefore straightforward. $\omega(\text{det}(h_{\lambda_i - i + j})_{1 \leq i, j \leq n}) = \text{det}(e_{\lambda_i - i + j})_{1 \leq i, j \leq n} = \text{det}(h_{\tilde{\lambda}_i - i + j})_{1 \leq i, j \leq n} $. Therefore, the $\omega$-involution acts on the Schur polynomials by taking the transpose of the Young diagram $\omega(s_\lambda(x)) = s_{\tilde{\lambda}}(x)$. This will be essential for deriving Pieri's rule.
    \end{corollary}

Finally with these tools we are in a position to prove Pieri's rule, which is an essential tool in deriving the branching structure of the Paldus transform

\begin{theorem}[Dual Pieri Rule]
  Given a partition $\lambda$ and a positive integer $k$, the product of the Schur polynomial $s_\lambda$ with the elementary symmetric polynomial $e_k$ is given by:
  
  \begin{equation}
  s_\lambda \cdot e_k = \sum_{\mu} s_\mu,
  \end{equation}
  
  \noindent where $\mu$ is a partition obtained by adding at most $k$ boxes to each (possibly empty) row of $\lambda$. 
  \label{thm:dualpieri}
  \end{theorem}

\begin{proof}

\noindent Following Macdonald~\cite{MacdonaldSymPoly} (p.\ 73), there is a bijection between the Schur polynomial and the alternant polynomial. This implies that multiplying the Schur function $s_\lambda$ by the elementary symmetric function $e_k$ yields a sum of Schur functions $s_\mu$ for which $\mu/\lambda$ forms a vertical strip of length $k$ with no two boxes in the same row. We use the determinantal identity multiplied by $e_k$:

\begin{equation}
\begin{split}
a_{\lambda+\delta}\,e_r \;=\;& 
\sum_{\sigma \in S_n} \mathrm{sgn}(\sigma)\,\prod_{i=1}^n x_i^{\lambda_i + n - \sigma(i)}
\;\sum_{1 \leq i_1 < i_2 < \cdots < i_r \leq n} x_{i_1} x_{i_2} \cdots x_{i_r} \\
=\;& \sum_{\alpha} \sum_{\sigma \in S_n} \mathrm{sgn}(\sigma)\,\prod_{i=1}^n x_i^{\lambda_i + \alpha_i + n - \sigma(i)} 
\;=\; \sum_{\mu} a_{\mu+\delta}.
\end{split}
\end{equation}

\noindent Applying the determinantal identity then gives

\begin{equation}
s_\lambda \cdot e_k 
= \frac{a_{\lambda+\delta}}{a_{\delta}} \, e_k
= \sum_{\mu} \frac{a_{\mu+\delta}}{a_{\delta}}
= \sum_{\mu} s_\mu.
\end{equation}

\noindent Since $e_k$ requires strictly increasing indices, the added boxes cannot occupy the same row, thus forming a vertical strip in the Young diagram. This completes the proof of the dual Pieri's rule.
\end{proof}

  \begin{example}
      Consider the complete elementary product of schur functions $s_{(2,1)} \cdot e_{(2)}$. The product is given graphically by the following diagram:
      
    \begin{equation}
      \begin{ytableau}
      \, & \, \\
      \,
      \end{ytableau}
      \otimes
      \begin{ytableau}
      \, \\
      \,
      \end{ytableau} \quad = \quad
      \begin{ytableau}
        \, & \, \\
        \, \\
        \star \\
        \star
      \end{ytableau}
      \oplus 
      \begin{ytableau}
        \, & \, \\
        \, & \star \\
        \star
      \end{ytableau}
      \oplus 
      \begin{ytableau}
      \, & \, & \star \\
      \, & \star
      \end{ytableau}
      \oplus 
      \begin{ytableau}
      \, & \, & \star \\
      \, \\
      \star
      \end{ytableau}
    \end{equation}
  
    \noindent Where the boxes denoted $\star$ are the vertical strip of length 2.
    \label{ex:dual_pieri}
  \end{example}

  \begin{theorem}[Pieri's Rule]
  Let $\lambda$ be a partition and $k$ a positive integer. The product of the Schur polynomial $s_\lambda$ with the complete homogeneous symmetric polynomial $h_k$ is given by
  \[
  s_\lambda \cdot h_k = \sum_{\mu} s_\mu,
  \]
  where $\mu/\lambda$ is a horizontal strip of length $k$ in which no two boxes lie in the same column.
  \end{theorem}

  \begin{proof}
  The proof of Pieri's rule parallels that of the dual Pieri's rule. It can be achieved using the $\omega$- involution acting on the dual Pieri's rule $s_\lambda \cdot e_k = \sum_{\mu} s_\mu$:

  \begin{equation}
    \begin{split}
      \omega(s_\lambda \cdot e_k) &= \omega(\sum_{\mu} s_\mu) = \sum_{\mu} \omega(s_\mu) = \sum_{\mu} s_{\tilde{\mu}}\\
      &= s_{\tilde{\lambda}} \cdot h_k = \sum_{\tilde{\mu}} s_{\tilde{\mu}}
    \end{split}
  \end{equation}
  
  \noindent Consequently, the partitions $\lambda + \alpha = \mu$ have the condition that no two boxes are added in the same column, as represented by a Young diagram. This condition characterises the horizontal strip of length $r$, thereby completing the proof of Pieri's rule.
  \end{proof}
    
      \begin{example}
          Consider the complete homogeneous product of schur functions $s_{(2,1)} \cdot h_{(2)}$. The product is given graphically by the following diagram:
          
        \begin{equation}
          \begin{ytableau}
          \, & \, \\
          \,
          \end{ytableau}
          \otimes
          \begin{ytableau}
          \, & \,
          \end{ytableau} \quad = \quad 
          \begin{ytableau}
          \, & \, & \star  & \star    \\
          \, 
          \end{ytableau}
          \oplus
          \begin{ytableau}
          \, & \, & \star  \\
          \, & \star 
          \end{ytableau}
          \oplus
          \begin{ytableau}
          \, & \, \\
          \, & \star  \\
          \star  
          \end{ytableau}
          \oplus
          \begin{ytableau}
          \, & \, & \star \\
          \, \\
          \star 
          \end{ytableau}
        \end{equation}
      
        \noindent Where the boxes denoted $\star$ are the added boxes.
      \end{example}
    
    Although very simple, Pieri's rule will be fundamental in deriving the branching rules underpin the Paldus transform. More complicated products of two arbitrary Schur polynomials can be found by recursively applying Pieri's rule row by row, this can result in multiplicities of greater than one. This is equivalent to the Littlewood-Richardson rule for the multiplication of Schur functions.
    \begin{remark}
    Pieri's rules are valid for any finite number of variables or infinitely many, provided that the number of variables exceeds the length of the partition $\ell(\lambda)$. 
    \end{remark}

   Another important basis of the symmetric polynomial are the power sum polynomials, they will be essential in deriving Schur-Weyl duality.

\begin{definition}[Power Sum Symmetric Polynomial] The power sum symmetric polynomial $p_k(x_1, x_2, \ldots, x_n)$ in $n$ variables is defined as the sum of the $k^\text{th}$ powers of the variables. Formally, it is given by:

\begin{equation}
  p_k(x_1, x_2, \ldots, x_n) = \sum_{i=1}^n x_i^k.
\end{equation}
\end{definition}

\noindent They form a generating set for the ring of symmetric polynomials. The Schur functions can be expressed in terms of the power sum symmetric polynomials via the following relation involving the characters of the symmetric group:

\begin{equation}
  s_\lambda(x) = \sum_\mu \frac{\chi^\lambda_\mu}{z_\mu} p_\mu(x),
\label{eqn:schur_power_sum}
\end{equation}

\noindent where $z_\mu = \prod_{i=1}^n i^{m_i} m_i!$ is the size of the conjugacy class of the symmetric group $S_n$ corresponding to the partition $\mu$, and $m_i$ is the number of parts of size $i$ in the partition $\mu$. This can be rearranged to express the power sum symmetric polynomials in terms of the Schur functions:

\begin{equation}
  p_\mu(x) = \sum_\lambda \chi^\lambda_\mu s_\lambda(x).
\label{eqn:power_sum_schur}
\end{equation}

The next section will proceed to derive the Cauchy identities for the product of two conjugate Schur functions $\lambda$ and $\tilde{\lambda}$. This is the essential identity which leads to the structure of the Paldus transform.
x
\begin{lemma} 
  \[
  \prod_{i,j} \frac{1}{1 - x_i y_j} = \sum_\lambda h_\lambda(x) m_\lambda(y) = \sum_\lambda m_\lambda(x) h_\lambda(y)
  \]
\label{lem:cauchy_help}
\end{lemma}

\begin{proof}
  By fixing $y_i$, we can form a product of generating functions:
  \begin{equation}
    \begin{split}
      \prod_{i,j} \frac{1}{1 - x_i y_j}  &= \prod_{j} H(y_j)\\
      &= \prod_{j} \sum_{r=0}^\infty h_r(x) y_j^r\\
    \end{split}
  \end{equation}
  
  \noindent Then, combining the products into $\alpha$ tuples by $h_{\alpha_1} h_{\alpha_2} \ldots h_{\alpha_n} = h_\alpha$ and $y_1^{\alpha_1} y_2^{\alpha_2} \ldots y_n^{\alpha_n} = y^\alpha$:
  
  \begin{equation}
    \begin{split}
      \prod_{j} \sum_{r=0}^\infty h_r(x) y_j^r &= \sum_{\alpha} h_\alpha(x) \prod_{j} y_j^{\alpha_j}\\
      &= \sum_{\alpha} h_\alpha(x) m_\alpha(y)
    \end{split}
  \end{equation}

  \noindent It can clearly be seen that the variables can be swapped to give the other form of the product.

\end{proof}

\subsection{Cauchy Identities}

Finally we arrive at the Cauchy identities, these will be essential in the recursive decomposition of the fermionic Fock space under $U(m) \otimes U(n)$, which the Paldus transform is built upon.

\begin{theorem}[Cauchy Identity] For two sets of variables $x = (x_1, x_2, \ldots, x_n)$ and $y = (y_1, y_2, \ldots, y_n)$, the Cauchy identity states that the generating functions of the power sum symmetric polynomials are related to the Schur functions by:

\[
  \prod_{i,j} \frac{1}{1 - x_i y_j} = \sum_\lambda s_\lambda(x) s_\lambda(y)
\]

\end{theorem}

\begin{proof} Using the theory of alternants, we define

\begin{equation}
\begin{split}
a_\delta(x)a_\delta(y)\prod_{i,j}\frac{1}{1 - x_i y_j} &= a_\delta(x)a_\delta(y)\sum_{\alpha} h_\alpha(x) \prod_{i} y_i^{\alpha_i} \\
&= a_{\delta}(x) \sum_{\sigma \in S_n} \text{sgn}(\sigma) \prod_{i=1}^n y_i^{n - \sigma(i)} \sum_{\alpha} h_\alpha(x) \prod_{i} y_i^{\alpha_i}\\
&= a_\delta(x) \sum_{\alpha} \sum_{\sigma \in S_n}   \text{sgn}(\sigma) \prod_{i=1}^n  h_{\alpha_i}(x) y_i^{\alpha_i + n - \sigma(i)} 
\end{split}
\end{equation}

\noindent Now, $\delta = (n-1, n-2, \ldots, 1, 0)$ is the staircase partition of $n$ parts, which in the Leibniz determinant form is analogous to $\prod_{i} (n - \sigma(i))$. Changing the variables to $\alpha_i + n - \sigma(i) = \beta_i$:

\begin{equation}
  \begin{split}
    a_\delta(x) \sum_{\alpha} \sum_{\sigma \in S_n}   \text{sgn}(\sigma) \prod_{i=1}^n  h_{\beta_i - n + \sigma(i)}(x) y_i^{\beta_i} 
    &= \sum_\beta a_\beta (x) y^\beta
  \end{split}
\end{equation}

\noindent Where we have used $a_\beta = a_\delta \sum_{\sigma \in S_n} \text{sgn}(\sigma) \prod_{i=1}^n h_{\beta_i - n + \sigma(i)}(x)$. Importantly, $y_i^{\beta_i}$ is being permuted for each $\beta$ due to the permutation with the index $\alpha_i + n - \sigma(i) = \beta_i$. Additionally, the action $\sum_{\sigma \in S_n} \text{sgn}(\sigma) \prod_{i=1}^n  y_i^{\beta_i} = a_\beta (y)$, where $\delta + \lambda = \beta$ and $\beta$ is strictly decreasing. This gives the Schur function in terms of the power sum symmetric polynomials:

\[
\sum_\beta a_\beta (x) y^\beta = \sum_\lambda a_{\delta + \lambda}(x) a_{\delta + \lambda}(y)
\]

\[
 a_\delta(x) a_\delta(y) \prod_{i,j}\frac{1}{1 - x_i y_j} = \sum_\beta a_\beta (x) y^\beta
\]

\noindent and therefore gives the Cauchy identity:

\[
  \prod_{i,j}\frac{1}{1 - x_i y_j} = \sum_\lambda s_\lambda(x) s_\lambda(y)
\]

\end{proof}

\begin{theorem}[Dual Cauchy Identity] For two sets of variables $x = (x_1, x_2, \ldots, x_n)$ and $y = (y_1, y_2, \ldots, y_n)$, the dual Cauchy identity states that the generating functions of the power sum symmetric polynomials are related to the Schur functions by:

\[
  \prod_{i,j}(1 + x_i y_j) = \sum_\lambda s_\lambda(x) s_{\tilde{\lambda}}(y)
\]
\label{thm:dualcauchy}
\end{theorem}
\begin{proof} By fixing $y_i$ and applying an $\omega$ involution we can form a product of generating functions, then using an involution

  \begin{equation}
    \begin{split}
      \prod_{i,j} (1 + x_i y_j)  &= \prod_{j} E(y_j)\\
      &= \prod_{j} \sum_{r=0}^N \omega(h_r(x) y_j^r)\\
    \end{split}
  \end{equation}

  \noindent where $\omega(h_r(x) y_j^r) = \omega(h_r(x)) y_j^r$.

  \begin{equation}
    \begin{split}
    a_\delta(x)a_\delta(y)\prod_{i,j} (1 + x_i y_j) &= a_\delta(x)a_\delta(y)\sum_{\alpha} \omega(h_\alpha(x) \prod_{i} y_i^{\alpha_i}) \\
    &= a_{\delta}(x) \sum_{\sigma \in S_n} \text{sgn}(\sigma) \prod_{i=1}^n y_i^{n - \sigma(i)} \sum_{\alpha} \omega(h_\alpha(x) \prod_{i} y_i^{\alpha_i})\\
    &= a_\delta(x) \sum_{\alpha} \sum_{\sigma \in S_n}   \text{sgn}(\sigma) \prod_{i=1}^n  \omega(h_{\alpha_i}(x) y_i^{\alpha_i + n - \sigma(i)}) 
    \end{split}
    \end{equation}

    \noindent which leads to

  \begin{equation}
    \begin{split}
      a_\delta(x) \sum_{\alpha} \sum_{\sigma \in S_n}   \text{sgn}(\sigma) \prod_{i=1}^n  \omega(h_{\beta_i - n + \sigma(i)}(x) y_i^{\beta_i})
      &= a_\delta(x) \sum_{\alpha} \sum_{\sigma \in S_n}   \text{sgn}(\sigma) \prod_{i=1}^n  e_{\beta_i - n + \sigma(i)}(x) y_i^{\beta_i} \\
      &= a_\delta(x) \sum_{\alpha} \sum_{\sigma \in S_n}   \text{sgn}(\sigma) \prod_{i=1}^n  h_{\tilde{\beta}_i - n + \sigma(i)}(x) y_i^{\beta_i} 
    \end{split}
  \end{equation}

  \noindent Where $\beta'$ is the conjugate partition of $\beta$. This then leads to the dual Cauchy identity:

  \[
    \prod_{i,j}(1 + x_i y_j) = \sum_\lambda s_{\tilde{\lambda}}(x) s_{\lambda}(y)
  \]

\end{proof}

The recursive structure of the dual Cauchy identity leads to the branching structure of the Paldus transform and is therefore the central tool of this work.

\section{Dualities}
\label{app:dualities}

Schur-Weyl duality and Howe duality (of which Paldus duality is an instance) are two important dualities in representation theory. They are both examples of the double centralizer theorem, which states that the centralizer of a subalgebra is itself a subalgebra. In this section, we will present Schur-Weyl duality and Howe duality, showing how they are examples of the double centralizer theorem. We will also present a proof of Schur-Weyl duality via the theory of characters.

\subsection{Double Centralizer Theorem}
\label{sec:double_centralizer}
The Double Centralizer Theorem is a fundamental result in representation theory that describes the relationship between two subalgebras of the endomorphism ring of a vector space. It is particularly useful in the context of Schur-Weyl duality and Howe duality, where it provides a framework for understanding the interplay between different symmetry groups acting on tensor products of vector spaces.

\begin{theorem}[Double Centralizer Theorem]
The set of all linear transformations from $V$ to itself forms the endomorphism ring $\mathrm{End}(V)$. If $A$ is a subalgebra of $\mathrm{End}(V)$, then its centralizer $C(A) = B$ is the set of all elements of $\mathrm{End}(V)$ that commute with every element of $A$:
\[
C(A) = B = \{X \in \mathrm{End}(V) \mid [X, A] = 0 \;\forall A \in A\}.
\]
The Double Centralizer Theorem states that for a subalgebra $A \subseteq \mathrm{End}(V)$, the centralizer of $B$ is $A$ itself:
\[
C(C(A)) = C(B) = A.
\]
\end{theorem}

\begin{remark}
Suppose a vector space $V$ decomposes under a subalgebra $A$ of $\mathrm{End}(V)$ into a direct sum of irreducible subspaces:
\[
V = \bigoplus_{i} m^a_i\,V_i,
\]
where each $V_i$ is irreducible under $A$ and $m^a_i$ is its multiplicity. If there is another subalgebra $B$ of $\mathrm{End}(V)$,
\[
B = \{ X \in \mathrm{End}(V) \mid [X, A] = 0 \;\forall A \in A \},
\]
then $V$ similarly decomposes under $B$ as
\[
V = \bigoplus_{j} m^b_j\,V_j,
\]
where each $V_j$ is irreducible under $B$ and $m^b_j$ is its multiplicity. By the Double Centralizer Theorem,
\[
C(C(A)) = C(B) = A.
\]
Schur's Lemma then implies that any operator commuting with all elements of $A$ must preserve the isotypic decomposition of $V$ under $A$ and cannot mix its irreps outside the multiplicity.
\end{remark}

\subsection{Schur-Weyl Duality}
\label{app:schur_weyl_duality}

Schur-Weyl duality is a fundamental framework illustrating how symmetries govern the behavior of tensor products of single-particle Hilbert spaces. It can be viewed as an instance of the double centralizer theorem, and we will first outline its general principles and then derive it using characters. A helpful perspective emerges from interpreting Schur-Weyl duality as a tensor network: when a tensor product wavefunction is acted upon by the same unitary on each index and then subjected to a permutation of indices, these operations commute.

\begin{figure}[!htbp]
  \centering
  \includegraphics[width=0.9\textwidth]{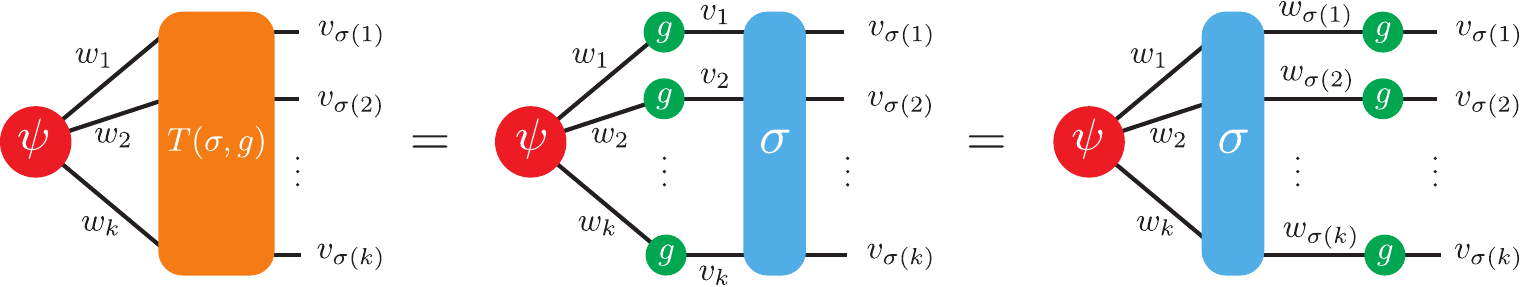}
  \caption{Schur-Weyl duality shown as a tensor network.}
  \label{fig:schur_weyl_tensornetwork}
\end{figure}

\noindent This is a case of the double centralizer theorem where $g\sigma = \sigma g$ for all $g \in U(n)$ and $\sigma \in S_n$. The action of the symmetric group $S_n$ permutes the indices of the tensor product and the action of the unitary group $U(n)$ acts on each index of the tensor product. The action of the symmetric group $S_n$ and the unitary group $U(n)$ commute and the tensor product is reducible under the joint action of $U(n)$ and $S_n$.

\begin{corollary}
Schur-Weyl duality is a special case of the Double Centralizer Theorem, where $A$ is the symmetric group $S_n$ and $B$ is the unitary group $U(n)$ acting on $V^{\otimes n}$.
\end{corollary}

In the Schur-Weyl setting, $A = \mathbb{C}[S_k]$ acts by permuting tensor factors, while $B = \mathrm{Im}(GL(n, \mathbb{C}))$ acts diagonally, which can be seen elegantly in Figure~\ref{fig:schur_weyl_tensornetwork}.

\begin{theorem}[Schur-Weyl Duality]
Schur-Weyl duality states that the tensor product of the irreducible representations of the symmetric group $S_k$ and the general linear group $GL(n, \mathbb{C})$ decomposes into a direct sum of irreducible representations of $GL(n, \mathbb{C})$ indexed by partitions of $k$:
\[
V^{\otimes n} \cong \bigoplus_{\lambda \vdash k,\, \ell(\lambda) \leq n} V_\lambda^{GL(n, \mathbb{C})} \otimes S^\lambda,
\]
where $V^{\otimes n}$ is the tensor product representation of $S_k$ and $GL(n, \mathbb{C})$, $V_\lambda^{GL(n, \mathbb{C})}$ are the irreps of $GL(n, \mathbb{C})$ indexed by $\lambda \vdash k$, and $S^\lambda$ are the irreps of $S_k$ indexed by the same partition.
\end{theorem}

This aligns with our general framework: the irreps of $A$ are $S^\lambda$ (partitions $\lambda \vdash k$), and those of $B$ are $V_\lambda^{GL(n, \mathbb{C})}$ (the same $\lambda$). The Double Centralizer Theorem enforces that these partitions match. The multiplicity of $S^\lambda$ in $V$ is exactly the dimension of $V_\lambda^{GL(n, \mathbb{C})}$, ensuring consistency in the irreducible decompositions.

The main focus of this work is the link between the symmetric polynomials and the characters of the representations, therefore we present a form of Schur-Weyl duality that uses the symmetric polynomials.

\begin{proposition}
  The character of the reducible representation \( V^{\otimes n} \) under the joint action of \( U(n) \) and \( S_n \) is given by:
  \[
    \chi_{V^{\otimes n}}(g, \sigma) = p_{\rho(\sigma)}(z),
  \]
  where \( p_{\rho} = p_{\rho_1} p_{\rho_2} \dots \) is the product of power sum polynomials for a partition \( \rho \) representing the cycle structure of the permutation \(\sigma\).
\end{proposition} 
  
\begin{proof}
The character of the reducible representation \( V^{\otimes n} \) under the joint action of \( U(n) \) and \( S_n \) is given by the trace
\[
    \chi_{V^{\otimes n}}(g, \sigma) = \operatorname{Tr}_{V^{\otimes n}}(g \sigma).
\]
The combined action of the operators is the composition of the diagonal action of \( g \) and the permutation action of \(\sigma\):
\[
    g \sigma \,\bigl(v_1 \otimes v_2 \otimes \dots \otimes v_n\bigr) 
    = g\,v_{\sigma^{-1}(1)} \otimes g\,v_{\sigma^{-1}(2)} \otimes \dots \otimes g\,v_{\sigma^{-1}(n)}.
\]
Consider a standard eigenbasis \(e_I = e_{i_1} \otimes e_{i_2} \otimes \dots \otimes e_{i_n}\) with \(g \, e_{i_k} = z_{i_k} \, e_{i_k}\). The character in this basis is
\[
  g \sigma \,\bigl(e_I\bigr) 
  = z_{i_{\sigma^{-1}(1)}} e_{i_{\sigma^{-1}(1)}} \otimes z_{i_{\sigma^{-1}(2)}} e_{i_{\sigma^{-1}(2)}} 
  \otimes \dots \otimes z_{i_{\sigma^{-1}(n)}} e_{i_{\sigma^{-1}(n)}}.
\]
Thus, 
\[
    \chi_{V^{\otimes n}}(g, \sigma) 
    = \sum_{I} \langle e_I \mid g \sigma \mid e_I \rangle.
\]

\noindent Where the sum is over all possible eigenbasis and since \(\{e_I\}\) is an eigenbasis for \(g\), \(\langle e_I \mid e_J \rangle = \delta_{IJ}\). For \(\langle e_I \mid g \sigma \mid e_I \rangle\) to be nonzero, we require \(\sigma \cdot I = I\), which implies \(i_k = i_{\sigma(k)}\). Consequently, \( i_k\) must be constant along the cycles of \(\sigma\). The trace becomes
\[
  \chi_{V^{\otimes n}}(g, \sigma) 
  = \sum_{I \,\text{fixed by}\, \sigma} 
    \bigl(\prod_{k=1}^n z_{i_{\sigma^{-1}(k)}} \bigr).
\]
For each cycle \(c\) in \(\sigma\), if all indices in that cycle are the same, the product is \(z_{i_c}^{|c|}\). Therefore, 
\[
  \chi_{V^{\otimes n}}(g, \sigma) 
  = \prod_{\text{cycles } c} 
    \Bigl(\sum_{i=1}^{n_c} z_i^{|c|}\Bigr)
  = \prod_{\text{cycles } c} p_{|c|}(z).
\]
Since the lengths of the cycles of \(\sigma\) form a partition \(\rho(\sigma)\) of \(n\), we write
\[
  \chi_{V^{\otimes n}}(g, \sigma) 
  = p_{\rho(\sigma)}(z) = \sum_\lambda \chi^\lambda_{\rho(\sigma)} s_\lambda(x),
\]
where \( p_{\rho} = p_{\rho_1} p_{\rho_2} \dots \) is the product of power sum polynomials. Consequently, using the identity from Equation~\ref{eqn:power_sum_schur}, we this can be realted to the character of the irreps of the symmetric group \(S_n\) and the unitary group \(U(n)\):
\[
  \chi_{V^{\otimes n}}(g, \sigma) 
  = \sum_{\lambda} \chi^\lambda_{U(n)}(g)\,\chi^\lambda_{S_n}(\sigma)  
\] 
\end{proof}

Therefore, using the relation 
\(\prod_{i} p_{\rho_i}(z)\) together with Schur-Weyl duality, we see that the irreducible representations of \(U(n)\) and \(S_n\) appear multiplicity-free in the decomposition of the reducible representation \(V^{\otimes n}\), linking the power sum symmetric polynomials and the corresponding characters. 

\subsection{Howe Duality}

Unitary-unitary and therefore Paldus duality can be viewed as a special case of $GL(n) \times GL(m)$ Howe duality \cite{Howe1989,Howe_1995}, which also is an example of a double centralizer. The general derivation is a consequence of the fundamental theorem of invariants presented by Weyl and uses quite different machinery to what was presented in section \ref{sec:uuduality}, however we note that our results hold for antisymmetric representations in this more general case. One way to think about why is to note that for matrix Lie groups, the invariant subspaces of irreps are the same as the invariant subspaces of their Lie algebra. On the other hand, the $\mathfrak{gl}(n)$ Lie algebra is obtained through the \textit{complexification} of the $\mathfrak{u}(n)$ Lie algebra, a process which does not affect the irrep structure (for this reason, mathematicians often choose to study the group $GL(n)$ over $U(n)$). In fact, Schur polynomials are also used to describe the characters of $GL(n)$ irreps, so the same derivation applies in that case. 

In short, $GL(n) \times GL(m)$ Howe duality can be summarised with the following theorem:

\begin{theorem}[Howe Duality]
  Consider the standard action of the group $GL(n) \times GL(m)$ on the vector space $\mathbb{C}^n \otimes \mathbb{C}^m$. The two groups form a reductive dual pair, and their action on the symmetric and exterior algebra of $\mathbb{C}^n \otimes \mathbb{C}^m$ forms the following multiplicity-free decomposition:

  \begin{align}
    \text{Sym}(\mathbb{C}^n \otimes \mathbb{C}^m) &\cong \bigoplus_{k=0}^\infty \bigoplus_{\substack{\lambda \vdash k \\ \ell(\lambda) \leq n \\ \ell(\lambda) \leq m}} 
 V^{GL(n)}_\lambda \otimes V^{GL(m)}_\lambda, \\
    \bigwedge(\mathbb{C}^n \otimes \mathbb{C}^m) &\cong \bigoplus_{k=0}^{n \times m} \bigoplus_{\substack{\lambda \vdash k \\ \ell(\lambda) \le n \\ \ell(\tilde{\lambda}) \le m }} W^{GL(n)}_{\lambda} \otimes W^{GL(m)}_{\tilde{\lambda}}
  \end{align}
\label{thm:howe_duality}
\end{theorem}

\noindent Howe duality has been remarked on in the context of quantum computing as recently as Gross, Nezami and Walter's work on Schur-Weyl duality for the Clifford group \cite{Gross_2021} and as early as Aram Harrow's PhD thesis \cite{harrowphd}. In fact, Harrow notes that implementing the symmetric algebra decomposition can be achieved with the Clebsch-Gordan transform, however this will always require a truncation of the direct sum to some value of $k$. In our case, we are interested in the antisymmetric algebra and thus obtain the \textit{full} decomposition, and to our knowledge this is the first time the antisymmetric variant has been considered in the context of quantum computing. We finally note that studying the symmetric algebra decomposition may be useful in the context of bosonic systems, where a similar irrep decomposition can be obtained.

\section{The Unitary Group \texorpdfstring{$U(n)$}{U(n)} and Branching Rules}
\label{app:unitarygroup}

This section delves into the foundational theory of representations of the unitary group of matrices denoted $U(n)$. There will be particular emphasis on the characters of the weight space representations which can expressed through Schur polynomials, as discussed in Appendix~\ref{app:symmetricpolynomials}.

\begin{definition}[Unitary Group of Matrices]
The unitary group $U(n)$ is the group of $n \times n$ unitary matrices:

\begin{equation}
    U(n) = \{ U \in \mathbb{C}^{n \times n} | UU^\dagger = I \}.
\end{equation}

\end{definition}

\noindent It is a subgroup of the general linear group $GL(n)$, which is the group of all invertible $n \times n$ matrices explained in more detail in Appendix \ref{app:liegroups}. The unitary group Lie algebra $\mathfrak{u}(n)$ is also defined in Definition \ref{def:unitaryliealgebra}, and the standard irrep of $U(n)$ is given in Definition \ref{def:unitary_standard_irrep} of the main text.

\begin{definition}[Cartan, Raising, and Lowering Operators] \label{def:cartan}
The Cartan subalgebra of the Lie algebra \(\mathfrak{u}(n)\) is spanned by the diagonal generators \(e_{ii}\), referred to as the Cartan operators. The remaining off-diagonal generators \(e_{ij}\) with \(i < j\) are called raising operators, while those with \(i > j\) are called lowering operators.
\end{definition}

\begin{definition}[Weight Vectors]
A weight vector of the Lie algebra $\mathfrak{u}(n)$ is an eigenvector of the Cartan operator $e_{ii}$, satisfying 
\[
e_{ii}\,|\lambda\rangle \;=\;\lambda_i\,|\lambda\rangle.
\]
Here, $|\lambda\rangle$ is labeled by the weight $\lambda = (\lambda_1, \lambda_2, \dots, \lambda_n)$ with $\lambda_1 \ge \lambda_2 \ge \cdots \ge \lambda_n$. The highest weight vector, annihilated by all raising operators, characterises an irreducible representation of $U(n)$ and is identified with a partition $\lambda$ of (up to) $n$ parts.
\end{definition}

\begin{remark}
  The Cartan operators commute with each other and therefore can be simultaneously diagonalised and share the same eigenvectors.
\end{remark}

\begin{definition}[Weyl Character Formula \cite{Weyl1925}]
  \label{def:weyl_character_formula}

The Weyl character formula gives the character of the irreducible representation of the unitary group $U(n)$ in terms of the Schur polynomials. The character of the irrep $V_\lambda$ of $U(n)$ is given by the Schur polynomial $s_\lambda$ evaluated at the eigenvalues of the Cartan operators:

\[
\chi^\lambda(x_1,\dots,x_n) = \frac{\displaystyle \sum_{w\in S_n} \text{sgn}(w)\, x_{w(1)}^{\lambda_1+n-1}\,x_{w(2)}^{\lambda_2+n-2}\cdots x_{w(n)}^{\lambda_n}}{\displaystyle \sum_{w\in S_n} \text{sgn}(w)\, x_{w(1)}^{n-1}\, x_{w(2)}^{n-2} \cdots x_{w(n)}^{0}} = \frac{a_{\lambda+\delta}(x_1,\dots,x_n)}{a_\delta(x_1,\dots,x_n)} = s_\lambda(x_1,\dots,x_n)
\]
This is equivalent to the bialternant definition of the Schur polynomials. For a full derivation of the Weyl character formula see Ref.~\cite{fultonandharris}.
\end{definition}

\subsection{\texorpdfstring{$U(n) \downarrow U(n-1)$}{U(n) to U(n-1)} Branching}

The foundational mathematical structure of the Paldus transform is based on a subgraph of the Bratelli diagrams within the $U(n) \supset U(n-1) \supset \cdots \supset U(1)$ chain. The branching of $U(n)\downarrow U(n-1)$ can  be explained with an inclusion map and Pieri's formula\cite{MacdonaldSymPoly,fultonandharris}. The inculsion map of elements of \( U(n-1) \hookrightarrow U(n) \) can be represented as block-diagonal matrices in \( U(n) \) with a \( U(n-1) \) matrix embedded in the top left, acting trivially on the $n$-th element:
\[
U(n-1) \hookrightarrow U(n):
\begin{pmatrix}
  \mathbf{U}_{(n-1) \times (n-1)} & \mathbf{0} \\
  \mathbf{0}^\dagger & 1
  \end{pmatrix}
\]
This allows the characters in the restriction of the irrep  $V_\lambda: U(n) \downarrow U(n-1)$ to be given as the Schur polynomials where the last entry is fixed to 1:
\[
\chi_{\lambda}^{U(n)}(z_1, z_2, \dots, z_{n-1},1) = s_\lambda(z_1, z_2, \dots, z_{n-1}, 1)
\]
It is straightforward to see that fixing of the diagonal elements will give a subgroup of $U(n)$ conjugacy classes. $\lambda$ is a partition given by $\lambda = (\lambda_1, \lambda_2, \dots, \lambda_k)$ with $k \leq n$. The Branching rule is given by the restriction:

\[
  \chi_{\lambda}^{U(n)}(z_1, z_2, \dots, z_{n-1},1) = \sum_{\mu} \chi_{\mu}^{U(n-1)}(z_1, z_2, \dots, z_{n-1})
\]

\begin{theorem}
  The branching rule of $U(n)$ is given by the betweenness condition of the partitions $\lambda$ and $\mu$:

  \[
  V_\lambda: U(n) \downarrow U(n-1) = \bigoplus_{\mu} V_\mu
  \]
  
  \noindent where the sum runs over all partitions $\mu$ which obey the betweenness (interleaving) conditions given by:

  \[
  \lambda_1 \geq \mu_1 \geq \lambda_2 \geq \mu_2 \geq \dots \geq \lambda_{k-1} \geq \mu_{k-1} \geq \lambda_k \geq \cdots \geq \lambda_{n-1} \geq \mu_{n-1} \geq \lambda_n
  \]
  
  \noindent where $ \mu_{k-1}$ and $\lambda_k$ are the last non zero parts of their respective partitions.
\label{thm:gt_branching_betweenness}
\end{theorem}

\begin{proof}
The proof expands and generalises the example given by Procotor~\cite{PROCTOR1989135}. By using the properties of determiants we will show:
\[
  \chi_{\lambda}^{U(n)}(z_1, z_2, \dots, z_{n-1},1) = \sum_{\mu} \chi_{\mu}^{U(n-1)}(z_1, z_2, \dots, z_{n-1})
\]
where $\lambda$ and the set $\{\mu\}$ obey the intervealing conditions. Using the Bialternant definition given in Theorem \ref{thm:bialternant} of the Schur polynomials:

\[
s_\lambda(x_1,..,x_n) = \frac{
  \begin{vmatrix}
  x_1^{\lambda_1+n-1} & x_2^{\lambda_1+n-1} & \cdots & x_n^{\lambda_1+n-1} \\
  x_1^{\lambda_2+n-2} & x_2^{\lambda_2+n-2} & \cdots & x_n^{\lambda_2+n-2} \\
  \vdots & \vdots & \ddots & \vdots \\
  x_1^{\lambda_n +n -n} & x_2^{\lambda_n +n -n} & \cdots & x_n^{\lambda_n +n -n} \\
  \end{vmatrix}
}{
  \Delta(x_1, x_2, \ldots, x_n)
}
\]
Where $\Delta(x_1, x_2, \ldots, x_n)$ is the Vandermonde determinant given by $\prod_{1 \leq i < j \leq n} (x_i - x_j)$. Applying the restriction property by setting $x_n = 1$ and using the properties of determinants we can subtract the last column from the first $n-1$ columns without changing the value of the determinant:
\[
\frac{
  \begin{vmatrix}
  x_1^{\lambda_1+n-1} & x_2^{\lambda_1+n-1} & \cdots & 1 \\
  x_1^{\lambda_2+n-2} & x_2^{\lambda_2+n-2} & \cdots & 1 \\
  \vdots & \vdots & \ddots & \vdots \\
  x_1^{\lambda_n +n -n} & x_2^{\lambda_n +n -n} & \cdots & 1 \\
  \end{vmatrix}
}{
  \Delta(x_1, x_2, \ldots, x_{n-1},1)
} =
\frac{
  \begin{vmatrix}
  x_1^{\lambda_1+n-1}-1 & x_2^{\lambda_1+n-1}-1 & \cdots & 1 \\
  x_1^{\lambda_2+n-2}-1 & x_2^{\lambda_2+n-2}-1 & \cdots & 1 \\
  \vdots & \vdots & \ddots & \vdots \\
  x_1^{\lambda_n +n -n}-1 & x_2^{\lambda_n +n -n}-1 & \cdots & 1 \\
  \end{vmatrix}
}{
  \Delta(x_1, x_2, \ldots, x_{n-1},1)
}
\]
In the numerator a factor of $(x_i-1)$ has been extracted from each column respectively using the identity $x_i^n - 1 = (x_i-1)(x_i^{n-1} + x_i^{n-2} + \cdots + x_i + 1) = (x_i-1)\sum_{m=0}^{n-1} x_i^m$. The same factor can be exracted from the denominator using the following identity:
\[
\Delta(x_1, x_2, \ldots, x_{n-1},1) = \prod_{i=1}^{n-1} (x_i - 1) \cdot \prod_{1 \leq i < j \leq n-1} (x_i - x_j) =  \prod_{i=1}^{n-1} (x_i - 1) \Delta(x_1, x_2, \ldots, x_{n-1})
\]
As the same factor is present in the numerator and denominator they cancel out: 
\[
\frac{\prod_{i=1}^{n-1} (x_i - 1)
  \begin{vmatrix}
  \sum^{\lambda_1+n-1-1}_{m=0} x_1^m & \sum^{\lambda_1+n-1-1}_{m=0} x_2^m & \cdots & 1 \\
  \sum^{\lambda_2+n-2-1}_{m=0} x_1^m & \sum^{\lambda_2+n-2-1}_{m=0} x_2^m & \cdots & 1 \\
  \vdots & \vdots & \ddots & \vdots \\
  \sum^{\lambda_n+n-n-1}_{m=0} x_1^m & \sum^{\lambda_n+n-n-1}_{m=0} x_2^m & \cdots & 1 \\
  \end{vmatrix}
}{
  \prod_{i=1}^{n-1} (x_i - 1) \Delta(x_1, x_2, \ldots, x_{n-1})
}
=
\frac{
  \begin{vmatrix}
  \sum^{\lambda_1+n-1-1}_{m=0} x_1^m & \sum^{\lambda_1+n-1-1}_{m=0} x_2^m & \cdots & 1 \\
  \sum^{\lambda_2+n-2-1}_{m=0} x_1^m & \sum^{\lambda_2+n-2-1}_{m=0} x_2^m & \cdots & 1 \\
  \vdots & \vdots & \ddots & \vdots \\
  \sum^{\lambda_n+n-n-1}_{m=0} x_1^m & \sum^{\lambda_n+n-n-1}_{m=0} x_2^m & \cdots & 1 \\
  \end{vmatrix}
}{
  \Delta(x_1, x_2, \ldots, x_{n-1})
}
\]

\noindent If the $i^\text{th}$ row is recursively subtracted from the $(i-1)^\text{th}$ row until the first row is reached, this leaves only a single nonzero term in the final row of the last column. Therefore, applying a cofactor expansion along the last column can reduce the dimensionality of the determinant matrix on the numerator to $n-1$:
\[
  \frac{
    \begin{vmatrix}
    \sum^{\lambda_1+n-1-1}_{m=\lambda_2+n-2} x_1^m & \sum^{\lambda_1+n-1-1}_{m=\lambda_2+n-2} x_2^m & \cdots & 0 \\
    \sum^{\lambda_2+n-2-1}_{m=\lambda_3+n-3} x_1^m & \sum^{\lambda_2+n-2-1}_{m=\lambda_3+n-3} x_2^m & \cdots & 0 \\
    \vdots & \vdots & \ddots & \vdots \\
    \sum^{\lambda_n+n-n-1}_{m=\lambda_{n-1}+n-{n-1}} x_1^m & \sum^{\lambda_n+n-n-1}_{m=\lambda_{n-1}+n-{n-1}} x_2^m & \cdots & 1 \\
    \end{vmatrix}
  }{
    \Delta(x_1, x_2, \ldots, x_{n-1})
  }
  =
  \sum_{\substack{\lambda_1 \geq \mu_1 \geq \lambda_2 \\
  \lambda_2 \geq \mu_2 \geq \lambda_3 \\
  \vdots \\
  \lambda_{n-1} \geq \mu_{n-1} \geq \lambda_{n}}
  }
  \frac{
  \begin{vmatrix}
      x_1^{\mu_1+(n-1)-1} & x_2^{\mu_1+(n-1)-1} & \cdots & x_n^{\mu_1+(n-1)-1} \\
      x_1^{\mu_2+(n-1)-2} & x_2^{\mu_2+(n-1)-2} & \cdots & x_n^{\mu_2+(n-1)-2} \\
      \vdots & \vdots & \ddots & \vdots \\
      x_1^{\mu_n +(n-1) -n} & x_2^{\mu_n +(n-1) -n} & \cdots & x_n^{\mu_n +(n-1) -n} \\
    \end{vmatrix}
  }{
    \Delta(x_1, x_2, \ldots, x_{n-1})
  }
\]
Crucially the row $i$ summation $\sum^{\lambda_i+n-1-i}_{m=\lambda_{i+1}+n-(i+1)} x_i^m$ can be rewritten as $\sum_{\mu_i = \lambda_i}^{\lambda_{i+1}} x^{\mu_i+(n-1)-i}$. Furthermore, due to the linearity of a summation over a row we can extract the sum out of the determinant, repeating this for multiple rows leads to a nested summation. This leads to the Schur polynomial and the interleaving condition:
\begin{equation}
  \begin{split}
s_{(\lambda_1,..,\lambda_n)} (x_1,..,x_{n-1},1) &= \sum_{\substack{\lambda_1 \geq \mu_1 \geq \lambda_2 \\
  \lambda_2 \geq \mu_2 \geq \lambda_3 \\
  \vdots \\
  \lambda_{n-1} \geq \mu_{n-1} \geq \lambda_{n}}
}
\frac{\det(x_i^{\mu_j +(n-1)-j})_{1 \leq i, j \leq n-1}}
{\Delta(x_1, x_2, \ldots, x_{n-1})} \\
&= \sum_{\substack{\lambda_1 \geq \mu_1 \geq \lambda_2 \\
  \lambda_2 \geq \mu_2 \geq \lambda_3 \\
  \vdots \\
  \lambda_{n-1} \geq \mu_{n-1} \geq \lambda_{n}}
} s_{(\mu_1,..,\mu_{n-1})}(x_1,..,x_{n-1})
\end{split}
\end{equation}
In the general case, for each row \(i\), the summation ranges from \(\lambda_{i+1} + n - (i+1)\) to \(\lambda_i + n - 1 - i\). This range enforces the betweenness condition on the partitions \(\lambda\) and \(\mu\), ensuring that only new terms not included in the preceding row contribute to the determinant. Consequently, all permissible \(\mu\) partitions satisfy
\[
\lambda_1 \;\ge\; \mu_1 \;\ge\; \lambda_2 \;\ge\; \mu_2 \;\ge\;\cdots\;\ge\; \lambda_{n-1}\;\ge\;\mu_{n-1}\;\ge\;\lambda_n,
\]
which completes the proof.
\end{proof}

\begin{proposition} The dimension of the irrep $V_\lambda$ of $U(n)$ is the sum of the dimensions of the irreps $V_\mu$ of $U(n-1)$ that obey the betweenness condition of the partitions $\lambda$ and $\mu$: 
\end{proposition}

\begin{proof}

The general case can be proven using the Weyl dimension formula~\cite{fultonandharris}:

\begin{equation}
  \begin{split}
  \prod_{1 \leq i < j \leq n} \frac{(\lambda_i - \lambda_j) + (j - i)}{j - i} &= \sum_{\mu} \prod_{1 \leq i < j \leq n-1} \frac{(\mu_i - \mu_j) + (j - i)}{j - i} 
  \end{split}
\end{equation}

\noindent The left hand side is the dimension of the irrep $V_\lambda$ of $U(n)$, and the right hand side is the sum of the dimensions of the irreps $V_\mu$ of $U(n-1)$ that obey the betweenness condition of the partitions $\lambda$ and $\mu$.

\end{proof}

\begin{example}
  For $n=4$, $\lambda = (3,2,1)$ and $k=6$ the branching rule is given by the following partitions:

\begin{equation}
  \begin{split}
  \begin{ytableau}
    \, & \, & \, \\
    \, & \, \\
    \,
  \end{ytableau} : U(4) \downarrow U(3)
&= \begin{ytableau}
  \, & \, & \, \\
  \, & \, \\
  \,
\end{ytableau} \\
&\oplus \begin{ytableau}
  \, & \, \\
  \, & \, \\
  \,
\end{ytableau}
\oplus \begin{ytableau}
  \, & \, & \, \\
  \,  \\
  \,
\end{ytableau}
\oplus \begin{ytableau}
  \, & \, & \, \\
  \, & \, \\
\end{ytableau} \\
& \oplus \begin{ytableau}
  \, & \, \\
  \, \\
  \, 
\end{ytableau}
\oplus \begin{ytableau}
  \, & \, \\
  \, & \,  \\
\end{ytableau}
\oplus \begin{ytableau}
  \, & \, & \,  \\
  \, \\
\end{ytableau} \\
& \oplus \begin{ytableau}
  \, & \, \\
  \, \\
\end{ytableau}
\end{split}
\end{equation}
Resulting in the set of partitions $\{\mu\}$. The dimensions of the irreps can be calculated using the Weyl dimension formula resulting in $64 = 8 + 3 + 6 + 15 + 3 + 6 + 15 + 8$

\end{example}

\begin{theorem}
The branching rule for $U(n) \downarrow U(n-1)\times U(1)$ is given by:

\[
V_\lambda: U(n) \downarrow U(n-1) \times U(1) = \bigoplus_{\mu} \left(V_\mu \otimes D_\nu\right)
\]

\noindent where the direct sum runs over all partitions $\mu$ which obey the betweenness (interleaving) conditions and $D_\nu$ is the standard 1-dimensional irrep of $U(1)$ with weight $|\nu| = |\lambda| - |\mu|$.
\label{thm:un_branching}
\end{theorem}

\begin{proof}
We can obtain $V_\lambda$ of $U(n)$ by lifting from interleaving irreps $\mu$ of $\lambda:U(n-1)$ via multiplication of the standard irrep of $U(1)$ with weight $|\nu| = |\lambda| - |\mu|$, where the variable of the $U(1)$ irrep is $x_n$:
\[
s_\lambda(x_1,\ldots,x_n) = \sum_\mu s_\mu(x_1,\ldots,x_{n-1})\,x_n^{|\lambda| - |\mu|}
\]
The left-hand side is a homogeneous function of degree $|\lambda|$ in the variables $x_1,\ldots,x_n$, and the right-hand side is a homogeneous function of degree $|\lambda|$ in the variables $x_1,\ldots,x_{n-1}$ and $x_n$.
\(\sum_\mu s_\mu(x_1,\ldots,x_{n-1})\) in the restriction is the LHS with \(x_n = 1\). Therefore, the two sides can be made equal by including \(x_n\) on the right-hand side with the appropriate degree to recover the LHS degree.
\end{proof}

\begin{remark}
  This provides a rigorous framework for the recursive lifting of representations of \( U(n-1) \) to \( U(n) \), and is used in the quantum Paldus transform to lift the representations of \( U(1) \) to \( U(d) \).
\end{remark}

\begin{example}
  For $n=3$, $\lambda = (4,2,1)$ the ratio of alternants has the following form:
  \[
\frac{
    \begin{vmatrix}
    x^{4+2} & y^{4+2} & z^{4+2} \\
    x^{2+1} & y^{2+1} & z^{2+1} \\
    x^{1+0} & y^{1+0} & z^{1+0}
    \end{vmatrix}
}{
    \begin{vmatrix}
    x^2 & y^2 & z^2 \\
    x^1 & y^1 & z^1 \\
    x^0 & y^0 & z^0
    \end{vmatrix}
}
=
\sum_{
    \substack{
    4 \geq \mu_1 \geq 2 \\
    2 \geq \mu_2 \geq 1
    }
}
\frac{
    \begin{vmatrix}
    x^{\mu_1+1} & y^{\mu_1+1} \\
    x^{\mu_2+0} & y^{\mu_2+0}
    \end{vmatrix}
}{
    \begin{vmatrix}
    x^1 & y^1 \\
    x^0 & y^0
    \end{vmatrix}
}
z^{\lambda_1 - \mu_1 + \lambda_2 - \mu_2 + \lambda_3}.
\]
To obtain this we can use the restriction property by setting $z=1$ and using the properties of determinants we can subtract the last column from the first two columns without changing the value of the determinant:
\[
\frac{
  \begin{vmatrix}
  x^{4+2} & y^{4+2} & 1 \\
  x^{2+1} & y^{2+1} & 1 \\
  x^{1+0} & y^{1+0} & 1
  \end{vmatrix}
}{
  \begin{vmatrix}
  x^2 & y^2 & 1 \\
  x^1 & y^1 & 1 \\
  x^0 & y^0 & 1
  \end{vmatrix}
}
=
\frac{
  \begin{vmatrix}
  x^{4+2}-1 & y^{4+2}-1 & 1 \\
  x^{2+1}-1 & y^{2+1}-1 & 1 \\
  x^{1+0}-1 & y^{1+0}-1 & 1
  \end{vmatrix}
}{
  \begin{vmatrix}
  x^2-1 & y^2-1 & 1 \\
  x^1-1 & y^1-1 & 1 \\
  x^0-1 & y^0-1 & 1
  \end{vmatrix}
}.
\]
The factors $(x-1)$ and $(y-1)$ can be factored out of the first and second column respectively using the identity $x^n - 1 = (x-1)(x^{n-1} + x^{n-2} + \cdots + x + 1)$:
\[
\frac{
    \begin{vmatrix}
    x^5 + x^4 + \cdots + x + 1 & y^5 + y^4 + \cdots + y + 1 & 1 \\
    x^2 + x + 1 & y^2 + y + 1 & 1 \\
    1 & 1 & 1
    \end{vmatrix}
}{
    \begin{vmatrix}
    x+1 & y+1 & 1 \\
    1 & 1 & 1 \\
    0 & 0 & 1
    \end{vmatrix}
}.
\]
Then, we subtract the $n^\text{th}$ row from the $(n-1)^\text{th}$ row for $n=2$ and $n=1$ recursively until the second row is subtracted from the first row:
\[
  \frac{
    \begin{vmatrix}
    x^5 + x^4 + x^3 & y^5 + y^4 + y^3 & 0 \\
    x^2 + x  & y^2 + y & 0 \\
    1 & 1 & 1
    \end{vmatrix}
}{
    \begin{vmatrix}
    x+1 & y+1 & 0 \\
    1 & 1 & 0 \\
    0 & 0 & 1
    \end{vmatrix}
}
=
\frac{
    \begin{vmatrix}
    x^5 + x^4 + x^3 & y^5 + y^4 + y^3 \\
    x^2 + x^1 & y^2 + y^1
    \end{vmatrix}
}{
    \begin{vmatrix}
    x^1 & y^1 \\
    x^0 & y^0
    \end{vmatrix}
}
=
\sum_{\mu_1 = 2}^{4}\sum_{\mu_2 = 1}^{2}
\frac{
    \begin{vmatrix}
    x^{\mu_1 + 1} & y^{\mu_1 + 1} \\
    x^{\mu_2 + 0} & y^{\mu_2 + 0}
    \end{vmatrix}
}{
    \begin{vmatrix}
    x^1 & y^1 \\
    x^0 & y^0
    \end{vmatrix}
}.
\]
The dimension reduction is obtained via a cofactor expansion along the final column. The result is therefore a sum over Schur polynomials and we have expressed the restriction from $V_{(4,2,1)}: U(3) \downarrow U(2)$ as a sum of Schur polynomials. From the example it is clear that the betweenness condition is a consequence of the row subtractions. Because the original expression is a homogeneous polynomial, we can substitute the correct degree of $z$ to match the overall degree of the expression:
\[
  \sum_{
    \substack{
    4 \geq \mu_1 \geq 2 \\
    2 \geq \mu_2 \geq 1
    }
}
\frac{
    \begin{vmatrix}
    x^{\mu_1+1} & y^{\mu_1+1} \\
    x^{\mu_2+0} & y^{\mu_2+0}
    \end{vmatrix}
}{
    \begin{vmatrix}
    x^1 & y^1 \\
    x^0 & y^0
    \end{vmatrix}
}
z^{4 - \mu_1 + 2 - \mu_2 + 1}.
\]

Where $(\lambda_1, \lambda_2, \lambda_3) = (4,2,1)$ and $\{(\mu_1, \mu_2)\} = \{(4,2), (4,1), (3,2), (3,1), (2,2), (2,1)\}$. The inbetweenness condition is satisfied by the partitions $\mu$ and $\lambda$. This can be represented by the Young diagrams:

\begin{equation}
  \begin{split}
  \begin{ytableau}
    \, & \, & \, & \,\\
    \, & \, \\
    \,
  \end{ytableau} : U(3) \downarrow U(2) \times U(1)
&= \begin{ytableau}
  \, & \, & \, & \,\\
  \, & \, \\
  3 \\
\end{ytableau} 
\oplus \begin{ytableau}
  \, & \, & \, & \,\\
  \, & 3 \\
  3 \\
\end{ytableau}
\oplus \begin{ytableau}
  \, & \, & \, & 3\\
  \, & \, \\
  3 \\
\end{ytableau} \\
&\oplus \begin{ytableau}
  \, & \, & \, & 3\\
  \, & 3 \\
  3 \\
\end{ytableau} 
\oplus \begin{ytableau}
  \, & \, & 3 & 3\\
  \, & \, \\
  3 \\
\end{ytableau} 
\oplus \begin{ytableau}
  \, & \, & 3 & 3\\
  \, & 3 \\
  3 \\
\end{ytableau} 
\end{split}
\end{equation}

\noindent with the boxes labeled $3$ denoting the order of the $z$ variable in the resulting polynomial. Hence, removing the boxes labeled $3$ gives the partitions $\mu$.

\end{example}

\begin{theorem} 
  $s_\lambda(x_1,,.., x_n) = \sum_T x^{T}$, where the sum is over all SSYTs $T$ of shape $\lambda$ and $x^T = x_1^{a_1(T)} x_2^{a_2(T)} \cdots x_n^{a_n(T)} $. $ a_i(T) $ denotes the number of times $ i $ appears in the tableau $T$.
  \label{thm:un_branching_poly}
\end{theorem}
  
\begin{proof}
  Using Theorem~\ref{thm:un_branching}, the proof proceeds via induction starting with a partition $\lambda$ if we restrict to the partition the sets of partitions $\{\mu\}$ that obey the betweenness condition of the partitions $\lambda$ and $\mu$ is we fill the missing boxes between $\lambda$ and $\mu$ with the number $n$ and recursively down to $n=1$ we can recover the SSYTs of shape $\lambda$ as each SSYT corresponds to a unique branching down the recursion.
\end{proof}

\begin{remark}
  Each SSYT is isomorphic to a monomial, which in turn corresponds to a basis element of the irreducible representation $V_\lambda$ of $U(n)$. Hence, the number of SSYTs of shape $\lambda$ is equal to the dimension of the irreducible representation $V_\lambda$ of $U(n)$, where each monomial is a term in the trace of the irreducible representation, giving the character. This recursion through sequential addition of boxes can be seen in Figure~\ref{fgr:CSFALL}.
\end{remark}

The Branching rule can be applied in a recursive fashion for all subgroups of $ U(n) $ until the irreps of $ U(1) $ are reached, resulting in a Bratelli diagram

\begin{equation}
  U(n) \supset U(n-1) \supset \cdots \supset U(1).
\label{eq:unbratelli}
\end{equation}

\noindent Interestingly, a unique branch of along the Bratelli diagram of the following branching rule of $U(n)$
\[
U(n) \supset U(n-1) \times U(1) \supset U(n-2) \times U(1) \times U(1) \supset \cdots \supset U(1)^{\times n},
\]
leads to an irrep of $U(1)^{\times n}$, whose character is a monomial in the Schur polynomial of the character of the original irrep of $U(n)$. 

Following a unique branch along the subgroup chain in Equation~\ref{eq:unbratelli}, at each level $n$ we select a single irreducible representation $\lambda_i$. This procedure yields the Gelfand tableaux, which provide a canonical, one-to-one correspondence with the basis elements of the irreducible representation of $U(n)$.

\begin{definition}[Gelfand Tableaux]
  The Gelfand tableau is a unique representation of the basis elements of the $U(n)$ irrep. It is built through a unique branch of the recursion of the branching rule of $U(n) \downarrow U(n-1)$, obeying the interleaving constraint. The $n^\text{th}$ row of the tableau corresponds to a partition $\lambda_n$ corresponding to an irrep of $U(n)$.

  \begin{equation}
    \mathbf{m} = 
  \begin{bmatrix}
  \lambda_n\\
  \lambda_{n-1}\\\
  \vdots\\
  \lambda_2\\
  \lambda_1
  \end{bmatrix} 
  = 
  \begin{bmatrix}
  \lambda_{1,n} &  & \lambda_{2,n} &  & \cdots & & \lambda_{n-1,n} &  & \lambda_{n,n}  \\
   & \lambda_{1,n-1} & &  \lambda_{2,n-1} & & \cdots &  &  \lambda_{n-1,n-1} &  \\
   &  & \ddots &  & \cdots &  & \iddots &  &  \\
   &  &  & \lambda_{1,2} &  & \lambda_{2,2} &  &  &  \\
   &  &  &  & \lambda_{1,1} &  &  &  &   
  \end{bmatrix} 
  \end{equation}

\noindent Where the $\mathbf{m}$ represents and orthogonal basis element of the $U(n)$ irrep. An example is given in Figure \ref{fgr:gt_example}.

\end{definition}
  
    \begin{figure}
      \centering
      \includegraphics[width=\textwidth]{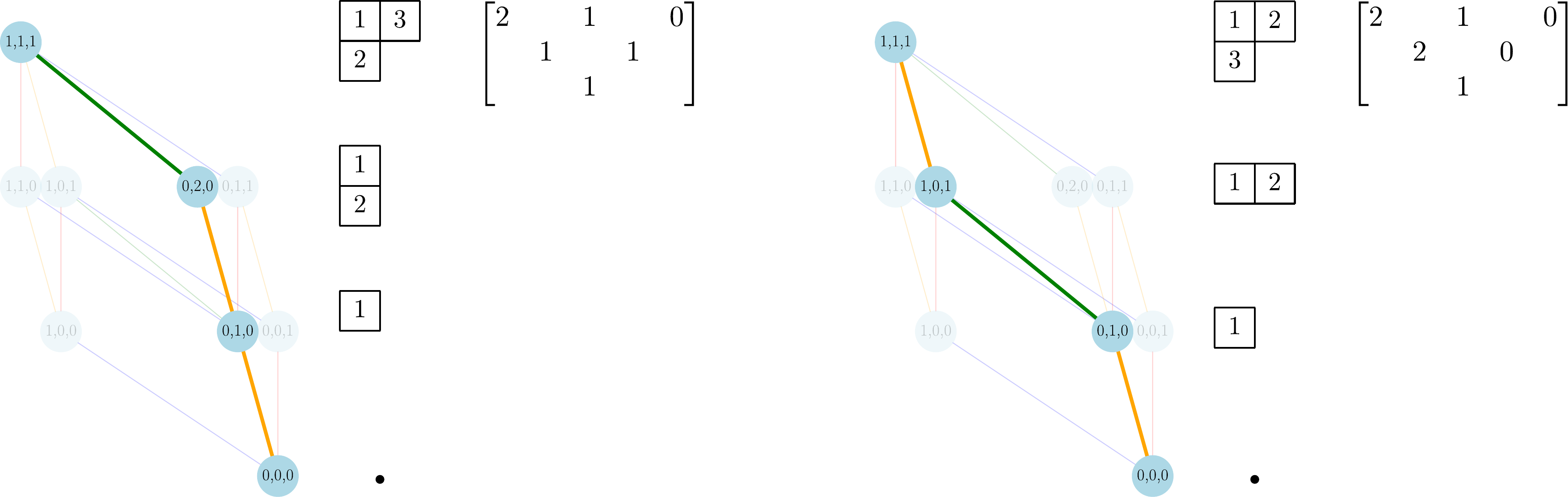}
      \caption{Two examples of Gelfand Tableaux for the irrep of $U(3)$ with partition $(2,1,0)$ and their corresponding walks on the Shavitt graph.} 
      \label{fgr:gt_example}
    \end{figure}

\subsection{\texorpdfstring{$U(2) \downarrow SU(2)$}{U(2) to SU(2)} Branching}\label{app:u2su2}

The $SU(2)$ matrices are given by the set $\{ U \in U(2) \mid \det(U) = 1 \}$. The $U(2)$ matrices are given by the set $\{ U \in GL(2) \mid U U^\dagger = I, \ \det(U) = e^{i\theta} \}$. The branching rule describes how the \( U(2) \) irrep \( V^{U(2)}_{(\lambda_1, \lambda_2)} \) decomposes into irreps of \( SU(2) \) when restricted to \( SU(2) \). \(U(2)\) can be described as a quotient of \(SU(2)\) by kernel of the homomorphism, the center of \(SU(2)\), \(Z(SU(2)) = \{ \pm I \}\):

\[
U(2) = SU(2) / Z(SU(2))
\]

\noindent This is a $2$-to-$1$ homomorphism, \(\phi: SU(2) \rightarrow U(2)\), where \(\phi(V) = e^{i\theta} V\). This extra phase in $U(2)$ is physically irrelevant upon measurement and this serves as some intuition as to why $U(2)$ can often be used interchangeably with $SU(2)$. The standard irrep of U(2) with all conjugacy classes parameterised by \(\theta_1\) and \(\theta_2\) is given by:

\[ 
U(e^{i\theta_1}, e^{i\theta_2}) = \begin{pmatrix}
  e^{i\theta_1} & 0 \\
  0 & e^{i\theta_2}
\end{pmatrix}
\]

\noindent For the Lie algebra \( \mathfrak{su}(2) \), the Cartan subalgebra is the maximal Abelian (commuting) subalgebra. In \( \mathfrak{su}(2) \), the Cartan subalgebra is generated by a single element, typically chosen to be the \( z \)-component of the angular momentum operator, \( J_z \). Thus, we can consider \( J_z \) as generating the Cartan subalgebra of \( \mathfrak{su}(2) \). The \( J_z \) element is given by:

\[
J_z = \frac{1}{2} \begin{pmatrix}
  1 & 0 \\
  0 & -1
\end{pmatrix}
\]

\noindent For an irrep with spin \( j \), the \( J_z \) matrix is given by:

\[ 
J_z = \begin{pmatrix}
  j & 0 & \cdots & 0 \\
  0 & j-1 & \cdots & 0 \\
  \vdots & \vdots & \ddots & \vdots \\
  0 & 0 & \cdots & -j
\end{pmatrix}
\]

\noindent This generates the diagonal elements of $SU(2)$ via the exponential map \( e^{i\theta J_z} \), and since all conjugacy classes are represented by a diagonal matrix, the character of the standard irrep of $SU(2)$ is given by:

\[
U(g) = \begin{pmatrix}
  e^{i j \theta} & 0 & \cdots & 0 \\
  0 & e^{i (j-1) \theta} & \cdots & 0 \\
  \vdots & \vdots & \ddots & \vdots \\
  0 & 0 & \cdots & e^{-i j \theta}
\end{pmatrix}
\]

\noindent Therefore, the trace is given by:

\[
\chi^{SU(2)}_j(e^{i\theta}) = \sum_{m=-j}^{j} e^{i m \theta}
\]

\noindent where \( m \) is incremented by $1$ from \( -j \) to \( j \).

\begin{theorem}
  The branching rule for an irrep of \( U(2) \) decomposing an irreps of \( SU(2) \) is given by:
  \[
    V^{U(2)}_{(\lambda_1, \lambda_2)} \downarrow V^{SU(2)}_j
  \]
  where \((\lambda_1, \lambda_2)\) are the partitions labeling the highest weight of \(U(2)\), and \(j = \frac{\lambda_1 - \lambda_2}{2}\) labels the irreps of \(SU(2)\).
  \label{thm:u2su2}
\end{theorem}

\begin{proof}
Using the Weyl character formula \cite{Weyl1925}, the character of the \(U(2)\) irrep is given by:

\[
  \chi^{U(2)}_{\lambda_1,\lambda_2}(z_1,z_2) = \frac{z_1^{\lambda_1+1}z_2^{\lambda_2} - z_1^{\lambda_2}z_2^{\lambda_1+1}}{z_1 - z_2}.
\]

\noindent Then making the substitution \(z_1 = e^{i\theta_1/2}\) and \(z_2 = e^{i\theta_2/2}\) and setting \(\theta_2 = -\theta_1\), the character of the \(SU(2)\) irrep is given by:

\[
  \chi^{SU(2)}_{\lambda_1,\lambda_2}(\theta_1) = \frac{e^{i(\lambda_1 - \lambda_2 +1)\theta_1/2} - e^{i(\lambda_2 - \lambda_1 +1)\theta_1/2}}{e^{i\theta_1/2} - e^{-i\theta_1/2}} = \frac{e^{i(2j +1)\theta_1/2} - e^{i(2j +1)\theta_1/2}}{e^{i\theta_1/2} - e^{-i\theta_1/2}} = \frac{\sin((2j +1)\theta_1/2)}{\sin(\theta_1/2)} =\sum_{m = -j}^{j} e^{im\theta_1}
\]

\noindent Which is the known character of $SU(2)$, where \(j = \frac{\lambda_1 - \lambda_2}{2}\) labels the irreps of $SU(2)$
\end{proof}

\noindent When restricting $U(2) \downarrow SU(2)$, the coupling of irreps described was described by the Pieri rule for the characters in Theorem \ref{thm:dualpieri}, however in $SU(2)$ irrep coupling can be described by explicitly by Clebsch Gordon coupling \cite{Brink:1975}:

\begin{equation}
  V^{SU(2)}_j \otimes V^{SU(2)}_{j'} = \bigoplus_{j'' = |j - j'|}^{j + j'} V^{SU(2)}_{j''}
\end{equation}
This Clebsch-Gordan coupling replaces the Pieri rule when restricting to \( SU(2) \), because it also describes the coupling of a pair of irreducible representations into a direct sum of irreducible components. This is evident because the $j$ values can be derived directly from the partition $\lambda$ of the $U(2)$ irrep via $j = \frac{\lambda_1 + \lambda_2}{2}$, so the Young diagram notation can be used interchangeably with the spin notation. This results in the $U(n) \times SU(2)$ of the Paldus transform. Hence this leads to the presence of the genealogical coupling of $SU(2)$ in the Paldus transform.

\subsection{Shavitt Branching}
\label{app:shavitt_branching}

In this section we derive the branching rule for the antisymmetric representation of $U(d) \times U(2)$ restricted to $U(d-1) \times U(2)$, which is used in the Paldus transform. The branching rule is derived through recursively restricting the antisymmetric representation, and given by the following theorem:

\begin{theorem}[Shavitt Branching]
  Let $\mathbb{C}^d$ and $\mathbb{C}^2$ be $d$ and $2$-dimensional Hilbert spaces respectively, and let $U(d)$ denote the unitary group on a $d$–dimensional complex space. The \textit{lifting} of the antisymmetric representation of $U(d-1)\times U(2)$ to $U(d)\times U(2)$ is achieved by four branching constraints via the subgroup $(U(d-1)\times U(1))\times U(2)$:
  
  \begin{align}
  \bigwedge (\mathbb{C}^d \otimes \mathbb{C}^2) \cong \bigoplus_{\lambda}\Bigl( V^{U(d)}_{\lambda} \otimes V^{U(2)}_{\tilde{\lambda}} \Bigr) &\cong\bigoplus_{\mu}\Biggl\{ \Bigl( V^{U(d-1)\times U(1)}_{\mu+d_0} \otimes V^{U(2)}_{\tilde{\mu}+(0,0)} \Bigr)
  \oplus \Bigl( V^{U(d-1)\times U(1)}_{\mu+d_1} \otimes V^{U(2)}_{\tilde{\mu}+(1,0)} \Bigr) \\
  &\quad\quad\quad\oplus \Bigl( V^{U(d-1)\times U(1)}_{\mu+d_2} \otimes V^{U(2)}_{\tilde{\mu}+(0,1)} \Bigr)
  \oplus \Bigl( V^{U(d-1)\times U(1)}_{\mu+d_3} \otimes V^{U(2)}_{\tilde{\mu}+(1,1)} \Bigr)
  \Biggr\},
  \end{align}
  
  \noindent where $\mu + \{d_0, d_1, d_2, d_3\}$ represents the additions of boxes to the two–column irreducible representations of $U(d-1)$ that are conjugate to the additions of boxes to the two–row $U(2)$ Young diagrams given by $\tilde{\mu} + \{(0,0), (1,0), (0,1), (1,1)\}$. Here, $\mu$ denotes the Young diagram corresponding to an irrep of $U(d-1)$, and $\lambda$, $\tilde{\lambda}$ denote the Young diagrams corresponding to the irreps of $U(d)$ and $U(2)$ respectively. The terms in the direct sum must not violate the constrain of the $U(2)$ irrep having at most at most two rows, and the conjugate $U(d)$ irreps having at most two columns. 
  \label{thm:shavitt_branching}
  \end{theorem}

  \begin{proof} By using the internal recursion of the exterior product and then decomposing it into irreps using Theorem~\ref{thm:ffspace_umn_decomp} we obtain:
  
  \begin{align}
  \bigwedge^{\bullet 2(d-1)} (\mathbb{C}^d \otimes \mathbb{C}^2) \otimes \bigwedge^{\bullet 2} (\mathbb{C}^d \otimes \mathbb{C}^2) &\cong \bigoplus_{\mu}\Biggl\{V^{U(d-1 )}_\mu \otimes V^{U(2)}_{\tilde{\mu}}\Biggl\}\otimes \Biggl\{I \oplus \Big(D_{(1)}^{U(1)} \otimes  V^{SU(2)}_{(1)}\Big) \oplus \Big(D_{(2)}^{U(1)} \otimes  V^{U(2)}_{(1,1)}\Big)\Biggl\}
  \end{align}
  
  \noindent where the expansion of the second factor in the tensor product,
  
  \begin{equation}
  \bigwedge^{\bullet 2} (\mathbb{C}^d \otimes \mathbb{C}^2)
  \cong \Bigl( I \otimes I \Bigr) \oplus \Bigl( V_{(1)}^{U(1)} \otimes V^{U(2)}_{(1)} \Bigr)
  \oplus \Bigl( V_{(2)}^{U(1)} \otimes V^{U(2)}_{(1,1)} \Bigr),
  \end{equation}
  
  \noindent follows from Theorem \ref{thm:paldus_duality}. Then, using the characters of the exterior power representation and the dual Cauchy identity in Appendix~\ref{thm:dualcauchy}, with the $U(d-1)$ variables given as $(x_1,...,x_{d-1})$, the $U(2)$ variables as $y_1,y_2$ and $U(1)$ variable as $x_d$, we decompose the character as:
  
    \begin{align}
    \chi_{U(d) \times U(2)}\Bigl( \bigwedge^{\bullet 2d} (\mathbb{C}^d \otimes \mathbb{C}^2) \Bigr) &= \chi_{U(d-1) \times U(2)}\left(  \bigwedge^{\bullet 2(d-1)} (\mathbb{C}^d \otimes \mathbb{C}^2) \right) \cdot \chi_{U(1) \times U(2)}\Bigl( \bigwedge^{\bullet 2} (\mathbb{C}^d \otimes \mathbb{C}^2) \Bigr) \\
    &= \prod_{i=1}^{d-1} \prod_{j=1}^{2} \bigl(1 + x_i y_j\bigr) \cdot \prod_{k=1}^{2} \bigl(1 + x_d y_k\bigr) \\
    &= \left(\sum_{\mu} s_\mu(x_1,\dots,x_{d-1})\, s_{\tilde{\mu}}(y_1,y_2)\right) \cdot \left(1 + x_d y_1 + x_d y_2 + x_d^2 y_1 y_2\right),
    \end{align}
  
    \noindent with $\mu$ running over the partitions related to the $(d-1)$-dimensional sum. This decomposition leads to four distinct ways to extend $\mu$ and its conjugate $\tilde{\mu}$ in the $d^\text{th}$ recursion step:

  \begin{align}
    \chi_{U(d) \times U(2)}\Bigl( \bigwedge^{\bullet 2d} (\mathbb{C}^d \otimes \mathbb{C}^2) \Bigr) = &\sum_{\mu}\Bigl(s_\mu(x_1,\dots,x_{d-1})\, s_{\tilde{\mu}}(y_1,y_2) \\
    & \quad \quad + s_\mu(x_1,\dots,x_{d-1})\, x_d \cdot s_{\tilde{\mu}}(y_1,y_2)\, y_1 \\
    & \quad \quad + s_\mu(x_1,\dots,x_{d-1})\, x_d \cdot s_{\tilde{\mu}}(y_1,y_2)\, y_2 \\
    & \quad \quad + s_\mu(x_1,\dots,x_{d-1})\, x_d^2 \cdot s_{\tilde{\mu}}(y_1,y_2)\, y_1 y_2
    \Bigr) \\
    & = \sum_\lambda \Bigl( s_\lambda(x_1,\dots,x_{n})\, s_{\tilde{\lambda}}(y_1,y_2) \Bigr).
  \end{align}
  The $U(d-1)$ to $U(d)$ lifting is achieved via a direct sum over intermediate representations of $U(d-1)\times U(1)$ -- as derived in Theorem~\ref{thm:un_branching} -- where the products $s_\mu(x_1,\dots,x_{d-1})\, x_d$ are expressed in terms of Young diagrams with additional elements corresponding to $d$. The Schur polynomials $s_{\tilde{\mu}}$ are multiplied by monomials in $y_1^i$ and $y_2^j$, with $i,j \in \{0,1\}$, as dictated by the dual Pieri rule coupling derived in Theorem~\ref{thm:dualpieri}. The four terms correspond to the four distinct methods for extending the Young diagram $\mu$ and its conjugate $\tilde{\mu}$ in the $d^\text{th}$ recursion step. The constraint that the partition $\tilde{\mu}$ has at most two rows implies, by conjugation, that the partition $\mu$ contains at most two columns. Each term in the sum represents the character of an irreducible representation of $(U(d-1) \times U(1)) \times U(2)$, where the $U(1)$ character is realised as the Schur polynomial $x_d$. Since the branching is multiplicity-free by Howe duality (see Theorem~\ref{thm:howe_duality}), the characters yield the irreducible representation decomposition precisely, thereby completing the proof.
  \end{proof}

  \begin{figure}[htbp!]
    \centering
      \includegraphics[width=15cm]{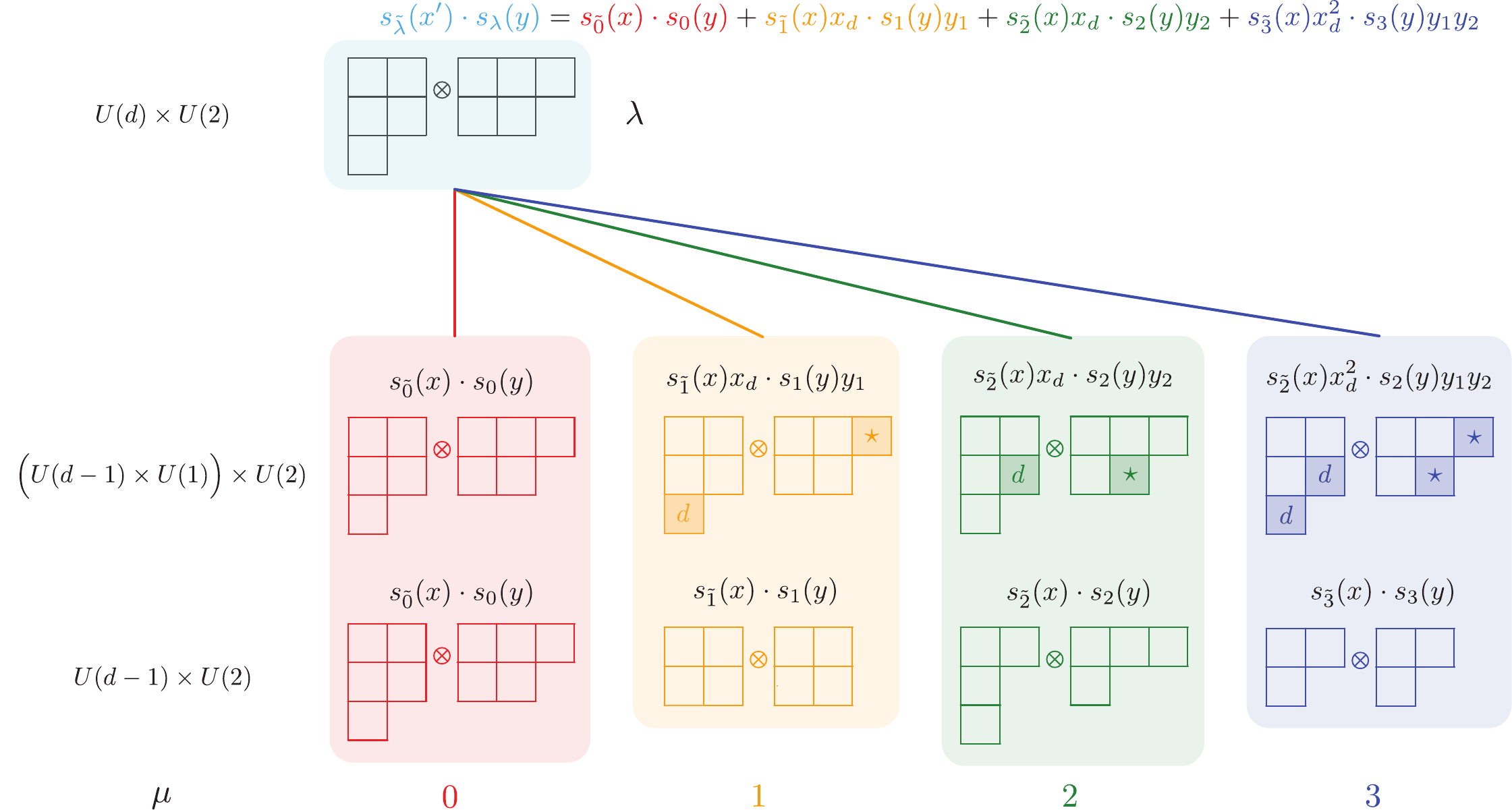}
      \caption{Creation of a new $U(d) \times U(2)$ irrep $\lambda$ via incoming $U(d-1) \times U(2)$ irreps $\{\mu_0, \mu_1, \mu_2, \mu_3 \}$, where the $U(d) \times U(2)$ irrep is created from an $U(d-1) \times U(2)$ irrep via a summation over $\left(U(d-1) \times U(1) \right) \times U(2)$ irreps (directed upwards in the diagram). The $U(1)$ variable is the Schur polynomial $x_d$ and the $\star$ is the Pieri rule coupling for the $U(2)$ group (See Example~\ref{ex:dual_pieri}). $s_{\tilde{\lambda}}(x')$ is the Schur polynomial character of the $U(d)$ irrep, where $x' = x_1, x_2, \hdots, x_{d}$. $s_{\tilde{\mu}}(x)$ is the Schur polynomial character of the $U(d-1)$ irrep, where $x = x_1, x_2, \hdots, x_{d-1}$. $s_{\mu}(y)$ and $s_{\lambda}(y)$ are the Schur polynomial characters of $U(2)$ irreps, where $y = y_1, y_2$.}
      \label{fgr:schur_poly_d_branching_new}
    \end{figure}
  
    \begin{figure}[htbp!]
      \centering
        \includegraphics[width=12cm]{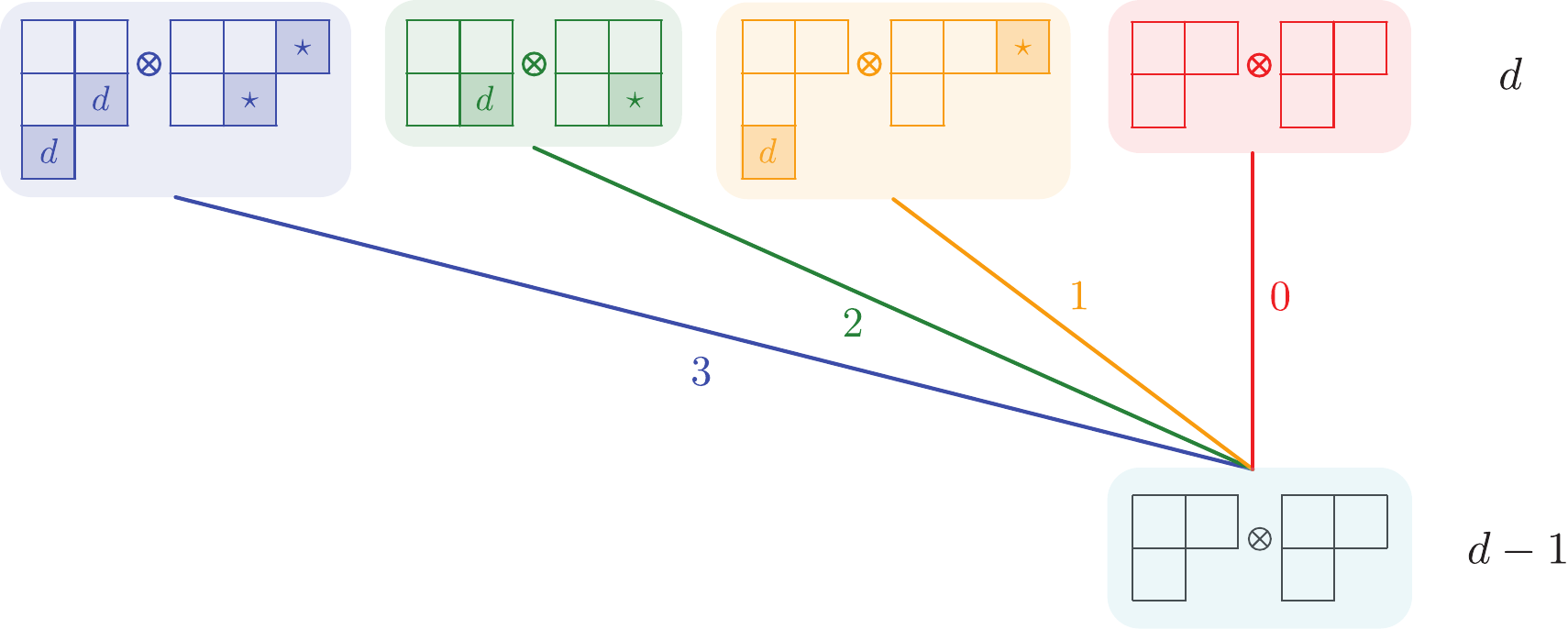}
        \caption{A graphical representation of lifting a $U(d-1) \otimes U(2)$ irrep pair $\tilde{\mu}, \mu$ to $U(d) \otimes U(2)$ irreps via $\mathbf{d}_i$ steps. The $\star$ indicates an application of the Pieri rule. The $\mathbf{d}_i$ represent four possible couplings of a $1$-dimensional irrep of $U(1)$ to the $U(d-1)$ irrep.}
        \label{fgr:yt_d_branching}
      \end{figure}
  
  This branching process is illustrated in Figure~\ref{fgr:schur_poly_d_branching_new}, where the decomposition from $U(d-1) \times U(2)$ to $U(d) \times U(2)$ irreducible representations is achieved via an intermediate $U(d-1) \times U(1)$ representation. In the figure, the Young diagrams correspond to the Schur polynomials as discussed previously (as well as in Appendix~\ref{app:symmetricpolynomials} and Theorem \ref{thm:un_branching}). The $U(1)$ variable is represented by the Schur polynomial $x_d$, and the $\star$ symbol denotes the dual Pieri rule coupling for the $U(2)$ group (see Theorem~\ref{thm:dualpieri}). The coupling of the $U(d-1) \times U(1)$ representations is also detailed in Theorem~\ref{thm:gt_branching_betweenness}. Furthermore, the Schur polynomial character of the $U(d)$ irreducible representation is given by $s_{\tilde{\lambda}}(x')$, where $x' = (x_1,x_2,\dots,x_{d})$, while that of the $U(d-1)$ irreducible representation is given by $s_{\tilde{\mu}}(x)$, where $x = (x_1,x_2,\dots,x_{d-1})$. The Schur polynomial characters of the $U(2)$ representations are denoted by $s_{\mu}(y)$ and $s_{\lambda}(y)$. Theorem~\ref{thm:shavitt_branching} thus identifies four possible ${d}_i$ branching directions (subject to constraints), where the index $d_i$ indicates the coupling of a one-dimensional representation of $U(1)$ to the $U(d-1)$ representation.

  Starting with a $U(1)$ irrep, the consecutive lifting up the subgroup chain $U(1) \subset  U(2) \subset  \cdots \subset  U(d)$ can be described by a sequence of \textit{step vectors} $\mathbf{d}$, where $\mathbf{d}_i \in \{0, 1, 2, 3\}$ represents the four possible ways to embed the $U(i)$ irrep into a $U(i+1)$ irrep (i.e, the four different ways to `add a box' to its Young diagram). This can be represented graphically, in what is known in the UGA literature as a \textit{Shavitt graph} \cite{Shavitt1978,Shepard2006}, leading to a Bratellli diagram-like structure and forming a graphical representation of the GT basis deployed in the Paldus transform. A single embedding is shown in Figure \ref{fgr:yt_d_branching}, and a Shavitt graph with multiple embeddings is shown in Figure \ref{fgr:CSFALL}.

\section{Summation of the Dimension Formula} \label{app:dimension_formula}

Recall that in the decomposition of the vector space into antisymmetric irreps of $U(d) \times U(2)$, every irreducible representation can be labelled by quanum numbers $N, S$. For every irrep, the number of distinct $U(2)$ GT patterns is given by $2S + 1$, whereas the number of distinct $U(d)$ GT paterns (i.e. step vectors $\mathbf{d}$) can be found from the Weyl dimension formula \cite{PALDUS1976131}:

\begin{equation}
  T^d_{S, N} = \frac{2S+1}{d+1} {d + 1 \choose \frac{1}{2} N - S} {d + 1 \choose d - \frac{1}{2} N - S}, \label{eq:weyl_dimension}
\end{equation}

\noindent where for a given $S$ the permitted particle numbers are $N \in \{ 2S, 2S + 2, ..., 2d - 2S\}$. Taking a sum of $T^d_{S, N}$ over all $N$ gives us $T^d_{S}$, which is the number of all distinct step vectors with a common value of $S$:

\begin{align*}
  T^d_{S} &= \frac{2S+1}{d+1} \sum_{N}  {d + 1 \choose \frac{1}{2} N - S} {d + 1 \choose d - \frac{1}{2} N - S} \\
&= \frac{2S+1}{d+1} \sum_{k=0}^{d - 2S} {d + 1 \choose k} {d + 1 \choose d - 2S - k} \\
&= \frac{2S+1}{d+1} {2d + 2 \choose d - 2S}
\end{align*}

\noindent where in the first line we have set $N = 2S + 2k$ and in the second line we have used the Vandermonde identity. Using basic binomial identities, $T^d_{S}$ can be rewritten as a difference of two binomial coefficients:

\begin{align*}
  T^d_{S} &= \frac{2S+1}{d+1} \times \frac{(2d + 2)}{(2d+2) - 2(d - 2S)} \left( {(2d + 2) -1 \choose d - 2S} - {(2d + 2) - 1  \choose (d - 2S) - 1}\right) \\
  &=  {2d + 1 \choose d - 2S} - {2d + 1  \choose d - 2S - 1}
\end{align*}

We can sum $T^d_{S}$ over all values of $S$ to obtain $T^d$, which is the total number distinct step vectors (i.e, the sum of dimensions of all \textit{inequivalent} $U(d)$ irreps appearing in the decomposition). This is a simple telescoping sum:

\begin{align*}
  T^d &= \sum_S T_{d, S} \\
  &= \sum_{2S = 0}^d \left( {2d + 1 \choose d - 2S} - {2d + 1  \choose d - 2S - 1} \right) \\
  &= {2d + 1 \choose d} - {2d + 1 \choose d-1} + {2d + 1 \choose d - 1} + \hdots + {2d + 1 \choose 0} -  {2d + 1 \choose 0}\\
  &= {2d + 1 \choose d}.
\end{align*}

To convince ourselves that the total number of $U(d) \times U(2)$ GT states is indeed $2^{2d}$, we can count them by taking the following sum:

\begin{align*}
  \sum_{S, N} (2S+1) \times T^d_{S, N} &= \sum_{S} (2S+1) \times T_{d, S} \\
   &= \sum_{2S=0}^d (2S+1)  \left( {2d + 1 \choose d - 2S} - {2d + 1  \choose d - 2S - 1} \right) \\
  & = \sum_{k=0}^d {2d + 1 \choose k} \\
  & = \frac{1}{2} \sum_{k=0}^{2d + 1} {2d + 1 \choose k} \\
  & = 2^{2d}.
\end{align*}

Where the first line is again a telescoping sum, which may be re- written as the second line. The second line forms exactly one half of the sum over $2d+2$ binomial coefficients, and can be re-written as the third line.

\section{Circuits} \label{sec:circuit_compilation}
In this section we introduce some of the leading compilation strategies for the abstract primitives utilised in the quantum Paldus transform. These strategies are essential for optimizing the implementation of the transform, and the corresponding gate counts are provided to support the resource estimation discussed in the main text.
\subsection{\texorpdfstring{$\select$}{SELECT}}

The $\select$ circuit depicted in Figure~\ref{fgr:select_box} implements a series of indexed multi-controlled unitaries, denoted as $\{U_i\}$. In this framework, the control register encodes a binary index that sequentially activates the corresponding unitary operation, thereby realizing the transformation known as a $\select(U)$ operation. The $\select(U)$ operation is defined as
\[
\select(U) = \sum^{I}_i |i\rangle \langle i| \otimes U_i.
\] 
where the index $i$ corresponds to the binary expansion of the integer stored within the control register of $m = \lceil \log_2 I \rceil$ qubits, which acts on a state register containing $|\psi\rangle$ on $q$ qubits. Figure \ref{fgr:select_box} provides a schematic representation of the $\select(U)$ construction.

\begin{figure}[htbp!]
    \centering
    \begin{adjustbox}{width=\textwidth}
    \begin{quantikz}
    \lstick{$|i\rangle$}     &\qwbundle{m}  &   \mctrl{1} & \qw \\
    \lstick{$|\psi\rangle$}  &\qwbundle{q}  & \gate{\sel(U)}  &\qw \\ 
    \end{quantikz}
    \hspace{0.25cm}
    =
    \begin{quantikz}[wire types={q,n,q,q,q,q}]
    \lstick{$|i\rangle_0$}     &             & \octrl{1}     & \octrl{1}             & \octrl{1}      & \octrl{1}    &   \octrl{1}   & \octrl{1}     & \octrl{1}      & \octrl{1}     & \cdots\\
                               &             & \vdots\vqw{1} & \vdots\vqw{1}         & \vdots\vqw{1}  & \vdots\vqw{1}& \vdots\vqw{1} & \vdots\vqw{1}  &\vdots\vqw{1}  &\vdots\vqw{1}  & \\
    \lstick{$|i\rangle_{m-3}$} &             & \octrl{1}     & \octrl{1}             & \octrl{1}      & \octrl{1}    & \ctrl{1}      & \ctrl{1}       & \ctrl{1}      & \ctrl{1}      & \cdots\\
    \lstick{$|i\rangle_{m-2}$} &             &\octrl{1}      & \octrl{1}             & \ctrl{1}       & \ctrl{1}     & \octrl{1}     & \octrl{1}      & \ctrl{1}      & \ctrl{1}      & \cdots\\
    \lstick{$|i\rangle_{m-1}$} &             &\octrl{1}      & \ctrl{1}              & \octrl{1}      & \ctrl{1}     &  \octrl{1}    & \ctrl{1}       & \octrl{1}     & \ctrl{1}      & \cdots\\
    \lstick{$|\psi\rangle$}    &\qwbundle{q} &\gate{U_0}     & \gate{U_1}            & \gate{U_2}     &  \gate{U_3}  & \gate{U_4}    &  \gate{U_5}    &  \gate{U_6}   &  \gate{U_7}   & \cdots& \rstick{$U_i|\psi\rangle$} \\
    \end{quantikz}
    \end{adjustbox}

    \caption{$\select(U)$ construction over the set of unitaries $\{U_i\}$. $|i\rangle$ are the $m$ index qubits and the $|\psi\rangle$ is the $q$ qubit state register that $\{U_i\}$ acts on.}
    \label{fgr:select_box}
\end{figure}
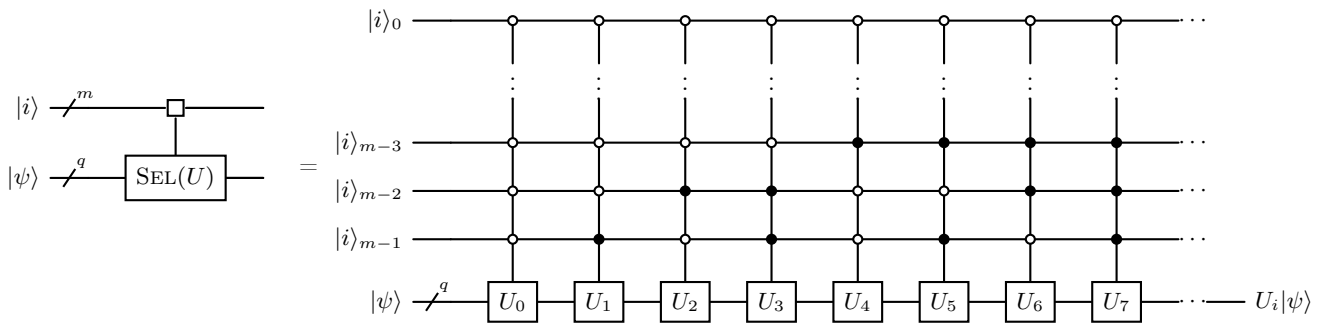

The $\select(U)$ can be compiled efficiently using the unary iteration technique~\cite{Babbush2018LinearT}. This method implements the $\select(U)$ with a decreased gate count and reduced circuit depth. The unary iteration approach leverages the observation that an $m$-controlled unitary can be synthesised via a Toffoli cascade, with $m-1$ work qubits. Moreover, the sequential application of this cascade induces multiple adjacent cancellations, ultimately lowering the overall $T$ gate count to approximately $4I$.  
    
    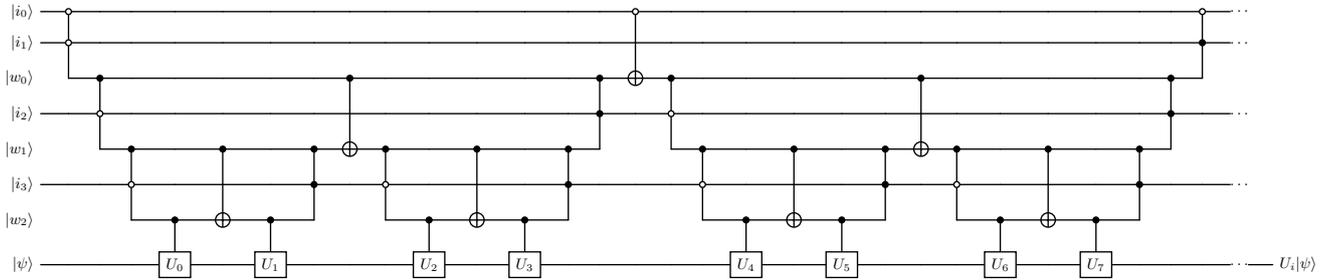
\begin{figure}[htbp!]
        \centering
        \begin{adjustbox}{width=\textwidth}
        \begin{quantikz}
        \lstick{$|i_0\rangle$}  & \octrl{1} &           &           &           &           &           &           &           &           &           &           &           &           &           & \octrl{2} &           &           &           &           &           &           &           &           &           &           &           &           &           & \octrl{1} & \cdots \\
        \lstick{$|i_1\rangle$}  & \octrl{1} &           &           &           &           &           &           &           &           &           &           &           &           &           &           &           &           &           &           &           &           &           &           &           &           &           &           &           & \ctrl{1}  & \cdots \\
        \lstick{$|w_0\rangle$}  & \nw       & \ctrl{1}  &           &           &           &           &           & \ctrl{2}  &           &           &           &           &           & \ctrl{1}  & \targ{0}  & \ctrl{1}  &           &           &           &           &           & \ctrl{2}  &           &           &           &           &           & \ctrl{1}  &           & \nw    \\
        \lstick{$|i_2\rangle$}  &           & \octrl{1} &           &           &           &           &           &           &           &           &           &           &           & \ctrl{1}  &           & \octrl{1} &           &           &           &           &           &           &           &           &           &           &           & \ctrl{1}  &           & \cdots \\
        \lstick{$|w_1\rangle$}  & \nw       & \nw       & \ctrl{1}  &           & \ctrl{2}  &           & \ctrl{1}  & \targ{0}  & \ctrl{1}  &           & \ctrl{2}  &           & \ctrl{1}  &           & \nw       & \nw       & \ctrl{1}  &           & \ctrl{2}  &           & \ctrl{1}  & \targ{0}  & \ctrl{1}  &           & \ctrl{2}  &           & \ctrl{1}  &           &  \nw      & \nw    \\
        \lstick{$|i_3\rangle$}  &           &           & \octrl{1} &           &           &           & \ctrl{1}  &           & \octrl{1} &           &           &           & \ctrl{1}  &           &           &           & \octrl{1} &           &           &           & \ctrl{1}  &           & \octrl{1} &           &           &           & \ctrl{1}  &           &           & \cdots \\
        \lstick{$|w_2\rangle$}  &\nw        &\nw        & \nw       & \ctrl{1}  & \targ{0}  & \ctrl{1}  &           & \nw       & \nw       & \ctrl{1}  & \targ{0}  & \ctrl{1}  &           & \nw       &  \nw      & \nw       &   \nw     & \ctrl{1}  & \targ{0}  & \ctrl{1}  &           &  \nw      & \nw       & \ctrl{1}  & \targ{0}  & \ctrl{1}  &           &  \nw      &  \nw      & \nw    \\
        \lstick{$|\psi\rangle$}    &           &           &           & \gate{U_0}&           & \gate{U_1}&           &           &           & \gate{U_2}&           & \gate{U_3}&           &           &           &           &           & \gate{U_4}&           & \gate{U_5}&           &           &           & \gate{U_6}&           & \gate{U_7}&           &           &           & \cdots&\rstick{$U_i|\psi\rangle$} \\ 
        \end{quantikz}
        \end{adjustbox}
        \caption{\textit{Unary Iteration} construction on an Index Box over the set of unitaries $\{U_i\}$ showning 8 indices of a 16 element set. $|w_i\rangle$ are $3$ work qubits, $|i\rangle$ are the $4$ index qubits and the $|\psi\rangle$ is the state register which $\{U_i\}$ act on.} 
        \label{fgr:unary_iteration}
    \end{figure}

    The $4I$ $T$ gate complexity arises from the construction of the compute Toffoli acting on the $|0\rangle$ state, as shown in Figure~\ref{fig:compute_toffoli}. The uncompute Toffoli can be further optimised using the measurement-based uncomputation technique, as illustrated in Figure~\ref{fig:uncompute_toffoli}, which results in it not being included in the overall Toffoli cost. Consequently, the total cost of the $\select(U)$ operation is $I$ Toffoli gates (excluding the cost of the controlled unitaries $U_i$). Alternatively, if mid-circuit measurement is unavailable, a standard Toffoli gate can be employed, which decomposes into $7$ $T$ gates.
    \begin{figure}[htbp!]
        \centering
        \begin{quantikz}
             & \ctrl{1}  & \\
             & \ctrl{1}  & \\
             & \nw       & \\
        \end{quantikz}
        =
        \begin{quantikz}
                                & \ctrl{1}  & \\
                                & \ctrl{1}  & \\
        \lstick{$|0\rangle$}    & \targ{0}  & \\
        \end{quantikz}
        =
        \begin{quantikz}
                                &            &                  & \ctrl{2}  &                   &               &                   &           &                   &     \\
                                & \ctrl{1}   &                  &           &                   & \ctrl{1}      &                   &           &                   &     \\
        \lstick{$|T\rangle$}    & \targ{0}   & \gate{T^\dagger} & \targ{0}  & \gate{T^\dagger}  & \targ{0}      & \gate{T^\dagger}  & \gate{H}  & \gate{S^\dagger}  &      \\
        \end{quantikz}
        \caption{Computation Toffoli, where the target state is always $|0\rangle$ using a $T$ state and $3$ $T$ Gates}
        \label{fig:compute_toffoli}
        \end{figure}
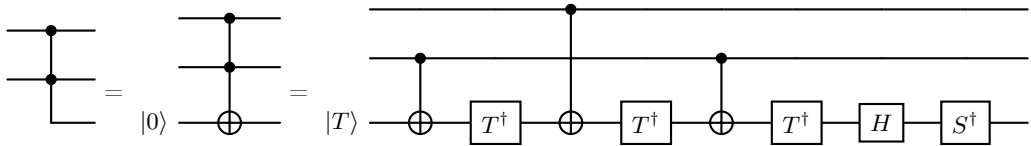

        \begin{figure}[!htbp]
            \centering
            \begin{quantikz}
                & \ctrl{1}   & \\
                & \ctrl{1}   & \\
                &           & \nw   \\
            \end{quantikz}
            =
            \begin{quantikz}
                & \ctrl{1}   & \\
                & \ctrl{1}   & \\
                & \targ{1}   & \rstick{$|0\rangle$} \\
        \end{quantikz}
        =
        \begin{quantikz}
            &           &           & \ctrl{1}                 &\\
            &           &           & \gate{Z}                  &\\
            &\gate{H}   & \meter{}  & \ctrl[vertical wire=c]{-1} \setwiretype{c} \\
        \end{quantikz}
            \caption{Measurement-based uncomputation Toffoli gate, where the final target state is $|0\rangle$}

            \label{fig:uncompute_toffoli}
        \end{figure}
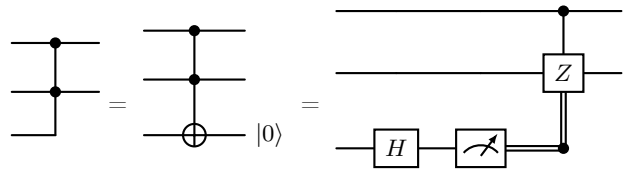

\subsection{Data Lookup}
\label{sec:data_lookup}

This section presents methods for storing indexed binary data in a quantum register. The data lookup operation is a fundamental building block in quantum algorithms, and it will be used in the quantum Paldus transform to implement controlled rotations. The $\select$ construction can be used to implement a $\data$ look up operation (sometimes called QROM~\cite{Babbush2018LinearT}), where the unitaries $\{U_i\}$ encode binary data $D_i$. We can write the data as $D_i = \bigotimes^{b-1}_{j=0} X^{i}_j$, which is just a series of $X$ gates on the data qubits that encodes the binary number into the qubit register. This can be seen in Figure~\ref{fig:data_box}.

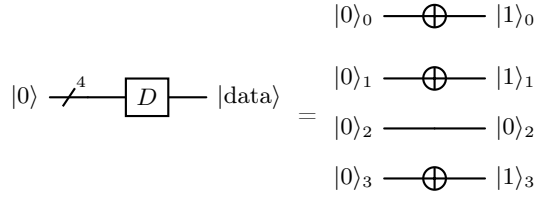
\begin{figure}[htbp!]
    \centering
    \begin{quantikz}
        \lstick{$|0\rangle$} &\qwbundle{4} & \gate{D}   & \rstick{$|\text{data}\rangle$} \\
    \end{quantikz}
    =
    \begin{quantikz}
       \lstick{$|0\rangle_0$}  &\targ{}    & \rstick{$|1\rangle_0$} \\
       \lstick{$|0\rangle_1$}  &\targ{}    & \rstick{$|1\rangle_1$} \\
       \lstick{$|0\rangle_2$}  &           & \rstick{$|0\rangle_2$}  \\
       \lstick{$|0\rangle_3$}  &\targ{}    & \rstick{$|1\rangle_3$}   \\
    \end{quantikz}
    
    \caption{Data box circuit for a $4$-bit index $(1,1,0,1)$ } \label{fig:data_box}
\end{figure}

\noindent In the data lookup oracle the initial state must be in the $|0\rangle^{\otimes n}$ state to encode the data, which can be seen in Figure~\ref{fgr:controlled_data_box}. Choosing an index $i$ on the control register will apply the data box $D_i$ to the target state $|0\rangle^{\otimes n}$.
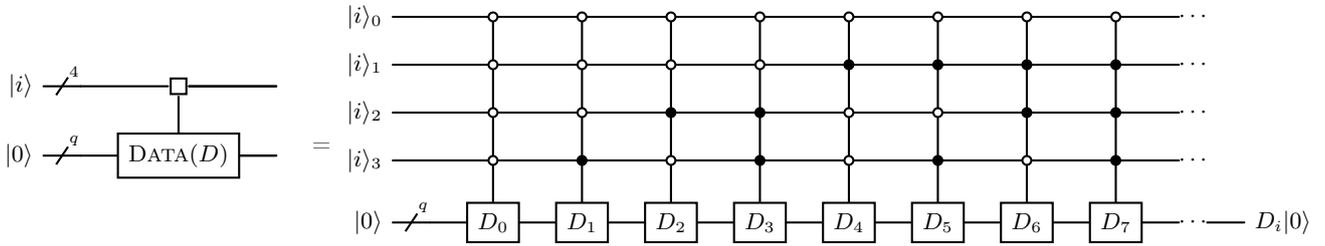
\begin{figure}[htbp!]
    \centering
    \begin{adjustbox}{width=\textwidth}
    \begin{quantikz}
    \lstick{$|i\rangle$}  & \qwbundle{4} &    \mctrl{1} & \qw \\
    \lstick{$|0\rangle$}  &\qwbundle{q}  & \gate{\data(D)}  &\qw \\ 
    \end{quantikz}
    \hspace{0.25cm}
    =
    \begin{quantikz}
    \lstick{$|i\rangle_0$} &              &  \octrl{1} & \octrl{1} & \octrl{1} & \octrl{1} &  \octrl{1} & \octrl{1} & \octrl{1} & \octrl{1} & \cdots\\ 
    \lstick{$|i\rangle_1$} &              &  \octrl{1} & \octrl{1} & \octrl{1} & \octrl{1} & \ctrl{1} & \ctrl{1} & \ctrl{1} & \ctrl{1} & \cdots\\
    \lstick{$|i\rangle_2$} &              & \octrl{1} & \octrl{1} & \ctrl{1} & \ctrl{1} & \octrl{1} & \octrl{1} & \ctrl{1} & \ctrl{1} & \cdots\\
    \lstick{$|i\rangle_3$} &              & \octrl{1} & \ctrl{1} & \octrl{1} & \ctrl{1} & \octrl{1} & \ctrl{1} & \octrl{1} & \ctrl{1} & \cdots\\
    \lstick{$|0\rangle$}   &\qwbundle{q}  & \gate{D_0}  &  \gate{D_1}  &  \gate{D_2}  &  \gate{D_3}  & \gate{D_4}  &  \gate{D_5}  &  \gate{D_6}  &  \gate{D_7}  & \cdots & \rstick{$D_i|0\rangle$}\\
    \end{quantikz}
    \end{adjustbox}
    \caption{\textit{Data Lookup} over the set of unitaries $\{U_i\}$ showing 8 indices of a 16 element set, which can be compiled by unary iteration, where we choose an index $i$ and apply the data box $D_i$ to the target $|0\rangle^{\otimes q}$ state.} \label{fgr:controlled_data_box}
    \end{figure}

\subsubsection{Unary Iteration}
\label{sec:unary_iteration}

The most straightforward compilation of the data lookup operation is to use unary iteration, which is depicted in Figure~\ref{fgr:unary_iteration}. Compared to the $\select$ operation, the main difference here is that the controlled unitaries are now $D_i$, and the target state is $|0\rangle$.

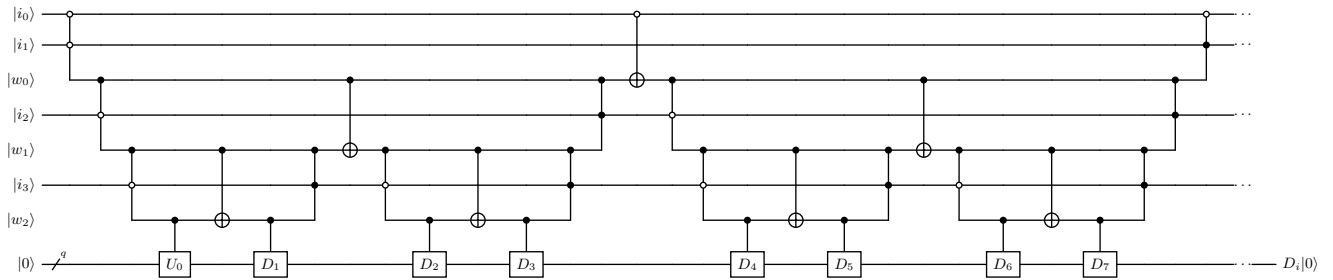
\begin{figure}[htbp!]
    \centering
    \begin{adjustbox}{width=\textwidth}
    \begin{quantikz}
    \lstick{$|i_0\rangle$}  & \octrl{1} &           &           &           &           &           &           &           &           &           &           &           &           &           & \octrl{2} &           &           &           &           &           &           &           &           &           &           &           &           &           & \octrl{1} & \cdots \\
    \lstick{$|i_1\rangle$}  & \octrl{1} &           &           &           &           &           &           &           &           &           &           &           &           &           &           &           &           &           &           &           &           &           &           &           &           &           &           &           & \ctrl{1}  & \cdots \\
    \lstick{$|w_0\rangle$}  & \nw       & \ctrl{1}  &           &           &           &           &           & \ctrl{2}  &           &           &           &           &           & \ctrl{1}  & \targ{0}  & \ctrl{1}  &           &           &           &           &           & \ctrl{2}  &           &           &           &           &           & \ctrl{1}  &           & \nw    \\
    \lstick{$|i_2\rangle$}  &           & \octrl{1} &           &           &           &           &           &           &           &           &           &           &           & \ctrl{1}  &           & \octrl{1} &           &           &           &           &           &           &           &           &           &           &           & \ctrl{1}  &           & \cdots \\
    \lstick{$|w_1\rangle$}  & \nw       & \nw       & \ctrl{1}  &           & \ctrl{2}  &           & \ctrl{1}  & \targ{0}  & \ctrl{1}  &           & \ctrl{2}  &           & \ctrl{1}  &           & \nw       & \nw       & \ctrl{1}  &           & \ctrl{2}  &           & \ctrl{1}  & \targ{0}  & \ctrl{1}  &           & \ctrl{2}  &           & \ctrl{1}  &           &  \nw      & \nw    \\
    \lstick{$|i_3\rangle$}  &           &           & \octrl{1} &           &           &           & \ctrl{1}  &           & \octrl{1} &           &           &           & \ctrl{1}  &           &           &           & \octrl{1} &           &           &           & \ctrl{1}  &           & \octrl{1} &           &           &           & \ctrl{1}  &           &           & \cdots \\
    \lstick{$|w_2\rangle$}  &\nw        &\nw        & \nw       & \ctrl{1}  & \targ{0}  & \ctrl{1}  &           & \nw       & \nw       & \ctrl{1}  & \targ{0}  & \ctrl{1}  &           & \nw       &  \nw      & \nw       &   \nw     & \ctrl{1}  & \targ{0}  & \ctrl{1}  &           &  \nw      & \nw       & \ctrl{1}  & \targ{0}  & \ctrl{1}  &           &  \nw      &  \nw      & \nw    \\
    \lstick{$|0\rangle$}    & \qwbundle{q}   &           &           & \gate{D_0}&           & \gate{D_1}&           &           &           & \gate{D_2}&           & \gate{D_3}&           &           &           &           &           & \gate{D_4}&           & \gate{D_5}&           &           &           & \gate{D_6}&           & \gate{D_7}&           &           &           & \cdots & \rstick{$D_i|0\rangle$} \\ 
    \end{quantikz}
    \end{adjustbox}
    \caption{$\data(D)$ operation \textit{Unary Iteration} over the set of unitaries $\{D_i\}$ showning 8 indices of a 16 element set. $|w_i\rangle$ are work qubits of the $3$ qubit register, $|i\rangle$ are the index qubits of the $4$ qubit index register and the $|\psi\rangle$ is the state register that $\{D_i\}$ act on.} 
    \label{fgr:unary_iteration_data}
\end{figure}

\noindent The cost of the control compilation will be $4I$ $T$ gates or $I$ Toffolis. The controlled data boxes $\{D_i\}$ are a sequential set of $CX$ gates depending on the binary representation chosen. We assume that the mean binary string will have half of the number of data qubits $q$ encoded with $X$ gates. We can thus approximate that each controlled $U$ will be $q / 2$ $CX$ gates. For the purpose of our resource estimation, we denote this as $\Theta(In)$ Clifford gates.

\subsubsection{Clean \texorpdfstring{$\selswap$}{SelectSwap}}

        Using $k$ registers of $q$ clean qubits, the gate count can be reduced following the procedure of \cite{Low2024tradingtgatesdirty}. This combines a $\select$ and a series of controlled register $\swup$ operations over multiple registers.

            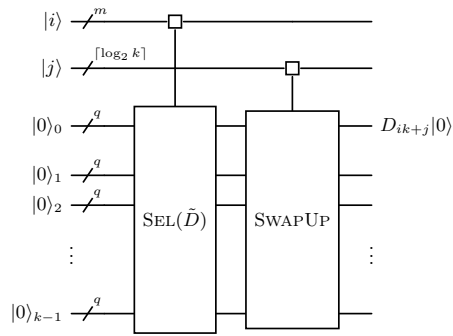
\begin{figure}[!htbp]
                \centering
            \begin{adjustbox}{width=6cm}
            \begin{quantikz}[wire types={q,q,q,q,q,n,q}]
            \lstick{$|i\rangle$}      & \qwbundle{m}               &\mctrl{2}      &                        & \\
            \lstick{$|j\rangle$}      & \qwbundle{\lceil\log_2 k \rceil}              &               &\mctrl{1}               & \\
            \lstick{$|0\rangle_0$}      & \qwbundle{q}&\gate[5]{\sel(\tilde{D})} &\gate[5]{\swup}           & \rstick{$D_{ik+j}|0\rangle$}  \\
            \lstick{$|0\rangle_1$}      & \qwbundle{q}&               &                        & \\
            \lstick{$|0\rangle_2$}      & \qwbundle{q}&               &                        & \\
            \vdots                    &             &               &                        & \vdots \\
            \lstick{$|0\rangle_{k-1}$}      & \qwbundle{q}&               &                        & \\
            \end{quantikz}
            \end{adjustbox}
            \caption{Clean $\selswap$ operation for a $data$ lookup. $k$ is the number of clean target registers, $q$ is the number of qubits in each register. $\sel(\tilde{D})$ is the combined product of data boxes $\tilde{D}_i = \otimes^k_j D_{ik+j}$. The $\swup$ operation swaps the $j^\text{th}$ element of the set of $k$ indexed with $i$ to the top register and results in $D_{ik+j}|0\rangle$.}
            \end{figure}

            The $\select(\tilde{D})$ operation can be compiled using unary iteration (Section~\ref{sec:unary_iteration}); however, now each controlled operation is $\tilde{D}_i = \otimes^k_j D_{ik+j}$. This means that every controlled element contains $k$ data boxes, where $k$ must be a power of $2$ (See Ref~\cite{Low2024tradingtgatesdirty} for values not powers of $2$). This is illustrated in Figure~\ref{fgr:sel_swap_example}. The controlled $\swup$ box essentially takes the register $|q\rangle_i$, with $i$ indexed on the control register, and swaps it up to the top qubit. A register swap is simply a qubit-wise swap operation for each pair of qubits.

            \begin{figure}[!htbp]
                \centering
            \begin{adjustbox}{width=9cm}
            \begin{quantikz}
                \lstick{$|i\rangle$}   & \qwbundle{3}   &\mctrl{1}      &  \\
                \lstick{$|q\rangle_0$} &                &\gate[8]{\swup}  & \rstick{$|q\rangle_i$}  \\ 
                \lstick{$|q\rangle_1$} &                &               &  \\ 
                \lstick{$|q\rangle_2$} &                &               &  \\ 
                \lstick{$|q\rangle_3$} &                &               &  \\ 
                \lstick{$|q\rangle_4$} &                &               &  \\ 
                \lstick{$|q\rangle_5$} &                &               &  \\ 
                \lstick{$|q\rangle_6$} &                &               &  \\ 
                \lstick{$|q\rangle_7$} &                &               &  \\ 
                \end{quantikz}
                \hspace{0.25cm}
                =
            \begin{quantikz}
            \lstick{$|i\rangle_0$} & {} & {} & {} & {} & {} & {} & \ctrl{3} & \\
            \lstick{$|i\rangle_1$} & {} & {} & {} & {} & \ctrl{3} & \ctrl{2} & {} & \\
            \lstick{$|i\rangle_2$} & \ctrl{4} & \ctrl{3} & \ctrl{2} & \ctrl{1} & {} & {} & {} & \\
            \lstick{$|q\rangle_0$} & {} & {} & {} & \swap{4} & {} & \swap{2} & \swap{1} & \rstick{$|q\rangle_i$} \\
            \lstick{$|q\rangle_1$} & {} & {} & \swap{4} & {} & \swap{2} & {} & \targX{} & \\
            \lstick{$|q\rangle_2$} & {} & \swap{4} & {} & {} & {} & \targX{} & {} & \\
            \lstick{$|q\rangle_3$} & \swap{4} & {} & {} & {} & \targX{} & {} & {} &\\
            \lstick{$|q\rangle_4$} & {} & {} & {} & \targX{} & {} & {} & {}& \\
            \lstick{$|q\rangle_5$} & {} & {} & \targX{} & {} & {} & {} & {}& \\
            \lstick{$|q\rangle_6$} & {} & \targX{} & {} & {} & {} & {} & {}& \\
            \lstick{$|q\rangle_7$} & \targX{} & {} & {} & {} & {} & {} & {}& \\
            \end{quantikz}   
            \end{adjustbox}
            \caption{$\swup$ operation, where the $|q\rangle_i$ is the $i^\text{th}$ qubit of the set of 8 qubits indexed with $i$ and is swapped to the top position.} \label{fig:swap}
            \end{figure}
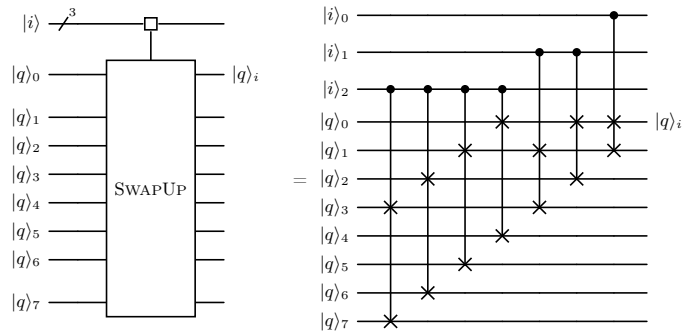

            To make this clearer, an example is presented in Figure~\ref{fgr:sel_swap_example} with an $I=16$ element set spread across $k=4$ registers. Hence, each control on the $i$ register has 4 data boxes. The controlled $\swup$ operation takes the $j^\text{th}$ element of the set of 4 indexed with $i$ and brings it to the top register.

            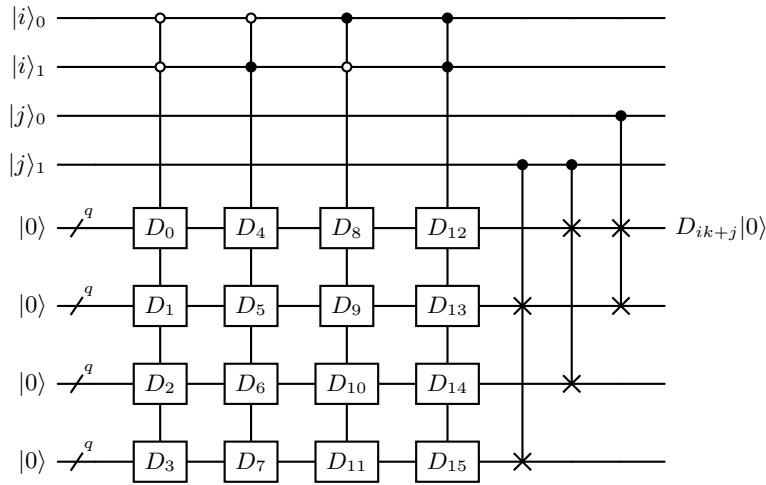
\begin{figure}[htbp!]
                \centering
                \begin{quantikz}
                \lstick{$|i\rangle_0$}  &            & \octrl{1}   & \octrl{1}      & \ctrl{1}          & \ctrl{1}          &         &         &          & \\
                \lstick{$|i\rangle_1$}  &            & \octrl{6}   & \ctrl{6}       & \octrl{6}         & \ctrl{6}          &         &         &          & \\
                \lstick{$|j\rangle_0$}  &            &             &                &                   &                   &         &         &\ctrl{2}  & \\
                \lstick{$|j\rangle_1$}  &            &             &                &                   &                   &\ctrl{2} &\ctrl{1} &          & \\
                \lstick{$|0\rangle$}    &\qwbundle{q}& \gate{D_0}  &  \gate{D_4}    &  \gate{D_8}       &\gate{D_{12}}      &         &\swap{2} &\swap{1}  & \rstick{$D_{ik+j}|0\rangle$} \\
                \lstick{$|0\rangle$}    &\qwbundle{q}& \gate{D_1}  &  \gate{D_5}    &  \gate{D_{9}}     &\gate{D_{13}}      &\swap{2} &         &\targX{}  & \\
                \lstick{$|0\rangle$}    &\qwbundle{q}& \gate{D_2}  &  \gate{D_6}    &  \gate{D_{10}}    &\gate{D_{14}}      &         &\targX{} &          & \\
                \lstick{$|0\rangle$}    &\qwbundle{q}& \gate{D_3}  &  \gate{D_7}    &  \gate{D_{11}}    &\gate{D_{15}}      &\targX{} &         &          & \\
            \end{quantikz}
            
            \caption{Example of $\selswap$ operation with $L=16$ elements and $k=4$ registers. The $D_{ik+j}$ data boxes are indexed via the double index $i,j$. It can be seen that each $i$ index has 4 data boxes, one for each register. The $j$ index then performs the $\swup$ operation to the top qubit to obtain $D_{ik+j}|0\rangle$.} \label{fgr:sel_swap_example}
        \end{figure}

        \begin{figure}[!htbp]
            \centering
            \begin{quantikz}
                & \ctrl{1}   & \\
                & \swap{1}   & \\
                & \targX{}   & \\
        \end{quantikz}
        =
        \begin{quantikz}
           & & \ctrl{1}   &  &\\
           & \targ{1} & \ctrl{1}   &\targ{1} &\\
           & \ctrl{-1} & \targ{1}   &\ctrl{-1}& \\
        \end{quantikz}
        $\approx$
        \begin{quantikz}
            &           &  &          & &\ctrl{2} &   &           & &             &\\
            & \targ{1}  &  &\ctrl{1}  & &         &   &\ctrl{1}   & &\targ{1}     &\\
            & \ctrl{-1} &\gate{G^\dagger}   &\targ{1}  &\gate{G^\dagger}  &\targ{1} &\gate{G}    &\targ{1}   & \gate{G} &\ctrl{-1}    & \\
         \end{quantikz}
            \caption{Controlled SWAP $G = S^\dagger H T H S$, correct up to phase. Where a single application in the clean $\selswap$ acts on $|0\rangle$ or $|1\rangle$ and the phase is negligible. In the dirty $\selswap$, this is applied in pairs of conjugate operations so the phase error cancels.}
            \label{fig:enter-label}
        \end{figure}
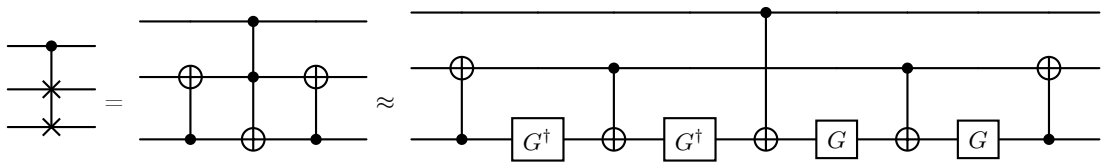
        Alternatively, it can also be helpful to view the data boxes as elements in a rectangular array with elements $D_{i,j}$, which is a recasting of the original linear array, with the $i$ and $j$ registers indexing the columns and rows respectively. The Toffoli cost of computing the $\selswap$ operation is $\lceil I / k \rceil + q(k - 1)$, where the $\select$ operation is $I/k$ Toffolis and the $\swup$ operation is $q(k-1)$ Toffolis. The Toffoli cost can be minimised by choosing the correct $I$ and $k$ values, depending on the clean qubit resources available. The cost of the uncompute $\selswap$ is $\lceil I/k\rceil + k$ Toffolis, which uses a measurement-based uncomputation procedure, compiled with a phase fixup operation (see Appendix C of \cite{Berry2019qubitizationof} for more details).

\subsubsection{Dirty \texorpdfstring{ $\selswap$}{SelectSwap}}

When there are $k-1$ registers of $q$ idling \textit{dirty} qubits (already in some state $|\psi\rangle$), they can be employed in a modified $\selswap$ workflow with a single clean register and $k-1$ dirty registers. These dirty qubits are used for computation and then restored to their original state $|\psi\rangle_i$, which still results in an overall gate complexity comparable to unary iteration. This is illustrated in Figure~\ref{fgr:dirty_qubit_selswap}. There are $k-1$ dirty qubit registers of $q$ qubits. In this approach, the clean register is initialised in the $|0\rangle$ state, then an $H$ gate is applied to all qubits in the register, and the $\swup$ operation is reversed to move the clean qubits in the $|+\rangle$ state to the other positions. The data boxes are applied, and the $\swup$ operation is performed again to return the clean qubits to the top register; finally, $H$ gates are applied to return the clean qubits to their original state $|0\rangle$. Because the sequence of operations satisfies $HXH|0\rangle = |0\rangle$, the top qubit remains in the $|0\rangle$ state after the first half of the procedure. However, dirty qubits that were not acted on by the $H$ gate may have accumulated various $X$ operations from the data box. The operations in the first and second parts of the algorithm are conjugate from the perspective of the dirty qubits, so applying them again returns these qubits to $|\psi\rangle_i$. In the second half of the procedure, the clean qubits do not undergo $H$ gates. A swap-down step disperses the clean qubits among all registers, analogous to the beginning of a clean $\selswap$ operation. The combined product of data boxes $\tilde{D}_i = \otimes^k_j D_{ik+j}$ is then applied according to the index $i$, and the element $j$ is swapped up to the top register. This achieves the same result as a clean $\selswap$ and returns the dirty qubits to their original state. There are $4$ $\swup$ operations costing $4n(k-1)$ Toffolis and $2$ $\select$ operations costing $2I/k$ Toffolis. The total cost of the dirty $\selswap$ is $4q(k-1) + 2I/k$ Toffolis and $\Theta(In)$ Clifford gates.

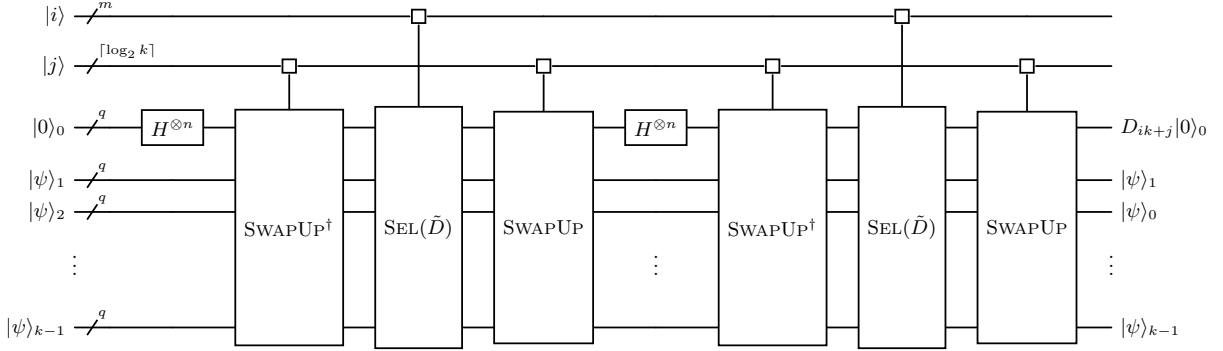
\begin{figure}[!htbp]
    \centering
\begin{adjustbox}{width=16cm}
\begin{quantikz}[wire types={q,q,q,q,q,n,q}]                    
\lstick{$|i\rangle$}       &  \qwbundle{m}                &                       &                        &\mctrl{2}                  &                       &                       &                        &\mctrl{2}       &                      & \\
\lstick{$|j\rangle$}       &  \qwbundle{\lceil\log_2 k \rceil}              &                       &\mctrl{1}               &                           &\mctrl{1}              &                       &\mctrl{1}               &               &\mctrl{1}              & \\
\lstick{$|0\rangle_0$}       &\qwbundle{q}    &\gate{H^{\otimes n}}   &\gate[5]{\swup^\dagger} &\gate[5]{\sel(\tilde{D})}  &\gate[5]{\swup}        &\gate{H^{\otimes n}}   &\gate[5]{\swup^\dagger} &\gate[5]{\sel(\tilde{D})}  &\gate[5]{\swup}& \rstick{$D_{ik+j}|0\rangle_0$}   \\
\lstick{$|\psi\rangle_1$}  &\qwbundle{q}    &                       &                        &                           &                       &                       &                        &               &                      &\rstick{$|\psi\rangle_1$} \\
\lstick{$|\psi\rangle_2$}  &\qwbundle{q}    &                       &                        &                           &                       &                       &                        &               &                      &\rstick{$|\psi\rangle_0$} \\
\vdots                     &                &                       &                        &                           &                       &\vdots                 &                        &               &                      &\vdots \\
\lstick{$|\psi\rangle_{k-1}$}  &\qwbundle{q}    &                       &                        &                           &                       &                       &                        &               &                      &\rstick{$|\psi\rangle_{k-1}$} \\
\end{quantikz}
\end{adjustbox}
\caption{Dirty $\selswap$ operation. There are a single clean register in the $|0\rangle$ state, $k-1$ dirty target registers, and $q$ is the number of qubits in each register. $\sel(\tilde{D})$ is the combined product of data boxes $\tilde{D}_i = \otimes^k_j D_{ik+j}$. The $\swup$ operation swaps the $j^\text{th}$ element of the set of $k$ indexed with $i$ to the top register and results in $D_{ik+j}|0\rangle$.}
\end{figure}

Depending on the number of dirty qubits available, the gate complexity can be minimised by choosing the correct $I$ and $k$ values. The cost of the uncompute is $2\lceil I/k \rceil + 4k$ Toffolis, which uses an advanced measurement-based uncomputation procedure, compiled with phase fix-up operations (see Appendix C of \cite{Berry2019qubitizationof} for more details). 
The scaling of the data lookup methods is summarised in Table~\ref{tab:table_lookup_methods}.

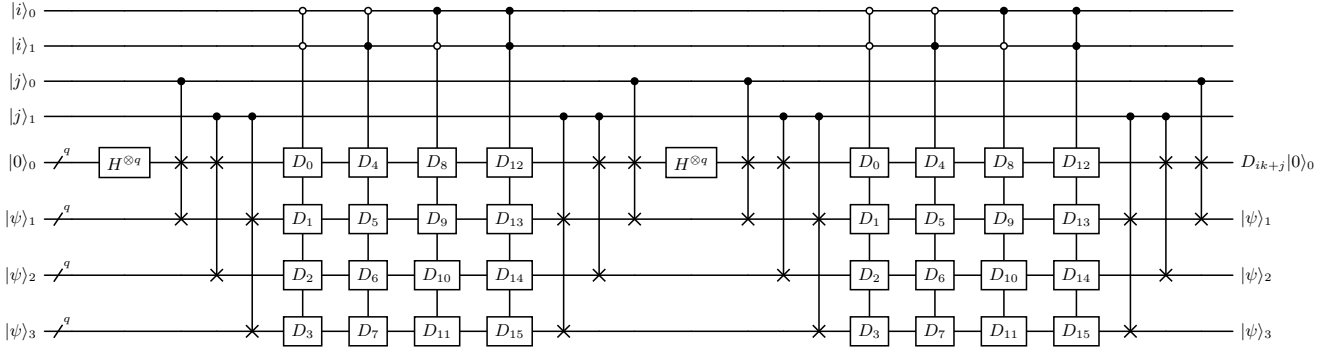
\begin{figure}[htbp!]
    \centering
    \begin{adjustbox}{width=\textwidth}
    \begin{quantikz}
    \lstick{$|i\rangle_0$}   &              &                      &           &         &         & \octrl{1} & \octrl{1}      & \ctrl{1}          & \ctrl{1}          &         &         &           &                          &           &         &         & \octrl{1} & \octrl{1}      & \ctrl{1}          & \ctrl{1}          &         &         &           & \\
    \lstick{$|i\rangle_1$}   &              &                      &           &         &         & \octrl{6} & \ctrl{6}       & \octrl{6}         & \ctrl{6}          &         &         &           &                          &           &         &         & \octrl{6} & \ctrl{6}       & \octrl{6}         & \ctrl{6}          &         &         &           & \\
    \lstick{$|j\rangle_0$}   &              &                      &\ctrl{2}   &         &         &           &                &                   &                   &         &         &\ctrl{2}   &                          &\ctrl{2}   &         &         &           &                &                   &                   &         &         &\ctrl{2}   & \\
    \lstick{$|j\rangle_1$}   &              &                      &           &\ctrl{1} &\ctrl{2} &           &                &                   &                   &\ctrl{2} &\ctrl{1} &           &                          &           &\ctrl{1} &\ctrl{2} &           &                &                   &                   &\ctrl{2} &\ctrl{1} &           & \\
    \lstick{$|0\rangle_0$}     &\qwbundle{q}  &\gate{H^{\otimes q}}  &\swap{1}   &\swap{2} &         & \gate{D_0}&  \gate{D_4}    &  \gate{D_8}       &  \gate{D_{12}}    &         &\swap{2} &\swap{1}   &\gate{H^{\otimes q}}      &\swap{1}   &\swap{2} &         & \gate{D_0}&  \gate{D_4}    &  \gate{D_8}       &  \gate{D_{12}}    &         &\swap{2} &\swap{1}   & \rstick{$D_{ik+j}|0\rangle_0$}  \\
    \lstick{$|\psi\rangle_1$}&\qwbundle{q}  &                      &\targX{}   &         &\swap{2} & \gate{D_1}&  \gate{D_5}    &  \gate{D_{9}}     &  \gate{D_{13}}    &\swap{2} &         &\targX{}   &                          &\targX{}   &         &\swap{2} & \gate{D_1}&  \gate{D_5}    &  \gate{D_{9}}     &  \gate{D_{13}}    &\swap{2} &         &\targX{}   &\rstick{$|\psi\rangle_1$}  \\
    \lstick{$|\psi\rangle_2$}&\qwbundle{q}  &                      &           &\targX{} &         & \gate{D_2}&  \gate{D_6}    &  \gate{D_{10}}    &  \gate{D_{14}}    &         &\targX{} &           &                          &           &\targX{} &         & \gate{D_2}&  \gate{D_6}    &  \gate{D_{10}}    &  \gate{D_{14}}    &         &\targX{} &           &  \rstick{$|\psi\rangle_2$} \\
    \lstick{$|\psi\rangle_3$}&\qwbundle{q}  &                      &           &         &\targX{} & \gate{D_3}&  \gate{D_7}    &  \gate{D_{11}}    &  \gate{D_{15}}    &\targX{} &         &           &                          &           &         &\targX{} & \gate{D_3}&  \gate{D_7}    &  \gate{D_{11}}    &  \gate{D_{15}}    &\targX{} &         &           & \rstick{$|\psi\rangle_3$} \\
    \end{quantikz}
    \end{adjustbox}
    
    \caption{Example of dirty $\selswap$ operation with $I=16$ elements and $k=4$ registers. The $D_i$ are the data boxes for each index $i$. It can be seen that each $i$ index has 4 data boxes, one for each register. The $j$ index then performs the $\swup$ operation to the top qubit to obtain $D_{ik+j}|0\rangle$. }
    \label{fgr:dirty_qubit_selswap}
\end{figure}

\begin{table}[!htbp]
    \centering
    \begin{ruledtabular}
    \begin{tabular}{l r r r}
    Compilation & Toffoli Count & Clean Qubits & Dirty Qubits \\
    \colrule
        Unary Iteration~\cite{Babbush2018LinearT}& $I$ & $2 \lceil \log_2(I) \rceil + q - 1$ & 0\\
        Clean $\selswap$ (Compute)~\cite{Low2024tradingtgatesdirty} & $\lceil I / k \rceil + q(k - 1)$  & $\lceil \log_2(I) \rceil  + \lceil \log_2(\lceil I / k \rceil) \rceil  + kq -1$& 0\\
        Clean $\selswap$ (Uncompute)~~\cite{Berry2019qubitizationof}  &$\lceil I/k\rceil + k$  & & \\
        Dirty $\selswap$ (Compute)~\cite{Berry2019qubitizationof} & $2\lceil I/k\rceil  + 4q(k - 1)$  & $\lceil \log_2(I) \rceil +  \lceil \log_2(\lceil I / k \rceil) \rceil + q -1$& $(k -1)q$\\
        Dirty $\selswap$ (Uncompute)~\cite{Berry2019qubitizationof} & $2\lceil I/k \rceil + 4k$  &  & \\
    \end{tabular}
    \end{ruledtabular}
        \caption{Resource estimation for various data lookup methods. Here, $I$ denotes the number of elements in the data set, $q$ is the number of qubits per register, and $k$ is the number of registers. The table reports the counts for Toffoli gates (with measurement-based uncomputation), where the $T$ count is approximately 4 times the Toffoli count, as well as the numbers of clean and dirty qubits required for each operation. The factor of $ \lceil \log_2(\lceil I / k \rceil) \rceil -1$ comes from the anncllas in the unary iteration over the smaller set of $\lceil I/k \rceil$ elements.}
    \label{tab:table_lookup_methods}
    \end{table}

Further to the tabulated results using the construction presented in Appendix C of \cite{Berry2019qubitizationof}, the clean $\selswap$ uncompute can be done with $\lceil I/k\rceil + k$ Toffoli and dirty $\selswap$ $2\lceil I/k \rceil + 4k$ using measurement-based uncomputation and phase correction. There is no $q$ dependence on the uncompute, which leads to a significant reduction in Toffoli complexity.

\subsection{Adder}
\label{sec:adder}

Adder circuits enable the bitwise addition of one register to another. If a carry bit is included, overflow becomes possible. However, for the purpose of this work modular addition is considered, as it is used in conjunction with a phase gradient to implement rotations. Specifically, a doubly-controlled adder is required, and its construction is shown in Figure~\ref{fig:gidney_adder_2controlled}. Only in-place adders are considered, where one of the input registers is overwritten with the sum. The basic concept is illustrated in Figure~\ref{fig:intro_adder}. The adder circuit takes two registers $|a\rangle$ and $|b\rangle$ and adds them together, producing a new register $|b + a\rangle$. The $|a\rangle$ register remains unchanged.

\begin{figure}[!htbp]
    \centering
\begin{quantikz}
    \lstick{$|a\rangle$} & \gate[2]{\adder(a + b)}   & \rstick{$|a\rangle$} \\
    \lstick{$|b\rangle$} &                          & \rstick{$|b+a\rangle$}\\
\end{quantikz}
    \caption{An in-place adder dircuit takes the $|a\rangle$ and adds $b$ to it. The $|a\rangle$ register is left unchanged.}
    \label{fig:intro_adder}
\end{figure}
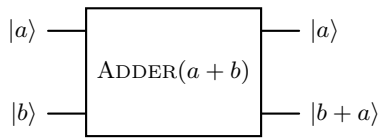

The Gidney adder is the most efficient way to perform a controlled addition of two qubit registers. This is a measurement-based adder that scales with a cost of $q$ Toffolis, where $q$ is the number of qubits used to store the integers $a$ and $b$. It is based on the Cuccaro adder~\cite{cuccaro2004newquantumripplecarryaddition}, but uses measurement-based uncomputation to remove the cost of the uncompute Toffoli. A carry-out bit can be incorporated to allow for integer overflow; however, here we use a modular adder. The Gidney adder is shown in Figure~\ref{fig:gidney_adder}. It is defined recursively as using a carry-in bit $c_{k}$ and a carry-out bit $c_{k+1}$, which is the input to the next step of the iterations. The single controlled and the double controlled adders are shown in Figures~\ref{fig:gidney_adder_1controlled} and \ref{fig:gidney_adder_2controlled} respectively, which when using measurement-based uncomputation and dynamic qubit allocation scale as $2n$ and $3n$ respectively.

\begin{figure}[!htbp]
    \centering
    \begin{quantikz}
    \lstick{$c_k$}     & \ctrl{2}  &           & \ctrl{3}  &                 &          &                 & \ctrl{3}  &           & \ctrl{1} &         &\rstick{$c_k$} \\
    \lstick{$a_k$}     & \targ{}   &\ctrl{1}   &           &                 &          &                 &           & \ctrl{1}  & \targ{}  &\ctrl{1} &\rstick{$a_k$}   \\ 
    \lstick{$b_k$}     & \targ{}   &\ctrl{1}   &           &                 &          &                 &           & \ctrl{1}  &          &\targ{}  &\rstick{$b_k + a_k$}   \\ 
    \nw                & \nw       &\nw        & \targ{}   & \push{c_{k+1}}  & \cdots   &\push{c_{k+1}}   & \targ{}   &           & \nw      &\nw      &\nw               \\
    \end{quantikz}
    \caption{$q$ qubit in-place Gidney adder circuit bitwise component. The $c_k$ qubit is the carry-in bit, the qubit $a_k$ is added to the $b_k$ qubit, forming the sum bit $b_k + a_k$ where the bit overflow goes to $c_{k+1}$. The process then repeats onto the next bit. If a final carry bit is not provided, the adder is modular. When measurement-based uncomputation is used only the compute Toffoli needs to be counted, giving a Toffoli cost of $q$, with a single dynamically allocated qubit.}
    \label{fig:gidney_adder}
\end{figure}
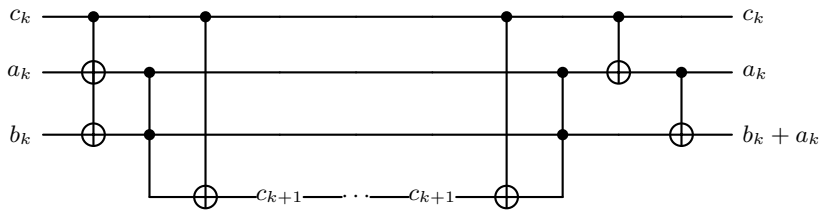

\begin{figure}[!htbp]
    \centering
    \begin{quantikz}
    \lstick{$z$}       &           &           &           &                 &          &                 &           &           &\ctrl{2}  &        &\ctrl{2}  &          &         &\rstick{$z$} \\
    \lstick{$c_k$}     & \ctrl{3}  &           & \ctrl{4}  &                 &          &                 & \ctrl{4}  &           &          &        &          & \ctrl{3} &         &\rstick{$c_k$} \\
    \lstick{$a_k$}     & \targ{}   &\ctrl{2}   &           &                 &          &                 &           & \ctrl{2}  &\ctrl{1}  &        &\ctrl{1}  & \targ{}  &         &\rstick{$a_k$}   \\ 
    \nw                &\nw        &\nw        &\nw        &\nw              &\nw       &\nw              &\nw        & \nw       &\nw       &\ctrl{1}&          & \nw      & \nw     &\nw             \\ 
    \lstick{$b_k$}     & \targ{}   &\ctrl{1}   &           &                 &          &                 &           & \ctrl{1}  &          &\targ{} &          & \targ{}  &         &\rstick{$b_k + a_k\cdot z$}   \\ 
    \nw                & \nw       &\nw        & \targ{}   & \push{c_{k+1}}  & \cdots   &\push{c_{k+1}}   & \targ{}   &           &\nw       &\nw     &\nw       & \nw      &\nw      &\nw               \\
    \end{quantikz}
     \caption{Controlled $q$ qubit in-place Gidney adder circuit bitwise component. The ancilla is $z$, the $c_k$ qubit is the carry-in bit, the qubit $a_k$ is added to the $b_k$ qubit, forming the sum bit $b_k + a_k$ where the bit overflow goes to $c_{k+1}$. The process then repeats onto the next bit. If a final carry bit is not provided, the adder is modular. When measurement-based uncomputation is used only the compute Toffoli needs to be counted, giving a Toffoli cost of $2n$, with two dynamically allocated qubits.}
    \label{fig:gidney_adder_1controlled}
\end{figure}

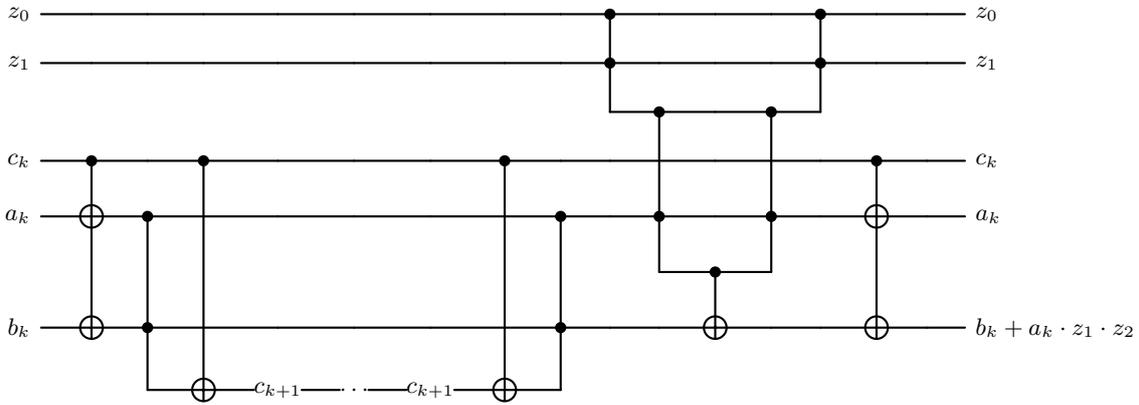
\begin{figure}[!htbp]
    \centering
    \begin{quantikz}
    \lstick{$z_0$}       &           &           &           &                 &          &                 &           &           &\ctrl{1} &          &        &         &\ctrl{1}  &          &         &\rstick{$z_0$} \\
    \lstick{$z_1$}       &           &           &           &                 &          &                 &           &           &\ctrl{1} &          &        &         &\ctrl{1}  &          &         &\rstick{$z_1$} \\
    \nw                &\nw        &\nw        &\nw        &\nw              &\nw       &\nw              &\nw        & \nw       &\nw      &\ctrl{2}  &        &\ctrl{2} &          & \nw      & \nw     &\nw             \\ 
    \lstick{$c_k$}     & \ctrl{3}  &           & \ctrl{4}  &                 &          &                 & \ctrl{4}  &           &         &          &        &         &          & \ctrl{3} &         &\rstick{$c_k$} \\
    \lstick{$a_k$}     & \targ{}   &\ctrl{2}   &           &                 &          &                 &           & \ctrl{2}  &         &\ctrl{1}  &        &\ctrl{1} &          & \targ{}  &         &\rstick{$a_k$}   \\ 
    \nw                &\nw        &\nw        &\nw        &\nw              &\nw       &\nw              &\nw        & \nw       &\nw      &\nw       &\ctrl{1}&         &\nw       &\nw       & \nw     &\nw             \\ 
    \lstick{$b_k$}     & \targ{}   &\ctrl{1}   &           &                 &          &                 &           & \ctrl{1}  &         &          &\targ{} &         &          & \targ{}  &         &\rstick{$b_k + a_k \cdot z_1 \cdot z_2$}   \\ 
    \nw                & \nw       &\nw        & \targ{}   & \push{c_{k+1}}  & \cdots   &\push{c_{k+1}}   & \targ{}   &           &\nw       &\nw       &\nw     &\nw      &\nw       & \nw      &\nw      &\nw               \\
    \end{quantikz}
    \caption{Doubly-controlled $q$-qubit in-place Gidney adder circuit bitwise component. The ancillas are $z_1, z_2$, the $c_k$ qubit is the carry-in bit, the qubit $a_k$ is added to the $b_k$ qubit, forming the sum bit $b_k + a_k$ where the bit overflow goes to $c_{k+1}$. The process then repeats onto the next bit. If a final carry bit is not provided, the adder is modular. When measurement-based uncomputation is used only the compute Toffoli needs to be counted, giving a Toffoli cost of $3n$, with three dynamically allocated qubits.}
    \label{fig:gidney_adder_2controlled}
\end{figure}

\subsection{Multi Index Data \texorpdfstring{$\select$}{SELECT}}
\label{sec:multi_index_data_select}

This section follows the construction in Appendix VII.B of \cite{MSFTQROMRotations}. Naively, for a data lookup over a qubit register $i$ over $I$ elements and qubit register $j$ over $J$ elements one could combine then into a single index and perform unary iteration over $IJ$ elements. This is illustrated in Figure~\ref{fig:double_index_unary}.

\begin{figure}[!htbp]
    \centering
\begin{quantikz}
    \lstick{$|i\rangle$}    &\qwbundle{\lceil \log_2 I \rceil} & \mctrl{1}             & \\
    \lstick{$|j\rangle$}    &\qwbundle{\lceil \log_2 J \rceil} & \mctrl{1}             & \\
    \lstick{$|0\rangle$} &\qwbundle{q} & \gate{\data(D)}   & \rstick{$D_{ij}|0\rangle$} \\
 \end{quantikz}
    \caption{Circuit for iterating over $IJ$ elements for a data lookup over all $K$ and $I$ elements. The data lookup is performed over the combined index $i,j$ and the data boxes are applied to the $|0\rangle$ register.}
    \label{fig:double_index_unary}
\end{figure}
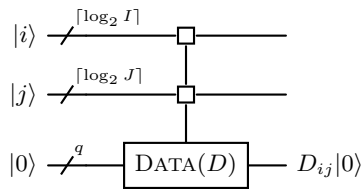

In the case of the quantum Paldus transform we do not iterate over all index pairs, as there are only certain allowed values for $S,M$. Consider the case where for each $i$ where only $J_i \le J$ elements are allowed. This means that the data lookup is defined over $L = \sum^{J}_i I_i \le I J  $ elements, hence the circuit shown in Figure~\ref{fig:double_index_unary} is over inefficiently many $IJ$ Toffolis. We therefore seek a circuit which will only encode valid $L$ elements. The double index $(i,j)$ must be mapped to $l = j + \sum_{\alpha \in [i-1]} I_\alpha$, where $ \sum_{\alpha \in [i-1]} I_\alpha = L_j$, which is stored as a binary $\lceil \log_2 L \rceil$ bit integer controlled on the $|i\rangle$ register. The set $\{L_j\}$ is calculated in the pre-computation step and stored in a data lookup. The circuit is shown in Figure~\ref{fig:double_index_unary}. 

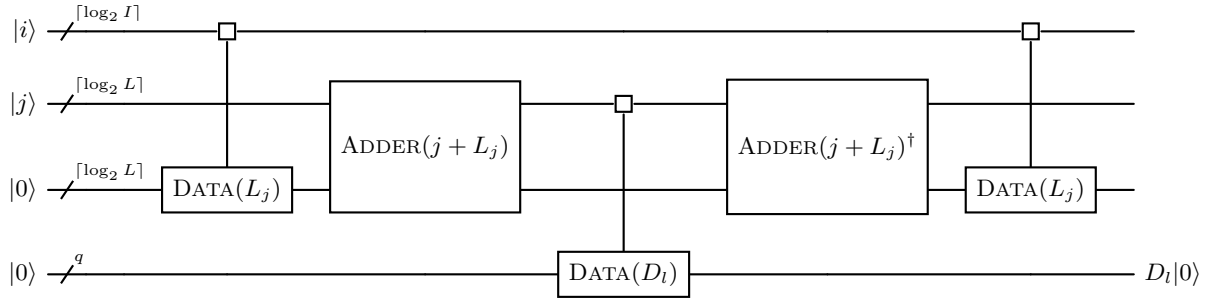
\begin{figure}[!htbp]
    \centering
\begin{quantikz}
    \lstick{$|i\rangle$}    &\qwbundle{\lceil \log_2 I \rceil} & &  \mctrl{2}       &                          &                        &                           &\mctrl{2}              & \\
    \lstick{$|j\rangle$}    &\qwbundle{\lceil \log_2 L \rceil} & &                  &\gate[2]{\adder(j + L_j)} &\mctrl{2}               &\gate[2]{\adder(j + L_j)^\dagger}   &                       & \\
    \lstick{$|0\rangle$}    &\qwbundle{\lceil \log_2 L \rceil} & & \gate{\data(L_j)} &                          &                        &                           &\gate{\data(L_j)}   & \\
    \lstick{$|0\rangle$}    &\qwbundle{q}                    & &                  &                          &\gate{\data(D_l)}     &                           &                       & \rstick{$D_{l}|0\rangle$}\\
 \end{quantikz}
    \caption{Circuit for iterating over sparse multicontrolled $data$ operation over double index set $(i,j) \in IJ$, where only $L$ defined elements are allowed, which runs over the $q$ index where $L < IJ$. $L_j = \sum_{\alpha \in [i-1]} I_\alpha$ encodes the shift along a flattened array indexed by $l$ for each set of allowed $L_j$ elements. The adder implements this index shift to the new index $l$. The adder and the $\data(L_j)$ operation must be uncomputed to return the index registers to their original state.}
    \label{fig:double_index_data_lookup}
\end{figure}

There is clearly lots of scope to choose how the controlled data boxes are implemented from Table~\ref{tab:table_lookup_methods}. The $\data(Q_j)$ operation is done with clean qubit $\selswap$ and the $\sel(U)$ operation is done with dirty qubit $\selswap$. The $\adder$ operation is a Gidney adder \cite{Gidney2018halvingcostof} at the cost of $\mathcal{O}(\lceil \log_2 L \rceil)$ Toffolis. 

It is assumed that $\lceil \log_2 L \rceil > \log_2 J$ and the cost of the adder is $\lceil \log_2 L \rceil$ Toffolis (with $4\lceil \log_2 L \rceil$ $T$ gates) using the Gidney adder, which uses measurement-based uncomputation~\cite{Gidney2018halvingcostof} (The $j$ register can be padded with extra bits). The overall Toffoli cost of the double index data lookup is dominated by 4 components: 2 adders, the compute and uncompute of $\data(Q_j)$ using either a clean $\selswap$ or unary iteration, and the $\sel(U)$ dirty $\selswap$. These components are summarised in Table~\ref{tab:double_index_cost}. Typically, the term involving $L$ is dominant in the $\data(D_l)$ operation that runs over $l$, however, due to the extra qubit register needed, this should only be used for problems where $L \ll IJ$.

\begin{table}[!htbp]
    \centering
    \begin{ruledtabular}
    \begin{tabular}{l r}
        \textbf{Component} & \textbf{Toffoli Cost} \\
        \hline
        2 adders & $2\lceil \log_2 L \rceil$ \\
        Compute $\data(L_j)$ (clean \selswap) & $\lceil I/k_c \rceil+ \lceil \log_2 L \rceil(k_c-1)$ \\
        Uncompute $\data(L_j)$ & $\left\lceil I/k_c \right\rceil + k_c$ \\
        Compute/Uncompute $\data(L_j)$ (unary iteration) & $I$ \\
        Compute $\data(D_l)$ (dirty \selswap) & $2 \lceil L/k_d \rceil + 4q(k_d-1)$ \\
        \hline
        \textbf{Total}$^{*}$ & $2\lceil \log_2 L \rceil + 2\lceil L/k_d \rceil + 4q(k_d-1) + 2\lceil I/k_c \rceil + \lceil \log_2 L \rceil(k_c-1) + k_c$ \\
        \textbf{Total}$^{**}$ & $2\lceil \log_2 L \rceil + 2\lceil L/k_d \rceil + 4q(k_d-1) + 2I$ \\
    \end{tabular}
    \end{ruledtabular}
    \caption{Summary of Toffoli costs for each component in a double index data lookup indexed by $i \in \{I\}$ and $j \in \{J\}$, with $L < IJ$ their overall set. $k_c$ and $k_d$ are the number of clean and dirty qubit registers respectively (which must be a power of $2$). $I$ is the number of elements in the larger index data set, $\lceil \log_2 L \rceil$ is the number of qubits encoding the data lookup $L_j$, which is the cost of the adder, and $q$ is the number qubits in the data lookup over the flattened indices $L$. $^*$ denotes the compilation using clean $\data(L_i)$ $\selswap$ and dirty $\data(D_l)$ $\selswap$. $^{**}$ denotes the compilation using $\data(L_j)$ unary iteration and dirty $\data(D_l)$ $\selswap$. The $\data(L_j)$ only carries a compute step.}
    \label{tab:double_index_cost}
\end{table}

\subsection{Rotations}
\label{sec:rotations}

In the quantum Paldus transform one of the essential primitives is a multiplexed sequence of controlled $R_y$ rotations. This can be represented as a $\select$ operation:

\[
\select(R_y) = \sum^{I}_i |i\rangle \langle i| \otimes R_y(\theta_i)
\] 

In a fault-tolerant setting, rotations must be discretised because continuous angles are not directly accessible. Therefore, it is beneficial to store the angles as fixed-point binary fractions using a data lookup object, which then enables the parallelisation of rotations via controlled increments of the rotation angle. Following the presented implementation in Appendix VII.B of \cite{MSFTQROMRotations}, the lookup stores the angles as a fixed-point binary fraction with $q$ bits of precision, i.e. $[\theta_0, \theta_1, ..., \theta_{q-1}]$. The rotation is then approximated by a series of controlled increments, where each controlled rotation increment activates according to the $q^\text{th}$ binary fraction in the lookup table. This is illustrated in Figure~\ref{fig:QROM_rotations}.
\begin{figure}[!htbp]
    \centering
    \includegraphics[width=\textwidth]{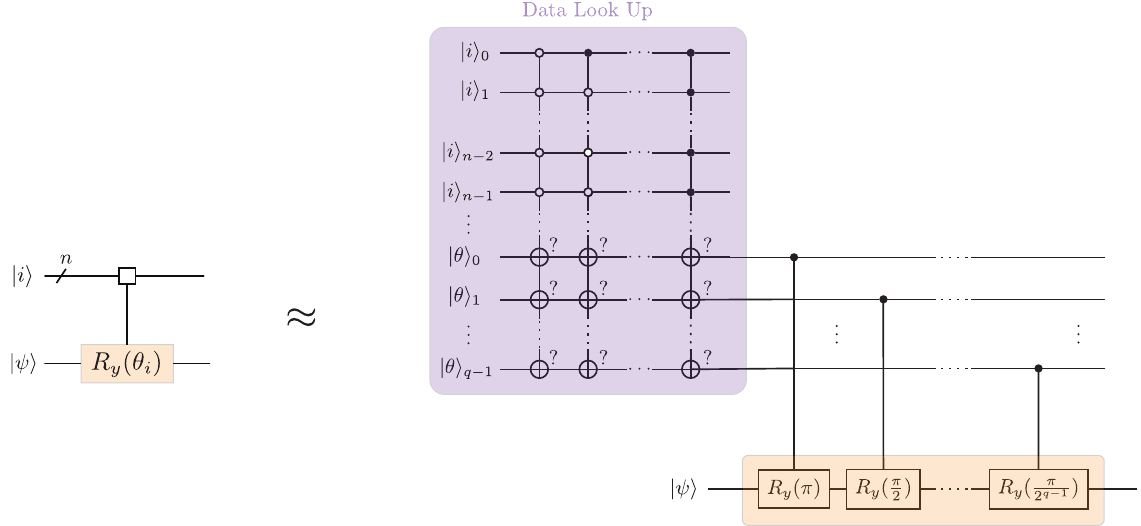}
    \caption{Data Lookup rotations. The angles are stored as fixed-point binary fractions with $q$ bits of precision. The rotation is approximated by a series of controlled increments, where each controlled rotation increment activates according to the $q^\text{th}$ binary fraction in the lookup table. Each $\theta$ qubit in the data lookup is initialised to $|0\rangle$.}
    \label{fig:QROM_rotations}
\end{figure}

An efficient way to calculate the multicontrolled rotations is with the phase gradient method~\cite{Gidney2018halvingcostof}. The phase gradient method is a technique for implementing controlled rotations using a combination of controlled adder and phase gradient registers. The basic idea is to use a controlled adder to add the fixed-point binary fraction of the angle of rotation to a phase gradient register, which is then used to implement the rotation on a target qubit. The phase gradient applies an incremented phase on each bitstring of a register of $q$ qubits:

\begin{equation}
    \mathcal{F} = \frac{1}{\sqrt{2^q}} \sum_{k=0}^{2^q-1} e^{-2\pi i k / 2^q} \ket{k}.
\end{equation}

This can be implemented very efficiently as a product of $R_z$ rotations on each qubit at the cost of $\mathcal{O}(q\log(1/\epsilon))$ $T$ gates~\cite{Gidney2018halvingcostof}. The circuit for this shown in Figure~\ref{fig:phase_gradient}. 

\begin{figure}[!htbp]
    \centering
\begin{quantikz}
    \lstick{$|0\rangle$} & \gate{\Delta} & \rstick{$|\mathcal{F}\rangle$} \\
\end{quantikz}
= 
\begin{quantikz}[wire types={q,q,n,q,q}]
    \lstick{$|0\rangle_0$} & \gate{R_z(\pi/2^0)} &  \\
    \lstick{$|0\rangle_1$} & \gate{R_z(\pi/2^1)} &  \\
    \vdots & \vdots & \vdots \\
    \lstick{$|0\rangle_{d-2}$} & \gate{R_z(\pi/2^{d-2})} &  \\
    \lstick{$|0\rangle_{d-1}$} & \gate{R_z(\pi/2^{d-1})} &  \\
\end{quantikz}
    \caption{\textit{Phase Gradient Circuit}. The $R_z$ gates are applied to each qubit in the register. }
    \label{fig:phase_gradient}
\end{figure}
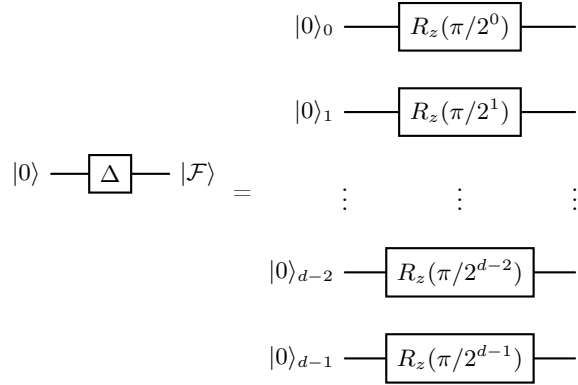

\noindent Using a an adder circuit at the cost of $\mathcal{O}(q)$ Toffolis \cite{Gidney2018halvingcostof} performs modular addition on the $q$ qubit register.

\begin{equation}
    \adder\ket{x} \ket{y} = \ket{x} \ket{y + x \mod 2^q}
\end{equation}

\noindent The adder circuit can be applied to the phase gradient register to obtain the following transformation:

\begin{equation}
    \begin{split}
    \adder |x\rangle\ket{\mathcal{F}} &= |x\rangle \frac{1}{\sqrt{2^d}} \sum_{k=0}^{2^d-1} e^{-2\pi i k / 2^d} \ket{k+x \mod 2^d} \\
    &=  e^{2\pi i x} |x\rangle \frac{1}{\sqrt{2^d}} \sum_{k=0}^{2^d-1} e^{-2\pi i (k +x) / 2^d} \ket{k+x \mod 2^d} \\
    &= e^{2\pi i x / 2^d} \ket{x} \ket{\mathcal{F}}.
    \end{split}
\end{equation}

Where the constant phase $e^{2\pi i x}$ is factored out. The adder circuit can be used to implement the rotation $R_y(\theta)$ as a controlled operation on the phase gradient register, the addition register $|x\rangle$ and the control register $|z\rangle$. The controlled adder circuit is shown in Figure~\ref{fig:addder_qrom_rotations}.

\begin{equation}
    C\adder \ket{x} \ket{\mathcal{F}} \ket{z} = \ket{x} \ket{\mathcal{F}} e^{2\pi i x / 2^d Z} \ket{z}.
\end{equation}

In this case, the controlled adder circuit is applied to the phase gradient register and the $|z\rangle$ register, which has an $R_z(x/2^d)$ rotation applied to it. Basis changes can be also applied to the target register to obtain $R_y$ or $R_x$ rotations. This strategy can now be applied in conjunction with the data lookup methods shown in Section~\ref{sec:data_lookup} to do addititon in superposition, enabling the application of multiple rotations.

\begin{figure}[!htbp]
    \centering
\begin{quantikz}
    \lstick{$|i\rangle$} & \mctrl{1} & \rstick{$|i\rangle$} \\
    \lstick{$|\psi\rangle$} & \gate{R_y(\theta_i)} & \rstick{$ R_y(\theta_i)|\psi\rangle$} \\
\end{quantikz}
=
\begin{quantikz}
    \lstick{$|i\rangle$}            &\qwbundle{\lceil \log_2 I \rceil}  & \mctrl{1}             &           &                                       &           & \mctrl{1}             & \rstick{$|i\rangle$}   \\
    \lstick{$|0\rangle$}            &\qwbundle{q}                       & \gate{\data(\theta)}  &           &\gate[2]{\adder(\theta + \mathcal{F})} &           & \gate{\data(\theta)}  & \rstick{$|0\rangle$}   \\
    \lstick{$|\mathcal{F}\rangle$}  &\qwbundle{q}                       &                       &           &                                       &           &                       & \rstick{$|\mathcal{F}\rangle$}  \\
    \lstick{$|\psi\rangle $}        &                                   &                       & \gate{H}  &\ctrl{-1}                              &\gate{H}   &                       & \rstick{$R_y(\theta_i)|\psi\rangle $}\\
 \end{quantikz}
    \caption{Circuit for the controlled adder. The $|i\rangle$ register is the index register, the $|\mathcal{F}\rangle$ register is the phase gradient register and the $|\psi\rangle$ register is the target register. The data lookup is performed over the $\theta$ angles.}
    \label{fig:addder_qrom_rotations}
\end{figure}
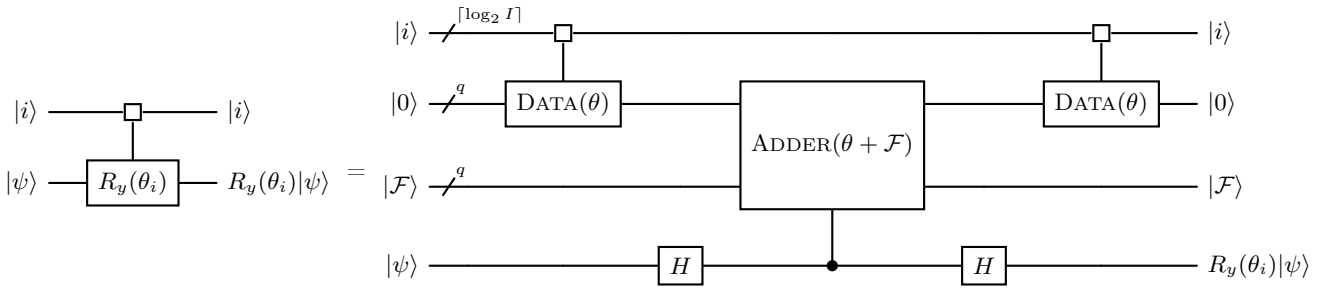

The Toffoli cost is calculated by assuming the phase gradient $\mathcal{F}$ as a resource of constant $\mathcal{O}(q\log(1/\epsilon))$ $T$ gate cost which is ignored. The main cost is the computation and uncomputation of the data lookup and the controlled adder. From Ref~\cite{Gidney2018halvingcostof} the cost of the controlled adder is $2q$ Toffolis (with measurement-based uncomputation). The number of data qubits $q$ is directly related to the error by $\log_2(1/\epsilon)$ (the precision of encoding the rotation angle). Assuming that the Fourier state is prepared perfectly, this approximates \( U \) with error
\[
\left\| U - U' \right\| \leq 
\left\| \sum_{i=0}^{I} |i\rangle\langle i| \otimes \left( e^{i 2\pi \theta_i Z} - e^{i 2\pi \theta'_i / 2^q Z} \right) \right\| 
\leq \max_{|y| < 1/2^q} \left| e^{i 2\pi y} - 1 \right| < \frac{2\pi}{2^q}.
\]

Where $\theta'_i$ is the binary approximation of the continuous angle. The cost of this operation for some fixed error $\epsilon$ resulting in $q$ data qubits and $I$ element is shown in Table~\ref{tab:rotation_cost}.

\begin{table}[!htbp]
    \centering
    \begin{ruledtabular}
    \begin{tabular}{l r r r}
    Compilation & Toffoli Count & Clean Qubits & Dirty Qubits \\
    \colrule
        Unary Iteration~\cite{Babbush2018LinearT}                                   & $2I + 2q$                                        & $2\lceil \log_2(I) \rceil + 2q$           & 0 \\
        Clean $\selswap$~\cite{Low2024tradingtgatesdirty,Berry2019qubitizationof}   & $ 2\lceil I/k\rceil + q(k - 1)+ k + 2q$          & $2\lceil \log_2(\lceil I\rceil) \rceil + \lceil \log_2(\lceil I/k\rceil) + \rceil + (k+2)q $   & 0\\        
        Dirty $\selswap$~\cite{Berry2019qubitizationof}                             & $2\lceil I/k\rceil+ 4q(k - 1) + 4k + 4q$         & $2\lceil \log_2(\lceil I\rceil) \rceil + \lceil \log_2(\lceil I/k\rceil) + 3q $          & $(k -1)q$\\
    \end{tabular}
    \end{ruledtabular}
        \caption{Resource estimation for data lookup rotation methods. $I$ is the number of controlled rotations, $q$ is the number of qubits which is equal to $\log_2(1/\epsilon)$, $k$ is the number of extra data lookup registers used in the $\selswap$. The Toffoli count is the number of Toffoli gates required for the operation. $T$ counts can be approximated from this as $4 \times$ Toffoli gates. These resource estimates include measurement-based uncomputation. The factor of $\lceil \log_2(\lceil I/k\rceil) -1$ comes from the unary iteration of a reduced set of elements. The adder contributes $3q$ clean qubits due to the extra $d$ ancilla cost in the Toffoli cascade.}
    \label{tab:rotation_cost}
    \end{table}

These results are used in slightly modified fashion in the main text for the quantum Pladus transform, which is presented using controlled Givens rotations.

\subsection{Incrementer}
\label{sec:incrementer}

In the quantum Paldus transform, the incrementer is the primitive that enables traversing along the GT state branching diagrams and spin coupling graphs. The $(N, S, M)$ update operations are all constructed from combinations of controlled incrementer circuits. The incrementer is a circuit that takes a qubit register $|y\rangle$ and increments it by one (modular), $|y\rangle \rightarrow |y+1\rangle$. This is done by applying a series of multi-controlled $X$ gates. A controlled incrementer is a circuit that takes a qubit register $|y\rangle$ and increments it by one, controlled on a qubit $|x\rangle$, resulting in the transformation $|x\rangle |y\rangle \rightarrow |x\rangle |y \oplus x\rangle$. The controlled incrementer is a primitive that enables graph traversal where the control qubit is the branch direction on the graph, as illustrated in Figure~\ref{fig:incrementer}. 
\begin{figure*}[htb!]
    \includegraphics[width=8cm]{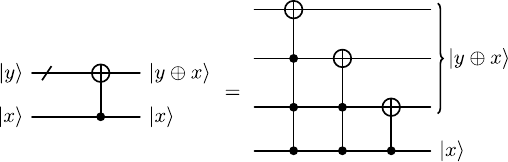}
    \caption{\textit{Controlled Incrementer Circuit} used a a component of the $N,S,M$ branching operations in the quantum Paldus transform.}
    \label{fig:incrementer}
  \end{figure*}

 An efficient method for compiling the $C_nX$ gates employs a technique known as conditionally clean qubits \cite{nie2024quantumcircuitmultiqubittoffoli}. This is demonstrated in Figure~\ref{fig:clean_mcx}.

\begin{figure}[!htbp]
    \centering
\begin{adjustbox}{width=\textwidth}
\begin{quantikz}
    \lstick{$c_0$}    & \ctrl{1}  & \\
    \lstick{$c_1$}    & \ctrl{1}  & \\
    \lstick{$c_2$}    & \ctrl{1}  & \\
    \lstick{$c_3$}    & \ctrl{1}  & \\
    \lstick{$c_4$}    & \ctrl{1}  & \\
    \lstick{$c_5$}    & \ctrl{1}  & \\
    \lstick{$c_6$}    & \ctrl{1}  & \\
    \lstick{$c_7$}    & \ctrl{1}  & \\
    \lstick{$|t\rangle$}    & \targ{}   & \\
\end{quantikz} 
=
\begin{quantikz}
    \lstick{$c_0$}    &\ctrl{1}   &       &           &       &           &       &         &         &         &\targ{}  &\targ{}    &\ctrl{9}   &\targ{}    &\targ{}    &           &           &          &        &           &           &           &           &\ctrl{1} &           \\
    \lstick{$c_1$}    &\ctrl{7}   &\targ{}&\targ{}    &       &           &       &         &         &         &         &\ctrl{-1}  &           &\ctrl{-1}  &           &           &           &          &        &           &           &\targ{}    &\targ{}    &\ctrl{7} &           \\
    \lstick{$c_2$}    &           &       &\ctrl{-1}  &       &           &       &         &\targ{}  &\targ{}  &         &\ctrl{-1}  &           &\ctrl{-1}  &           &\targ{}    &\targ{}    &          &        &           &           &\ctrl{-1}  &           &         &           \\
    \lstick{$c_3$}    &           &       &\ctrl{-1}  &\targ{}&\targ{}    &       &         &         &\ctrl{-1}&         &           &           &           &           &\ctrl{-1}  &           &          &        &\targ{}    &\targ{}    &\ctrl{-1}  &           &         &           \\
    \lstick{$c_4$}    &           &       &           &       &\ctrl{-1}  &       &         &         &         &         &           &           &           &           &           &           &          &        &\ctrl{-1}  &           &           &           &         &           \\
    \lstick{$c_5$}    &           &       &           &       &\ctrl{-1}  &\targ{}&\targ{}  &         &\ctrl{-2}&         &           &           &           &           &\ctrl{-2}  &           &\targ{}   &\targ{} &\ctrl{-1}  &           &           &           &         &           \\
    \lstick{$c_6$}    &           &       &           &       &           &       &\ctrl{-1}&         &         &         &           &           &           &           &           &           &\ctrl{-1} &        &           &           &           &           &         &           \\
    \lstick{$c_7$}    &           &       &           &       &           &       &\ctrl{-1}&         &         &         &           &           &           &           &           &           &\ctrl{-1} &        &           &           &           &           &         &           \\
 \lstick{$|0\rangle$}   &\targ{}    &       &           &       &           &       &         &         &         &         &           &\ctrl{1}   &           &           &           &           &          &        &           &           &           &           &\targ{}  &\meterD{0} \\
 \lstick{$|t\rangle$}   &           &       &           &       &           &       &         &         &         &         &           &\targ{}    &           &           &           &           &          &        &           &           &           &           &         &           \\
 \end{quantikz}
\end{adjustbox}
    \caption{$C_8 X$ gate using a single conditionally clean qubit. The compute stage is shown on the left and the uncompute stage is shown on the right. The $C_8 X$ gate is implemented using a single Toffoli, as the other qubits are conditioned to $|1\rangle$ on the clean qubit.}
    \label{fig:clean_mcx}
\end{figure}
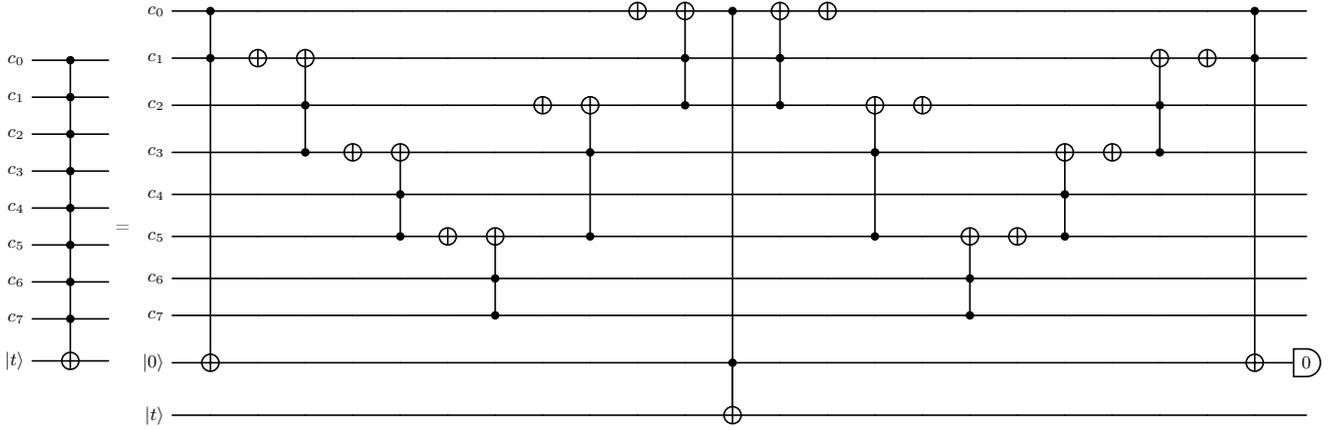
  
A leading implementation of the incrementer is given by Khattar and Gidney~\cite{khattar2024riseconditionallycleanancillae}. Their approach requires an extra $\log_2(n)$ conditionally clean qubits (max $5$) and scales with $3(n)$ Toffoli gates, as shown in Figure~\ref{fig:incrementer_circuit}, where $n$ denotes the number of target qubits. As a controlled incremented on $(n-1)$ elements has the same cost as a incremementer on $n$ elements (minus a single not gate). We assume the cost of the controlled $(n-1)$ qubit incrementer is the same as $n$ qubit incrementer (see Section \ref{sec:incrementer} for the $n$ incrementer circuit). 

\begin{figure}[htbp!]
    \centering
    \begin{adjustbox}{width=\textwidth}
    \begin{quantikz}
    \lstick{$q_0$} &  \ctrl{1}        & \ctrl{1}       & \ctrl{1}       &  \ctrl{1}      & \ctrl{1}       & \targ{}         &   \\
    \lstick{$q_1$} &  \ctrl{1}        & \ctrl{1}       & \ctrl{1}       &  \ctrl{1}      & \targ{}        &        &  \\
    \lstick{$q_2$} &  \ctrl{1}        & \ctrl{1}       & \ctrl{1}       &  \targ{}        &        &        &  \\
    \lstick{$q_3$} &  \ctrl{1}        & \ctrl{1}       & \targ{}       &        &        &        &  \\
    \lstick{$q_4$} &  \ctrl{1}        & \targ{}        &        &        &        &        &     \\
    \lstick{$q_5$} &  \targ{}        &                &        &        &        &        &    \\
    \nw &  \nw &  \nw &  \nw &  \nw &  \nw &  \nw &  \nw & \nw \\
    \end{quantikz} =
    \begin{quantikz}
    \lstick{$q_0$}          &\ctrl{1}   &\targ{}    &           &\targ{}    &\ctrl{5}   &\targ{}    &\targ{}    &           &           &           &\ctrl{1}   &\ctrl{1}   &\targ{}    & \\
    \lstick{$q_1$}          &\ctrl{5}   &\targ{}    &\targ{}    &\ctrl{-1}  &           &\ctrl{-1}  &\ctrl{3}   &\targ{}    &\targ{}    &           &\ctrl{5}   &\targ{}    &           & \\
    \lstick{$q_2$}          &           &           &\ctrl{-1}  &           &           &           &           &\ctrl{-1}  &\ctrl{1}   &\targ{}    &           &           &           & \\
    \lstick{$q_3$}          &           &           &\ctrl{-1}  &           &           &           &           &\ctrl{-1}  &\targ{}    &           &           &           &           & \\
    \lstick{$q_4$}          &           &           &           &\ctrl{-3}  &           &\ctrl{-3}  &\targ{}    &           &           &           &           &           &           & \\
    \lstick{$q_5$}          &           &           &           &           &\targ{}    &           &           &           &           &           &           &           &           & \\
    \lstick{$|0\rangle$}    &\targ{}    &           &           &           &\ctrl{-1}  &           &\ctrl{-2}  &           &\ctrl{-3}  &\ctrl{-4}  &\targ{}    &\meterD{0} &\nw        &\nw  \\
    \end{quantikz}
    \end{adjustbox}
    \caption{Six qubit conditionally clean incrementater circuit. The LHS of the circuit shows the accumulated conditonally clean qubits and the RHS shows the step-wise uncomputation of the conditionally clean qubits.}
    \label{fig:incrementer_circuit}
\end{figure}
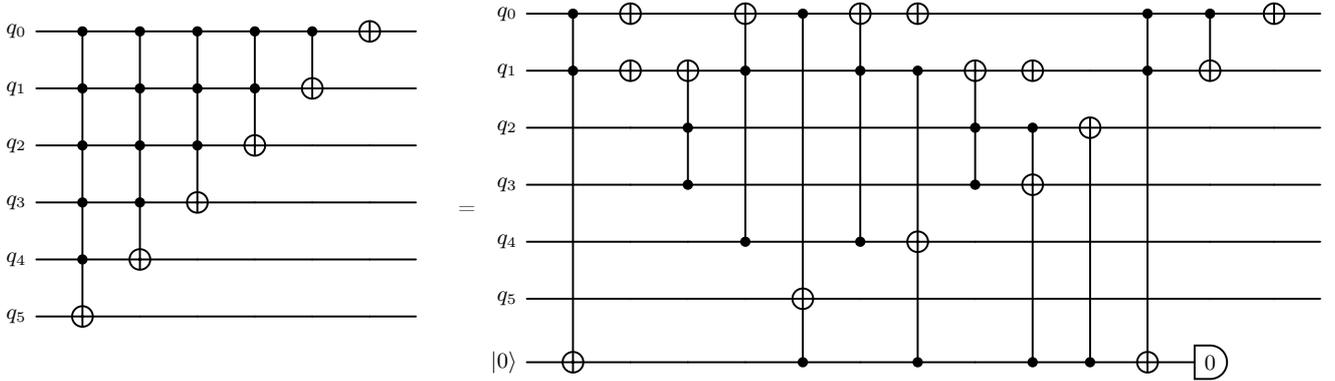

\noindent For large incrementers a recursive approach is used which scales at a $3n$ Toffoli cost (omitting constant factors) -- see \cite{khattar2024riseconditionallycleanancillae} for more details -- and $12n$ $T$ cost, where a maximum of 5 conditionally clean qubits are needed in practice.

\subsection{Resource Estimates}
\label{sec:total_gate_complexity}

The total Toffoli gate complexity can be determined by summing the costs of the incrementers and controlled Givens rotations, multiplied by $d$ steps of the algorithm. We assume the cost of incrementation remains constant at each step $d$, while the cost of the controlled Givens rotations scales with $d$. Specifically, the cost of the controlled Givens rotation is $O(d^2)$ Toffolis, and the cost of incrementation is $O(d)$ Toffolis. Consequently, the asymptotic complexity of the Paldus transform is $O(d^3)$ Toffolis. In this work, we present four distinct compilation strategies and analyze their impact on the constant factors of the Toffoli gate complexity.

Using the conditionally clean compilation shown in Figure~\ref{fig:incrementer_circuit}, a controlled incrementer on $(n-1)$ elements has the same cost as an incrementer on $n$ elements (minus a single $NOT$ gate). We assume the cost of the controlled $(n-1)$ qubit incrementer is the same as the $n$ qubit incrementer (see Section~\ref{sec:incrementer} for the $n$ incrementer circuit). The number of $N$ elements to be incremented is $2d +1$, as the $d$ spatial orbitals can be doubly-occupied, and there is also a $0$ particle state. These can therefore be contained in $\log_2(2d +1)$ qubits. As there are two controlled increments acting on this set, there are $6(\log_2(2d +1) +1)$ Toffoli gates. The number of conditionally clean qubits is then $\log_2(\log_2(2d +1)+1)$. The number of $S$ elements to be incremented is the same as the $N$ set. These can therefore be contained in $\log_2(2d +1)$. As there are also two controlled incrementers acting on this set, the number of conditionally clean qubits is $\log_2(\log_2(2d +1)+1)$. The number of $M$ elements in this set is double the number of $S$ elements minus one. This gives $2(2d +1) - 1 = 4d + 1$ elements. This can be contained in $\log_2(4d + 1)$ qubits. This also contains two controlled incrementers; therefore, the number of Toffoli gates is $6(\log_2(4d + 1) +1)$, and the number of conditionally clean qubits is $\log_2(\log_2(4d +1)+1)$. This is summarised in Table~\ref{tab:inc_cost}.

\begin{table}[!htbp]
    \centering
    \begin{ruledtabular}
    \begin{tabular}{l r r r}
    Register & Elements & Toffoli Count & Conditionally Clean Qubits \\
    \colrule
        $N$   & $2d + 1$      & $6(\log_2(2d +1) +1)$     & $\log_2(\log_2(2d +1)+1)$  \\
        $S$   & $2d + 1$      & $6(\log_2(2d +1) +1)$     & $\log_2(\log_2(2d +1)+1)$  \\
        $M$   & $4d + 1$      & $6(\log_2(4d + 1) +1)$    & $\log_2(\log_2(4d +1)+1)$  \\
    \end{tabular}
    \end{ruledtabular}
        \caption{Cost of incrementation for $d$ spatial orbitals using conditionally clean qubits \cite{khattar2024riseconditionallycleanancillae}.}
    \label{tab:inc_cost}
    \end{table}
  
\noindent Therefore, the total cost of the incrementers in each step of the Paldus transform using a maximum of $\log_2(\log_2(4d +1)+1)$ reusable conditionally clean qubits is $C_{inc} = 12(\log_2(2d +1) +1) + 6(\log_2(4d + 1) +1)$ Toffoli gates. 

The cost of the controlled Givens rotation is more complicated to analyse, as there are variable numbers of clean and dirty qubit registers that can be integrated into the workflow to manipulate the Toffoli cost. Some possible compilation strategies are shown in Table~\ref{tab:full_gate_counts} in the main text. The gate counts can therefore be calculated by combining the cost of the controlled Givens rotations with the cost of the incrementers and summing over all intermediate $d$ values in the iteration.

\begin{table}[!htbp]
    \centering
    \begin{ruledtabular}
    \begin{tabular}{l l}
    Compilation & Toffoli Count  \\
    \colrule
        Unary Iteration~\cite{Babbush2018LinearT}                                   & $\sum^{\max(d)}_d(2I_d + 3q + C_{inc})$                                              \\
        Clean $\selswap$~\cite{Low2024tradingtgatesdirty,Berry2019qubitizationof}   & $\sum^{\max(d)}_d(2\lceil I_d/k\rceil + q(k - 1)+ k + 3q + C_{inc})$                 \\        
        Dirty $\selswap$~\cite{Berry2019qubitizationof}                             & $\sum^{\max(d)}_d(2\lceil I_d/k\rceil+ 4q(k - 1) + 4k + 3q + C_{inc})$               \\
        $^{**}$Multi Index $\data$ (Section~\ref{sec:multi_index_data_select})                                                 & \makecell{$\sum^{\max(d)}_d (2(2\lceil \log_2 L_d \rceil + 2\lceil L_d/k \rceil$ \\ $+ 4q(k-1) + 4(d+2))+ 3q + C_{inc})$}            \\
    \end{tabular}
    \end{ruledtabular}
        \caption{\textit{Costs for the Fault-Tolerant Quantum Paldus Transform.} Here, $I = 8d^2 + 6d + 1$, which loops over all $SM$ combinations, and $L = \frac{1}{2}(d^2 + 3d + 2)$, which includes only the allowed $(S, M)$ combinations. $C_{inc} = 12(\log_2(2d +1) +1) + 6(\log_2(4d + 1) +1)$ Toffoli gates, which remains constant for each $d$ iteration. The algorithm applies an incrementation step and a controlled Givens rotation step, scaling with the intermediate $d$ values. The cost is calculated by summing over all $d$ values up to a maximum $\max(d)$. The $^{**}$Multi Index $\data$ Look Up is used from Table~\ref{tab:double_index_cost}, derived in Appendix~\ref{sec:multi_index_data_select}, where the smaller index corresponds to the $S$ elements, scaling as $2d +1$, and has been substituted into $I$ of Table~\ref{tab:double_index_cost}. $T$ counts can be approximated as $4\times$ the Toffoli gates. These resource estimates include measurement-based uncomputation. $q$ is the number of qubits required to implement the rotation up to some precision $\log_2(1/\epsilon)$. $k$ is the number of qubit registers used in the data lookup (either dirty or clean).}
    \label{tab:full_gate_counts}
    \end{table}
The number of clean and dirty qubits can be determined from the maximum $d$ value by referring to Table~\ref{tab:resouces_cost_gives}. Due to the leading constant factor of the Multi Index $\data$ Lookup, which scales with $L$ and is approximately 16 times smaller than the other methods that scale with $I$, this approach is expected to be the most efficient. Furthermore, it offers the flexibility to utilise dirty qubits, which are expected to be abundant due to the cascading nature of the algorithm, resulting in many idle dirty qubits.

\bibliography{references.bib}

\end{document}